\documentclass[12pt]{report}
\usepackage{graphicx} 
\usepackage{amsmath,amssymb,amsthm,amsfonts,mathtools}
\usepackage{textcomp}
\usepackage{mdframed}
\usepackage[a4paper]{geometry}
\usepackage{graphicx}
\usepackage[dvipsnames]{xcolor}
\usepackage{lmodern}
\usepackage{appendix}
    \usepackage{cancel} 
    \usepackage{tikz}
    \usetikzlibrary{matrix,arrows.meta}
    \usepackage{mathtools}
    \usepackage{booktabs}
    \usepackage[utf8]{inputenc}
    \usepackage{enumerate}
    \usepackage[shortlabels]{enumitem}
    \usepackage{bm}
    \usepackage{quiver}
    \usepackage{stmaryrd}
        \usetikzlibrary{backgrounds}
\usepackage{enumitem}
\usepackage{float}
\usepackage[english]{babel}
\usepackage{csquotes}
\usepackage{mathtools,todonotes} 
    \usepackage{tikz-cd}  
    \usepackage[all,cmtip]{xy} 
    \usepackage[bbgreekl]{mathbbol} 
    \usepackage{array} 

    \usepackage{hyperref} 

    \usepackage{authblk} 

\makeatletter
\def\blx@noerroretextools{}
\makeatother
\usepackage[
  bibstyle=alphabetic,
  citestyle=alphabetic,
  useprefix,
  giveninits=true,
  minbibnames=99,
  maxbibnames=99,
  sorting=nyt,
  sortcites=true,
  backend=biber,
  refsection=chapter
]{biblatex}
\renewbibmacro{in:}{}
\DeclareSourcemap{
  \maps[datatype=bibtex]{
    \map[overwrite]{
      \step[fieldsource=doi, final]
      \step[fieldset=url, null]
      \step[fieldset=eprint, null]
    }
    \map[overwrite]{
      \step[fieldsource=eprint, final]
      \step[fieldset=pages, null]
      \step[fieldset=eid, null]
      \step[fieldset=journal, null]
    }
  }
}
\bibliography{bibliography}

\usepackage{authblk}
\usepackage{titlesec}
\usepackage{abstract}

\usepackage[dvipsnames]{xcolor}
\usepackage{mathtools,todonotes}
\usepackage{autonum}

\usepackage{tikz}
\usetikzlibrary{backgrounds}
\usepackage{tikz-cd}
\usepackage[all,cmtip]{xy}
\usepackage[bbgreekl]{mathbbol}
\usepackage{mathrsfs}
\usepackage{bm}
\newcommand{\vb}[1]{\bm{#1}}

\usepackage{aliascnt}

\theoremstyle{plain}
\newtheorem{theorem}{Theorem}[section]

\newaliascnt{proposition}{theorem}
\newtheorem{proposition}[proposition]{Proposition}
\aliascntresetthe{proposition}

\newaliascnt{lemma}{theorem}
\newtheorem{lemma}[lemma]{Lemma}
\aliascntresetthe{lemma}

\newaliascnt{corollary}{theorem}
\newtheorem{corollary}[corollary]{Corollary}
\aliascntresetthe{corollary}

\newaliascnt{claim}{theorem}

\aliascntresetthe{claim}

\theoremstyle{definition}
\newaliascnt{definition}{theorem}
\newtheorem{definition}[definition]{Definition}
\aliascntresetthe{definition}

\theoremstyle{remark}
\newaliascnt{example}{theorem}
\newtheorem{example}[example]{Example}
\aliascntresetthe{example}

\newaliascnt{remark}{theorem}
\newtheorem{remark}[remark]{Remark}
\aliascntresetthe{remark}

\newaliascnt{axiom}{theorem}

\aliascntresetthe{axiom}
\newtheorem*{axiom*}{Axiom}

\usepackage[nameinlink,noabbrev,capitalize]{cleveref}
\Crefname{theorem}{Theorem}{Theorems}
\Crefname{proposition}{Proposition}{Propositions}
\Crefname{lemma}{Lemma}{Lemmas}
\Crefname{corollary}{Corollary}{Corollaries}
\Crefname{claim}{Claim}{Claims}
\Crefname{definition}{Definition}{Definitions}
\Crefname{example}{Example}{Examples}
\Crefname{remark}{Remark}{Remarks}
\Crefname{axiom}{Axiom}{Axioms}

\newcommand{\qsp}[2]{,\ensuremath{\raise.5ex\hbox{$#1$}\big\slash\raise-.5ex\hbox{$#2$}}}

\makeatletter
\newcommand{\zzlabel}[1]{\ifmeasuring@\else\ltx@label{#1}\fi}
\makeatother

\newcounter{terms}[equation]

\RequirePackage{tikz-cd}
\RequirePackage{amssymb}
\usetikzlibrary{calc}
\usetikzlibrary{decorations.pathmorphing}

\tikzset{curve/.style={settings={#1},to path={(\tikztostart)
    .. controls ($(\tikztostart)!\pv{pos}!(\tikztotarget)!\pv{height}!270:(\tikztotarget)$)
    and ($(\tikztostart)!1-\pv{pos}!(\tikztotarget)!\pv{height}!270:(\tikztotarget)$)
    .. (\tikztotarget)\tikztonodes}},
    settings/.code={\tikzset{quiver/.cd,#1}
        \def\pv##1{\pgfkeysvalueof{/tikz/quiver/##1}}},
    quiver/.cd,pos/.initial=0.35,height/.initial=0}

\tikzset{tail reversed/.code={\pgfsetarrowsstart{tikzcd to}}}
\tikzset{2tail/.code={\pgfsetarrowsstart{Implies[reversed]}}}
\tikzset{2tail reversed/.code={\pgfsetarrowsstart{Implies}}}
\tikzset{no body/.style={/tikz/dash pattern=on 0 off 1mm}}
    
\title{
   General Relativity for Mathematicians -- Part I
}

\author{
   Filippo Fila--Robattino and Manuel Tecchiolli\\
}

\date{Universität Zürich -- HS 2025}

\begin{document}

\maketitle
\pagenumbering{roman}

\begin{abstract}
These notes cover the first part of a course of the same title taught by Filippo Fila--Robattino and Manuel Tecchiolli at the University of Zurich in the Winter Semester of 2025. They are also an experiment: much of their content was developed from the lecturers’ own knowledge with the assistance of OpenAI’s ChatGPT models.
\end{abstract}

\tableofcontents

\chapter*{Introduction}
\pagenumbering{arabic}
\addcontentsline{toc}{chapter}{Introduction}

A stone falls, a clock ticks, and a beam of light travels from one observer to another. These are familiar occurrences, and yet a precise account of them requires us to decide what we mean by motion, duration, and observation. We assign positions to bodies, compare the readings of clocks, and write equations relating these quantities. In doing so, we already make assumptions about the structure of the world: which events can be regarded as simultaneous, how distances are measured, and which motions require an explanation. Much of the history of relativity consists in making such assumptions explicit, examining their physical meaning, and discovering that some of them can be replaced.

The purpose of these notes is to follow this process up to the geometric and variational formulation of general relativity. Our main interest is in understanding how physical questions lead us to mathematical structures, and how those structures, once introduced, change the questions we are able to ask. The resulting theory brings together ideas that initially seem to belong to different subjects: the motion of a free particle, the propagation of light, the geometry of smooth manifolds, and the symmetries of an action. We will develop the connections between them step by step.

\medskip
\noindent\textbf{Motion, descriptions, and symmetry.}
The position of a body is always given relative to some description. Two observers may assign different coordinates to the same event, different velocities to the same particle, and even different temporal orderings to certain pairs of events. To compare their accounts, we need to know how their measurements are related and which relations among them have physical significance. A principle of relativity specifies an equivalence among suitable physical descriptions and constrains the laws that can be formulated within them.

This point of view gives symmetry a central role from the beginning. In Newtonian mechanics, the Galilean transformations relate inertial descriptions and preserve the structure of free motion. Electromagnetism brings a different structure into the discussion: the Maxwell equations single out a characteristic propagation speed, and their transformation properties lead us to reconsider the relation between space and time. Special relativity organizes these ideas in Minkowski space--time, where durations and spatial separations enter a common geometric framework. The metric encodes the causal cones and the proper time measured along timelike worldlines; the Poincar\'e transformations preserve this structure.

Notice that a change of coordinates and a change of physical observer are distinct notions. Coordinates label events, while an observer is itself a physical system whose motion and measurements must be described. Their relationship will matter throughout these notes. A coordinate expression can change while the geometric object it represents remains the same; conversely, physically different motions can be compared within a single coordinate system. Keeping track of these distinctions helps us formulate precisely what a symmetry preserves.

\medskip
\noindent\textbf{From a prescribed geometry to a dynamical one.}
In special relativity, the Minkowski metric is given before we solve the equations of motion. It provides a fixed causal and metrical structure within which particles and fields evolve. General relativity changes the status of this structure: the space--time metric becomes one of the fields whose behavior the theory must determine. The geometry that enters the description of motion also participates in the dynamics.

The universality of free fall provides a physical motivation for this step. Freely falling test bodies, released with the same initial position and velocity and subject to no other interactions, follow the same trajectory independently of their mass. This suggests that their common motion can be described through a structure of space--time itself. An affine connection supplies a notion of parallel transport and geodesic motion; a Lorentzian metric supplies causal cones and proper time. In the metric formulation of general relativity, these structures are related by the Levi--Civita connection.

A freely falling description can make the metric Minkowskian and its first derivatives vanish at an event. Curvature, however, can still be present. It governs tidal effects, including the relative acceleration of nearby freely falling particles, and records information that a change of coordinates cannot eliminate. Thus the passage from special to general relativity requires us to understand both the local simplifications available in a suitable description and the geometric information that survives them.

\medskip
\noindent\textbf{Why differential geometry?}
For a mathematician, the statement that gravity is described by geometry is an invitation to identify exactly which geometric structures are involved. A smooth manifold allows us to label events locally and compare smooth descriptions on overlapping domains. Its tangent spaces provide the setting for velocities and directional derivatives. Tensor fields, differential forms, and connections then allow us to formulate relations among these objects intrinsically. Each additional structure supplies something that the preceding ones did not determine.

We will take time to distinguish these levels. A smooth manifold alone has no preferred metric, no prescribed notion of parallel transport, and no distinguished family of freely falling trajectories.

The geometric development is therefore part of the argument of these notes. We will move between intrinsic definitions and coordinate expressions, using the former to identify the objects and the latter to compute with them. Definitions will be accompanied by examples, remarks, and proofs so that the formalism remains connected to the questions that motivated it. Familiarity with linear algebra and multivariable calculus will be assumed; the differential geometry needed for the construction will be developed along the way.

\medskip
\noindent\textbf{The variational point of view.}
A second thread running through the notes is the principle of stationary action. Already in mechanics, it gives us a way to organize equations of motion and study their symmetries. We assign an action to possible trajectories and ask for its first variation to vanish under admissible deformations. As we pass to field theory, trajectories are replaced by fields, and the same question becomes a variational problem on a space of sections.

In general relativity, we apply this idea to the metric itself. The Einstein--Hilbert action uses the scalar curvature and the metric volume form to turn geometric data into a functional. Its variation, together with the variation of a matter action, yields the Einstein equations. This construction also exposes the relation between diffeomorphism symmetry, geometric identities, and the constraints on the dynamics. Understanding the equations includes understanding why their components cannot all be treated as independent evolution equations.

The path through the notes reflects these two threads. We begin with mechanics, electromagnetism, and special relativity, emphasizing the role of symmetry and the assumptions entering the description of motion. We then develop the differential geometry required to formulate these ideas on a manifold. Finally, we bring geometry and variational calculus together to study geodesics, the gravitational action, and the resulting field equations.

Our aim is to reach a point where the mathematical formulation of general relativity can be read with a clear understanding of the role of each of its ingredients. The connection appearing in a geodesic equation, the metric measuring proper time, and the curvature entering the gravitational action belong to a single construction. Following that construction reveals the depth of Einstein's proposal: the structure by which we compare events and describe motion is itself subject to physical laws.
\chapter{Principles of Relativity}
\section{The problem of motion}
The problem of motion is one of the oldest and most persistent questions in the history of science. It asks in the most general way: how do bodies change their positions in space with the passage of time, and why? The attempt to answer this question forced philosophers and scientists to define what they meant by space, time, cause, and matter. It also forced them to develop mathematical tools to represent trajectories and to formulate laws that could be applied to concrete situations. The history of physics can therefore be read as the progressive refinement of our understanding of motion.

In the Aristotelian framework, motion was conceived in a world that was finite, ordered, and qualitative. The cosmos, in Aristotle’s view, was composed of concentric spheres, with the Earth at the center. Above the Earth extended the celestial spheres, incorruptible and eternal, made of a special element called aether.\footnote{Not to get confused with the “luminiferous one” that we will mention in the following sections.} Below the sphere of the Moon lay the sublunar world, composed of the four elements: earth, water, air, and fire. Each element had a natural place: earth tended to fall toward the center of the cosmos, water above earth, air above water, and fire above all, rising upward. Motion, therefore, was classified into two kinds. Natural motion was spontaneous: it was the tendency of each element to seek its proper place. Violent motion, by contrast, was motion contrary to this natural tendency, and it required continuous external action, a mover, to sustain it. In this way, Aristotle explained why stones fall downward, why flames rise upward, and why a cart must be pulled by a horse to continue moving.

This picture was embedded in a teleological philosophy. Everything in nature had a purpose (a \emph{telos}). Rest was not considered a special case of motion but the completion of motion, the realization of a thing’s essence. A stone at the center of the Earth was at rest not because of a balance of forces but because it had reached its final cause. Motion was not described quantitatively but qualitatively. Time, for Aristotle, was the number of motion according to before and after. It was dependent on change, not an independent entity. Space, similarly, was not conceived as a homogeneous container but as an order of places, each defined by the body that occupied it. In this picture, there was no notion of inertia. A body did not continue moving unless something continued to move it. The medium, such as air or water, played a role in sustaining motion as well as resisting it. Kinematics and dynamics were not separated; to describe motion was to speak of its cause. Celestial and terrestrial motions were fundamentally different: celestial motion was eternal and circular, terrestrial motion was corruptible and linear. Thus the Aristotelian problem of motion was solved by referring to essences, causes, and purposes, not by equations or universal laws.

The transition to modern physics required a radical change. Galileo Galilei represents the decisive break. He introduced idealization as a method: one strips away the accidental features of the world to isolate its fundamental laws. In studying bodies on inclined planes, he reduced friction and resistance; in thought experiments, he imagined removing the medium entirely. He insisted on mathematization: motion must be described by quantitative relations among measurable magnitudes, expressed geometrically and later algebraically. He separated kinematics, the description of motion, from dynamics, the study of its causes. He discovered quantitative regularities: in free fall from rest, distance grows as the square of time, velocity increases linearly with time, and the composition of uniform horizontal motion with accelerated vertical motion yields a parabolic trajectory. Galileo also articulated the relativity of motion. In his ship thought experiment, he argued that no internal mechanical experiment can distinguish a ship at rest from one moving with constant velocity in a straight line. This showed that uniform rectilinear motion is, in a sense, indistinguishable from rest. The medium no longer sustained motion; it only resisted it. Motion, once acquired, persisted without the need for a mover.

The culmination of Galileo’s insights is the principle of inertia. In modern form, we state:

\begin{axiom*}[Galilean Law of Inertia]
In an inertial reference frame, a body not subjected to any external force preserves its state of motion.
\end{axiom*}

At this point, one may naturally ask: what exactly is meant by an inertial reference frame? How can we recognize that a system of coordinates is truly “inertial”? These questions are unavoidable, because the very formulation of the law of inertia presupposes the existence of such privileged frames.

From a pedagogical perspective, one can say that an inertial reference frame is a coordinate system in which the Galilean law of inertia holds: if a body is isolated from all interactions, it moves uniformly in a straight line. However, this definition is somewhat circular: we define inertial frames through the validity of the law of inertia, and at the same time we state the law of inertia only in inertial frames.

Isaac Newton brought the study of motion into a unified mathematical framework through his three laws of motion and the law of universal gravitation, describing terrestrial and celestial phenomena within the same theory. In this framework, he introduced the notion of an \emph{absolute space}: a universal arena that exists independently of material bodies and within which true inertial motion takes place.

In such an absolute space, the notion of an inertial system acquires an absolute meaning: it is simply a system of coordinates at rest or in uniform translation with respect to absolute space. Of course, absolute space cannot be observed directly; it is a metaphysical construct introduced to avoid logical regress.

Thus, while the law of inertia seems conceptually simple, it inevitably pushes us towards subtle philosophical considerations: either we accept the necessity of postulating an absolute background (as Newton did), or we accept a pragmatic, operational definition of inertial frames by approximating them with real physical systems (for example, using the fixed stars as reference).

This axiom reframes the problem of motion. Dynamics becomes the study of deviations from inertial motion. Forces are precisely what cause such deviations. To mathematically formulate this in a systematic way, we introduce the principle of stationary action, which provides the general foundation for deriving equations of motion.

The idea is the following: instead of writing down the equations of motion directly, we assign to each possible path of a system a single number, called the \emph{action}. The action is a functional, meaning it is a rule that takes as input a function $q=q(t)$ (a path in configuration space) and produces a number. The physical path is then singled out as the one that makes this number stationary under small deformations.

Let us begin with the definition. For a system described by coordinates $q=(q^1,\dots,q^n)$, velocities $\dot q=(\dot q^1,\dots,\dot q^n)$, and a Lagrangian function $\mathcal L=\mathcal L(t, q,\dot q)$, we define the action of a path $q=q(t)$ on a time interval $[t_1,t_2]$ as
\begin{align}
S[q] = \int_{t_1}^{t_2} \mathcal L(t,q(t),\dot q(t))\,dt.
\end{align}
Notice carefully what this means. If you hand me a trial trajectory $q(t)$, I can compute at each instant the value of $\mathcal L(q,\dot q,t)$, and then integrate over the interval. The result is a number $S[q]$. If you change the path, the number changes. In this sense $S$ is not a function but a \emph{functional}.\footnote{In a more formal sense, a real functional $F$ is a map 
$F \colon \mathrm{Map}(A,B) \to \mathbb{R}$,
where $\mathrm{Map}(A,B)$ denotes a space of maps between two sets $A$ and $B$, 
which in general may carry additional geometric structures (e.g. topological or differentiable structures).}

The \emph{principle of stationary action} states: among all smooth paths $q=q(t)$ that connect the two fixed endpoints $q(t_1)=q_1$ and $q(t_2)=q_2$, the physical path is the one for which $S[q]$ does not change to first order when “we deform the path slightly”. In other words, if $q$ is the physical trajectory, and $q+\varepsilon\eta$ is a nearby path with $\eta=\eta(t)$ such that $\eta(t_1)=\eta(t_2)=0$ (the endpoints are fixed), then, defining
\begin{align}
\delta S\coloneqq\left.\frac{d}{d\varepsilon} S[q+\varepsilon\eta]\right|_{\varepsilon=0},
\end{align}
we have
\begin{align}
    \boxed{\delta S=0.}
\end{align}
This condition is what we mean by the action being \emph{stationary}. It could be a minimum, a maximum, or a saddle point, but the crucial fact is that the first-order variation vanishes.\footnote{Notice that in the literature (especially in physics) one can find the notation $\delta x(t)$ for $\varepsilon\eta(t)$. Technically, it makes sense. In fact, in field theory, the symbol $\delta$ is usually interpreted as the de Rham differential operator on the space of fields.}

\begin{remark}[Summation convention]
Throughout these notes we adopt the Einstein summation convention: whenever an index variable appears twice in a single term, once as an upper index and once as a lower index, it is implicitly summed over all its possible values. In the present context of classical mechanics, since we have not defined yet what the position of the indices means, we shall simply agree that repeated indices are summed. For example,
\begin{align}
a_i b_i = \sum_{i} a_i b_i.
\end{align}
This convention simplifies notation by avoiding the need to write explicit summation signs.
\end{remark}

Let us now compute this variation explicitly. By definition,
\begin{align}
\delta S &= \left.\frac{d}{d\varepsilon} \int_{t_1}^{t_2} \mathcal L(t,q+\varepsilon\eta,\dot q+\varepsilon \dot\eta)\,dt\right|_{\varepsilon=0}.
\end{align}
Differentiating under the integral sign gives
\begin{align}
\delta S = \int_{t_1}^{t_2}\left( \frac{\partial \mathcal L}{\partial q^i}\eta^i + \frac{\partial \mathcal L}{\partial \dot q^i}\dot \eta^i \right) dt.
\end{align}
The first term is already in the desired form, but the second term involves $\dot\eta^i$, the derivative of the variation. To handle this, we integrate by parts:
\begin{align}
\int_{t_1}^{t_2} \frac{\partial \mathcal L}{\partial \dot q^i}\dot \eta^i\,dt
= \Big[\frac{\partial \mathcal L}{\partial \dot q^i}\eta^i\Big]_{t_1}^{t_2}
- \int_{t_1}^{t_2} \frac{d}{dt}\left(\frac{\partial \mathcal L}{\partial \dot q^i}\right)\eta^i\,dt.
\end{align}
Because the variations vanish at the endpoints ($\eta(t_1)=\eta(t_2)=0$), the boundary term is zero. Thus, the variation of the action becomes
\begin{align}
\delta S = \int_{t_1}^{t_2} \left(\frac{\partial \mathcal L}{\partial q^i} - \frac{d}{dt}\frac{\partial \mathcal L}{\partial \dot q^i}\right)\eta^i\,dt.
\end{align}

Now comes a crucial mathematical fact, called the \emph{fundamental lemma of the calculus of variations}: if a function $f(t)$ satisfies
\begin{align}
\int_{t_1}^{t_2} f(t)\eta(t)\,dt = 0
\end{align}
for every smooth $\eta$ that vanishes at the endpoints, then $f=0$ on the interval.

Applying this result, we conclude that the stationarity condition $\delta S=0$ for all variations implies
\begin{align}
\frac{d}{dt}\frac{\partial \mathcal L}{\partial \dot q^i} - \frac{\partial \mathcal L}{\partial q^i} = 0, \qquad i=1,\dots,n.
\end{align}
These are called the Euler--Lagrange equations, the equations of motion (e.o.m.) of our theories.
\begin{remark}
Why should nature select stationary-action paths? Up to this point, the principle of stationary action may appear as a mathematical postulate: we write down an action, require it to be stationary, and obtain the equations of motion. This is powerful, but one may ask: why should nature care about stationarity of this functional? The answer comes from quantum mechanics, and in particular from the path integral formulation introduced by Richard Feynman.\\
In quantum mechanics, the fundamental object is not the trajectory of a particle but its wavefunction, or more generally the amplitude for the particle to go from one configuration to another. If a particle starts at position $q_a$ at time $t_a$ and is later found at position $q_b$ at time $t_b$, quantum mechanics assigns a complex number $\mathcal{A}(q_b,t_b;q_a,t_a)$, called the transition amplitude, whose modulus squared gives the probability for this process.\\
Feynman proposed a striking way to compute this amplitude: instead of selecting a single trajectory, we must sum over \emph{all possible trajectories} connecting $(q_a,t_a)$ to $(q_b,t_b)$. Each trajectory $q=q(t)$ contributes with a complex phase
\begin{align}
\exp\!\left(\tfrac{i}{\hbar} S[q]\right),
\end{align}
where $S[q]$ is exactly the classical action of that path. The transition amplitude is then formally given by the path integral
\begin{align}
\mathcal{A}(q_b,t_b;q_a,t_a) = \int \mathcal{D}q \, \exp\!\left(\tfrac{i}{\hbar} S[q]\right).
\end{align}
At first sight this looks mysterious: what does it mean to integrate over all possible paths? To approach this expression, we discretize time into small intervals. Suppose we divide $[t_a,t_b]$ into $N$ slices of width $\Delta t$. Instead of specifying a continuous trajectory, we specify the positions of the particle at the intermediate times, $(q_1,q_2,\dots,q_{N-1})$. For each short step, quantum mechanics provides a short-time propagator, which depends on the Lagrangian. Multiplying all these propagators and integrating over all intermediate positions gives an approximation to the amplitude. Formally, the limit $N\to\infty$ yields the path integral. The existence and meaning of this limit require additional hypotheses and a suitable prescription; discretization alone does not establish them.\footnote{In particular, this construction does not automatically define a measure on the space of paths.}\\
The crucial point is what happens to the path integral in the classical limit $\hbar\to0$.\footnote{The reduced Planck constant $\hbar=h/(2\pi)$ is the fundamental constant setting the scale of quantum effects. For example, the Heisenberg uncertainty relation $\Delta x\,\Delta p\geq\hbar/2$ limits how sharply position and momentum can be specified simultaneously. Formally taking $\hbar\to0$ removes this quantum restriction, allowing the classical description in terms of definite positions and momenta.} Formally, since each path contributes a phase $e^{iS[q]/\hbar}$, if $S[q]$ changes even a little when we deform the path, then the phase oscillates extremely rapidly across nearby paths, as $\hbar$ approaches zero. The stationary-phase argument suggests that contributions from nonstationary regions cancel by destructive interference, while neighborhoods of stationary paths determine the leading semiclassical behavior.\footnote{This is the formal path-integral analogue of the \emph{method of stationary phase}; a rigorous application requires control of the integral and its asymptotics.}\\
In precise terms, “stationary" means that the variation of the action vanishes, i.e.
\begin{align}
\delta S= 0.
\end{align}
This condition is exactly the Euler--Lagrange equations, so the stationary paths are the classical trajectories.\\
This provides a semiclassical explanation of the principle of stationary action: under suitable assumptions, classical trajectories emerge from the stationary-phase behavior of the quantum amplitude.
\end{remark}
Now, consider a particle of mass $m$ moving in $\mathbb{R}^3$ with potential $V=V(\mathbf x,t)$. The Lagrangian is
\begin{align}
\mathcal L(\mathbf x,\dot{\mathbf x},t) = \tfrac{1}{2}m \dot{\mathbf x}^2 - V(\mathbf x,t).
\end{align}
For each component $x_k$,
\begin{align}
\frac{\partial \mathcal L}{\partial \dot x_k} = m \dot x_k, \qquad
\frac{d}{dt}\frac{\partial \mathcal L}{\partial \dot x_k} = m \ddot x_k, \qquad
\frac{\partial \mathcal L}{\partial x_k} = -\frac{\partial V}{\partial x_k}.
\end{align}
The Euler--Lagrange equation gives
\begin{align}
m \ddot x_k + \frac{\partial V}{\partial x_k} = 0,
\end{align}
or in vector form,
\begin{align}
m \ddot{\mathbf x} = -\nabla V(\mathbf x,t) \eqqcolon \mathbf F(\mathbf x,t).
\end{align}
This is Newton’s second law $\mathbf F=m\mathbf a$.

We now introduce what is arguably the most important concept in physics: \emph{symmetry}. Very informally, by a symmetry we mean a specified set of transformations acting on the mathematical objects of the theory (coordinates, fields, parameters), and we say:\\

\emph{A theory is said to be symmetric under a given set of transformations if, when we apply those transformations to the variables that enter the action, the equations of motion that follow from the variational principle do not change.}\\

This definition emphasizes the strength of the variational formalism: it isolates dynamics in a single scalar functional $S$, and invariances of $S$ (up to boundary terms/total derivatives\footnote{In general, if two action functionals differ by a boundary term, or, equivalently, if two Lagrangian functions differ by a total derivative, then they lead to the same equations of motion, since boundary terms do not contribute to the variation of $S$, being the endpoints fixed.}) guarantee invariance of the Euler--Lagrange equations.

Let us make this precise for the simplest nontrivial example: a free particle of mass $m$ moving in three-dimensional Euclidean space. The action is given by
\begin{align}
S[\mathbf{x}] = \int_{t_1}^{t_2} \tfrac{1}{2}m\,|\dot{\mathbf{x}}(t)|^2\,dt,
\end{align}
and the Euler--Lagrange equations give $m\,\ddot{\mathbf{x}}=\mathbf{0}$. We now ask: what are the symmetries of this theory?
\begin{enumerate}
    \item \textbf{Rotations:} A rotation of the spatial coordinates changes $\mathbf{x}(t)$ to a rotated vector, but it does not change its length. In particular, the squared velocity $|\dot{\mathbf{x}}(t)|^2$ remains the same. Therefore, the Lagrangian and the action are invariant, and the equations of motion are unchanged. This expresses the fact that the theory sees all directions in space as equivalent.

    \item \textbf{Translations:} A shift $\mathbf{x}(t)\mapsto\mathbf{x}(t)+\mathbf a$, where $\mathbf a$ is a constant vector, leaves the velocity invariant and hence the Lagrangian and the action. Also, a time translation $t\mapsto t+b$ gives the same solution set for the equations of motion. In fact. a time-translation can be regarded as a linear change in the parametrization of the trajectories. For simplicity of notation, consider the $1$-dimensional case. Notice that, in general, for $t' = at + b$, we have
    \begin{align}
    q'(t') := q\!\left(\tfrac{t'-b}{a}\right).
    \end{align}
    Differentiating, we obtain the relation between velocities
    \begin{align}
    \dot q(t) = \frac{dq}{dt} 
    = \frac{dq'}{dt'}\frac{dt'}{dt} 
    = a\,\dot q'(t').
    \end{align}
    Substituting in the action, we get
    \begin{align}
    S[q] &= \int_{t_1}^{t_2} \mathcal L(t, q(t),\dot q(t))\,dt \\
    &= \int_{t_1'}^{t_2'} \mathcal L\!\left(\tfrac{t'-b}{a},q'(t'),\,a\,\dot q'(t')\right)\,\frac{dt'}{a}.
    \end{align}
    This implies that, if we set $a=1$, and if the Lagrangian does not depend explicitly on time\footnote{If the Lagrangian does not depend explicitly on time, a change $t'-b\mapsto t'$ is immaterial.} (like the case at hand), we get
    \begin{align}
    S[q] &= \int_{t_1'}^{t_2'} \mathcal L\!\left(t', q'(t'),\,\dot q'(t')\right)\,dt'.
    \end{align}
    The form of the action functional is equivalent, hence leading to the same equations of motion for the trajectory $q'$.\\
    These invariances express the homogeneity of space and time. 

    \item \textbf{Galilean boosts:} For a constant velocity $\mathbf{v}$, a \emph{Galilean boost} acts by
    \begin{align}
    \mathbf{x}(t)\mapsto \mathbf{x}'(t) = \mathbf{x}(t) + \mathbf{v}\,t, 
    \qquad
    \dot{\mathbf{x}}(t)\mapsto \dot{\mathbf{x}}'(t) = \dot{\mathbf{x}}(t)+\mathbf{v}.
    \end{align}
    The Lagrangian transforms as
    \begin{align}
    \mathcal L' = \tfrac{1}{2}m\,|\dot{\mathbf{x}}+\mathbf{v}|^{2} 
    = \tfrac{1}{2}m\,|\dot{\mathbf{x}}|^{2} + m\,\mathbf{v}\cdot\dot{\mathbf{x}} + \tfrac{1}{2}m\,|\mathbf{v}|^{2}
    = \mathcal L + \frac{d}{dt}\big(m\,\mathbf{v}\cdot \mathbf{x} + \tfrac{1}{2}m\,|\mathbf{v}|^{2}t\big).
    \end{align}
    Therefore
    \begin{align}
    S' = \int_{t_1}^{t_2} \mathcal L'\,dt 
    = \int_{t_1}^{t_2} \mathcal L\,dt + \Big[m\,\mathbf{v}\cdot \mathbf{x} + \tfrac{1}{2}m\,|\mathbf{v}|^{2}t\Big]_{t_1}^{t_2},
    \end{align}
    so the action changes only by a boundary term, vanishing when considering the variations with fixed endpoints. Hence, the free particle is invariant under Galilean boosts: moving from one inertial frame to another, with constant relative velocity, does not change the form of the equations of motion, according to what Galileo observed.
\end{enumerate}
Collecting these observations, the symmetries of the free-particle action in three-dimensional Euclidean space are: spatial rotations (isotropy, $3$ degrees of freedom), spatial translations and time translations (homogeneity of space and time, $4$ degrees of freedom), and Galilean boosts (relativity of uniform motion, $3$ degrees of freedom). Together they generate the \emph{Galilean transformations} ($10$ degrees of freedom), the kinematical symmetry of Newtonian mechanics. In the variational language, each of these transformations leaves the action invariant up to a total derivative, and therefore leaves the Euler--Lagrange equations unchanged.
\begin{remark}[Symmetries and external forces]\label{remsymmforc}
The introduction of a fixed external potential generally reduces the symmetries of the
free-particle action. For example, a fixed central potential $V=V(|\mathbf{x}|)$ preserves rotational invariance
but breaks invariance under spatial translations and Galilean boosts. A uniform gravitational potential
$V(\mathbf{x}) = m g x_3$ breaks rotational invariance while preserving translations in the horizontal
directions. In each case, the equations of motion still follow from the variational principle.

It is important to distinguish the symmetries of a theory from those of a fixed external background or a particular configuration. Interactions do not necessarily destroy Galilean invariance: for example, a system of particles with interaction potential
\begin{align}
V=\sum_{i<j}V_{ij}(|\mathbf{x}_i-\mathbf{x}_j|)
\end{align}
is invariant under simultaneous Galilean transformations of all particles, up to the boundary terms arising from their kinetic energies. Likewise, a fixed electromagnetic background need not possess all the spacetime symmetries of the full electromagnetic theory.

In general relativity, Einstein’s revolutionary step was to make spacetime geometry itself dynamical and to describe free fall geometrically. In this way, free fall in a gravitational field and uniform motion in the absence of gravity both appear as natural manifestations of the underlying geometry. This does not mean that every gravitational field has the same symmetries as flat spacetime: the covariance of the field equations must be distinguished from the isometries of a particular metric. But before we dive into these concepts, there are still some steps to take in our conceptual journey.
\end{remark}

\section{The Maxwell equations}
By the beginning of the nineteenth century, electricity, magnetism, and optics had developed as largely independent branches of natural philosophy. Each possessed a growing body of experimental results and empirical laws, but their mutual relations were not yet understood.

Electrostatics was grounded on Coulomb’s careful experiments, which had established that the force between two point charges is inversely proportional to the square of their separation. This law provided the electric analogue of Newton’s law of gravitation. Meanwhile, the invention of the voltaic pile by Alessandro Volta (1800) made it possible to produce sustained electric currents, which in turn led to a new range of phenomena. Georg Simon Ohm (1827) discovered the proportionality between current, voltage, and resistance, later codified as Ohm’s law, thereby laying the foundations of circuit theory.

Magnetism, studied for centuries through the behavior of lodestones and compass needles, also acquired new clarity in this period. The decisive event came in 1820, when Hans Christian Ørsted observed that a current-carrying wire deflects a magnetic needle. This showed that electricity could directly influence magnetism. Building on this discovery, André-Marie Ampère conducted systematic experiments and formulated quantitative laws for the forces between currents, thereby founding electrodynamics. These results suggested that currents generate magnetic fields and that magnetism and electricity are not entirely separate.

A second decisive advance was made by Michael Faraday in the 1830s. Faraday discovered electromagnetic induction: a changing magnetic flux through a circuit induces an electromotive force. His experiments revealed a profound reciprocity between electricity and magnetism. Just as currents produce magnetic effects (Ampère’s law), changing magnetic fields can produce electric currents (Faraday’s law). Faraday also emphasized the concept of lines of force, suggesting that electric and magnetic phenomena should be understood in terms of \emph{fields} permeating space rather than merely as actions at a distance.

Meanwhile, the study of light had advanced along a separate trajectory. Thomas Young’s double-slit experiment (1801) and Augustin Fresnel’s subsequent work had established the wave nature of light, accounting for interference and diffraction phenomena. Polarization experiments further showed that light waves are transverse rather than longitudinal. By the mid-nineteenth century, the wave theory of light was widely accepted.

By the 1850s, therefore, a rich but fragmented picture had emerged:
\begin{itemize}
\item Coulomb’s law described the force between charges.
\item Ørsted and Ampère had shown that currents create magnetic effects.
\item Faraday had demonstrated that changing magnetic flux induces currents.
\item Optics had revealed that light is a transverse wave, but its medium of propagation was unknown.
\end{itemize}

What was missing was a unified theory that connected these disparate results. The empirical laws of electricity and magnetism were individually successful, but they lacked a common mathematical framework. Optics was understood as the study of waves, but the physical nature of these waves was obscure. The stage was thus set for a synthesis: a single theory that would explain both electromagnetic phenomena and the propagation of light.

It was James Clerk Maxwell who finally achieved the synthesis that unified electricity, magnetism, and optics into a single theoretical framework. Building on the experimental laws of Coulomb, Ampère, and Faraday, and guided by Faraday’s intuition of fields as genuine physical entities, Maxwell formulated a set of differential equations that describe how the electric and magnetic fields evolve in space and time.

One of Maxwell’s key insights was the recognition that Ampère’s law, as originally formulated, was incomplete. Ampère’s law stated that the circulation of the magnetic field around a closed curve is proportional to the electric current passing through the enclosed surface. However, when applied to situations involving time-varying electric fields, this formulation appeared to violate charge conservation. Maxwell resolved this inconsistency by introducing an additional term, the \emph{displacement current}, proportional to the time derivative of the electric field. This modification ensured that the theory respected the continuity equation for electric charge and made the set of laws self-consistent.

Maxwell’s equations in vacuum, written in modern differential form, are the following:
\begin{align}
\nabla \cdot \mathbf{E} &= \frac{\rho}{\varepsilon_0}, 
&\qquad \nabla \cdot \mathbf{B} &= 0, \\
\nabla \times \mathbf{E} &= - \frac{\partial \mathbf{B}}{\partial t}, 
&\qquad \nabla \times \mathbf{B} &= \mu_0 \mathbf{J} + \mu_0 \varepsilon_0 \frac{\partial \mathbf{E}}{\partial t},
\end{align}
where we define the \emph{nabla} operator $\nabla = \big(\partial_{x_1},\partial_{x_2},\partial_{x_3}\big)$.

The first line contains Gauss’s laws. The divergence of the electric field is proportional to the local charge density, with proportionality constant $\varepsilon_0$, the permittivity of free space. The divergence of the magnetic field vanishes, which means that there are no isolated magnetic charges (monopoles); this does not imply that all magnetic field lines are closed loops, as illustrated by a nonzero uniform magnetic field, whose field lines are straight lines.

The second line contains the dynamical laws. Faraday’s law states that a time-varying magnetic field produces a circulating electric field. The Ampère–Maxwell law states that magnetic fields are produced both by electric currents and by time-varying electric fields, the latter contribution being Maxwell’s displacement current.

Together, these four equations encapsulate all of classical electromagnetism. They unify in a compact and elegant form the empirical discoveries of Coulomb, Ørsted, Ampère, and Faraday. The laws of electrostatics, magnetostatics, and electromagnetic induction appear as limiting cases, while the full equations capture the dynamical interplay of electric and magnetic fields. For the first time, electricity and magnetism were no longer two separate branches of physics, but two aspects of the same field theory.

One of the great conceptual shifts brought by nineteenth-century physics, and in particular by electromagnetism, was the realization that the fundamental laws of nature cannot be expressed solely in terms of forces acting at a distance between isolated particles. Instead, they require the introduction of new mathematical objects: \emph{fields}. A field is something defined at every point of space (and time), capable of encoding how physical quantities vary continuously across regions rather than existing only at discrete locations. 

Electrostatics already pointed in this direction. Coulomb’s law says that charges exert forces on each other, but to describe the influence of a distribution of charges it is far more efficient to introduce the electric field $\mathbf{E}$, which, at each point $x$, specifies the force per unit charge that would act on a test particle. Similarly, magnetostatics showed that currents create forces on other currents, but Faraday and Ampère’s work made it clear that these effects can be described in terms of a magnetic field $\mathbf{B}$. The decisive step came with Faraday’s induction experiments: changing magnetic fluxes induce electric effects even where no charge is present. This demanded a language in which electric and magnetic influences are present throughout space and time, not just as instantaneous actions at a distance. 

Maxwell embraced and mathematized Faraday’s vision by treating fields as primary physical quantities, obeying differential equations that relate their spatial variation to their temporal evolution and to sources (charges and currents). This conceptual leap — elevating fields to the same status as particles — is what led to the mature form of classical electromagnetism. From the mathematical point of view, it also marked the systematic entry of vector calculus into physics. The notions of gradient, divergence, curl, and Laplacian, first developed in mathematical analysis, became indispensable tools for expressing the structure of Maxwell’s equations.

For this reason, before we proceed further, we set down a precise definition of what we mean by a vector field in the Euclidean context, and we recall the operators that act on such fields.

\begin{definition}[Vector field in Euclidean space]
A \emph{smooth vector field} on $\mathbb{R}^3$ is a smooth map
\begin{align}
\mathbf{A}: \mathbb{R}^3 \longrightarrow \mathbb{R}^3,
\end{align}
which assigns to each point $\mathbf{x}=(x_1,x_2,x_3)\in\mathbb{R}^3$ a vector 
\begin{align}
\mathbf{A}(\mathbf{x}) = \big(A_1(\mathbf{x}),A_2(\mathbf{x}),A_3(\mathbf{x})\big),
\end{align}
where each component $A_i:\mathbb{R}^3\to\mathbb{R}$ is assumed to be in $\mathcal C^\infty(\mathbb R^3)$. 

In this Euclidean setting, all vector calculus operators (gradient, divergence, curl, Laplacian) are understood with respect to the gradient/nabla operator
\begin{align}
\nabla = \big(\partial_{x_1},\partial_{x_2},\partial_{x_3}\big).
\end{align}
\end{definition}

In order to gain some familiarity with vector calculus identities, let us prove the following identity.

\begin{lemma}\label{id:curlcurl}
Let $\mathbf{A}:\mathbb{R}^3 \to \mathbb{R}^3$ be a smooth vector field. Then, we have
\begin{align}
\nabla \times (\nabla \times \mathbf{A}) \;=\; \nabla(\nabla \cdot \mathbf{A})\;-\;\nabla^2 \mathbf{A},
\end{align}
where $\nabla^2 \mathbf{A}$ denotes the componentwise Laplacian, i.e. $(\nabla^2 \mathbf{A})_i=\displaystyle\sum_{j=1}^3 \partial_j\partial_j A_i$.
\end{lemma}
\begin{proof}
The curl of $\mathbf{A}$ reads in components
\begin{align}
(\nabla \times \mathbf{A})_i = \varepsilon_{ijk}\,\partial_j A_k,
\end{align}
where $\varepsilon_{ijk}$ is the Levi--Civita symbol, defined by\footnote{This symbol is completely antisymmetric in its indices, and it provides a compact way to express the cross product and the curl in Cartesian coordinates. For example,
\begin{align}
(\mathbf{A}\times \mathbf{B})_i &= \varepsilon_{ijk}\,A_j B_k, \\
(\nabla \times \mathbf{A})_i &= \varepsilon_{ijk}\,\partial_j A_k.
\end{align}
Contractions of two Levi--Civita symbols yield Kronecker deltas, namely
\begin{align}
\varepsilon_{ijk}\,\varepsilon_{i\ell m} = \delta_{j\ell}\delta_{km} - \delta_{jm}\delta_{\ell k}.
\end{align}}
\begin{align}
\varepsilon_{ijk} = 
\begin{cases}
+1 & \text{if $(i,j,k)$ is an even permutation of $(1,2,3)$},\\[6pt]
-1 & \text{if $(i,j,k)$ is an odd permutation of $(1,2,3)$},\\[6pt]
0 & \text{if any two indices coincide}.
\end{cases}
\end{align}
Therefore,
\begin{align}
\big(\nabla \times (\nabla \times \mathbf{A})\big)_i 
&= \varepsilon_{ijk}\,\partial_j\big(\varepsilon_{k\ell m}\,\partial_\ell A_m\big).
\end{align}
Using the contraction identity 
\begin{align}
\varepsilon_{ijk}\varepsilon_{k\ell m} = \delta_{i\ell}\delta_{jm} - \delta_{im}\delta_{j\ell},
\end{align}
we obtain
\begin{align}
\big(\nabla \times (\nabla \times \mathbf{A})\big)_i 
&= (\delta_{i\ell}\delta_{jm} - \delta_{im}\delta_{j\ell})\,\partial_j\partial_\ell A_m \\
&= \partial_j\partial_i A_j \;-\; \partial_j\partial_j A_i.
\end{align}
The first term is exactly $(\nabla(\nabla \cdot \mathbf{A}))_i$, and the second is $(\nabla^2 \mathbf{A})_i$. Thus, in terms of $\nabla$,
\begin{align}
\big(\nabla \times (\nabla \times \mathbf{A})\big)_i = (\nabla(\nabla \cdot \mathbf{A}))_i - (\nabla^2 \mathbf{A})_i,
\end{align}
which proves the identity.
\end{proof}

We now turn to Maxwell’s equations in vacuum, where no charges or currents are present. In this case, $\rho=0$ and $\mathbf{J}=0$, and the equations take the simplified form
\begin{align}\label{Maxwell's equations}
\nabla \cdot \mathbf{E} &= 0, 
&\qquad \nabla \cdot \mathbf{B} &= 0, \\
\nabla \times \mathbf{E} &= - \frac{\partial \mathbf{B}}{\partial t}, 
&\qquad \nabla \times \mathbf{B} &= \mu_0 \varepsilon_0 \frac{\partial \mathbf{E}}{\partial t}.
\end{align}

The first pair tells us that both the electric and magnetic fields are divergence-free: there are no sources or sinks for the fields in empty space, no charges and no magnetic monopoles. The second pair describes their dynamics: a time-varying magnetic field generates an electric field, and a time-varying electric field generates a magnetic field.

Let's now consider Faraday’s law,
\begin{align}
\nabla \times \mathbf{E} = - \frac{\partial \mathbf{B}}{\partial t}.
\end{align}
Taking the curl of both sides gives
\begin{align}
\nabla \times (\nabla \times \mathbf{E}) = - \frac{\partial}{\partial t}(\nabla \times \mathbf{B}),
\end{align}
and, using \cref{id:curlcurl}, we obtain
\begin{align}
\nabla \times (\nabla \times \mathbf{E}) = \nabla(\nabla \cdot \mathbf{E}) - \nabla^2 \mathbf{E}.
\end{align}
Now, by recalling that $\nabla \cdot \mathbf{E} = 0$ and that $\displaystyle{\nabla \times \mathbf{B} = \mu_0 \varepsilon_0 \frac{\partial \mathbf{E}}{\partial t}}$ in vacuum, we find
\begin{align}
- \nabla^2 \mathbf{E} = - \frac{\partial}{\partial t}\big(\mu_0 \varepsilon_0 \frac{\partial \mathbf{E}}{\partial t}\big),
\end{align}
which becomes
\begin{align}
\nabla^2 \mathbf{E} - \frac1{c^2} \frac{\partial^2 \mathbf{E}}{\partial t^2} = 0.
\end{align}
This is the standard wave equation, with wave speed
\begin{align}
c = \frac{1}{\sqrt{\mu_0 \varepsilon_0}}.
\end{align}
A completely analogous computation applied to $\mathbf{B}$ yields the same equation. Therefore, both the electric and the magnetic fields propagate in vacuum as waves moving with speed $c$. The appearance of the constant $c$ in Maxwell’s theory was a striking result. Its numerical value matched the known speed of light, as measured by Fizeau, Foucault, and others in mid-nineteenth-century experiments. Since optical phenomena such as interference, diffraction, and polarization had already demonstrated the undulatory nature of light, it was natural to identify the electromagnetic waves predicted by Maxwell with light itself. Thus, the propagation speed emerging from the electromagnetic field equations was not just another constant of nature: it was precisely the velocity of light, providing the theoretical bridge between electromagnetism and optics.

This unification was one of Maxwell’s most profound insights. However, this immediately raised a deep conceptual question: \emph{with respect to what is this velocity measured?}

In mechanics, velocities are always relative to some reference frame. once again, if a train moves at $100\ \mathrm{km/h}$, this means relative to the ground, and if an observer moves alongside the train, they may describe the velocity differently. As we mentioned, the equations of Newtonian mechanics are invariant under Galilean transformations, which assume the existence of an absolute space and an absolute time though which the definition of inertial frames could be possible. Within this conceptual background, it seemed natural to assume that the electromagnetic waves predicted by Maxwell’s theory must propagate with velocity $c$ relative to some medium filling all of space and at rest with respect to the absolute space. This hypothetical medium was called the \emph{luminiferous ether}. Just as sound waves require air and water waves require water, light waves were thought to require a physical substrate that carried the oscillations of the electromagnetic field.

The idea of the ether thus served as a way to reconcile Maxwell’s theory with the Newtonian picture of space and time. The ether was imagined as a stationary, all-pervading substance, defining an absolute frame of rest relative to which the velocity of light was measured. Different inertial observers, moving relative to the ether, would then measure different effective velocities of light, in the same way that the apparent velocity of sound depends on whether one moves with or against the wind. In this sense, the ether reintroduced the notion of an absolute structure in space, much like the absolute space of Newton.

The Michelson--Morley (1887) interferometer tests the ether hypothesis by comparing light travel times along two perpendicular arms while the whole apparatus moves with speed $v$ through the ether. To avoid any ambiguity, we distinguish two frames and relate them by a \emph{Galilean boost}:
\begin{itemize}
\item Ether frame $\mathcal{E}$: the ether is at rest; the light propagates isotropically at speed $c$.
\item Laboratory frame $\mathcal{T}$: the interferometer is at rest; relative to $\mathcal{E}$ it moves with constant velocity $v$ along the $x$-axis.
\end{itemize}
A Galilean boost from $\mathcal{E}$ to $\mathcal{T}$ is
\begin{align}
x_{\mathcal{T}} = x_{\mathcal{E}} - v\,t,\qquad y_{\mathcal{T}}=y_{\mathcal{E}},\qquad t_{\mathcal{T}}=t_{\mathcal{E}},\qquad \mathbf{u}_{\mathcal{T}}=\mathbf{u}_{\mathcal{E}}-v\,\hat{\mathbf{x}}.
\end{align}

\begin{center}
\includegraphics[width=0.55\textwidth]{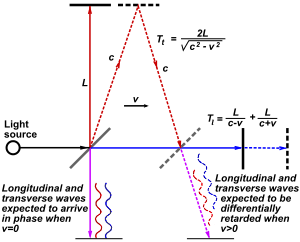}
\end{center}

Let each arm have length $L$. Two equivalent computations (w.r.t. the two frames) give the same result.

\emph{Longitudinal arm (parallel to $v$).} In $\mathcal{E}$ the mirrors and beam splitter translate at speed $v$. On the out-and-back path the light always moves at speed $c$ \emph{through the ether}, but the mirror is receding on the outbound leg and approaching on the return. The relevant relative speeds along the arm are $c-v$ (out) and $c+v$ (back), so
\begin{align}
T_{\parallel}=\frac{L}{c-v}+\frac{L}{c+v}=\frac{2Lc}{c^2-v^2}.
\end{align}
Equivalently, in $\mathcal{T}$, we can boost the light velocity by $\mathbf{u}_{\mathcal{T}}=\mathbf{u}_{\mathcal{E}}-v\hat{\mathbf{x}}$, which gives longitudinal speeds $c\mp v$.

\emph{Transverse arm (perpendicular to $v$).} In $\mathcal{E}$, since the mirror is moving, the light beam forms an angle (see the picture). Since $|\mathbf{u}_{\mathcal{E}}|=c$ and the longitudinal component has norm $v$, it must hold that the component along the arm has magnitude $\sqrt{c^2-v^2}$. Hence, the round-trip time is
\begin{align}
T_{\perp}=\frac{2L}{\sqrt{c^2-v^2}}.
\end{align}
On the other hand, in $\mathcal{T}$, the mirror is at rest, and the Galilean-boosted velocity has components $\mathbf{u}_{\mathcal{T}}=(u_x-v,\,u_y)$; choosing $u_x=v$ gives net longitudinal drift zero (the mirror is not moving) and an along-arm speed of $\sqrt{c^2-v^2}$.

\emph{Time difference and fringe shift.} The predicted time lag is
\begin{align}
\Delta T = T_{\parallel}-T_{\perp}
= \frac{2Lc}{c^2-v^2}-\frac{2L}{\sqrt{c^2-v^2}},
\end{align}
which for $v\ll c$ expands to $\Delta T \approx \frac{L}{c}\,\frac{v^2}{c^2}$. A time lag produces a phase lag; if the light wavelength is $\lambda$, the phase difference expressed in fringe units is
\begin{align}
N = \frac{\Delta\Phi}{2\pi}=\frac{\Delta T}{T_{\text{period}}}=\frac{c\Delta T}{\lambda}.
\end{align}
After a $90^\circ$ rotation, the two arms exchange roles, so the predicted magnitude of the fringe shift is
\begin{align}
|\Delta N|=\frac{2c\Delta T}{\lambda}
\approx \frac{2L}{\lambda}\frac{v^2}{c^2}.
\end{align}
\begin{remark}[Direction of $v$]
    Notice that we assumed that $v$ was in the direction of one of the arms, but one could simply do the calculation for a general direction making it dependent on an angle $\theta$.
\end{remark}
If one assumes that the velocity $v$ is given by the Earth's orbital velocity $v\approx 30\,\mathrm{km/s}$, the expected shift was of order $0.4$ fringes, upon a second measurement after a $90^\circ$ rotation of the apparatus.\\

\emph{The predicted shift of the interference fringes was not observed. This null result contradicted the prediction obtained from the Galilean ether model used above; the anticipated anisotropy of the round-trip speed of light was not seen.\footnote{There were, of course, attempts to preserve the notion of absolute space. It eventually became clear that insisting on this concept forces one to abandon other assumptions of the theory.}}\\

\begin{remark}[Cosmological perspective of the time]
The Earth's orbital speed $v\approx30\,\mathrm{km/s}$ provides a natural scale for estimating the expected effect, but it should not be identified automatically with the Earth's total velocity relative to a hypothetical ether. Within the Galilean picture, that velocity would also include any motion of the solar system relative to the ether. Using the orbital speed in an estimate therefore does not establish that the solar system was assumed to be at rest relative to the ether. Moreover, the Earth's orbital velocity changes direction over the year, so an unknown constant velocity of the solar system could not cancel it at all times.
\end{remark}

\section{The special theory of relativity}
We have reached a natural turning point. Starting from the Aristotelian picture, we adopted Galileo’s idealization and mathematization of motion, isolated dynamics in a single scalar functional (the action), derived the Euler–Lagrange equations, and identified the symmetries of the free particle; Maxwell’s equations then revealed light as a wave governed by precise field laws. The common thread is the search for structures that remain the same when we change our viewpoint. This motivates a unifying idea: \emph{the principle of relativity}, namely that the fundamental laws should have the same form for all admissible observers. In Newtonian kinematics, the admissible changes of viewpoint are the \emph{Galilean transformations} (spatial rotations and translations, time translations, and boosts with constant relative velocity)\footnote{Infinitesimally: $\vb x' = R\,\vb x + \vb v\,t + \vb a$, $t' = t + b$, with $R\in \mathrm{SO}(3)$ and constants $(\vb v,\vb a,b)$. These transformations map straight worldlines at constant velocity into straight worldlines at constant velocity and encode absolute time together with relative space.}. Operationally, Galileo’s ship thought experiment asserts the indistinguishability of uniform rectilinear motion from rest for purely mechanical phenomena, and this can be elevated to an axiom of Newtonian mechanics.
\begin{axiom*}[Galilean principle of relativity]
The laws of mechanics are equivalent in all inertial reference frames; or equivalently, their equations are invariant in form under Galilean transformations.
\end{axiom*}
Combined with the law of inertia, this axiom singles out an equivalence class of inertial frames and fixes the kinematics of classical mechanics: time is absolute (simultaneity is frame–independent), spatial distances at a given time can be compared across frames via Euclidean geometry, and velocities add linearly ($\vb u'=\vb u-\vb v$). Within this framework, a signal’s speed is not universal but frame–dependent; if light were an ordinary disturbance in a material medium, its measured speed would vary with the observer’s motion relative to that medium.

At this point, a tension emerges.  
Maxwell’s equations describe light as an electromagnetic wave propagating with a fixed velocity $c$, independent of the source.  
If one interprets this propagation as occurring in a material medium (the ``aether''), then the measured speed of light should depend on the observer’s velocity relative to that medium.  
Yet the Michelson–Morley experiment yielded a null result: no difference in the speed of light was observed along different directions of Earth’s motion.  
Attempts to reconcile these findings with the notion of an absolute space led to increasingly contrived hypotheses, ultimately undermining the very principle of relativity if kept in the Galilean framework.

As in Newtonian mechanics, the core of the problem lies in \emph{symmetry}.  
For a free particle, we saw that Galilean invariance dictated the admissible action and equations of motion.  
Now Maxwell’s theory forces us to ask: what is the correct symmetry underlying electromagnetism?  
If it were Galilean invariance, then the homogeneous wave equation should retain its form under Galilean boosts.  
We now show that this is not the case.

\begin{proposition}\label{prop:wave-not-galilean}
The scalar wave operator 
\begin{align}
\Box:=\nabla^{2}-\tfrac{1}{c^{2}}\partial_{t}^{2}
\end{align}
is not invariant under Galilean boosts.
\end{proposition}
\begin{proof}
Consider a Galilean boost with velocity $\vb v$.  
The coordinate change is
\begin{align}
\vb x'=\vb x-\vb v\,t, \qquad t'=t.
\end{align}
A scalar field transforms as $\phi'(\vb x',t')=\phi(\vb x,t)=\phi(\vb x'+\vb v t',t')$.  
By the chain rule, derivatives with respect to the unprimed coordinates can be expressed 
in terms of primed ones:
\begin{align}
\partial_{x_i}
&= \frac{\partial x'_j}{\partial x_i}\,\partial_{x'_j}
   +\frac{\partial t'}{\partial x_i}\,\partial_{t'}
= \delta_{ij}\,\partial_{x'_j}+0
= \partial_{x'_i}, \\[0.5em]
\partial_t
&= \frac{\partial x'_j}{\partial t}\,\partial_{x'_j}
   +\frac{\partial t'}{\partial t}\,\partial_{t'}
= (-v_j)\,\partial_{x'_j}+1\cdot\partial_{t'} \\
&= \partial_{t'} - \vb v\!\cdot\nabla'.
\end{align}
Thus
\begin{align}
\nabla = \nabla', \qquad \partial_{t}=\partial_{t'}-\vb v\!\cdot\nabla'.
\end{align}
Hence the wave operator becomes
\begin{align}
\Box \;=\; \nabla'^{2}-\frac{1}{c^{2}}\big(\partial_{t'}-\vb v\!\cdot\nabla'\big)^{2}
= \Big(\nabla'^{2}-\frac{1}{c^{2}}\partial_{t'}^{2}\Big)
+ \frac{2}{c^{2}}(\vb v\!\cdot\nabla')\,\partial_{t'}
-\frac{(\vb v\!\cdot\nabla')^{2}}{c^{2}}.
\end{align}
The extra terms vanish only if $\vb v=\vb 0$.  
Therefore $\Box$ does not retain its form under nontrivial Galilean boosts.
\end{proof}
This calculation reveals the crux: the symmetry of Maxwell’s equations is not the Galilean transformations.  
If we insist on Galilean relativity, the very structure of the wave operator would differ from one inertial frame to another, contradicting the empirical fact that the speed of light is measured to be the same for all observers.  Thus, we are led to abandon Galilean symmetry.
\begin{remark}
    As stressed in \cref{remsymmforc}, a fixed external field may break symmetries that are present in the complete theory. Electromagnetic fields can be treated as prescribed external fields when this is an appropriate approximation. However, in the complete description they have their own dynamics, governed by Maxwell’s equations, and interact with charged matter. The fundamental difficulty is therefore not the loss of symmetry caused by a particular external field, but the incompatibility between Galilean kinematics and Maxwell’s equations. We must develop a theory whose spacetime symmetry is coherent with both electrodynamics and the dynamics of matter, transforming the particles and the electromagnetic fields together.
\end{remark}
Einstein’s guiding idea is to take the constancy of the speed of light as a fundamental premise: in every inertial frame, light in vacuum propagates at speed $c$, independent of the source and of the observer.  
Together with the principle of relativity and the usual assumptions of spatial isotropy and spacetime homogeneity, this leads to three key kinematical effects: lack of absolute simultaneity, time dilation, and length contraction. In the following light--signal experiments, clocks within each frame are synchronized using Einstein’s convention. We also use the absence of a change in transverse lengths and assume that ideal clocks at rest in the same frame measure the same elapsed time, independently of their orientation or construction.\footnote{The absence of a transverse scale change can be justified using reciprocity and spatial isotropy. If transverse lengths are multiplied by a positive factor $k(v)$ when passing between two frames, exchanging the frames gives $k(v)k(-v)=1$, while isotropy gives $k(v)=k(-v)$; hence $k(v)=1$. These assumptions are additional to the constancy of the speed of light alone.}

We illustrate them with a boat–versus–shore setup.  
The boat is at rest in frame $\mathcal S'$, the shore defines frame $\mathcal S$, and the boat moves at constant velocity $v$ along the $X^1$--axis of $\mathcal S$.  

We identify three scenarios:
\begin{enumerate}
    \item \textbf{Relativity of simultaneity.}  
On the boat, two lamps at bow and stern (separated by the boat’s proper length $L_0$) are wired so that, at boat time $t'=0$, they emit simultaneously.  
A detector at the boat’s midpoint receives the two pulses together by symmetry in $\mathcal S'$.  
Consider the same process from the shore. Let $L$ be the boat’s length as judged in $\mathcal S$ (to be related to $L_0$ below).  
Denote by $t_{\rm bow}$ and $t_{\rm stern}$ the \emph{shore} emission times that lead the two pulses to reach the moving midpoint at the same shore instant. Set the reception event to $(t,x)=(0,0)$, so that the bow and stern follow $x_{\rm bow}(t)=vt+L/2$ and $x_{\rm stern}(t)=vt-L/2$.  
Because light has speed $c$ in $\mathcal S$, the backward–moving pulse (from the bow) and the forward–moving pulse (from the stern) must satisfy
\begin{align}
 \tfrac{L}{2}+v\,t_{\rm bow}\;=\;-c\,t_{\rm bow}
\quad\Rightarrow\quad
t_{\rm bow} \;=\; -\frac{L}{2(c+v)},\\
-\tfrac{L}{2}+v\,t_{\rm stern}\;=\;c\,t_{\rm stern}
\quad\Rightarrow\quad
t_{\rm stern} \;=\; -\frac{L}{2(c-v)}.
\end{align}
Therefore the shore time difference between the two \emph{emissions} is
\begin{align}
\Delta t \;=\; t_{\rm bow}-t_{\rm stern} 
\;=\; \frac{Lv}{c^{2}-v^{2}}
\;=\; \gamma^{2}\,\frac{v\,L}{c^{2}},
\end{align}
with $\gamma=(1-\frac{v^2}{c^2})^{-1/2}$. Thus, events arranged to be simultaneous on the boat (the two emissions) are not simultaneous on the shore.

\item \textbf{Time dilation.}  
Consider on the boat a \emph{light clock}: two mirrors a vertical distance $h$ apart.  
In $\mathcal S'$ one tick (up and down) takes the proper time
\begin{align}
\Delta t' \;=\; \frac{2h}{c}.
\end{align}
In $\mathcal S$ the apparatus drifts horizontally, so the light follows a longer zig–zag.  
For the half–tick, with horizontal drift $v\,\Delta t_{1/2}$ and vertical leg $h$,
\begin{align}
(c\,\Delta t_{1/2})^{2} \;=\; (v\,\Delta t_{1/2})^{2}+h^{2}
\quad\Rightarrow\quad
\Delta t_{1/2} \;=\; \frac{h}{\sqrt{c^{2}-v^{2}}} \;=\; \gamma\,\frac{h}{c}.
\end{align}
Hence, the full tick lasts
\begin{align}
\Delta t \;=\; 2\,\Delta t_{1/2} \;=\; \gamma\,\Delta t'.
\end{align}
A moving clock runs slower:
\begin{align}
\boxed{\ \Delta t \;=\; \gamma\,\Delta t'.\ }
\end{align}

\item \textbf{Length contraction.}  
Let $L_0$ be the boat’s \emph{proper} length (distance between stern and bow measured on the boat).  
In $\mathcal S'$ send a light pulse from stern to bow and back; the round–trip proper time is
\begin{align}
\Delta t' \;=\; \frac{2L_0}{c}.
\end{align}
In $\mathcal S$, the forward leg takes $t_{\rightarrow}=(L+vt_{\rightarrow})/c$ (the bow recedes), which gives $t_{\rightarrow}=L/(c-v)$, and, similarly, the return leg takes $t_{\leftarrow}=L/(c+v)$ (the stern approaches), so
\begin{align}
T \;=\; t_{\rightarrow}+t_{\leftarrow}
\;=\; \frac{L}{c-v}+\frac{L}{c+v}
\;=\; \frac{2Lc}{c^{2}-v^{2}}
\;=\; \frac{2L}{c}\,\gamma^{2}.
\end{align}
But the two events “emit at stern” and “receive at stern” are recorded by the \emph{same} boat clock, so the between them measured from the shore must be dilated by
\begin{align}
T \;=\; \gamma\,T' \;=\; \gamma\,\frac{2L_0}{c}.
\end{align}
Equating the two expressions gives
\begin{align}
\boxed{\ L \;=\; \frac{L_0}{\gamma},\ }
\end{align}
i.e.\ moving rods are shorter along the direction of motion by the factor $1/\gamma$.
\end{enumerate}
\textit{Summary.} Under the assumptions stated above, elementary light–signal experiments give
\begin{align}
&\text{Relativity of simultaneity: } 
\Delta t \;=\; \gamma^{2}\,\frac{v\,L}{c^{2}},
\quad\hfill\\
&\text{Time dilation: } \Delta t \;=\; \gamma\,\Delta t',\\
&\text{Length contraction: } L \;=\; \dfrac{L_0}{\gamma}.
\end{align}
The principle of relativity must be reformulated on a new foundation—one that treats the speed of light as invariant and ensures that the wave equation, or better Maxwell’s equations, retain the same form in all inertial frames. This step leads directly to Einstein’s special theory of relativity.
\begin{axiom*}[Einstein’s principle of relativity]
The laws of mechanics and electrodynamics are equivalent in all inertial reference frames; or equivalently, their equations are invariant in form under Poincaré transformations.
\end{axiom*}

\begin{remark}
    Notice that, with regard to \cref{remsymmforc}, the complete system of relativistic particles and electromagnetic fields possesses Poincaré symmetry when all dynamical variables are transformed together. A particular fixed electromagnetic background may have fewer symmetries. Therefore, the special theory of relativity not only described the space--time symmetries by means of different transformations, but also placed mechanics and electrodynamics within a common spacetime symmetry framework.
\end{remark}

In Newtonian mechanics, space and time are regarded as separate: three coordinates $x,y,z$ describe position, and a universal time $t$ flows independently.  

The three thought experiments above show that this separation cannot be maintained.  
Time dilation and length contraction reveal that temporal intervals and spatial lengths are not absolute but depend on the state of motion of the observer.  
The relativity of simultaneity demonstrates that space and time are not merely relative, but fundamentally intertwined: whether two events occur at the same time depends on the frame of reference.  

Taken together, these effects compel us to treat space and time as aspects of a four--dimensional object, which we shall call \emph{Minkowski spacetime}.
\begin{definition}[Affine space]
Let $V$ be a vector space over a field $\mathbb{K}$.  
An \emph{affine space} modeled on $V$ is a set $A$ together with an operation
\begin{align}
(p,v) \ &\longmapsto\ p+v\in A , \qquad p\in A,\ v\in V ,
\end{align}
such that:

1. For every $p\in A$ and every $v,w\in V$,
\begin{align}
(p+v)+w &= p+(v+w).
\end{align}

2. For every $p\in A$,
\begin{align}
p+0 &= p .
\end{align}

3. For every $p,q\in A$ there exists a unique $v\in V$ such that
\begin{align}
q &= p+v .
\end{align}

The vector $v$ is denoted by
\begin{align}
v &= q-p ,
\end{align}
and is called the \emph{displacement} from $p$ to $q$.  
\end{definition}
\begin{remark}
    Thus, points of $A$ are not vectors themselves, but differences of points define vectors in $V$.  
Choosing an origin $p_{0}\in A$ identifies $A$ with $V$ via
\begin{align}
p &\longmapsto p-p_{0},
\end{align}
though this identification depends on the choice of origin.
\end{remark}
\begin{definition}[Minkowski spacetime]\label{defmink}
Minkowski spacetime is a structure $(\mathcal M, \mathcal Q)$ where 
$\mathcal M$ is a four--dimensional affine space modeled on a vector space 
$V\cong\mathbb R^4$, whose points are called \emph{events}.  For any two events $p,q\in\mathcal M$, the displacement $X=p-q\in V$ is called 
a \emph{four--vector}.
 
$\mathcal Q$ is a quadratic form\footnote{A quadratic form on a real 
$n$-dimensional vector space $V$ is a map $Q:V\to\mathbb{R}$ such that, in a basis of $V$, 
$Q(v)=v^{T} B v$ for some symmetric $n\times n$ matrix $B$ w.r.t. the chosen basis. More precisely, a quadratic form is defined via a symmetric bilinear form $B$, such that $Q(v)=B(v,v)$.} defined on $V$. In a suitable basis of $V$, where its associated matrix is $\eta=\mathrm{diag}(-1,1,1,1)$,
it takes the canonical form\footnote{By Sylvester's law of inertia, every real quadratic form $Q$ on a 
finite--dimensional vector space $V$ can be represented, in a suitable basis, 
by a symmetric matrix $B$. Any two such matrices associated to $Q$ are 
congruent, i.e.\ related by $B' = P^T B P$ for some invertible matrix $P$. 
Sylvester's theorem states that there exists a basis in which $B$ is diagonal 
with entries $\pm 1$ (or $=0$), and that the numbers of positive and negative entries 
are invariants of $Q$. In the Minkowski case, $Q$ is chosen to have signature 
$(3,1)$, i.e.\ three positive and one negative directions, so that in a 
suitable basis its associated matrix is $\eta=\mathrm{diag}(-1,1,1,1)$.}
\begin{align}
\mathcal Q(X) = (X^1)^2+(X^2)^2+(X^3)^2-(X^0)^2
\qquad \forall \,X\in V,
\end{align}
with $X=(X^0,X^1,X^2,X^3)$.
\end{definition}
\begin{remark}
It is crucial that $\mathcal M$ is an affine space: we cannot add or scale events in a meaningful way, but we can take differences of events.  
These differences form vectors in an associated four--dimensional vector space $V\cong\mathbb{R}^4$.\\
In the physics literature it is customary to denote a four-vector with an upper index, writing
\begin{align}
X^\mu = (X^0,X^1,X^2,X^3),
\end{align}
where usually $X^i$ ($i=1,2,3$) denotes the spatial components.
\end{remark}
\begin{remark}[Coordinates in Minkowski spacetime]
Once a basis $\{e_\mu\}$ for $V$ is fixed, every four--vector $X\in V$ can be written as
\begin{align}
    X &= X^\mu e_\mu ,
\end{align}
and the tuple $(X^0,X^1,X^2,X^3)$ are the \emph{components} of $X$ in that basis.\\
If one also chooses an origin $O\in\mathcal M$, it is possible to define \emph{coordinates}
for events $p\in\mathcal M$. Each event $p$ can then be written uniquely as
\begin{align}
    p &= O + X, \qquad X\in V ,
\end{align}
and the components of the displacement/four--vector $X$ relative to $(O,\{e_\mu\})$ are identified with the coordinates of $p$. In particular, if we choose the \emph{canonical basis},\footnote{Notice that, in this basis, the quadratic form is represented by the canonical matrix 
$\eta=\mathrm{diag}(-1,1,1,1)$. A general basis in which $\eta$ takes this form is called an 
$\eta$--\emph{orthonormal basis} (see next definition).} we conventionally denote the components as
\begin{align}
    (X^0,X^1,X^2,X^3) &= (ct,x,y,z),
\end{align}
with $X^0=ct$ interpreted as the temporal part and $(X^1,X^2,X^3)=(x,y,z)$ as the spatial parts.  
In these coordinates, the quadratic form reads
\begin{align}
    \mathcal Q(X) &= x^2+y^2+z^2-c^2t^2.
\end{align}
The coordinates of $p$ with respect to the canonical basis are called \emph{canonical coordinates}.
\end{remark}
\begin{definition}[Orthonormal basis]\label{etaortbas}
    Every basis of $V$ for which 
    \begin{align}
        \eta = \mathrm{diag}(-1,1,1,1)
    \end{align}
    is called an $\eta$--\emph{orthonormal basis}.
\end{definition}
\begin{definition}[Reference frame]\label{defreffram}
     A \emph{reference frame} for $\mathcal M$ is the choice of a point in $\mathcal M$ together with $\eta$-orthonormal basis of $V$.
\end{definition}
\begin{remark}
    Notice that the concept of coordinates is more general that the one of a reference frame. In fact, in all reference frames one has coordinates, but not all coordinates come from a reference frame, since general ones might be with respect to a non $\eta$-orthonormal basis of $V$.
\end{remark}
The following result lies at the heart of Special Relativity. It expresses the fundamental symmetry of Minkowski spacetime and shows that any two $\eta$--orthonormal bases are related by a Lorentz transformation.
\begin{definition}[Lorentz transformations]\label{lorinv}
A \emph{Lorentz transformation} is an invertible linear map
\begin{align}
\Lambda:V \to V
\end{align}
that preserves the Minkowski quadratic form $\mathcal Q$, i.e.
\begin{align}
\mathcal Q(\Lambda X)=\mathcal Q(X)\qquad \text{for all } X\in V.
\end{align}
Equivalently, in an $\eta$--orthonormal basis where $\mathcal Q$ is represented by the 
matrix $\eta=\mathrm{diag}(-1,1,1,1)$, this condition reads
\begin{align}
\Lambda^{T}\eta \Lambda = \eta.
\end{align}
\end{definition}
\begin{proposition}\label{prop:lorentz-with-PT}
The linear maps on four--vectors that preserve the Minkowski quadratic form are exactly those obtainable as finite compositions of:
\begin{itemize}
\item \textbf{Spatial rotations}:
\begin{align}
ct'&=ct,\qquad
\begin{pmatrix}x'\\y'\\z'\end{pmatrix}=R\begin{pmatrix}x\\y\\z\end{pmatrix},
\end{align}
where $R$ is a $3\times 3$ real orthogonal matrix with $\det R = +1$.
\item \textbf{Boosts} (e.g.\ along $x$):
\begin{align}
ct'&=\gamma\,(ct-\beta x),\qquad
x'=\gamma\,(x-\beta ct),\qquad
y'=y,\ z'=z,\quad \gamma=(1-\beta^2)^{-1/2};
\end{align}
\item \textbf{Parity} $P$ (space inversion):
\begin{align}
ct'&=ct,\qquad x'=-x,\ y'=-y,\ z'=-z;
\end{align}
\item \textbf{Time reversal} $T$:
\begin{align}
ct'&=-ct,\qquad x'=x,\ y'=y,\ z'=z.
\end{align}
\end{itemize}
\end{proposition}
\begin{proof}
Left to future chapters.
\end{proof}
\begin{remark}[Lorentz transformations as hyperbolic rotations]
As mentioned in the proof, Lorentz transformations can be viewed as a combination of 
ordinary and \emph{hyperbolic rotations} in four dimensions.  
Spatial rotations act as ordinary circular rotations on the three spatial components, 
preserving the Euclidean length $x^2+y^2+z^2$.  
Boosts, instead, are hyperbolic rotations in the $(ct,x^i)$--plane (with $i=1,2,3$), 
preserving the difference $-(ct)^2+(x^i)^2$.
\end{remark}
\begin{definition}[Proper and orthochronous]
A Lorentz transformation $\Lambda$ is called
\begin{itemize}
    \item Proper:\footnote{This condition means that the four--dimensional orientation is preserved. It does not imply preservation of time orientation: the composition $PT=-I$ is proper, but reverses time orientation and the orientation of the spatial subspace separately.} $\det \Lambda = +1$;
    \item Improper: $\det \Lambda = -1$;
    \item Orthochronous: $\Lambda_0^0 > 0$;
    \item Non--orthochronous: $\Lambda_0^0 < 0$.
    \end{itemize}
The \emph{proper orthochronous} ones are indeed both proper and orthochronous.
\end{definition}
\begin{definition}[Poincaré transformations]
Recall that Lorentz transformations are linear maps acting on the vector space $V\cong\mathbb R^4$ of four--vectors, preserving the Minkowski interval.  
The \emph{Poincaré transformations} are affine transformations acting on Minkowski spacetime $\mathcal M$: they consist of a Lorentz transformation on the displacement vectors, together with a translation of events.

Concretely, if $p\in\mathcal M$ is an event, a Poincaré transformation sends
\begin{align}
p \;\mapsto\; p' = p_0 + a + \Lambda (p-p_0),
\end{align}
where $a\in V$ is a translation vector, $\Lambda:V\to V$ a Lorentz transformation, and $p_0\in\mathcal M$ an arbitrary origin chosen to identify $\mathcal M$ with $V$.  
A difference choice $p'_0$ is related to the old one via a Poincaré transformation. In fact, this would simply result in a translation, i.e. a Poincaré transformation, since $p'_0+\Lambda (p-p'_0)=p'_0+\Lambda(p_0-p'_0)+\Lambda (p-p_0)$.
\end{definition}
\begin{remark}
    By definition, Poincaré transformations preserve the orthonormality of four--vectors in $V$ with respect to $\eta$, and hence the structure of a reference frame defined in \cref{defreffram}. 
    Therefore, they provide the natural spacetime transformations relating different reference frames, and they are the natural symmetry transformations of special relativity, extending the Galilean transformations of Newtonian mechanics.
\end{remark}
\begin{definition}[Active and passive transformations]
    Let $v\in V$ be a four--vector.
    \begin{itemize}
        \item A Lorentz \emph{passive transformation} consists in replacing the basis vectors $e_\mu$ by $e'_\mu=\Lambda(e_\mu)$, where $\Lambda:V\to V$ is a Lorentz transformation. If $\Lambda^\nu{}_\mu$ are its matrix entries in the original basis, this explicitly reads\footnote{We will define the meaning of the position of the indices more rigorously in the next chapters.}
        \begin{align}
            e_\mu \;\mapsto\; e'_\mu = e_\nu\,\Lambda^\nu{}_\mu.
        \end{align}
        Notice that we are just changing the basis in which we represent $v$, so the vector itself, as an element of $V$, remains unchanged: $v = v^\mu e_\mu = v'^\mu e'_\mu$. Therefore, the components of the vector in the new basis are obtained by the inverse matrix:
        \begin{align}
            v'^\mu= (\Lambda^{-1})^\mu{}_\nu v^\nu.
        \end{align}
        \item A Lorentz \emph{active transformation} consists in transforming the vector itself while keeping the basis fixed,
        \begin{align}
            v \;\mapsto\; v' = \Lambda v,
            \qquad
            v'^\mu = \Lambda^\mu{}_\nu v^\nu .
        \end{align}
    \end{itemize}
    Note that a Poincaré transformation in Minkowski spacetime induces a Lorentz active transformation on four--vectors.
\end{definition}
\begin{remark}
    For those familiar with the notions of covariant and contravariant vectors, let us stress that a passive transformation here simply means a change of basis by means of a linear map --- in this case restricted to be a Lorentz transformation. It should not be confused with the transformation law of a covariant vector under a change of coordinates.
\end{remark}
The Minkowski quadratic form defines a pseudo-inner product.\footnote{Technically, on a real vector space $V$, every quadratic form $Q$ determines a unique 
symmetric bilinear form $B$ (expressed by an associated symmetric matrix on a certain basis) via the polarization identity 
$B(u,v)=\tfrac12\big(Q(u+v)-Q(u)-Q(v)\big)$, so that $Q(v)=B(v,v)$. 
Conversely, every symmetric bilinear form defines a quadratic form in this way. Moreover,
if $B$ is positive definite, it is an (Euclidean) inner product; if $B$ has mixed 
signature $(p,q)$ with $p,q>0$, it is an indefinite (pseudo--)inner product---for 
instance, in Minkowski spacetime $B$ has Lorentzian signature $(3,1)$. 
We will return to these notions later on.}
\begin{definition}[Minkowski inner product and norm]\label{minkscpr}
On $V\cong\mathbb{R}^{4}$, for every two four--vectors 
$X=(X^{0},X^{1},X^{2},X^{3})$ and $Y=(Y^{0},Y^{1},Y^{2},Y^{3})$, we define 
their \emph{Minkowski inner product}. In an $\eta$--orthonormal basis of $V$, 
where the quadratic form $\mathcal Q$ is represented by the matrix 
$\eta=\mathrm{diag}(-1,1,1,1)$, it is given by\footnote{We recall Einstein's 
summation convention: any index repeated once up and once down is implicitly 
summed over its range.}
\begin{align}
\langle X,Y\rangle := \eta_{\mu\nu}\,X^{\mu}Y^{\nu}
= -X^{0}Y^{0}+X^{1}Y^{1}+X^{2}Y^{2}+X^{3}Y^{3}.
\end{align}
This rule is bilinear (linear in each entry) and symmetric 
($\langle X,Y\rangle=\langle Y,X\rangle$), just like the Euclidean dot product, 
but it is not positive definite.\\
For a vector $X$, the associated \emph{Minkowski norm} is given by
\begin{align}
\|X\|^{2}:=\langle X,X\rangle.
\end{align}
\end{definition}
\begin{remark}
    In this setting, the invariance property stated in \cref{lorinv} reads that $\langle X,Y\rangle$ is preserved by Lorentz transformations, i.e.\ by linear maps $L:V\to V$ such that
    \begin{align}
    \langle LX,LY\rangle = \langle X,Y\rangle, \qquad \forall\,X,Y\in V.
    \end{align}
\end{remark}
\begin{remark}[The light-cone]
The Minkowski inner product induces a natural classification of four--vectors, and hence of the possible relations between events in spacetime.  
Given $X\in V$ we say:
\begin{itemize}
\item $X$ is \emph{timelike} if $\|X\|^{2}<0$,
\item $X$ is \emph{spacelike} if $\|X\|^{2}>0$,
\item $X$ is \emph{lightlike (or null)} if $\|X\|^{2}=0$.
\end{itemize}
The set of all null vectors through the origin forms a double cone, called the \emph{light cone}.  
Vectors inside the cone are timelike, and represent displacements that a massive particle could follow.  
Vectors outside the cone are spacelike, corresponding to displacements that would require faster--than--light travel. 
\begin{figure}[h]
\centering
\includegraphics[width=0.55\textwidth]{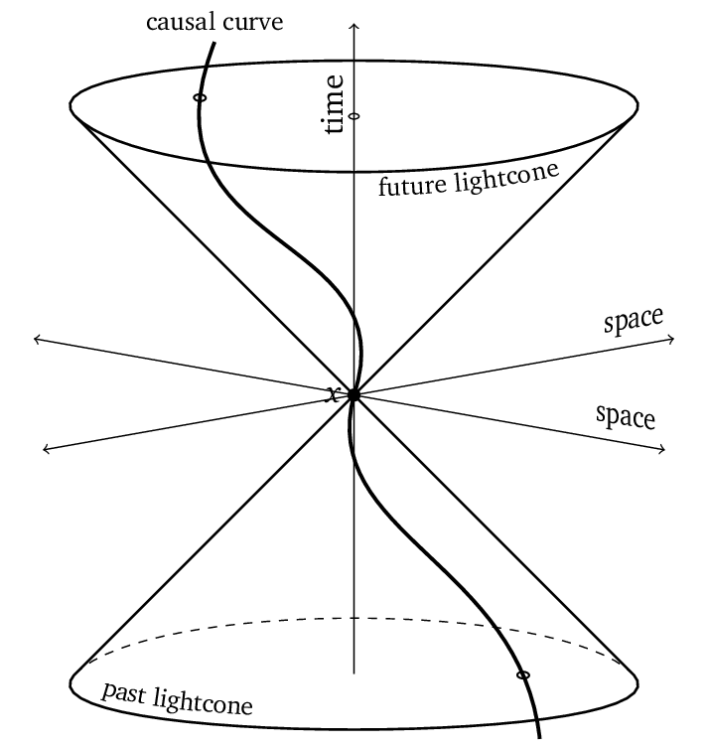}
\caption{Light cone structure in Minkowski spacetime. 
The future cone contains all possible influences of the event at the apex, the past cone all possible causes. 
Timelike curves remain within the cones, while spacelike directions lie outside. Source: Journal of Futures Studies, September 2021, 26(1): 83–90.}
\end{figure}
\end{remark}
\begin{lemma}[Existence of the rest frame]\label{exrestfr}
    Let $X \in V$. Then, according to the sign of $\|X\|^{2}$, there exists a Lorentz transformation $\Lambda$ such that $\Lambda X$ takes one of the following forms:
    \begin{itemize}
        \item If $X$ is timelike, $\|X\|^{2}<0$, then
        \begin{align}
            \Lambda X = \big(\sqrt{-\|X\|^{2}},\,0,\,0,\,0\big).
        \end{align}
        (If one restricts to the proper and orthochronous Lorentz transformations, the first component becomes $\operatorname{sgn}(X^0)\sqrt{-\|X\|^{2}}$.)

        \item If $X$ is spacelike, $\|X\|^{2}>0$, then
        \begin{align}
            \Lambda X = \big(0,\sqrt{\|X\|^{2}},\,0,\,0\big).
        \end{align}

        \item If $X$ is null, $\|X\|^{2}=0$ and $X\neq 0$, then
        \begin{align}
            \Lambda X = (\alpha,\,\alpha,\,0,\,0), \qquad \alpha\neq 0.
        \end{align}
        (Notice that the sign of $\alpha$ is determined by whether the vector is future-- or past--directed. In the case of proper and orthochronous Lorentz transformations, this time orientation cannot be changed.)
    \end{itemize}
    In particular, every timelike vector admits a rest frame, i.e.\ a Lorentz frame in which it is purely temporal.
\end{lemma}
\begin{proof}
Write $X=(X^{0},X^{1},X^{2},X^{3})$ and 
\begin{align}
    \|X\|^{2}=-(X^{0})^{2}+(X^{1})^{2}+(X^{2})^{2}+(X^{3})^{2}.
\end{align}
If $\vec X=(X^{1},X^{2},X^{3})\neq 0$, we can perform a spatial rotation that sends $\vec X$ to $(\|\vec X\|,0,0)$.\footnote{Spatial rotations act transitively on the set of spatial vectors with the same Euclidean length.}
Being part of Lorentz transformations, such a rotation does not change the value of $\|X\|^{2}$.  
Hence, without loss of generality, we may assume
\begin{align}
    X=(X^{0},\,\|\vec X\|,\,0,\,0).
\end{align}
A Lorentz boost, for a real parameter $\beta$ with $|\beta|<1$, acts as
\begin{align}
    X'^{0} &= \gamma\,(X^{0}-\beta\,\|\vec X\|),\\
    X'^{1} &= \gamma\,(\|\vec X\|-\beta\,X^{0}),\\
    X'^{2} &= 0,\qquad X'^{3}=0,
\end{align}
where $\gamma=(1-\beta^{2})^{-1/2}$. We know that a boost preserves the Minkowski norm, therefore it holds $\|X'\|^{2}=\|X\|^{2}$.

\smallskip
\noindent\textit{Timelike case} ($\|X\|^{2}<0$).  
Here $|X^{0}|>\|\vec X\|$. For $\|\vec X\|=0$, it is already in the form of the statement. In the other case, choosing $\beta=\|\vec X\|/X^{0}$ gives $X'^{1}=0$ and
\begin{align}
    X'^{0}=\operatorname{sgn}(X^{0})\sqrt{(X^{0})^{2}-\|\vec X\|^{2}}
    =\operatorname{sgn}(X^{0})\sqrt{-\|X\|^{2}}.
\end{align}
With a time reversal operation, obtain the claimed form.

\smallskip
\noindent\textit{Spacelike case} ($\|X\|^{2}>0$).  
Here $\|\vec X\|>|X^{0}|$. Apart from the trivial case where $X^0=0$, choosing $\beta=X^{0}/\|\vec X\|$ gives $X'^{0}=0$ and
\begin{align}
    X'^{1}=\sqrt{\|\vec X\|^{2}-(X^{0})^{2}}=\sqrt{\|X\|^{2}}.
\end{align}
Thus $X'=(0,\sqrt{\|X\|^{2}},0,0)$.

\smallskip
\noindent\textit{Null case} ($\|X\|^{2}=0$, $X\neq 0$).  
Then $\|\vec X\|=|X^{0}|$, so we have
\begin{align}
    X=(\alpha,\pm\alpha,0,0), \qquad \alpha\neq 0.
\end{align}
Then, we just apply a rotation, if necessary, to bring it in the desired form.
\end{proof}
The following result shows that a space--like separation violates causality.
\begin{proposition}[Spacelike separations admit time-order reversal]\label{prop:spacelike-reversal}
Let $A$ and $B$ be two events with spacelike separation $\Delta X=(c\,\Delta t,\Delta x,\Delta y,\Delta z)$, i.e.
\begin{align}
\|\Delta X\|^{2} = -c^{2}(\Delta t)^{2} + (\Delta x)^{2} + (\Delta y)^{2} + (\Delta z)^{2} > 0.
\end{align}
Then there exists an inertial frame (obtained by a Lorentz boost) in which the time order of $A$ and $B$ is reversed.
\end{proposition}

\begin{proof}
By a spatial rotation we can align the spatial separation along the $x$–axis, so that
\begin{align}
\Delta X=(c\,\Delta t,\Delta x,0,0), \qquad \text{with } (\Delta x)^{2} > c^{2}(\Delta t)^{2}.
\end{align}
Consider a Lorentz boost along $x$ with velocity $v<c$. 
The transformed time difference is
\begin{align}
\Delta t'=\gamma\!\left(\Delta t - \frac{v}{c^{2}}\Delta x\right), \qquad \gamma=\frac{1}{\sqrt{1-v^{2}/c^{2}}}.
\end{align}
Choose
\begin{align}
v_{0} := \frac{c^{2}\,\Delta t}{\Delta x}.
\end{align}
Because the separation is spacelike, we have $|v_{0}|=\dfrac{c^{2}|\Delta t|}{|\Delta x|}<c$, so $v_{0}$ is an admissible boost speed. 
For this choice,
\begin{align}
\Delta t'=\gamma\!\left(\Delta t - \frac{v_{0}}{c^{2}}\Delta x\right)=0,
\end{align}
i.e. $A$ and $B$ are simultaneous in the primed frame. 
Now, vary the boost velocity slightly: if $v>v_{0}$ (with the same sign as $\Delta x$), then
\begin{align}
\Delta t'=\gamma\!\left(\Delta t - \frac{v}{c^{2}}\Delta x\right)<0,
\end{align}
whereas if $v<v_{0}$ one gets $\Delta t'>0$. 
Hence there are inertial frames in which $B$ occurs before $A$ and others in which $A$ occurs before $B$.
Therefore, for spacelike separations the temporal order is not invariant under Lorentz transformations. 
Since causal influence is restricted to timelike or lightlike separations (signals cannot exceed $c$), spacelike–separated events cannot be causally related without violating this causal structure.
\end{proof}
\begin{remark}
Notice that, even if the boost is an orthochronous transformation, the sign of the time difference 
between two spacelike separated events may change. This is not a property of the 
boost itself, but of the fact that spacelike separations have no invariant 
time ordering.
\end{remark}
Therefore, the light cone determines the \emph{causal structure} of spacetime: 
an event $B$ can be influenced by another event $A$ only if the four-vector $B-A$ is timelike or lightlike.  
Spacelike separation means that no signal limited by the speed of light can connect the two events, hence no causal influence is possible.
\begin{remark}[Lorentz invariance and causal preservation]
By \cref{lorinv}, we know that Lorentz transformations preserve the Minkowski norm $\|X\|^{2}=\langle X,X\rangle$ Therefore, they preserve the causal classification of vectors and separations between events: timelike ($\|X\|^{2}<0$), lightlike ($\|X\|^{2}=0$), and spacelike ($\|X\|^{2}>0$). 
Consequently, the light cone at each event is mapped into itself. 
In particular, any timelike worldline remains timelike in every inertial frame, and no Lorentz transformation can turn a lightlike separation into a timelike or spacelike one (and vice versa). 
Thus the causal structure determined by the light cones is invariant under changes of inertial frame.
\end{remark}
\begin{definition}[Invariant separation and proper time]
Given a timelike four--vector $X=(ct,x,y,z)$, the Minkowski interval
\begin{align}
s^{2} = -c^{2}t^{2}+\|x\|^{2}
\end{align}
is strictly negative.  
For a worldline $X=X(t)$ the infinitesimal interval is
\begin{align}
ds^{2} = -c^{2}dt^{2}+dx^{2}+dy^2+dz^2.
\end{align}
The \emph{proper time} $\tau$ along the worldline is defined by
\begin{align}
d\tau = \frac{1}{c}\sqrt{-ds^{2}}
= \sqrt{\,1-\frac{\dot{x}^{\,2}}{c^{2}}}\,dt.
\end{align}
It is invariant under Lorentz transformations: all inertial observers agree on the elapsed proper time between two timelike separated events.
\end{definition}
\begin{remark}
In the Newtonian framework, the admissible set of transformations on space and time was given by the Galilean ones, reflecting the Galilean principle of relativity.  
As we saw earlier, this symmetry was encoded directly in the free action of a particle: the Lagrangian $\mathcal{L}=\tfrac{1}{2}m\,|\dot{\mathbf{x}}|^{2}$ is invariant under Galilean transformations up to a total time derivative, and this invariance embodies the fact that the equations of motion are the same in all inertial frames connected by such transformations.  
\\
Special relativity requires us to switch from the Galilean to the Poincaré transformations as the fundamental spacetime symmetry in order to preserve $c$.
Accordingly, the free action of a particle must now be built from Poincaré--invariant quantities.  
Within the class of first-order worldline actions with no additional fields or degrees of freedom, Poincaré invariance together with reparametrization invariance determines the action on timelike curves up to an overall factor, for a strictly invariant representative; adding an irrelevant boundary term does not change the equations of motion.
\end{remark}
The theorem below makes this statement mathematically precise.
\begin{remark}
The proper time provides the natural invariant measure of length along timelike curves, 
in the same way that the Euclidean norm provides the natural rotation--invariant measure of length in $\mathbb{R}^3$.
\end{remark}
\begin{proposition}
Consider first-order actions
\begin{align}
S[x]=\int_{\lambda_1}^{\lambda_2}L(x,\dot x)\,d\lambda
\end{align}
on regular timelike curves, with no additional fields or degrees of freedom and no explicit dependence on $\lambda$. Assume that the Lagrangian density is strictly Poincaré invariant and that the action is invariant under orientation-preserving reparametrizations. Then the action is proportional to $\int d\tau$. With the normalization for a particle of mass $m>0$, it is
\begin{align}
S_{\mathrm{SR}} = -mc^{2}\int d\tau.
\end{align}
After subtracting its constant rest-energy contribution on a fixed coordinate-time interval, it reduces to the Newtonian free-particle action as $c\to+\infty$. Adding a boundary term depending only on the endpoints does not alter the equations of motion.
\end{proposition}
\begin{proof}[Sketch of the proof]
\leavevmode\newline
\vspace{-0.55cm}
\begin{enumerate}
    \item \emph{First-order dependence.}
    By assumption, $L$ depends only on $x$ and $\dot x$. This is stronger than locality alone, which would also allow dependence on higher derivatives. Translation invariance gives $L(x+a,\dot x)=L(x,\dot x)$ for every constant $a$, so $L$ is independent of $x$.

    \item \emph{Poincaré invariance and reparametrization invariance.}
    The Lorentz group acts transitively on timelike vectors with a fixed value of $-\eta(\dot x,\dot x)$. Hence
    \begin{align}
    L(\dot x)=f(s),\qquad s=-\eta_{\mu\nu}\dot x^\mu\dot x^\nu>0.
    \end{align}
    Under an orientation-preserving reparametrization, the tangent vector rescales by a positive factor. Invariance of $L\,d\lambda$ therefore requires
    \begin{align}
    L(\alpha\dot x)=\alpha L(\dot x),\qquad \alpha>0.
    \end{align}
    Consequently,
    \begin{align}
    f(\alpha^2s)=\alpha f(s),
    \end{align}
    which implies $f(s)=C\sqrt{s}$ for a constant $C$. Thus
    \begin{align}
    S=C\int\sqrt{-\eta_{\mu\nu}\dot x^\mu\dot x^\nu}\,d\lambda
    =Cc\int d\tau.
    \end{align}
    This proves uniqueness within the stated class.

    \item \emph{Fixing the normalization.}
    Write
    \begin{align}
    S=A\int d\tau.
    \end{align}
    Using coordinate time as parameter and writing $\mathbf v=d\mathbf x/dt$, we obtain
    \begin{align}
    d\tau &= dt\,\sqrt{1-\tfrac{|\mathbf v|^2}{c^2}}\\
    &= dt\big(1-\tfrac{|\mathbf v|^2}{2c^2}+O(c^{-4})\big).
    \end{align}
    Matching the coefficient of $|\mathbf v|^2$ to the Newtonian kinetic term $\tfrac12m|\mathbf v|^2$ fixes
    \begin{align}
    A=-mc^2.
    \end{align}
    Therefore, on a fixed coordinate-time interval,
    \begin{align}
    S_{\mathrm{SR}}
    &=-mc^2(t_2-t_1)+\int_{t_1}^{t_2}\tfrac12m|\mathbf v|^2\,dt+O(c^{-2}).
    \end{align}
    The first term is independent of the spatial trajectory when $t_1$ and $t_2$ are fixed. After subtracting it, the action tends to the Newtonian free-particle action.
\end{enumerate}
\end{proof}
\begin{remark}
Notice that, once again, we are requiring the symmetries of spacetime, i.e. the change from one reference frame to another, to be respected by the action functional. This is a deep geometric insight: it suggests that the action of a particle must ultimately be constructed from geometric objects that encode the very symmetries of spacetime itself. This idea will become one of the central pillars of General Relativity.
\end{remark}
\begin{remark}[On the notion of four--vector]
    In special relativity, a \emph{four--vector} is simply any object $w$
    described by four components with respect to some basis, such that under an active Lorentz transformation represented by a matrix $\Lambda$, its components in the fixed basis change according to
    \begin{align}
        w' = \Lambda w.
    \end{align}
    Under the passive change of basis $e'_\mu=\Lambda(e_\mu)$, the same vector instead has components $w'=\Lambda^{-1}w$. Equivalently, if a change of inertial coordinates is written as $x'=\Lambda x+a$, four--vector components transform as $w'=\Lambda w$, and the associated basis change uses $\Lambda^{-1}$. The vector transformation law must always be read together with the convention for the basis or coordinate transformation.
\end{remark}
\begin{remark}
At first sight, parametrization invariance seems to have nothing to do with physics, appearing as a mere redundancy in the description of the action functional.
However, it has a concrete consequence: the degree-one homogeneity of the Lagrangian makes the Legendre transformation singular and leads to a constraint on the canonical momenta. For the free relativistic particle, this constraint can be computed explicitly.

For a regular timelike worldline parametrized by $\lambda$, write
\begin{align}
\mathcal L=-mc\sqrt{-\eta_{\rho\sigma}\dot x^\rho\dot x^\sigma}.
\end{align}
The canonical momentum is the covector
\begin{align}
p_\mu
:=\frac{\partial\mathcal L}{\partial\dot x^\mu}
=mc\,\frac{\eta_{\mu\nu}\dot x^\nu}
{\sqrt{-\eta_{\rho\sigma}\dot x^\rho\dot x^\sigma}}.
\end{align}
Raising its index with the inverse Minkowski metric gives the four--momentum $P^\mu=\eta^{\mu\nu}p_\nu$. Direct substitution yields
\begin{align}\label{rel4mom}
\eta^{\mu\nu}p_\mu p_\nu
=\eta_{\mu\nu}P^\mu P^\nu
=-m^2c^2.
\end{align}
This is a constraint on the image of the Legendre transformation, rather than an additional equation of motion.\footnote{The second Noether theorem expresses the same gauge redundancy through an identity among the Euler--Lagrange expressions. With $\mathcal E_\mu=\frac{\partial\mathcal L}{\partial x^\mu}-\frac{d}{d\lambda}\frac{\partial\mathcal L}{\partial\dot x^\mu}$, infinitesimal reparametrizations give $\dot x^\mu\mathcal E_\mu\equiv0$. The specific mass-shell relation \eqref{rel4mom} follows from the particular free-particle Lagrangian above.}

To identify the energy, choose a future-directed worldline and use coordinate time as parameter. The resulting Lagrangian is
\begin{align}
L_t=-mc^2\sqrt{1-\frac{|\mathbf v|^2}{c^2}},
\qquad \mathbf v=\frac{d\mathbf x}{dt}.
\end{align}
The energy associated with time-translation invariance is
\begin{align}
E
=\mathbf v\cdot\frac{\partial L_t}{\partial\mathbf v}-L_t
=\frac{mc^2}{\sqrt{1-|\mathbf v|^2/c^2}}.
\end{align}
Consequently, with our signature convention,
\begin{align}
P^0=\frac Ec,\qquad
p_0=\frac{\partial\mathcal L}{\partial\dot x^0}=-\frac Ec.
\end{align}
Writing $\mathbf p=(P^1,P^2,P^3)$, the mass-shell constraint becomes
\begin{align}
E^2=c^2|\mathbf p|^2+m^2c^4.
\end{align}
Since $P$ is timelike, \cref{exrestfr} provides a rest frame in which its spatial components vanish. In that frame,
\begin{align}\label{rel4mom_0}
(P^0)^2=p_0^2=m^2c^2.
\end{align}
On the positive-energy branch, the energy in the rest frame is therefore
\begin{align}\label{E=mc2}
\boxed{\ E_{\mathrm{rest}}=mc^2.\ }
\end{align}
Thus reparametrization invariance is associated with a structural constraint, whose explicit form for the free relativistic action gives the relativistic energy--momentum relation and, in the rest frame, Einstein’s mass--energy relation.
\end{remark}
Special relativity provides a coherent kinematical framework in which space and time are unified into Minkowski spacetime, and it reconciles mechanics with electrodynamics by establishing the invariance of the speed of light. Nevertheless, the theory is not complete. For example, the formulation is tied to the notion of inertial reference frames, defined as systems in which free particles move along straight worldlines at constant velocity. While operationally useful, this notion remains conceptually circular and lacks an intrinsic dynamical characterization within the theory itself.

Einstein was fully aware of these limitations. He considered special relativity as a decisive but provisional step: it solved the conflict between mechanics and electrodynamics, but at the same time it pointed to a more general framework in which spacetime itself would become a dynamical entity, accounting for the inclusion of gravity. This is precisely the path that leads to theory of general relativity.

\printbibliography[heading=subbibliography]
\chapter{Elements of Differential Geometry}
\noindent\fbox{%
\parbox{\dimexpr\textwidth-2\fboxsep-2\fboxrule\relax}{%
\textbf{From now on we adopt the \emph{Einstein summation convention}:}
whenever an index appears once in an upper and once in a lower position,
summation over its full range is understood without explicitly writing the sum.
For example, $v = v^{i}\,\partial_{i}$, $\alpha = \alpha_{i}\,dx^{i}$,
and $\langle \alpha, v\rangle = \alpha_{i} v^{i}$ all implicitly involve
a sum over $i$. We will keep using the summation symbol when the indices run over numbers different from the typical dimension.}}
\section{Smooth manifolds}
The concept of a smooth manifold (or $C^\infty$--differentiable manifold) is that of a space which locally resembles 
$\mathbb R^n$. But what does “locally” mean, and in what sense does it “look like” 
$\mathbb R^n$? And what does \emph{smooth} stay for? A common misconception is to think of this as an approximation to a 
\emph{flat} space. In fact, it has nothing to do with the notion of 
\emph{curvature} that we will introduce later. Rather, it concerns the 
differential--topological structure: the possibility of assigning $n$--tuples in 
$\mathbb R^n$ that label points in open neighborhoods of the manifold in a way 
compatible with a smooth structure.  

To clarify what a smooth manifold is, let's start with saying what a manifold \textbf{is not}. Here are some common misconceptions:
\begin{itemize}
    \item[\textbf{X}] A smooth manifold is \textbf{not} locally “flat” in the geometric sense. It has nothing to do with the fact that we can locally approximate an embedded curved surface with its tangent plane. For instance, the surface of the Earth can be locally charted with planar 
    coordinates, but this has nothing to do with its curvature.
    
    \item[\textbf{X}] A smooth manifold is \textbf{not} locally a vector space in the linear sense.  
    Being locally homeomorphic to $\mathbb R^n$ only means that neighborhoods of 
    points can be labeled by $n$--tuples of real numbers. This is a topological concept, not a linear one. It does not endow the manifold with addition or scalar multiplication of points.  
    For instance, while $\mathbb R^n$ itself carries a natural vector space 
    structure, generic manifolds (like the sphere $S^n$ or the projective space 
    $\mathbb{RP}^n$) do not admit any compatible linear structure.
    
    \item[\textbf{X}] A smooth manifold is \textbf{not} simply an “embedded surface” in 
    some higher--dimensional Euclidean space. While many examples are realized 
    this way,\footnote{In fact, Whitney's embedding theorem guarantees that every smooth manifold of dimension $n>0$ admits a smooth embedding into $\mathbb R^{2n}$ and can therefore be identified with a smooth submanifold of that space. Such an embedding, however, is a choice and is not part of the definition of a smooth manifold.} the definition of a manifold is intrinsic and does not require 
    any ambient space.
\end{itemize}
Now, we can ask ourselves: \emph{what is a smooth manifold?}
\begin{definition}[Smooth manifold]
A \emph{smooth manifold} of dimension $n$ is a pair $(M,[\mathcal A])$ where:
\begin{itemize}
    \item $M$ is a topological space that is Hausdorff\footnote{A topological space $X$ is \emph{Hausdorff} (or $T_2$) if for any two distinct points 
    $x,y\in X$ there exist disjoint open neighborhoods $U\ni x$, $V\ni y$.} 
    and second countable;\footnote{A topological space $X$ is \emph{second countable} if its topology admits a countable basis, 
    i.e.~there exists a countable family of open sets $\{U_i\}_{i\in\mathbb N}$ such that every open set in $X$ 
    is a union of some of the $U_i$.}
    \item $[\mathcal A]$ is an equivalence class\footnote{Every pair $(U_i,\varphi_i)$ is called \emph{chart}, and two atlases $\mathcal A$ and $\mathcal B$ on $M$ are said to be \emph{equivalent} if every chart of $\mathcal A$ is \emph{smoothly compatible} with every chart of $\mathcal B$, that is, if all transition maps between charts of $\mathcal A$ and charts of $\mathcal B$ are smooth, and vice versa. Each equivalence class contains a unique maximal atlas, called a \emph{smooth structure} on $M$.} of atlases 
    $\mathcal A=\{(U_i,\varphi_i)\}_{i\in I}$ such that:
    \begin{itemize}
        \item the open sets $\{U_i\}_{i\in I}$ cover $M$;
        \item each $\varphi_i:U_i\to\varphi_i(U_i)\subset\mathbb R^n$ is a homeomorphism onto an open subset of $\mathbb R^n$;
        \item for every $i,j\in I$ with $U_i\cap U_j\neq\varnothing$, the transition maps
        \begin{align}
            \varphi_j\circ\varphi_i^{-1}:\ \varphi_i(U_i\cap U_j) \longrightarrow  \varphi_j(U_i\cap U_j)
        \end{align}
        are smooth ($C^\infty$).
    \end{itemize}
\end{itemize}
\end{definition}
We denote the unique maximal atlas in $[\mathcal A]$ by $\mathcal A_{\max}$; a chart of the smooth structure means a chart in $\mathcal A_{\max}$.
\begin{figure}[h!]
\centering
\includegraphics[width=0.7\textwidth]{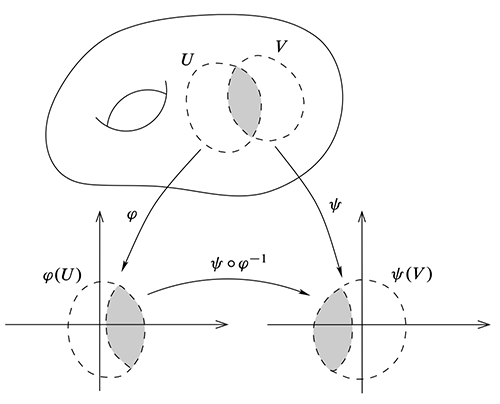}
\caption{Change of charts over two open subspaces $U,V\in M$. Source: https://xuan-li.github.io/post/properties-of-manifolds.}
\end{figure}
\begin{definition}[Smooth map]
    A map between smooth manifolds $F:M\to N$ is \emph{smooth} if, for all charts 
$(U,\varphi)\in\mathcal A_{\max}$ of $M$ and $(V,\psi)\in\mathcal B_{\max}$ of $N$, the coordinate 
representation
\begin{align}
    \psi\circ F\circ \varphi^{-1}:\ \varphi(U\cap F^{-1}(V)) \longrightarrow \psi(V)
\end{align}
is a smooth map between open subsets of $\mathbb R^n$ and $\mathbb R^m$. Moreover, a bijective smooth map $F:M\to N$ whose inverse $F^{-1}:N\to M$ is also smooth is called a \emph{diffeomorphism} of manifolds.
\end{definition}
\begin{remark}
    Notice that, by definition, the map between two charts $(U_i,\varphi_i)$ and $(U_j,\varphi_j)$ given by
    \begin{align}
        \varphi_j\circ\varphi_i^{-1}:\ \varphi_i(U_i\cap U_j) \longrightarrow  \varphi_j(U_i\cap U_j)
    \end{align}
    is a diffeomorphism of open subsets of $\mathbb R^n$ for all $i,j$.
\end{remark}
\begin{proposition}[Equivalent atlases agree on smooth maps]\label{propeqatl}
Let $(M,[\mathcal A])$ and $(N,[\mathcal B])$ be smooth manifolds of, respectively, dimensions $n$ and $n'$.  
Let $\mathcal A,\mathcal A'$ be two atlases in the equivalence class $[\mathcal A]$, 
and $\mathcal B,\mathcal B'$ two atlases in $[\mathcal B]$.  
Then, a map $F:M\to N$ is smooth with respect to $(\mathcal A,\mathcal B)$ 
if and only if it is smooth with respect to $(\mathcal A',\mathcal B')$.
\end{proposition}
\begin{proof}
Suppose first that $F:M\to N$ is smooth with respect to $(\mathcal A,\mathcal B)$.  
Take a point $p\in M$ and charts $(U',\varphi')\in\mathcal A'$, $(V',\psi')\in\mathcal B'$ 
with $p\in U'$ and $F(p)\in V'$.  
Choose also charts $(U,\varphi)\in\mathcal A$ and $(V,\psi)\in\mathcal B$ 
with $p\in U\cap U'$ and $F(p)\in V\cap V'$.\\
Set $W=U\cap U'\cap F^{-1}(V\cap V')$, an open neighborhood of $p$. On $\varphi'(W)\subset\mathbb R^n$ we can write
\begin{align}
    \psi'\circ F\circ (\varphi')^{-1}
    &= (\psi'\circ\psi^{-1}) 
       \circ (\psi\circ F\circ \varphi^{-1})
       \circ (\varphi\circ(\varphi')^{-1}).
\end{align}
Here:
\begin{itemize}
    \item $\psi\circ F\circ \varphi^{-1}$ is smooth by assumption;
    \item $\varphi\circ(\varphi')^{-1}$ and $\psi'\circ\psi^{-1}$ are smooth by compatibility 
    of the atlases.
\end{itemize}
Therefore, $\psi'\circ F\circ (\varphi')^{-1}$ is smooth, showing that $F$ is smooth 
with respect to $(\mathcal A',\mathcal B')$.\\
The reverse implication is identical, meaning that smoothness of maps does not depend on the representative atlases, but only on the smooth structures $[\mathcal A]$ and $[\mathcal B]$.
\end{proof}
\begin{remark}
In the definition of a smooth manifold one speaks of an equivalence class of atlases. 
The reason is that the relevant objects are not single atlases, but all atlases that give rise to the same smooth structure on $M$. The point is that, as shown in \cref{propeqatl}, the notion of “smoothness” for maps $F:M\to N$ does not depend on the particular choice of an atlas in $[\mathcal{A}]$, but only on the whole smooth structure of the manifold. In other words, in the definition of a smooth manifold, it would be meaningless to consider as different two atlases that agree on the smoothness of every map on $M$.\\
For example, on the $2$--sphere $S^2$, one can define an atlas $\mathcal A$ 
consisting of the two stereographic projection charts (see \cref{exmanfs}). On the other hand, a further atlas $\mathcal B$ can be obtained by adding other local charts, 
e.g. spherical coordinates restricted to suitable domains. 
Although $\mathcal A \neq \mathcal B$ as sets of charts, they are equivalent since the additional charts in $\mathcal B$ are compatible with the stereographic ones. 
In particular, both atlases generate the same maximal atlas, and thus define the same smooth structure on $S^2$.
\end{remark}
\begin{definition}[Coordinates]
Let $(M,[\mathcal A])$ be a smooth manifold of dimension $n$, and let 
$(U,\varphi)\in\mathcal A_{\max}$ be a chart.  
Writing
\begin{align}
    \varphi(q) = \big(x^1(q),\dots,x^n(q)\big), \qquad q\in U,
\end{align}
the smooth functions
\begin{align}
    x^i:\ U \longrightarrow \mathbb R, \qquad x^i(q)=(p_i\circ\varphi)(q),
\end{align}
where $p_i$ is the natural projection to the $i$-th component, are called the \emph{coordinate functions} associated with the chart $(U,\varphi)$.  
The collection $\{x^1,\dots,x^n\}$ is called a \emph{coordinate system} on $U$.\footnote{In the literature one finds sometimes that the definition of coordinates simply coincides with a the one of a chart.}
\end{definition}
\begin{example}[Examples of smooth manifolds]\label{exmanfs}
\leavevmode\newline
\vspace{-0.55cm}
\begin{enumerate}
    \item \textbf{Euclidean space.}  
    $\mathbb R^n$ with its standard topology is a smooth manifold of dimension $n$.  
    The atlas consists of the single global chart
    \begin{align}
        (\mathbb R^n, \ \mathrm{id}_{\mathbb R^n}),
    \end{align}
    whose maximal extension is the canonical smooth structure on $\mathbb R^n$.
    \item \textbf{The circle $S^1$.}  
    Consider $S^1=\{(x,y)\in\mathbb R^2: x^2+y^2=1\}$ with the subspace topology.  
    An explicit smooth atlas is given by the two charts
    \begin{align}
        U_1 &= S^1\setminus\{(0,1)\}, 
        &\varphi_1(x,y)&=\frac{x}{1-y}, \\[0.5em]
        U_2 &= S^1\setminus\{(0,-1)\}, 
        &\varphi_2(x,y)&=\frac{x}{1+y},
    \end{align}
    which map $U_1,U_2$ homeomorphically onto open subsets of $\mathbb R$.  
    The transition map $\varphi_2\circ\varphi_1^{-1}:\mathbb R\setminus\{0\}\to\mathbb R\setminus\{0\}$ 
    is smooth, hence this defines a smooth structure on $S^1$.
    \item \textbf{The $n$-sphere $S^n$.}  
    Let $S^n=\{x=(x_1,\dots,x_{n+1})\in\mathbb R^{n+1}: \|x\|=1\}$ with the subspace topology.  
    Define the open sets $U_\pm = S^n\setminus\{(0,\dots,0,\pm 1)\}$.  
    The stereographic projections from the north and south poles are the maps
    \begin{align}
        \varphi_+: U_+ \ \longrightarrow\ \mathbb R^n, 
        &\qquad \varphi_+(x)=\frac{1}{1-x_{n+1}}(x_1,\dots,x_n), \\[0.5em]
        \varphi_-: U_- \ \longrightarrow\ \mathbb R^n, 
        &\qquad \varphi_-(x)=\frac{1}{1+x_{n+1}}(x_1,\dots,x_n).
    \end{align}
    These are smooth homeomorphisms onto $\mathbb R^n$.  
    On the overlap $U_+\cap U_-$ the transition map 
    $\varphi_+\circ\varphi_-^{-1}:\mathbb R^n\setminus\{0\}\to \mathbb R^n\setminus\{0\}$ is given by
    \begin{align}
        y \longmapsto \frac{y}{\|y\|^2},
    \end{align}
    which is smooth.  
    Thus $\{(U_+,\varphi_+),(U_-,\varphi_-)\}$ is an atlas inducing the standard smooth structure on $S^n$.
    \item \textbf{The torus $\mathbb T^2$.}  
    Define the torus as the quotient
    \begin{align}
        \mathbb T^2 = \mathbb R^2 / \mathbb Z^2,
    \end{align}
    where $(m,n)\in\mathbb Z^2$ acts on $\mathbb R^2$ by integer translations $(x,y)\mapsto(x+m,y+n)$.  
    The quotient inherits a natural topology making it compact and Hausdorff.  
    Charts are given by projecting open sets in $\mathbb R^2$ which inject into $\mathbb T^2$, for example 
    \begin{align}
        U = (0,1)\times(0,1) \subset \mathbb R^2, 
        \qquad \pi: U \to \pi(U)\subset \mathbb T^2,
    \end{align}
    where $\pi$ is the quotient map.  
    The induced maps $\pi(U)\to U\subset\mathbb R^2$ are homeomorphisms ($\pi$ forms a universal covering), and the transition maps 
    are smooth by construction.  
    This yields the standard smooth structure on $\mathbb T^2$.
    \item \textbf{The real projective space $\mathbb{RP}^n$.}
    Define
    \begin{align}
    \mathbb{RP}^n = \big(\mathbb R^{n+1}\setminus\{0\}\big)/\!\sim,
    \end{align}
    where $x\sim y$ iff $x=\lambda y$ for some $\lambda\in\mathbb R^\times$.
    Equivalently, points of $\mathbb{RP}^n$ are 1--dimensional linear subspaces of $\mathbb R^{n+1}$. For each $i=0,\dots,n$, let
    \begin{align}
        U_i = \{[x_0:\dots:x_n]\in\mathbb{RP}^n : x_i\neq 0\},
    \end{align}
    and define the chart
    \begin{align}
        \varphi_i: U_i \longrightarrow \mathbb R^n, \qquad
        [x_0:\dots:x_n] \longmapsto \left(\frac{x_0}{x_i},\dots,\widehat{\tfrac{x_i}{x_i}},\dots,\frac{x_n}{x_i}\right),
    \end{align}
    where the hat indicates omission of the $i$-th coordinate.\footnote{Equivalently, we scale the homogeneous  coordinates so that $x_i=1$.}
    Each $\varphi_i$ is a homeomorphism onto $\mathbb R^n$.
    On the overlap $U_i\cap U_j$ (so $x_i\neq 0$ and $x_j\neq 0$),
    write
    \begin{align}
        \varphi_i([x])=(y_0,\dots,\widehat{y_i},\dots,y_n), \qquad y_k:=\frac{x_k}{x_i}\ (k\neq i),
    \end{align}
    so in particular $y_j=\frac{x_j}{x_i}\neq 0$.
    Then,
    \begin{align}
    \varphi_j\circ\varphi_i^{-1}(y_0,\dots,\widehat{y_i},\dots,y_n)
    = (z_0,\dots,\widehat{z_j},\dots,z_n),
    \end{align}
    where
    \begin{align}
        z_i=\frac{1}{y_j},
        \qquad
        z_k=\frac{y_k}{y_j}\quad (k\neq i,j).
    \end{align}
    Since these transition maps are rational functions with nonzero denominators, they are smooth. Thus, the charts $\{(U_i,\varphi_i)\}_{i=0}^n$ defines a smooth atlas on $\mathbb{RP}^n$.
    \item \textbf{An affine space $A$.}  
    Let $A$ be an affine space modeled on a finite--dimensional real vector space 
    $V\cong\mathbb R^n$.  
    Fix a point $a_0\in A$ and define the map
    \begin{align}
        \varphi_{a_0}: A \ \longrightarrow\ V,\qquad 
        \varphi_{a_0}(a) := a-a_0\in V.
    \end{align}
    This is a bijection sending $a_0$ to the origin of $V$.  
    Endow $A$ with the topology and smooth structure induced from $V$ via 
    $\varphi_{a_0}$. If $a_1\in A$ is another reference point, the corresponding chart 
    $\varphi_{a_1}:A\to V$ is related to $\varphi_{a_0}$ by the transition map
    \begin{align}
        \varphi_{a_1}\circ\varphi_{a_0}^{-1}(v)  =  v+(a_0-a_1),
        \qquad v\in V,
    \end{align}
    which is a translation on $V$, hence smooth. Thus the structure $\{(A,\varphi_{a_0}) : a_0\in A\}$ defines an atlas on $A$, and every affine space modeled on $V\cong\mathbb R^n$ carries a canonical smooth structure diffeomorphic to $\mathbb R^n$.\footnote{Notice that a single choice of $a_0$ already suffices. That is, the single atlas $\{(A,\varphi_{a_0})\}$ is enough to define the smooth structure on $A$, since all other charts defined via a different $a_0\in A$ have been shown above to be compatible. Hence, although the choice of base point $a_0$ is arbitrary (the identification $A\cong V\cong\mathbb R^n$ depends on $a_0$), the induced smooth structure on $A$ is canonical.}\\
    Notice also that this structure is compatible with the concept of coordinates and reference frame we provided in \cref{defmink}.
\end{enumerate}
\end{example}
What happens instead when we have the same underlying topological structure? Can we find inequivalent smooth manifolds?
\begin{example}[Examples of different smooth structures]
\leavevmode\newline
\vspace{-0.55cm}
\begin{enumerate}
    \item \textbf{Exotic $\mathbb R^4$.}  
    The topological space $\mathbb R^4$ admits \emph{uncountably many} mutually non-diffeomorphic smooth structures (the exotic $\mathbb R^4$'s).  
    These arise by combining Freedman’s topological results in dimension $4$ with smooth obstructions from Donaldson theory and explicit constructions (e.g. via Casson handles and Gompf’s techniques).  
    Each exotic $\mathbb R^4$ is homeomorphic to the standard $\mathbb R^4$ but not diffeomorphic to it.

    \item \textbf{Exotic spheres.}  
    For $n=7$ there are exactly $28$ distinct smooth structures on $S^7$, up to orientation-preserving diffeomorphism (Milnor; Kervaire–Milnor, $\Theta_7\cong\mathbb Z/28\mathbb Z$).  
    More generally, many $S^n$ with $n\ge 7$ admit exotic smooth structures, whereas $S^1,S^2,S^3$ have a unique smooth structure.  
    For $S^4$ it remains open whether exotic structures exist.

    \item \textbf{Closed $4$–manifolds: homeomorphic but not diffeomorphic.}  
    There are infinitely many closed simply connected $4$–manifolds that are homeomorphic but not diffeomorphic.  
    A classical family is given by the Dolgachev surfaces $E(1)_{p,q}$, obtained from $E(1)=\mathbb{CP}^2\#9\,\overline{\mathbb{CP}}^{\,2}$ by logarithmic transforms along multiple fibers of coprime multiplicities $p,q$: these are all homeomorphic to $E(1)$ (by Freedman) but pairwise non-diffeomorphic and not diffeomorphic to $E(1)$ (distinguished, for instance, by Seiberg–Witten invariants).  
    Another example is the Barlow surface, homeomorphic to $\mathbb{CP}^2\#8\,\overline{\mathbb{CP}}^{\,2}$ but not diffeomorphic to it.

    \item \textbf{Exotic tori in high dimensions.}  
    In sufficiently high dimensions (e.g. $n\ge 5$), there exist manifolds homeomorphic to the standard torus $T^n$ but not diffeomorphic to it (“exotic tori”), obtained via surgery theory; thus the topological type does not determine the smooth structure.
\end{enumerate}
All these examples share the same topological manifolds. In particular, for dimensions $\le 3$ the smooth structure is unique; in dimension $4$ the landscape is especially rich; in higher dimensions there are many non-equivalent smooth structures (including exotic spheres and tori).
\end{example}
\begin{definition}[Submanifold]
Let $(M,[\mathcal A])$ be a smooth manifold of dimension $n$.  
A subset $S\subseteq M$ is called a \emph{smooth submanifold of dimension $k$} 
if for every $p\in S$ there exists a chart $(U,\varphi)\in\mathcal A_{\max}$ with 
$p\in U$ such that
\begin{align}
    \varphi(U\cap S)  =  \{(x^1,\dots,x^n)\in \varphi(U): x^{k+1}=\cdots=x^n=0\}.
\end{align}
In other words, near each point $p\in S$, the set $S$ looks like 
$\mathbb R^k\times\{0\}^{\,n-k}\subset\mathbb R^n$ in suitable coordinates.  
The integer $k$ is called the \emph{dimension} of $S$.
\end{definition}
\begin{definition}[Germ of a smooth function]
Let $(M,[\mathcal A])$ be a smooth manifold and $p\in M$. 
Define an equivalence relation $\sim_p$ on smooth functions defined on open neighborhoods of $p$ by
\begin{align}
    f \sim_p g 
    \ \Longleftrightarrow\ 
    \exists\,U\subset\operatorname{dom}(f)\cap\operatorname{dom}(g)\text{ open with }p\in U\text{ such that }f|_U=g|_U.
\end{align}
The \emph{germ} of $f$ at $p$ is the equivalence class $[f]_p$, and the set
\begin{align}
    C^\infty_p(M) := \{[f]_p:\ f\in C^\infty(U),\ U\subset M\text{ open},\ p\in U\}
\end{align}
is an $\mathbb R$--algebra (with addition and multiplication induced pointwise).
The evaluation at $p$ is defined by
\begin{align}
    \mathrm{ev}_p:\ C^\infty_p(M)\longrightarrow \mathbb R,\qquad \mathrm{ev}_p([f]_p)=f(p).
\end{align}
\end{definition}
\begin{remark}
    By abuse of notation, we write $f\in C^\infty_p(M)$ for the germ of $f$ at $p$.
\end{remark}
\begin{definition}[Tangent space]
Let $(M,[\mathcal A])$ be a smooth manifold and $p\in M$.  
The \emph{tangent space} of $M$ at $p$, denoted $T_pM$, is the real vector space of 
$\mathbb R$--linear maps
\begin{align}
    v:\ C^\infty_p(M)\ \longrightarrow\ \mathbb R
\end{align}
satisfying the Leibniz rule
\begin{align}
    v(fg) = f(p)\,v(g) + g(p)\,v(f),
    \qquad \forall\ f,g\in C^\infty_p(M).
\end{align}
\end{definition}
\begin{remark}[As equivalence of curves]
Note that $T_pM$ can be defined as the space of equivalence classes of smooth curves 
$\gamma:(-\varepsilon,\varepsilon)\to M$ with $\gamma(0)=p$, where $\gamma_1\sim\gamma_2$ iff 
for one (hence any) chart $(U,\varphi)$ with $p\in U$ we have
\begin{align}
    \left.\frac{d}{dt}\right|_{t=0}\big(\varphi\circ\gamma_1(t)\big)
     = 
    \left.\frac{d}{dt}\right|_{t=0}\big(\varphi\circ\gamma_2(t)\big)\ \in \mathbb R^n.
\end{align}
The correspondence is given by
\begin{align}
    v_\gamma(f) := \left.\frac{d}{dt}\right|_{t=0}\big(f\circ\gamma\big)(t),
\end{align}
which defines a derivation at $p$ and yields a vector space isomorphism with the definition above.
\end{remark}
\begin{proposition}[Coordinate basis of the tangent space]\label{coordTpM}
Let $(M,[\mathcal A])$ be a smooth manifold of dimension $n$, $p\in M$, and let 
$(U,\varphi)\in\mathcal A_{\max}$ be a chart with associated coordinate functions
$x^1,\dots,x^n:U\to\mathbb R$. For each $i=1,\dots,n$, we define the derivation\footnote{In order to use a lighter notation, we write $\partial_i$ instead of $\tfrac{\partial}{\partial x^i}|_p$ when the base point $p$ is clear from the context. We also use the same symbol $x^i$ both for the coordinate function $x^i:U\to\mathbb R$ and for its evaluation $x^i(q)$ at a point $q\in U$, so that $x=(x^1,\dots,x^n)=\varphi(q)\in\mathbb R^n$.}
\begin{align}
    \partial_i : C^\infty_p(M)\longrightarrow \mathbb R,
    \qquad
    \partial_i(f) := 
    \left.\frac{\partial}{\partial x^i}\right|_{\varphi(p)} \big(f\circ\varphi^{-1}\big)(x).
\end{align}
Then, we have that
\begin{enumerate}[label=\roman*)]
\item Each $\partial_i$ is an element of $T_pM$, hence a derivation at $p$,
\item The family $(\partial_1,\dots,\partial_n)$ is a basis of $T_pM$ (the coordinate basis),
\item Every $v\in T_pM$ has a unique expansion
\begin{align}
    v = v^i\,\partial_i,
    \qquad
    v^i = v(x^i),
\end{align}
were the coefficients $v^i$ are called the \emph{components} of $v$ in the coordinate basis $(\partial_i)$, and
\item If $(V,\psi)$ is another chart around $p$ with coordinate functions 
$x^{\prime 1},\dots,x^{\prime n}$, then the components transform by the Jacobian rule
\begin{align}
    v^{\prime i} = v(x^{\prime i}) =\frac{\partial x^{\prime i}}{\partial x^j}\,v^j,
\end{align}
where we omitted the dependence on $p$ in the Jacobian.\footnote{I.e., we intend $\frac{\partial x^{\prime i}}{\partial x^j}= \frac{\partial x^{\prime i}}{\partial x^j}(p)$.}
\end{enumerate}
\end{proposition}
\begin{proof}
In order to prove (i), by definition, we have
\begin{align}
    \partial_i(f)
    &= \left.\frac{\partial}{\partial x^i}\right|_{\varphi(p)} (f\circ\varphi^{-1})(x).
\end{align}
Linearity is immediate. For $f,g\in C^\infty_p(M)$, using Leibniz in coordinates,
\begin{align}
    \partial_i(fg)
    &= \left.\frac{\partial}{\partial x^i}\right|_{\varphi(p)}
       \big((f\circ\varphi^{-1})(x)\,(g\circ\varphi^{-1})(x)\big)\\
    &= (f\circ\varphi^{-1})(\varphi(p))\,
       \left.\frac{\partial}{\partial x^i}\right|_{\varphi(p)}(g\circ\varphi^{-1})(x)
       + (g\circ\varphi^{-1})(\varphi(p))\,
       \left.\frac{\partial}{\partial x^i}\right|_{\varphi(p)}(f\circ\varphi^{-1})(x)\\
    &= f(p)\,\partial_i(g) + g(p)\,\partial_i(f),
\end{align}
so $\partial_i$ satisfies the Leibniz rule and hence $\partial_i\in T_pM$.

For (ii), we notice that, for linear independence, if $a^i\partial_i=0$, applying to $x^k$ gives
\begin{align}
    0
    &= a^i\,\partial_i(x^k)
     = a^i\,\left.\frac{\partial}{\partial x^i}\right|_{\varphi(p)}(x^k\circ\varphi^{-1})(x)
     = a^i\,\delta_i^{\ k}
     = a^k,
\end{align}
thus all $a^i=0$. For spanning, let $v\in T_pM$. Leibniz gives $v(1)=2v(1)$, hence $v$ vanishes on constants. Shrink the coordinate image to a convex neighborhood of $a=\varphi(p)$. For a representative $f$, set $\hat f=f\circ\varphi^{-1}$ and
\begin{align}
g_i(q)=\int_0^1\partial_i\hat f\big(a+t(\varphi(q)-a)\big)\,dt.
\end{align}
The fundamental theorem of calculus gives
\begin{align}
f=f(p)+(x^i-x^i(p))g_i,\qquad g_i(p)=\partial_i(f).
\end{align}
Applying $v$ and Leibniz therefore yields
\begin{align}
v(f)=v(x^i)g_i(p)=v(x^i)\partial_i(f),
\end{align}
hence $v=v(x^i)\partial_i$.

(iii) is given by uniqueness of the expansion that follows from (ii). Namely, the coefficients are
\begin{align}
    v^i = v(x^i).
\end{align}
Whereas, (iv) follows from the fact that, on $U\cap V$, if $x^{\prime i}=x^{\prime i}(x^1,\dots,x^n)$, then
\begin{align}
    \partial'_i=\frac{\partial}{\partial x'^i}=\frac{\partial x^j}{\partial x'^i}\frac{\partial}{\partial x^j}=\frac{\partial x^j}{\partial x'^i}\partial_j,
\end{align}
by definition of derivative, and hence
\begin{align}
    v^{\prime i}= \frac{\partial x^{\prime i}}{\partial x^j}\,v^j,
\end{align}
since $v=v^i\partial_i=v'^i\partial'_i$. This is exactly the Jacobian transformation law.
\end{proof}
\begin{definition}[Cotangent space]\label{defcotspace}
Given the tangent space $T_pM$ at $p\in M$, the \emph{cotangent space} is its linear dual\footnote{$\mathrm{Hom}_\mathbb R(T_pM,\mathbb R)$ is a way to call the space of $\mathbb R$-linear maps $T_pM\to\mathbb R$.}
\begin{align}
    T_p^*M := \mathrm{Hom}_\mathbb R(T_pM,\mathbb R).
\end{align}
Equivalently, there is a canonical $\mathbb R$--linear map
\begin{align}
    d:\ C^\infty_p(M)\ \longrightarrow\ T_p^*M,
    \qquad f\ \longmapsto\ df,
\end{align}
characterized by
\begin{align}
    d(f+g)&=df+dg, \\
     d(\lambda f)&=\lambda\,df,\\
     d(fg)&=f(p)\,dg+g(p)\,df,
\end{align}
and by the pairing
\begin{align}
    \langle df,\,v\rangle  :=  v(f), \qquad v\in T_pM.
\end{align}
In local coordinates $(x^1,\dots,x^n)$ around $p$, the cotangent space $T_p^*M$ has a distinguished basis
\begin{align}
    dx^1,\dots,dx^n \ \in T_p^*M,
\end{align}
defined by $dx^i:=d(x^i)$.  
This basis is dual to the coordinate basis of $T_pM$, namely\footnote{By definition of the pairing, $\langle dx^i,\partial_j\rangle=\partial_j(x^i)$. But $\partial_j$ acts as $\tfrac{\partial}{\partial x^j}$ at $p$, so 
    $\partial_j(x^i)=\tfrac{\partial x^i}{\partial x^j}(p)=\delta^i_{\ j}$. This is precisely the dual condition for the pairing: for a basis $e_i$ of a vector space and its dual $e^i$, one has $\langle e^i,\,e_j\rangle=\delta^i_{\ j}$.}
\begin{align}
    dx^i(\partial_j) = \delta^i_{\ j}.
\end{align}
\end{definition}
\begin{remark}
The map 
\begin{align}
    d\colon C^\infty_p(M) \longrightarrow T_p^*M, \qquad f\mapsto df,
\end{align}
is called the \emph{differential} of $f$ at $p$.  
Given a smooth function $f$, 
its differential $df$ is a covector at $p$ which acts on tangent vectors by
\begin{align}
    df(v) = v(f), \qquad v\in T_pM.
\end{align}
In local coordinates $(x^1,\dots,x^n)$, this reads
\begin{align}
    df =\partial_i f\,dx^i.
\end{align}
Thus, $df$ is also called the \emph{differential of $f$ at $p$}.
\end{remark}
\begin{definition}[Pushforward and pull-back]\label{pushpullfibers}
Let $F:M\to N$ be a smooth map between smooth manifolds, and let $p\in M$.  
The \emph{pushforward} of $F$ at $p$ is defined as the linear map\footnote{Notice that the pushforward is often also denoted as $dF_p$.}
\begin{align}
    F_*|_p:\ T_pM \longrightarrow T_{F(p)}N
\end{align}
such that
\begin{align}
    (F_*|_p v)(f)  :=  v(f\circ F),
    \qquad f\in C^\infty_{F(p)}(N),\ v\in T_pM.
\end{align}
Whereas, the \emph{pull-back} of covectors is the linear map
\begin{align}
    F^*|_p:\ T_{F(p)}^*N \longrightarrow T_p^*M
\end{align}
defined by
\begin{align}
    (F^*|_p\alpha)(v)  :=  \alpha(F_*|_p v),
    \qquad \alpha\in T_{F(p)}^*N,\ v\in T_pM.
\end{align}
\end{definition}
\begin{proposition}[Identification of $T_pM$ with $\mathbb R^n$]\label{idTpM}
Let $(U,\varphi)$ be a chart around $p\in M$ with coordinates 
$(x^1,\dots,x^n)$.  
Then, the linear map\footnote{Notice that this map can be seen as a pushforward.}
\begin{align}
    d\varphi_p:\ T_pM \longrightarrow \mathbb R^n,
    \qquad
    v \longmapsto (v^1,\dots,v^n),
\end{align}
with $v^i=v(x^i)$ as defined in \cref{coordTpM}, is a vector space isomorphism.\\
Moreover, we have
\begin{align}
    d\varphi_p(\partial_i) = e_i,
\end{align}
where $(\partial_1,\dots,\partial_n)$ is the coordinate basis of $T_pM$ and 
$(e_1,\dots,e_n)$ is the canonical basis of $\mathbb R^n$.
\end{proposition}
\begin{proof}
Linearity is clear. Furthermore, if $v\in T_pM$ and $d\varphi_p(v)=0$, then $v(x^i)=0$ for all $i$. Hence, $d\varphi_p$ is injective. Since $\dim T_pM=\dim\mathbb R^n=n$, it is also surjective, so $d\varphi_p$ is an isomorphism.\\
Finally,
\begin{align}
    d\varphi_p(\partial_i)
    = \big(\partial_i(x^1),\dots,\partial_i(x^n)\big)
    = (\delta_i^{\ 1},\dots,\delta_i^{\ n})
    = e_i,
\end{align}
as claimed.
\end{proof}
\begin{remark}
Via a chart $(U,\varphi)$ around $p\in M$, the differential $d\varphi_p$ identifies 
$T_pM$ with $\mathbb R^n$ by sending the coordinate basis $(\partial_i)$ to the canonical basis $(e_i)$, i.e.
\begin{align}
    d\varphi_p(\partial_i) = e_i.
\end{align}
This implies that the pairing with the coordinate functions identifies $T_p^*M$ with 
$(\mathbb R^n)^*$, where $(dx^i)$ corresponds to the canonical dual basis $(e^i)$, because of
\begin{align}
    dx^i(\partial_j) = \langle dx^i,\partial_j\rangle=\delta^i_{\ j}.
\end{align}
However, this does \emph{not} provide a canonical isomorphism 
\begin{align}
    T_pM  \cong  T_p^*M,
\end{align}
since the missing step would be a canonical identification 
$\mathbb R^n \cong (\mathbb R^n)^*$. Such an isomorphism requires an additional choice, for instance a basis or a nondegenerate bilinear form.
What \emph{is} canonical is the identification
\begin{align}
    T_pM  \cong  T_p^{**}M, 
    \qquad 
    v \longmapsto (\alpha \mapsto \alpha(v)),
\end{align}
which is natural, namely independent of the choice of the basis.
\end{remark}

\section{Multilinear algebra and tensors}
\begin{definition}[Tensor product]
Let $V_1,\dots,V_k$ be finite--dimensional\footnote{Notice that this construction holds also for the infinite dimensional case. We just need the finite one for our later purposes.} real vector spaces.  
The \emph{tensor product} $V_1\otimes\dots\otimes V_k$ is a real vector space equipped with a multilinear map
\begin{align}
\otimes:  V_1\times\dots\times V_k \longrightarrow V_1\otimes\dots\otimes V_k,
\qquad (v_1,\dots,v_k)\longmapsto v_1\otimes\dots\otimes v_k,
\end{align}
such that for any multilinear map $B: V_1\times\dots\times V_k\to U$ into a vector space $U$ there exists a unique linear map 
\begin{align}
\tilde B:  V_1\otimes\dots\otimes V_k \longrightarrow U
\end{align}
with 
\begin{align}
\tilde B\circ \otimes  =  B.
\end{align}
Equivalently, the following diagram commutes:
$$
\begin{tikzcd}
V_1 \times \dots \times V_k \arrow[r, "\otimes"] \arrow[dr, "B"'] 
  & V_1 \otimes \dots \otimes V_k \arrow[d, dashed, "\exists!\,\tilde B"] \\
  & U.
\end{tikzcd}
$$
\end{definition}
\begin{remark}
The tensor product $V_1\otimes\dots\otimes V_k$ is the vector space that \emph{linearly represents} multilinear maps on $V_1\times\dots\times V_k$. By definition, the canonical map
\begin{align}
    \otimes:  V_1\times\dots\times V_k \longrightarrow V_1\otimes\dots\otimes V_k, 
    \qquad (v_1,\dots,v_k)\longmapsto v_1\otimes\dots\otimes v_k
\end{align}
collects all the “multilinear information” of $V_1\times\dots\times V_k$ into a single vector space: every multilinear map factors uniquely through it as a linear map. Therefore, the tensor product is itself a vector space, which linearly represents multilinear maps $V_1\times\dots\times V_k \to U$ as linear maps $V_1\otimes\dots\otimes V_k \to U$.\\
Moreover, the tensor product space itself can be realized in different ways. Formally, a \emph{model of the tensor product} is a pair $(T,\tau)$ consisting of a vector space $T$ together with a multilinear map $\tau:V_1\times\dots\times V_k\to T$ that satisfies the universal property. If $(T,\tau)$ and $(T',\tau')$ are two such models, then there exists a unique isomorphism $f:T\to T'$ with\footnote{Indeed, 
since $\tau':V_1\times\dots\times V_k\to T'$ is multilinear, the universal property of $(T,\tau)$ yields a unique 
linear $f:T\to T'$ with $f(\tau(v_1,\dots,v_k))=\tau'(v_1,\dots,v_k)$. Similarly, there is a unique 
$g:T'\to T$ with $g(\tau'(v_1,\dots,v_k))=\tau(v_1,\dots,v_k)$. One checks that $g\circ f=\mathrm{id}_T$ and 
$f\circ g=\mathrm{id}_{T'}$, hence $f$ is an isomorphism.}
\begin{align}
f(\tau(v_1,\dots,v_k))  =  \tau'(v_1,\dots,v_k).
\end{align}
In particular, for a fixed multilinear map $B:V_1\times\dots\times V_k\to U$, the corresponding factorizations $\tilde B_T:T\to U$ and $\tilde B_{T'}:T'\to U$ are compatible with this canonical isomorphism, i.e.
\begin{align}
\tilde B_{T'} \circ f  =  \tilde B_T.
\end{align}
Hence, different constructions of the tensor product yield the same universal factorization.
\end{remark}
\begin{example}[Different constructions of the tensor product]
For simplicity, suppose $k=2$. There are several equivalent ways to construct the tensor product $V\otimes W$, starting from two finite dimensional vector spaces $V$ and $W$.
We list three standard models, each with a concrete example for $V=W=\mathbb R^2$.

\begin{enumerate}
    \item \textbf{Generators and relations.}  
    Start from the free vector space\footnote{Given set $X$, the \emph{free vector space} $\mathbb R^{(X)}$ is the space of finitely supported functions $f:X\to\mathbb R$, i.e.\ those with $f(x)\neq 0$ only for finitely many $x\in X$. A canonical basis is given by the characteristic functions $e_x$ with $e_x(y)=\delta_{xy}$, for $x,y\in X$, so that every $f\in\mathbb R^{(X)}$ can be written uniquely as $\sum_{x\in X} f(x)\,e_x$.} $\mathbb R^{(V\times W)}$ generated by pairs $(v,w)$, and impose the relations
    \begin{align}
        (v_1+v_2,w) &\sim (v_1,w)+(v_2,w), \\
        (v,w_1+w_2) &\sim (v,w_1)+(v,w_2), \\
        (\lambda v,w) &\sim \lambda (v,w), \\
        (v,\lambda w) &\sim \lambda (v,w).
    \end{align}
    The resulting quotient space is $V\otimes W$, with canonical map 
    $(v,w)\mapsto v\otimes w$, which is given by the quotient map:
    $$
    \begin{tikzcd}
    V\times W \arrow[r, hook, "\iota"] \arrow[dr, "\otimes"'] 
    & \mathbb R^{(V\times W)} \arrow[d, twoheadrightarrow, "q"] \\
    & V\otimes W.
    \end{tikzcd}
    $$

    For example, let $V=W=\mathbb R^2$ with bases 
    $e_1=(1,0), e_2=(0,1)$ and $f_1=(1,0), f_2=(0,1)$. For $v=(1,2)$ and $w=(3,4)$,
    \begin{align}
        [(v,w)] 
        &= [(e_1+2e_2,  3f_1+4f_2)] \\
        &\sim [(e_1, 3f_1+4f_2)] + 2[(e_2, 3f_1+4f_2)] \\
        &\sim 3[(e_1,f_1)] + 4[(e_1,f_2)] + 6[(e_2,f_1)] + 8[(e_2,f_2)].
    \end{align}
    Passing to the quotient notation, this is
    \begin{align}
        (1,2)\otimes(3,4)
        &= 3\, (e_1\!\otimes f_1) + 4\, (e_1\!\otimes f_2) \nonumber\\
        &\quad + 6\, (e_2\!\otimes f_1) + 8\, (e_2\!\otimes f_2).
    \end{align}

    \item \textbf{Dual of bilinear forms.}  
    Consider the space $\mathrm{Bil}(V\times W,\mathbb R)$ of bilinear forms with real values, 
    and define
    \begin{align}
        V\otimes W  :=  \mathrm{Bil}(V\times W,\mathbb R)^*.
    \end{align}
    The canonical map 
    \begin{align}
    (v,w)  \longmapsto  \big(B \mapsto B(v,w)\big)
    \end{align}
    identifies pairs with evaluation functionals. In particular,
    \begin{align}
    v\otimes w &:  B \longmapsto B(v,w),
    \end{align}
    defines the evaluation functional. Namely, for every bilinear map $B$, we have the canonical linear map
    \begin{align}
    \tilde B:  V\otimes W \longrightarrow U,
    \end{align}
    which is defined by linear extension from its action on elementary tensors, i.e.
    \begin{align}
    \tilde B(v\otimes w) &= B(v,w).
    \end{align}
    This description is canonical within the dual model, since the evaluation maps
    $v\otimes w\colon B\mapsto B(v,w)$ is uniquely determined once we realize 
    $V\otimes W$ as the dual of $\mathrm{Bil}(V\times W,\mathbb R)$.\\
    For instance, with $V=W=\mathbb R^2$ and 
    $B(v,w)=v_1w_1+2v_2w_2$, we have
    \begin{align}
        (e_1\otimes f_2)(B) &= B(e_1,f_2)=0, \\
        (e_2\otimes f_2)(B) &= B(e_2,f_2)=2.
    \end{align}

    \item \textbf{Concrete basis construction.}  
    Suppose you already have the bilinear map $\otimes$ (e.g. as given by one of the cases above). If $\dim V=n$ with basis $\{e_i\}$ and $\dim W=m$ with basis $\{f_j\}$, 
    then $V\otimes W$ is the $nm$-dimensional vector space with basis 
    $\{\,e_i\otimes f_j : 1\leq i\leq n, 1\leq j\leq m\,\}$.
    This shows $V\otimes W \cong \mathbb R^{nm}$ in a natural way. Notice that in the preceding cases we also built explicitly the map $\otimes:V\times W\to V\otimes W$. In this case, we view the construction as a realization of such a map with respect to some bases of $V$ and $W$.\\
    For $V=W=\mathbb R^2$,
    \begin{align}
        (1,2)\otimes(3,4)
        &= 3\, (e_1\!\otimes f_1) + 4\, (e_1\!\otimes f_2) \nonumber\\
        &\quad + 6\, (e_2\!\otimes f_1) + 8\, (e_2\!\otimes f_2).
    \end{align}
\end{enumerate}
All these constructions satisfy the same universal property and are therefore canonically isomorphic.
\end{example}
\begin{remark}
In the constructions above we implicitly focused on bilinear maps with values in $\mathbb R$. This already suffices, since bilinear maps with values in an arbitrary vector space $U$ can be obtained by linear extension: one has the natural identification\footnote{For $V$ and $W$ finite dimensional.}
\begin{align}
\mathrm{Bil}(V\times W, U)  \cong  \mathrm{Bil}(V\times W,\mathbb R)\otimes U.
\end{align}
Thus it is enough to construct $V\otimes W$ so that it represents real--valued bilinear forms; the universal property then guarantees that the same object works for every target space $U$. Of course, this hold also for the general multilinear case.
\end{remark}
\begin{definition}[Symmetric and alternating tensors]
Let $V$ be a finite--dimensional real vector space and $k\geq 1$.  
\begin{itemize}
\item The \emph{symmetric power} of $V$ is the quotient
\begin{align}
    \mathrm{Sym}^k(V) 
    := V^{\otimes k}\big/\big\langle 
        v_{\sigma(1)}\otimes\dots\otimes v_{\sigma(k)} 
        - v_1\otimes\dots\otimes v_k
     \big| \sigma\in S_k\big\rangle.
\end{align}
\item The \emph{alternating (wedge) power} of $V$ is the quotient
\begin{align}
    \textstyle{\bigwedge^k} V 
    := V^{\otimes k}\big/\big\langle 
        v_{\sigma(1)}\otimes\dots\otimes v_{\sigma(k)} 
        - \mathrm{sgn}(\sigma)\,v_1\otimes\dots\otimes v_k
     \big| \sigma\in S_k\big\rangle.
\end{align}
\end{itemize}
These quotient spaces are canonically isomorphic to the images of the linear operators
\begin{align}
    \mathrm{Sym}(T) 
    &:= \frac{1}{k!}\sum_{\sigma\in S_k} T\circ\sigma, \\[4pt]
    \mathrm{Alt}(T) 
    &:= \frac{1}{k!}\sum_{\sigma\in S_k} \mathrm{sgn}(\sigma)\,T\circ\sigma,
\end{align}
where, viewing $T\in V^{\otimes k}$ as a multilinear map on $V^*$, $(T\circ\sigma)(\alpha_1,\dots,\alpha_k):=T(\alpha_{\sigma(1)},\dots,\alpha_{\sigma(k)})$.  
These operators are linear, idempotent, and satisfy
\begin{align}
    \mathrm{Sym}^2=\mathrm{Sym}, \qquad 
    \mathrm{Alt}^2=\mathrm{Alt}, \qquad 
    \mathrm{Sym}\circ\mathrm{Alt}=0=\mathrm{Alt}\circ\mathrm{Sym}\,\,\,\, (k\geq 2).
\end{align}
The quotient maps restrict to isomorphisms from these images onto $\mathrm{Sym}^k(V)$ and $\textstyle{\bigwedge^k} V$, respectively. We use these identifications when decomposing tensors. For covariant tensors, apply the same construction to $V^*$.\\
In coordinates, if $T=T_{i_1\dots i_k}\,dx^{i_1}\otimes\dots\otimes dx^{i_k}$, then
\begin{align}
    \mathrm{Sym}(T)_{i_1\dots i_k}
    &= \frac{1}{k!}\sum_{\sigma\in S_k} 
       T_{\,i_{\sigma(1)}\dots i_{\sigma(k)}}, \\[4pt]
    \mathrm{Alt}(T)_{i_1\dots i_k}
    &= \frac{1}{k!}\sum_{\sigma\in S_k} 
       \mathrm{sgn}(\sigma)\,T_{\,i_{\sigma(1)}\dots i_{\sigma(k)}}.
\end{align}
Symmetrization means invariance of the components under any permutation of indices, 
while alternization means a change of sign under odd permutations.
\end{definition}
\begin{remark}
For $k=2$, there is a direct sum decomposition
\begin{align}
    V\otimes V  =  \mathrm{Sym}^2(V) \oplus \textstyle{\bigwedge^2} V,
\end{align}
so every tensor $T\in V\otimes V$ decomposes uniquely as
\begin{align}
    T  =  \mathrm{Sym}(T) + \mathrm{Alt}(T).
\end{align}
For covariant tensors, if $T=T_{ij}\,e^i\otimes e^j\in V^*\otimes V^*$, then
\begin{align}
    \mathrm{Sym}(T) &= \tfrac{1}{2}(T_{ij}+T_{ji})\,e^i\otimes e^j, \\
    \mathrm{Alt}(T) &= \tfrac{1}{2}(T_{ij}-T_{ji})\,e^i\otimes e^j,
\end{align}
so symmetry means $T_{ij}=T_{ji}$ and antisymmetry means $T_{ij}=-T_{ji}$.
\end{remark}
\begin{definition}[Tensors on tangent and cotangent spaces]
Let $p\in M$.  
For integers $r,s\geq 0$, an \emph{$(r,s)$--tensor at $p$} is an element of
\begin{align}
    T^{(r,s)}(M)_p  :=  \underbrace{T_pM\otimes \cdots \otimes T_pM}_{r \text{ times}}
     \otimes 
    \underbrace{T_p^*M\otimes \cdots \otimes T_p^*M}_{s \text{ times}}.
\end{align}
Equivalently, an $(r,s)$--tensor at $p$ can be seen as a multilinear map
\begin{align}
    \underbrace{T_p^*M\times \cdots \times T_p^*M}_{r\text{ times}}
    \times
    \underbrace{T_pM\times \cdots \times T_pM}_{s\text{ times}}
     \longrightarrow  \mathbb R.
\end{align}
If $(\partial_i)$ is the coordinate basis of $T_pM$ and $(dx^i)$ the dual basis of $T_p^*M$, 
then $\{\partial_{i_1}\otimes\dots\otimes\partial_{i_r}\otimes dx^{j_1}\otimes\dots\otimes dx^{j_s}\}$
is a basis of $T^{(r,s)}(M)_p$.\\
Thus, any tensor $T\in T^{(r,s)}(M)_p$ has components $T^{i_1\dots i_r}_{\phantom{i_1\dots i_r}j_1\dots j_s}\in\mathbb R$ defined by
\begin{align}
    T = T^{i_1\dots i_r}_{\phantom{i_1\dots i_r}j_1\dots j_s}\,
    \partial_{i_1}\otimes\dots\otimes\partial_{i_r}\otimes dx^{j_1}\otimes\dots\otimes dx^{j_s}.
\end{align}
\end{definition}
\begin{remark}
    Notice that we will often use the simple notation $(r,s)$--type to indicate a tensor in $T^{(r,s)}(M)_p$.
\end{remark}
\begin{proposition}[A tensor transforms like a tensor]
Let $(U,\varphi)$ and $(V,\psi)$ be two overlapping charts around $p\in M$, 
with coordinate functions $(x^1,\dots,x^n)$ and $(x^{\prime 1},\dots,x^{\prime n})$.  
Let $T\in T^{(r,s)}(M)_p$ and write
\begin{align}
    T = T^{i_1\dots i_r}_{\phantom{i_1\dots i_r}j_1\dots j_s}\,
        \partial_{i_1}\otimes\dots\otimes \partial_{i_r}
        \otimes dx^{j_1}\otimes\dots\otimes dx^{j_s}.
\end{align}
Then, under the coordinate change $x\mapsto x'$, the components transform as\footnote{Notice that, once again, we are omitting the dependence on $p$. I.e., we mean $\frac{\partial x^{\prime i}}{\partial x^j}=\frac{\partial x^{\prime i}}{\partial x^j}(p)$}
\begin{align}
    T^{\prime\,i_1\dots i_r}_{\phantom{i_1\dots i_r}j_1\dots j_s}
     = 
    \frac{\partial x^{\prime i_1}}{\partial x^{k_1}} \cdots
    \frac{\partial x^{\prime i_r}}{\partial x^{k_r}} \,
    \frac{\partial x^{\ell_1}}{\partial x^{\prime j_1}} \cdots
    \frac{\partial x^{\ell_s}}{\partial x^{\prime j_s}} \,
    T^{k_1\dots k_r}_{\phantom{k_1\dots k_r}\ell_1\dots \ell_s}.
\end{align}
That is, contravariant indices transform with the Jacobian of the coordinate change,
while covariant indices transform with the inverse Jacobian.
\end{proposition}
\begin{proof}
The tangent basis transforms by
\begin{align}
    \partial'_j = \frac{\partial x^i}{\partial x^{\prime j}} \,\partial_i,
\end{align}
and the cotangent basis by
\begin{align}
    dx^{\prime i} = \frac{\partial x^{\prime i}}{\partial x^j} \,dx^j.
\end{align}
Substituting these relations into the expansion of $T$ in the primed basis yields the stated formula.
\end{proof}
\begin{remark}[Covariant and contravariant indices]
Given a tensor $T\in T^{(r,s)}(M)_p$, the upper indices ($i_1,\dots,i_r$) 
are called \emph{contravariant}, while the lower indices ($j_1,\dots,j_s$) 
are called \emph{covariant}. Thus, the type $(r,s)$ indicates $r$ contravariant 
and $s$ covariant indices.\\
Notice also that our notation for the coordinate functions 
$x^i:U\to\mathbb R$, which are not tensors themselves and do not transform 
tensorially, is chosen to naturally match the transformation laws of the associated bases
\begin{align}
    \partial_i := \frac{\partial}{\partial x^i}\Big|_p,
    \qquad
    dx^i := d(x^i),
\end{align}
which transform covariantly and contravariantly, respectively.\footnote{This is 
consistent with our earlier notation $\partial_i\coloneqq\tfrac{\partial}{\partial x^i}|_p$, 
where the lower index reflects their covariant transformation law. One could in principle 
adopt the opposite convention, writing $x_i$, $\partial^i$, and $dx_i$, with the roles of 
upper and lower indices reversed. The choice $x^i$, $\partial_i$, $dx^i$ is simply the 
standard convention.} Therefore, writing $x^i$ with an upper index ensures that the 
placement of indices in $\partial_i$ and $dx^i$ is consistent with their 
transformation rules.\footnote{In any case, the assignment of index positions is conventional up to a 
global swap of all upper and lower indices. What matters is the overall consistency.}
\end{remark}
\begin{definition}[Operations on tensors]\label{opetensors}
Let $p\in M$ and let $T\in T^{(r,s)}(M)_p$ be an $(r,s)$--tensor at $p$.
\begin{enumerate}[label=\roman*)]

\item \emph{Tensor product.}  
If $S\in T^{(r,s)}(M)_p$ and $T\in T^{(r',s')}(M)_p$, their tensor product is the
$(r+r',s+s')$--tensor
\begin{align}
    &(S\otimes T)(\alpha^1,\dots,\alpha^{r+r'},v_1,\dots,v_{s+s'})
    :=\\
    &S(\alpha^1,\dots,\alpha^r,v_1,\dots,v_s)\,
       T(\alpha^{r+1},\dots,\alpha^{r+r'},v_{s+1},\dots,v_{s+s'}).
\end{align}

\item \emph{Contraction of indices.}  
Contracting a contravariant with a covariant slot produces an $(r-1,s-1)$--tensor.  
In components, if
\begin{align}
    T=T^{i_1\dots i_r}_{\phantom{i_1\dots i_r}j_1\dots j_s}\,
        \partial_{i_1}\otimes\dots\otimes \partial_{i_r}
        \otimes dx^{j_1}\otimes\dots\otimes dx^{j_s},
\end{align}
then contracting $i_k$ with $j_\ell$ gives a new tensor
\begin{align}
    T^{i_1\dots \widehat{i_k}\dots i_r}_{\phantom{i_1\dots i_r}j_1\dots \widehat{j_\ell}\dots j_s}
    := T^{i_1\dots i_r}_{\phantom{i_1\dots i_r}j_1\dots j_s}\,\delta^{j_\ell}_{\ i_k}.
\end{align}

\item \emph{Interior product.}  
Given a vector $v\in T_pM$, the \emph{interior product} of $v$ into $T$ is the $(r,s-1)$--tensor
\begin{align}
    \iota_v T(\alpha^1,\dots,\alpha^r,v_2,\dots,v_s)
    := T(\alpha^1,\dots,\alpha^r,v,v_2,\dots,v_s).
\end{align}
Dually, given a covector $\alpha\in T_p^*M$, one defines it via the insertion in a contravariant slot by
\begin{align}
    \iota_\alpha T(\alpha^2,\dots,\alpha^r,v_1,\dots,v_s)
    := T(\alpha,\alpha^2,\dots,\alpha^r,v_1,\dots,v_s),
\end{align}
yielding an $(r-1,s)$--tensor.

\item \emph{Trace.}  
For a $(1,1)$--tensor $T=T^i_{\ j}\,\partial_i\otimes dx^j$, the contraction of its indices gives the
\emph{trace}
\begin{align}
    \mathrm{tr}(T) := T^i_{\ i}.
\end{align}
Equivalently, if $T:T_pM\to T_pM$ is seen as a linear map, then $\mathrm{tr}(T)$ is its ordinary trace. Notice that it is invariant with respect to the basis.
\end{enumerate}
\end{definition}

\section{Tensor fields}
\begin{definition}[Smooth fiber bundle]
A \emph{smooth fiber bundle} over a smooth manifold $M$\footnote{In the following, we drop the notation for the atlas.} (called the base space) is a quadruple $(E,M,\pi,F)$ where $E$ is a topological space, $F$ a smooth manifold, and $\pi:E\to M$ is a continuous surjective map, such that for every $p\in M$ there exists an open neighborhood $U\subseteq M$ and a homeomorphism
\begin{align}
    \phi_U:\ \pi^{-1}(U) \longrightarrow  U\times F
\end{align}
making the diagram commute:
$$
\begin{tikzcd}
\pi^{-1}(U) \arrow[r, "\phi_U"] \arrow[d, "\pi"'] & U \times F \arrow[d, "p_1"] \\
U \arrow[r, equal] & U
\end{tikzcd}
$$
where $p_1:U\times F\to U$ is the natural projection to the first component.
Moreover, on overlaps $U\cap V$, we have transition maps given by
\begin{align}
\phi_{U}\circ \phi_{V}^{-1} \colon &(U\cap V)\times F 
   \longrightarrow (U\cap V)\times F, \\[4pt]
&(x,f) \longmapsto \bigl(x,\, g_{UV}(x)(f)\bigr),
\end{align}
where $g_{UV}(x)\colon F\to F$ is a diffeomorphism for all $x\in U\cap V$, and the map $(x,z)\mapsto g_{UV}(x)(z)$ from $(U\cap V)\times F$ to $F$ is smooth.\\
In this way, the smooth structures of $F$ and $M$, together with the collection of all \textit{local trivializations} $\{(U,\phi_U)\}$, endow $E$ with a unique smooth structure such that $\pi:E\to M$ is smooth and surjective.
Here, $E$ is called the \emph{total space}, $M$ the \emph{base space}, 
$\pi:E\to M$ the \emph{bundle projection}, and $F$ the \emph{typical fiber}.
\end{definition}
\begin{remark}
    Notice that, by abuse of notation, a smooth fiber bundle is also commonly denoted in the same way of its projection, namely $\pi:E\to M$.
\end{remark}
\begin{remark}
    This construction is \textit{a posteriori}, namely we endowed the bundle with a smooth structure afterwards. Another standard definition is to take the fiber bundle to be smooth directly.
\end{remark}
\begin{definition}[Smooth vector bundle]\label{smoothfb}
A \emph{smooth vector bundle} of rank $k$ over a smooth manifold $M$ is a smooth fiber bundle 
$(E,M,\pi,\mathbb R^k)$ such that:
\begin{itemize}
    \item Each fiber $E_x := \pi^{-1}(x)$, for $x\in M$, is a real vector space of dimension $k$;
    \item For every local trivialization in the chosen vector bundle atlas
    $\phi:\pi^{-1}(U)\to U\times \mathbb R^k$, the restriction 
    \begin{align}
        \phi_x:\ E_x  \longrightarrow  \{x\}\times\mathbb R^k  \cong  \mathbb R^k
    \end{align}
    is a linear isomorphism;
    \item The transition maps on overlaps
    \begin{align}
        \phi_{U}\circ \phi_{V}^{-1}:\ (U\cap V)\times\mathbb R^k 
         \longrightarrow  (U\cap V)\times\mathbb R^k
    \end{align}
    are smooth maps of the form
    \begin{align}
        (x,v)\longmapsto (x,\,g_{UV}(x)v),
    \end{align}
    where $g_{UV}\colon x\mapsto g_{UV}(x)$ is a smooth map on $U\cap V$, and $g_{UV}(x)$ is an invertible matrix for all $x\in U\cap V$.
\end{itemize}
\end{definition}
\begin{definition}[Section of a smooth fiber bundle]
Let $(E,M,\pi,F)$ be a smooth fiber bundle.
A \emph{local section} over an open set $U\subseteq M$ is a smooth map
\begin{align}
s:\ U &\longrightarrow E,\\
\pi\circ s &= \mathrm{id}_U.
\end{align}
In a local trivialization $\phi_U:\pi^{-1}(U)\to U\times F$, a section $s$ over $U$
corresponds to a smooth map $\sigma_U:U\to F$ via
\begin{align}
(\phi_U\circ s)(x) &= \big(x,\sigma_U(x)\big).
\end{align}
On overlaps $U\cap V$, the change of trivialization satisfies
\begin{align}
\sigma_U(x) &= g_{UV}(x)\big(\sigma_V(x)\big), \qquad x\in U\cap V,
\end{align}
where $g_{UV}$ is given by the transition maps are
\begin{align}
\phi_U\circ \phi_V^{-1}(x,f) &= \big(x,\,g_{UV}(x)(f)\big),
\end{align}
as defined in \cref{smoothfb}. Moreover, we denote by $\Gamma(U,E)=\Gamma(E|_U)$ the set of smooth sections over $U$, and by $\Gamma(E)=\Gamma(M,E)$ the set of global smooth sections.
\end{definition}
\begin{remark}
A local section with domain $M$ is called a \emph{global section}; thus every $s\in\Gamma(E)$ satisfies $\pi\circ s=\mathrm{id}_M$.
\end{remark}
\begin{remark}
Notice that, for a vector bundle, the relation $\sigma_U(x)=g_{UV}(x)\,\sigma_V(x)$ is a linear change of coordinates in $\mathbb R^k$: here the assignment $g_{UV}:U\cap V\to \mathrm{Mat}_{k\times k}(\mathbb R)$ is smooth and, for every $x\in U\cap V$, the matrix $g_{UV}(x)$ is invertible.
\end{remark}
\begin{definition}[Vector bundle morphism and isomorphism]
Let $(E,M,\pi,\mathbb R^k)$ and $(E',M',\pi',\mathbb R^{k'})$ be smooth vector bundles.
A \emph{vector bundle morphism over a smooth map $f:M\to M'$} is a smooth map\footnote{If $E,E'$ are smooth vector bundles (so their total spaces are smooth manifolds), a map $\Phi:E\to E'$ is called \emph{smooth} iff it is a smooth map of manifolds. In other words, since, for every two local trivializations of $E$ and $E'$, we have that $\phi'_U\circ\Phi\circ\phi_U^{-1}\colon (x,v)\mapsto(f(x),A(x)(v))$, one checks that $x\mapsto A(x)\in\operatorname{Hom}(\mathbb R^k,\mathbb R^{k'})$ is smooth on the relevant base open set (since $f$ is smooth). Moreover, if $M=M'$ and $f=\mathrm{id}_M$, we say that $\Phi$ is a bundle morphism over $\mathrm{id}_M$.}
\begin{align}
\Phi:\ E \longrightarrow E'
\end{align}
such that:
\begin{enumerate}
\item[\emph{(i)}] It covers $f$, i.e.
\begin{align}
\pi'\circ \Phi  =  f\circ \pi,
\end{align}
so the diagram commutes:
$$
\begin{tikzcd}
E \arrow[r, "\Phi"] \arrow[d, "\pi"'] & E' \arrow[d, "\pi'"] \\
M \arrow[r, "f"'] & M'
\end{tikzcd}
$$
\item[\emph{(ii)}] It is fiberwise linear: for every $x\in M$, the restricted map
\begin{align}
\Phi_x:\ E_x \longrightarrow E'_{f(x)}
\end{align}
is linear.
\end{enumerate}
A vector bundle morphism $\Phi:E\to E'$ over $f$ is a \emph{vector bundle isomorphism} if there exists a vector bundle morphism
\begin{align}
\Psi:\ E' \longrightarrow E
\end{align}
over $f^{-1}$ (necessarily $f$ is a diffeomorphism) such that $\Psi\circ \Phi=\mathrm{id}_E$ and $\Phi\circ \Psi=\mathrm{id}_{E'}$.
Equivalently, $\Phi$ is a diffeomorphism of total spaces covering $f$ and each $\Phi_x:E_x\to E'_{f(x)}$ is a linear isomorphism.
\end{definition}
\begin{proposition}[Global frame iff trivial bundle]\label{prop:trivial-global-frame}
Let $(E,M,\pi,\mathbb R^k)$ be a smooth vector bundle of rank $k$.  
The following are equivalent:
\begin{enumerate}
\item[\emph{(i)}] $E$ is smoothly trivial, i.e. there exists a vector bundle isomorphism
\begin{align}
\Phi:\ E  \xrightarrow{\ \cong\ }  M\times \mathbb R^k
\end{align}
covering $\mathrm{id}_M$.
\item[\emph{(ii)}] There exist $k$ nowhere--vanishing global sections $s_1,\dots,s_k\in\Gamma(E)$ such that
\begin{align}
\{s_1(x),\dots,s_k(x)\}\ \text{is a basis of }E_x\quad\text{for all }x\in M.
\end{align}
\end{enumerate}
\end{proposition}
\begin{proof}
\emph{(i) $\Rightarrow$ (ii).}
Let $e_1,\dots,e_k$ be the standard basis of $\mathbb R^k$ and define the canonical sections of the trivial bundle:
\begin{align}
\iota_i &: M \longrightarrow M\times\mathbb R^k,
&
\iota_i(x) &= (x,e_i),
\qquad i=1,\dots,k.
\end{align}
Since $\Phi:E\to M\times\mathbb R^k$ is a bundle isomorphism covering $\mathrm{id}_M$, its inverse $\Phi^{-1}$ is smooth and covers $\mathrm{id}_M$. Set
\begin{align}
s_i &:= \Phi^{-1}\circ \iota_i \in \Gamma(E),
\qquad i=1,\dots,k,
\\
\Phi\bigl(s_i(x)\bigr) &= (x,e_i),
\qquad x\in M.
\end{align}
Then, for each $x\in M$, the fiber map
\begin{align}
\Phi|_{E_x}:\ E_x \xrightarrow{\ \cong\ } \{x\}\times\mathbb R^k \cong \mathbb R^k
\end{align}
is a linear isomorphism, hence
\begin{align}
\{s_1(x),\dots,s_k(x)\}
= \bigl(\Phi|_{E_x}\bigr)^{-1}\!\bigl(\{e_1,\dots,e_k\}\bigr)
\end{align}
is a basis of $E_x$.\\
\emph{(ii) $\Rightarrow$ (i).}
Let $s_1,\dots,s_k\in\Gamma(E)$ be global sections forming a basis of $E_x$ for all $x\in M$.
For $v\in E_x$, write uniquely
\begin{align}
v = v^i(x)\,s_i(x).
\end{align}
Define the bundle morphism
\begin{align}
\Phi:\ E \longrightarrow M\times\mathbb R^k,\qquad
v\in E_x \longmapsto \bigl(x,\,(v^1(x),\dots,v^k(x))\bigr).
\end{align}
Fix a local trivialization $\phi_U:\pi^{-1}(U)\to U\times\mathbb R^k$ and let
\begin{align}
\phi_U\circ s_i(x)  =  \bigl(x,\sigma_i(x)\bigr),
\end{align}
where, for every $i=1,\dots,k$ and $x\in U$, $\sigma_i(x)=(\sigma_i^1(x),\dots,\sigma_i^k(x))$ is a vector with $k$ components. Therefore, we have the invertible matrix
\begin{align}
    S_U(x) := [\sigma_1(x)^\top\ \cdots\ \sigma_k(x)^\top].
\end{align}
In order to identify the trivial bundle $M\times\mathbb R^k$ with its canonical trivialization
$\phi'_U: (x,w)\mapsto (x,w)$, for $\phi_U(v)=(x,w)$, we write\footnote{$\phi_U(v)=(x,w),\qquad 
w  =  S_U(x)v(x),
\quad\text{so}\quad
v(x)  =  S_U(x)^{-1} w.$}
\begin{align}
\phi'_U\circ \Phi \circ \phi_U^{-1}(x,w)
 = 
\bigl(x,\,S_U(x)^{-1}w\bigr).
\end{align}
Hence $\Phi$ covers $\mathrm{id}_M$, is fiberwise linear with fiber maps
$\Phi|_{E_x}:E_x\to\{x\}\times\mathbb R^k$ linear isomorphisms, and is smooth since
$x\mapsto S_U(x)$ and $x\mapsto S_U(x)^{-1}$ are smooth.
Therefore $\Phi$ is a smooth vector bundle isomorphism and $E\cong M\times\mathbb R^k$.
\end{proof}
\begin{definition}[Subbundle]
Let $(E,M,\pi,\mathbb R^k)$ be a smooth vector bundle. 
A \emph{(smooth) subbundle} of rank $r$ is a smooth embedded submanifold $S\subset E$ such that:
\begin{enumerate}
\item[\emph{(i)}] The restricted projection $\pi_S:=\pi|_{S}:S\to M$ is a smooth surjective submersion;
\item[\emph{(ii)}] For every $x\in M$, the fiber 
\begin{align}
S_x := S\cap E_x\ \subset\ E_x
\end{align}
is a linear subspace of dimension $r$ (hence the rank is constant).
\end{enumerate}
Equivalently, for every $x\in M$ there exists an open neighborhood $U\ni x$ and a local trivialization 
$\phi_U:\pi^{-1}(U)\to U\times\mathbb R^k$ such that
\begin{align}
\phi_U\bigl(S\cap \pi^{-1}(U)\bigr) = U\times\bigl(\mathbb R^r\times\{0\}\bigr)\ \subset\ U\times\mathbb R^k.
\end{align}
\end{definition}
\begin{corollary}[Splitting off a trivial summand from $j$ global sections]\label{cor:j-sections-splitting}
Let $(E,M,\pi,\mathbb R^k)$ be a smooth vector bundle of rank $k$. 
Assume there exist $j$ global sections $s_1,\dots,s_j\in\Gamma(E)$, with $1\le j\le k$, which are pointwise linearly independent (in particular, each $s_i$ is nowhere vanishing). Then:
\begin{enumerate}
\item[\emph{(a)}] The pointwise span
\begin{align}
S_x := \mathrm{span}\{s_1(x),\dots,s_j(x)\}\ \subset\ E_x
\end{align}
defines a smooth rank-$j$ subbundle $S\subset E$.
Moreover, the map
\begin{align}
\Psi:\ \underline{\mathbb R}^j = M\times\mathbb R^j \longrightarrow S,\qquad
(x,a_1,\dots,a_j)\longmapsto \sum_{i=1}^ja_i\,s_i(x),
\end{align}
is a smooth vector bundle isomorphism; in particular $S\cong \underline{\mathbb R}^j$ is a trivial subbundle.
\item[\emph{(b)}] There exists a smooth rank-$(k-j)$ subbundle $E'\subset E$ such that
\begin{align}
E\ \cong\ S\ \oplus\ E'\ \cong\ \underline{\mathbb R}^j\ \oplus\ E'.
\end{align}
In particular, if $j=k$, then $E$ is (smoothly) trivial.
\end{enumerate}
\end{corollary}
\begin{proof}
Define the smooth bundle map
\begin{align}
\Phi:\ M\times\mathbb R^j \longrightarrow E,\qquad
(x,a)\longmapsto \sum_{i=1}^j a_i\,s_i(x).
\end{align}
Fiberwise, $\Phi_x:\mathbb R^j\to E_x$ is linear and injective by pointwise linear independence, hence has constant rank $j$. By the constant rank theorem for vector bundles\footnote{Let $\Phi:E\to E'$ be a smooth vector bundle morphism covering $f:M\to M'$ and assume the fiberwise rank $\mathrm{rank}(\Phi_x)$ is constant $r$. Regard $\Phi$ as the morphism $E\to f^*E'$ over $\mathrm{id}_M$. Then, locally there exist trivializations of $E$ and $f^*E'$ and smooth changes of frame such that its matrix takes the block form $\begin{psmallmatrix} I_r & 0 \\ 0 & 0 \end{psmallmatrix}$. Consequences: (a) if $f=\mathrm{id}_M$, both $\ker\Phi\subset E$ and $\mathrm{im}\,\Phi\subset E'$ are smooth subbundles of constant ranks $k-r$ and $r$, and one obtains local splittings $E|_U\simeq \ker\Phi\oplus S$ and $E'|_U\simeq \mathrm{im}\,\Phi\oplus Q$ for suitable subbundles $S,Q$; (b) in general, the fiberwise images $\Phi_x(E_x)$ define a rank-$r$ subbundle of the pull-back bundle $f^{*}E'\to M$.} (or by a direct check in local trivializations), the image
\begin{align}
S = \mathrm{im}(\Phi) = \bigsqcup_{x\in M} S_x
\end{align}
is a smooth subbundle of rank $j$. The map $\Psi$ in \emph{(a)} is just $\Phi$ with codomain restricted to $S$; fiberwise it is a linear isomorphism, and smoothness of the inverse follows from writing coordinates in the smooth frame $\{s_i\}_{i=1}^j$. This proves \emph{(a)}.\\
For \emph{(b)}, fix a smooth bundle metric $h$ on $E$ (it exists by partition of unity). Set $E':=S^{\perp_h}$. Then $E'$ is a smooth subbundle of rank $k-j$, and the orthogonal projections are smooth, yielding a (metric) direct sum decomposition $E=S\oplus E'$. The final claim follows: if $j=k$, then $E\cong\underline{\mathbb R}^k$ is trivial.
\end{proof}
\begin{definition}[Tangent bundle]
Let $M$ be a smooth $n$–dimensional manifold. The \emph{tangent bundle} is the rank-$n$ smooth vector bundle 
\begin{align}
\pi:TM \longrightarrow M,
\end{align}
whose fiber at $p\in M$ is $\pi^{-1}(p)=T_pM$.\footnote{Since $\pi$ is surjective and we build the bundle such that $\pi^{-1}(p)=T_pM$ for all $p\in M$, this means that $TM=\bigsqcup_{p\in M}T_pM$ as a set. Then we need to give it a topology and a smooth structure.} Given a chart over $U$ with coordinates $x=(x^1,\dots,x^n)$ and, for each $T_pM$ with $p\in U$, coordinate basis $\partial_i|_p:=\tfrac{\partial}{\partial x^i}\big|_p$,\footnote{Notice that we now introduce a $p$-dependent notation because here the point $p\in U$ is not fixed.} the map $\phi_U:\pi^{-1}(U)\to U\times \mathbb R^n$, defined fiberwise for $p\in U$ as
\begin{align}
\phi_U:\ \pi^{-1}(p) &\longrightarrow \{p\}\times\mathbb R^n\cong\mathbb R^n\\
v_p = v^i\,\partial_i|_p &\longmapsto (v^1,\dots,v^n),
\end{align}
is a local trivialization. On overlaps, with respect to coordinates $(x^1,\dots,x^n)$ and $(x'^{1},\dots,x'^{n})$, the transition maps are the Jacobian matrices at $p\in U\cap V$:
\begin{align}
g_{UV}(p)^i_{\ j} \,=\, \frac{\partial x^{i}}{\partial x^{\prime j}}(p),
\qquad
\phi_U\circ\phi_V^{-1}(p,v) \,=\, \big(p,\, g_{UV}(p)\,v\big).
\end{align}
\end{definition}

\begin{definition}[Cotangent bundle]
Let $M$ be a smooth $n$–dimensional manifold. The \emph{cotangent bundle} is the rank-$n$ smooth vector bundle
\begin{align}
\pi:T^*M &\longrightarrow M,
\end{align}
whose fiber at $p\in M$ is $\pi^{-1}(p)=T^*_pM$.
Given a chart over $U$ with coordinates $x=(x^1,\dots,x^n)$ and, for each $T_p^*M$ with $p\in U$, coordinate basis $dx^i|_p$, the map $\phi_U:\pi^{-1}(U)\to U\times \mathbb R^n$, defined fiberwise for $p\in U$ as
\begin{align}
\phi_U:\ \pi^{-1}(p) &\longrightarrow \{p\}\times\mathbb R^n\cong\mathbb R^n, \\
\alpha_p = a_i\,dx^i|_p &\longmapsto (a_1,\dots,a_n),
\end{align}
is a local trivialization. On overlaps, with respect to coordinates $(x^1,\dots,x^n)$ and $(x'^{1},\dots,x'^{n})$, the transition maps are the inverse transposes of the tangent transition matrices at $p\in U\cap V$:
\begin{align}
g^{T^*}_{UV}(p)_i{}^{j} = \frac{\partial x^{\prime j}}{\partial x^{i}}(p),
\qquad \phi_U\circ\phi_V^{-1}(p,a) = \bigl(p,\, g^{T^*}_{UV}(p)\,a\bigr).
\end{align}
\end{definition}
\begin{definition}[Vector fields and covector fields]
A \emph{vector field} on a smooth manifold $M$ is a smooth section $X\in\Gamma(TM)$. In components, on a chart $(U,x)$, it reads\footnote{Note that the notation is implicit, but every quantity depends on points in $U$.}
\begin{align}
X  =  X^i\,\partial_i,\qquad X^i\in C^\infty(U),\qquad
X^{\prime i}  =  \frac{\partial x^{\prime i}}{\partial x^j}\,X^j.
\end{align}
A \emph{covector field} on $M$ is a smooth section $\alpha\in\Gamma(T^*M)$. On a chart $(U,x)$, it reads
\begin{align}
\alpha  =  \alpha_i\,dx^i,\qquad \alpha_i\in C^\infty(U),\qquad
\alpha'_i  =  \frac{\partial x^{j}}{\partial x^{\prime i}}\,\alpha_j.
\end{align}
\end{definition}
\begin{definition}[Tensor fields]
A \emph{$(r,s)$–tensor field} on a smooth manifold $M$ is a smooth section $T \in \Gamma(\,T^{(r,s)}(M))$, where
\begin{align}
T^{(r,s)}(M) := \underbrace{TM\otimes\cdots\otimes TM}_{r\ \text{times}}
\otimes \underbrace{T^*M\otimes\cdots\otimes T^*M}_{s\ \text{times}}
\end{align}
is the $(r,s)$–tensor bundle. In a chart $(U,x)$, $T$ can be expressed in components as
\begin{align}
T = T^{i_1\cdots i_r}{}_{j_1\cdots j_s} 
   \partial_{i_1}\otimes\cdots\otimes \partial_{i_r}
   \otimes dx^{j_1}\otimes\cdots\otimes dx^{j_s},\quad
T^{i_1\cdots i_r}{}_{j_1\cdots j_s} \in C^\infty(U).
\end{align}
On coordinate overlaps $(U,x)$ and $(V,x')$, the components transform as
\begin{align}
T^{\prime i_1\cdots i_r}{}_{j_1\cdots j_s}
&= 
\frac{\partial x^{\prime i_1}}{\partial x^{k_1}}
\cdots
\frac{\partial x^{\prime i_r}}{\partial x^{k_r}}
\frac{\partial x^{\ell_1}}{\partial x^{\prime j_1}}
\cdots
\frac{\partial x^{\ell_s}}{\partial x^{\prime j_s}} 
T^{k_1\cdots k_r}{}_{\ell_1\cdots \ell_s}.
\end{align}
\end{definition}
\begin{remark}
Notice that the constructions $\mathrm{Sym}^k$ and $\Lambda^k$, 
applied fiberwise to the cotangent spaces $T^*_pM$, generate smooth vector bundles
\begin{align}
\mathrm{Sym}^k(T^*M), \qquad \Lambda^kT^*M.
\end{align}
\end{remark}
\begin{definition}[Pushforward and pull-back of tensor fields]
Let $F:M\to N$ be a smooth map.
\begin{itemize}
\item If $S\in\Gamma(T^{(0,s)}(N))$, its \textit{pull-back} is $F^* S\in\Gamma(T^{(0,s)}(M))$ given by
\begin{align}
(F^* S)_p(v_1,\dots,v_s)
 :=  S_{F(p)}\bigl(F_*|_p\,v_1,\dots,F_*|_p\,v_s\bigr),
\end{align}
for $v_i\in T_pM$ and $F_*|_p$ defined in \cref{pushpullfibers}.

\item If $T\in\Gamma(T^{(r,0)}(M))$, its \textit{pushforward} along $F$ is
\begin{align}
(F_{\!* }T)_p(\alpha^1,\dots,\alpha^r)
 :=  T_p\bigl((F^* \alpha^1)_p,\dots,(F^* \alpha^r)_p\bigr),
\end{align}
for $\alpha^i\in T^* _{F(p)}(N)$. Thus $p\mapsto (F_{\!* }T)_p\in T^{(r,0)}_{F(p)}(N)$ is a section of $F^* T^{(r,0)}(N)$ (the \textit{pull-back bundle}). Moreover, if $F$ is a local diffeomorphism $F|_U:U\!\to\!V$,\footnote{Here “(local) diffeomorphism” means that for every $p\in M$ there are neighborhoods $U\ni p$, $V\ni F(p)$ with $F|_U:U\to V$ a diffeomorphism.} $(F_* T)_q:=(F_{\!* }T)_{(F|_U)^{-1}(q)}$, so this restricted pushforward is a section over $V$. These local fields need not agree for different inverse branches; a diffeomorphism $F:M\to N$ gives a global pushforward, while for a general map a field on $N$ requires projectability.
\end{itemize}
\end{definition}

\begin{remark}[Pushforward vs pull-back]
For a smooth map $F:M\to N$, the pushforward
\begin{align}
F_*:\ TM\longrightarrow TN,\qquad (x,v)\longmapsto\bigl(F(x),\,F_*|_x v\bigr)
\end{align}
is a vector bundle morphism covering $F$. Note that, although $F_*:TM\to TN$ is a vector bundle morphism covering $F$, it does not induce a natural map on sections $\Gamma(TM)\to\Gamma(TN)$ for arbitrary $F$. A diffeomorphism gives such a map, while a local diffeomorphism only gives it after restriction to an inverse branch; in general one only has a natural map $\Gamma(TM)\to \Gamma(F^{*}TN)$. By contrast, the pull--back on covectors is defined fiberwise by
\begin{align}
F^*|_x:T^{*}_{F(x)}N\longrightarrow T^{*}_xM,
\end{align}
and induces the natural map on sections
\begin{align}
F^{*}:\ \Gamma(T^*N)\longrightarrow \Gamma(T^*M),
\end{align}
but it does \textit{not} induce a bundle morphism 
\begin{align}
    T^{*}N\longrightarrow T^{*}M.
\end{align}
In general, there is \textit{no} such natural bundle map $T^{*}N \longrightarrow T^{*}M$ covering $F$. What naturally always exists is the bundle morphism
\begin{align}
\iota:\ F^{*}T^{*}N \longrightarrow T^{*}M,\qquad
(x,\alpha_{F(x)})\longmapsto \bigl(x,F^*|_x\alpha_{F(x)}\bigr),
\end{align}
which covers $\mathrm{id}_M$. Composing the natural map $\alpha\mapsto(x\mapsto(x,\alpha_{F(x)}))$ from $\Gamma(T^*N)$ to $\Gamma(F^*T^*N)$ with the map on sections induced by $\iota$ gives exactly the pull-back $F^*: \Gamma(T^*N)\to\Gamma(T^*M)$. However, if $F$ is a diffeomorphism, then the natural map of sections $F^*$ induces a bundle morphism
\begin{align}
    F^*\colon T^{*}N\longrightarrow T^{*}M\qquad \alpha_y \longmapsto (dF_{F^{-1}(y)})^{*}\alpha_y,
\end{align}
which covers $F^{-1}$.
\end{remark}
\begin{remark}
    Notice that the map $F_*$ is also denoted as $dF$ and it is also called the \textit{tangential map}.
\end{remark}
\begin{proposition}[Properties of pushforward and pull-back]
Let $F:M\to N$ and $G:N\to P$ be diffeomorphisms between smooth manifolds. The fiberwise and covariant pull-back statements below also hold for arbitrary smooth maps. Then:

\begin{enumerate}[label=\textnormal{(\roman*)}]
\item \emph{Linearity.} For every $p\in M$, the maps
\begin{align}
F_*|_p &: T_pM \longrightarrow T_{F(p)}N, \\
F^*|_p &: T^*_{F(p)}N \longrightarrow T^*_pM
\end{align}
are linear. Hence, since $F$ is a diffeomorphism, also
\begin{align}
F_* &: \Gamma(TM)\longrightarrow \Gamma(TN), \\
F^* &: \Gamma(T^*N)\longrightarrow \Gamma(T^*M)
\end{align}
are well--defined and $\mathbb R$–linear.

\item \emph{Identity.}
\begin{align}
(\mathrm{id}_M)_* &= \mathrm{id}_{\Gamma(TM)}, \\
(\mathrm{id}_M)^* &= \mathrm{id}_{\Gamma(T^*M)}.
\end{align}

\item \emph{Functoriality.}
\begin{align}
(G\circ F)_* &= G_* \circ F_*, \\
(G\circ F)^* &= F^* \circ G^*.
\end{align}
\end{enumerate}
\end{proposition}

\begin{proof} For the proof we will often write $dF=F_*$.\\
\textit{(i)} Linearity follows because $F_*|_p$ is a derivative (hence linear), and $F^*$ is precomposition with $F_*|_p$, which preserves linearity.  

\textit{(ii)} For the identity map $\mathrm{id}_M:M\to M$, one has
\begin{align}
d(\mathrm{id}_M)_p(v) = v, \qquad
(\mathrm{id}_M)^*(\alpha)(v) = \alpha(d(\mathrm{id}_M)_p(v)) = \alpha(v),
\end{align}
so both pushforward and pull-back act as the identity.

\textit{(iii)} For $p\in M$ and $v\in T_pM$,
\begin{align}
d(G\circ F)_p(v) = dG_{F(p)}\bigl(dF_p(v)\bigr),
\end{align}
so $(G\circ F)_*=G_*\circ F_*$.  
For $\alpha\in T^*_{G(F(p))}P$ and $v\in T_pM$,
\begin{align}
\bigl((G\circ F)^*\alpha\bigr)(v)
= \alpha\bigl(d(G\circ F)_p(v)\bigr)
= \alpha\bigl(dG_{F(p)}(dF_p(v))\bigr)
= \bigl(F^*(G^*\alpha)\bigr)(v),
\end{align}
so $(G\circ F)^*=F^*\circ G^*$.
\end{proof}

\begin{definition}[Distribution]
A \emph{smooth distribution} of rank $r$ on a smooth manifold $M$ is a rank-$r$ subbundle 
$D\subset TM$ (equivalently: for each $p\in M$ a subspace $D_p\subset T_pM$ of dimension $r$,
depending smoothly on $p$).  
It is \emph{involutive} if for all $X,Y\in\Gamma(D)$ one has
\begin{align}
[X,Y]\in\Gamma(D).
\end{align}
An \emph{integral submanifold} of $D$ is an immersed\footnote{A smooth map $F:N\to M$ is an \emph{immersion} if its pushforward $F_*|_p:T_pN\to T_{F(p)}M$ is injective for all $p\in N$ (equivalently, $\mathrm{rank}\,F_*|_p=\dim N$ for all $p\in N$). A subset $S\subset M$ is an \emph{immersed submanifold} if there exists a smooth structure on $S$ such that the inclusion $i:S\hookrightarrow M$ is an immersion. In general, an immersed submanifold need not be embedded; it is \emph{embedded} precisely when the inclusion is a topological embedding (i.e.\ a homeomorphism onto its image with the subspace topology).} submanifold $N\subset M$ of dimension $r$
such that
\begin{align}
T_pN = D_p,\qquad \forall\,p\in N.
\end{align}
A maximal connected integral submanifold is called a \emph{leaf}.
\end{definition}
\begin{definition}[Frame]
Let $U\subset M$ be open. 
A \emph{local frame of $TM$ on $U$} is a $n$–tuple of smooth vector fields 
$X_1,\dots,X_n\in\Gamma(TU)$ that are pointwise linearly independent and span each tangent fiber:
\begin{align}
T_pM=\mathrm{span}\{X_1(p),\dots,X_n(p)\},\qquad \forall\,p\in U.
\end{align}
If $D\subset TM$ is a rank-$r$ distribution, a \emph{local frame of $D$ on $U$} is
$X_1,\dots,X_r\in\Gamma(D|_U)$ such that
\begin{align}
D_p=\mathrm{span}\{X_1(p),\dots,X_r(p)\},\qquad \forall\,p\in U.
\end{align}
\end{definition}
\begin{definition}[Coframe]
Let $U\subset M$ be open.
A \emph{local coframe of $T^*M$ on $U$} is a $n$–tuple of smooth $1$–forms 
$\theta^1,\dots,\theta^n\in\Omega^1(U)$ that form a pointwise basis of $T_p^*M$ and are dual to a frame
$X_1,\dots,X_n$:
\begin{align}
\theta^{i}(X_j)=\delta^{i}_{\ j}\qquad \text{on }U.
\end{align}
\end{definition}
\begin{remark}
Not every local frame associated to a distribution $D$ is made of coordinate vector fields. Coordinate fields always commute,
\begin{align}
[\partial_{i},\partial_{j}]=0,
\end{align}
but a general local frame $X_1,\dots,X_r$ (e.g.\ for a distribution $D|_U$) need not be commuting, i.e.\ typically $[X_i,X_j]\neq 0$. The next theorem characterizes precisely when a commuting (hence coordinate) frame exists: this holds locally if and only if $D$ is involutive. Equivalently, there are coordinates $(x^1,\dots,x^n)$ on $U$ with
\begin{align}
D|_U=\mathrm{span}\{\partial_{1},\dots,\partial_{r}\}
\end{align}
and $[\partial_{i},\partial_{j}]=0$ for $1\le i,j\le r$.
\end{remark}
\begin{definition}[Bracket of vector fields]
Let $M$ be a smooth manifold. For $X,Y\in\Gamma(TM)$, their \textit{bracket} or \textit{commutator} $[X,Y]\in\Gamma(TM)$ is
the unique vector field characterized by
\begin{align}
[X,Y](f) = X\bigl(Y(f)\bigr) - Y\bigl(X(f)\bigr), \qquad \forall f\in C^\infty(M).
\end{align}
In local coordinates $(x^1,\dots,x^n)$, if $X=X^i\partial_i$ and $Y=Y^i\partial_i$, then
\begin{align}
[X,Y]  =  \Bigl(X^j\partial_jY^i - Y^j\partial_jX^i\Bigr)\,\partial_i.
\end{align}
The Lie bracket is $\mathbb{R}$–bilinear, antisymmetric ($[X,Y]=-[Y,X]$), and satisfies the Jacobi identity
\begin{align}
[X,[Y,Z]] + [Y,[Z,X]] + [Z,[X,Y]]  =  0.
\end{align}
Moreover, Lie brackets are natural under local diffeomorphisms, after restriction to an open set on which the map is a diffeomorphism, i.e.
\begin{align}
\phi_* \,[Y,Z]  =  [\phi_* Y,\phi_* Z].
\end{align}
\end{definition}
\begin{theorem}[Frobenius theorem]
Let $D\subset TM$ be a smooth distribution of rank $r$ on a smooth manifold $M$.
Then, the following statements are equivalent:
\begin{enumerate}[label=(\roman*)]
\item $D$ is involutive: for all $X,Y\in\Gamma(D)$ one has
\begin{align}
[X,Y]\in\Gamma(D).
\end{align}
\item For every $p\in M$ there exist a neighborhood $U$ of $p$ and a local frame
$X_1,\dots,X_r\in\Gamma(D|_U)$ such that on $U$
\begin{align}
[X_i,X_j]=0, 
\end{align}
for all $i,j=1,\dots,r$. In particular, there exist local coordinates $(x^1,\dots,x^n)$ on $U$ with
\begin{align}
X_i=\partial_{i}\qquad (i=1,\dots,r),
\end{align}
so the local plaques are the connected components of the level sets
\begin{align}
x^{r+1}=\mathrm{const},\ \dots,\ x^n=\mathrm{const}.
\end{align}
\item $D$ is locally integrable: for every $p\in M$ there exists an immersed $r$-submanifold
$N\subset M$ containing $p$ and
\begin{align}
T_qN = D_q,\qquad \forall\,q\in N.
\end{align}
\end{enumerate}
\end{theorem}
\begin{proof}
\emph{(i) $\Rightarrow$ (ii).} Fix $p$. After reordering coordinates and shrinking to a neighborhood of $p$, the projection of $D$ onto the first $r$ coordinate directions is an isomorphism. Thus $D$ has a smooth frame
\begin{align}
X_i=\partial_{x^i}+\sum_{\alpha=r+1}^n a_i^\alpha\partial_{x^\alpha},\qquad i=1,\dots,r.
\end{align}
Involutivity also holds for local sections, by multiplying them by cutoff functions near any point. Each $[X_i,X_j]$ has zero first $r$ components and belongs to $D$, so $[X_i,X_j]=0$.
Let $\Phi^i_t$ be their local flows, obtained by the smooth ODE existence theorem. Commuting fields have commuting local flows: in coordinates, both $d\Phi^i_t(X_j)$ and $X_j\circ\Phi^i_t$ solve the same variational equation, since $[X_i,X_j]=0$, and have the same initial value; uniqueness then shows that $\Phi^i_t$ preserves $X_j$ and its integral curves.
Translate the original coordinates so that $p$ is the origin, and let $q(z)$ parametrize the transverse slice $x^1=\cdots=x^r=0$. The map
\begin{align}
F(t^1,\dots,t^r,z)=\Phi^1_{t^1}\circ\cdots\circ\Phi^r_{t^r}(q(z))
\end{align}
has invertible differential at the origin. By the inverse function theorem it is a local coordinate parametrization. Commutation of the flows gives $F_*\partial_{t^i}=X_i$, proving (ii). The case $r=0$ is immediate.\\
\emph{(ii) $\Rightarrow$ (iii).} Fixing the transverse coordinates gives an integral submanifold through every point.\\
\emph{(iii) $\Rightarrow$ (i).} Let $p$ lie on an immersed integral manifold $N$. Restrict to a locally embedded neighborhood in $N$. Any local sections $X,Y$ of $D$ restrict to tangent fields on $N$, and their bracket restricts to the bracket on $N$. Hence $[X,Y]_p\in T_pN=D_p$. Since $p$ is arbitrary, $D$ is involutive.
\end{proof}
\begin{definition}
    A \textit{smooth regular foliation} of dimension $r$ on an $n$--dimensional smooth manifold $M$ is a partition of $M$ into connected immersed $r$--dimensional submanifolds (the \emph{leaves}) such that every point $p\in M$ has a neighborhood $U$ with a chart
\begin{align}
\phi:U\longrightarrow \mathbb R^r\times\mathbb R^{n-r}
\end{align}
for which the connected components of the intersections of leaves with $U$ are the plaques given by
\begin{align}
\{(y,z)\in \mathbb R^r\times\mathbb R^{n-r}:\ z=\text{const}\},
\end{align}
i.e. the connected components of the fibers of the projection onto the transverse coordinates.\footnote{Precisely: in a foliation chart $\phi:U\to \mathbb R^r\times\mathbb R^{n-r}$ write points as $(y,z)$ with
$y\in\mathbb R^r$ (tangential coordinates) and $z\in\mathbb R^{n-r}$ (transverse coordinates). Consider the projection
$\pi:\mathbb R^r\times\mathbb R^{n-r}\to\mathbb R^{n-r}$, $\pi(y,z)=z$. For a fixed $z_0$, the \emph{fiber}
$\pi^{-1}(z_0)=\{(y,z_0):y\in\mathbb R^r\}$ is an $r$–plane. The \textit{plaques} in $U$ are the \emph{connected components} of $\pi^{-1}(z_0)\cap \phi(U)$ as $z_0$ varies. We take connected components because $\phi(U)$ need not be a product $n$--box; thus a single fiber may meet $\phi(U)$ in several disconnected pieces, each piece being one plaque. Via
$\phi^{-1}$, these plaques correspond to the local leaves in $U$.}
\end{definition}
\begin{remark}[Local leaves and the induced foliation]
Assume \textnormal{(ii)} holds in a chart $\phi:U\to V\subset\mathbb R^n$ with coordinates $(x^1,\dots,x^n)$ so that
\begin{align}
D|_U=\mathrm{span}\{\partial_{x^1},\dots,\partial_{x^r}\}.
\end{align}
Set $x^\perp\coloneqq(x^{r+1},\dots,x^n)$ and, for $c=(c_{r+1},\dots,c_n)\in\mathbb R^{n-r}$, define the level set
\begin{align}
L_c  :=  (x^\perp)^{-1}(c)
       =  \{\,q\in U:\ x^{r+1}(q)=c_{r+1},\dots,x^n(q)=c_n\,\}.
\end{align}
Since $x^\perp:U\to\mathbb R^{n-r}$ is a submersion on $U$,\footnote{Let $F:M\to N$ be smooth. We say that $F$ is a \emph{submersion at $p\in M$} if
\begin{align}
F_*|_p:\ T_pM \longrightarrow T_{F(p)}N \quad \text{is surjective}
\end{align}
(equivalently, $\operatorname{rank} F_*|_p=\dim N$). We say that $F$ is a \emph{submersion on $U\subset M$} if it is a submersion at every $p\in U$.} each $L_c$ is a smooth $r$–dimensional submanifold and
\begin{align}
T_q L_c  =  \ker d(x^\perp)_q  =  \mathrm{span}\{\partial_{x^1},\dots,\partial_{x^r}\}  =  D_q,
\qquad \forall\,q\in L_c.
\end{align}
Hence each connected component of $L_c$ is an immersed integral manifold of $D$ (a \emph{local leaf}), and every connected immersed integral manifold contained in $U$ has image in one such component, though it need not be the whole component.
Equivalently, the covectors $dx^{r+1},\dots,dx^n$ vanish on vectors tangent to any leaf, so the transverse coordinates
\begin{align}
x^{r+1},\dots,x^n
\end{align}
are constant along leaves. In other words, Frobenius theorem guarantees that a smooth regular foliation exists if and only if the distribution is involutive (and of constant rank). Under the stated smoothness and constant-rank hypotheses, these local plaques fit together into the unique foliation whose leaves are the maximal connected integral manifolds of $D$.
\end{remark}
\begin{remark}
For $D=TM$ the involutivity condition is automatic (the bracket of any two vector fields is again a vector field), but this does \emph{not} imply that an arbitrary local frame $X_1,\dots,X_n$ must commute: in general $[X_i,X_j]\neq 0$.
What involutivity guarantees here is the \emph{existence} of a local \emph{coordinate} frame (hence commuting), i.e. a system of coordinates admitting a frame of the form $(\partial_1,\dots,\partial_n)$, not necessarily the initially chosen frame.
More precisely:
\begin{align}
\text{``coordinate frame''}\ \Longleftrightarrow\ \text{``pairwise commuting local frame''}.
\end{align}
The forward implication is immediate since $[\partial_i,\partial_j]=0$. Conversely, if a local frame $E_1,\dots,E_n$ satisfies $[E_i,E_j]=0$ on $U$, then by Frobenius there exist local coordinates $(x'^1,\dots,x'^n)$ with
\begin{align}
E_i=\partial'_i \qquad (i=1,\dots,n).
\end{align}
If $(x^1,\dots,x^n)$ and $(x'^1,\dots,x'^n)$ are two coordinate systems, the associated frames and coframes are related by the Jacobian:
\begin{align}
\partial'_i &= \frac{\partial x^j}{\partial x'^i}\,\partial_j,\\
dx'^i &= \frac{\partial x'^i}{\partial x^j}\,dx^j.
\end{align}
Hence, a non-commuting frame cannot arise from a coordinate change, otherwise it would still be a coordinate frame. In other terms, this is compatible with the fact that Lie brackets are natural under local diffeomorphisms, after restriction to an open set on which the map is a diffeomorphism, i.e.
\begin{align}
\phi_* \,[Y,Z]  =  [\phi_* Y,\phi_* Z].
\end{align}
In general, any two frames (coordinate or not) differ by a smooth $GL(n)$–valued map $B:U\to GL(n)$:\footnote{This reflects that the frame bundle $FM\to M$ is a principal $GL(n)$–bundle whose structure group acts freely and transitively on each fiber.}
\begin{align}
E'_i  =  B^j{}_i\,E_j.
\end{align}
In particular, on $TM$, there always exists a $B$ such that
\begin{align}
E_i  =  B^j{}_i\,\partial_j,
\end{align}
and, after shrinking the neighborhood, $B^j{}_i=\partial x^j/\partial y^i$ for coordinates $y$ with $E_i=\partial_{y^i}$ if and only if $[E_i,E_j]=0$. Equivalently, the dual forms $(B^{-1})^i{}_j\,dx^j$ are closed and locally equal to $dy^i$.
\end{remark}
\begin{definition}[Lie algebroid]
A \emph{Lie algebroid} over a manifold $M$ is a vector bundle $A$ over $M$ together with
\begin{align}
[\,,\,]:\Gamma(A)\times\Gamma(A)\longrightarrow \Gamma(A),\qquad
\rho:A\longrightarrow TM,
\end{align}
where $\rho$ is a smooth vector bundle morphism over $\mathrm{id}_M$, called the \emph{anchor}, such that for all $s,t\in\Gamma(A)$ and $f\in C^\infty(M)$:
\begin{align}
&\text{(i) } [\,,\,]\ \text{is $\mathbb{R}$–bilinear, skew, and satisfies the Jacobi identity;}\\
&\text{(ii) } \text{Leibniz rule: } [s, f t]  =  f[s,t] + (\rho(s)f)\, t.
\end{align}
\end{definition}
\begin{remark}[From commutator to Lie algebra and Lie algebroid]
The commutator of vector fields just defined endows $\Gamma(TM)$ with Lie algebra structure on $\mathbb R$, and it satisfies the derivation rule
\begin{align}
[X, fY]  =  f[X,Y] + X(f)\,Y, \qquad X,Y\in\Gamma(TM),\ f\in C^\infty(M).
\end{align}
Hence, $(TM,[\,,\,],\rho=\mathrm{id}_{TM})$ is the \emph{canonical Lie algebroid} over $M$, 
the natural generalization of a Lie algebra over $\mathbb{R}$ to the $C^\infty(M)$–module of sections.\footnote{This is a Lie–Rinehart structure: $\Gamma(TM)$ is a Lie algebra over $\mathbb{R}$ and a $C^\infty(M)$–module linked by $[X,fY]=f[X,Y]+X(f)\,Y$; it is \emph{not} a Lie algebra over $C^\infty(M)$.}\\
Moreover, if $D\subset TM$ is a smooth distribution that is involutive (the Frobenius condition),
then $\Gamma(D)$ is closed under the ambient commutator of vector fields (well–defined by extending sections locally),
the inclusion $\iota:D\hookrightarrow TM$ serves as anchor, and the same Leibniz rule holds with $\rho=\iota$. Thus,
\begin{align}
(D,\,[\,,\,]_D,\,\iota)
\end{align}
is a Lie \emph{sub}algebroid of $(TM,[\,,\,],\mathrm{id})$. This is precisely the algebraic content behind Frobenius.
\end{remark}

\section{Differential forms}
\begin{definition}[Differential form]
A \emph{differential $p$-form} on a smooth manifold $M$ is a smooth section 
\begin{align}
\omega  \in  \Omega^p(M)\coloneqq\Gamma(\Lambda^p T^*M).
\end{align}
On a chart $(U,x)$, it can be written as\footnote{The symbol $\wedge$ is the alternating tensor product. In particular, notice that $dx^{i}\wedge dx^{i}=0$.}
\begin{align}
\omega  =  \frac{1}{p!}\,\omega_{i_1\cdots i_p}\,dx^{i_1}\wedge\cdots\wedge dx^{i_p},
\qquad \omega_{i_1\cdots i_p}\in C^\infty(U),
\end{align}
with antisymmetric components
\begin{align}
\omega_{i_{\sigma(1)}\cdots i_{\sigma(p)}}  =  \mathrm{sgn}(\sigma)\,\omega_{i_1\cdots i_p},
\qquad \sigma\in S_p.
\end{align}
Under a change of coordinates $x \mapsto x'$ the components transform as
\begin{align}
\omega'_{j_1\cdots j_p}
 = 
\frac{\partial x^{i_1}}{\partial x^{\prime j_1}}
\cdots
\frac{\partial x^{i_p}}{\partial x^{\prime j_p}}\,
\omega_{i_1\cdots i_p}.
\end{align}
\end{definition}
\begin{definition}[Wedge product]
There is a bilinear operation
\begin{align}
\wedge:\ \Omega^p(M)\times \Omega^q(M)\longrightarrow \Omega^{p+q}(M)
\end{align}
defined on elementary forms by
\begin{align}
(\alpha_1\wedge\cdots\wedge\alpha_p)\wedge(\beta_1\wedge\cdots\wedge\beta_q)
=\alpha_1\wedge\cdots\wedge\alpha_p\wedge\beta_1\wedge\cdots\wedge\beta_q,
\end{align}
and extended by linearity. Here the exterior power of $T^*_pM$ is identified with alternating multilinear forms by sending $\alpha_1\wedge\cdots\wedge\alpha_k$ to $k!\,\mathrm{Alt}(\alpha_1\otimes\cdots\otimes\alpha_k)$. Thus
\begin{align}
\alpha\wedge\beta=\frac{(p+q)!}{p!\,q!}\,\mathrm{Alt}(\alpha\otimes\beta).
\end{align}
\end{definition}
\begin{definition}[Graded algebra]
A \emph{graded algebra} $(A^\bullet,\,\cdot\,)$ over $\mathbb{R}$ is the direct sum of vector spaces $\{A^p\}_{p\in\mathbb{Z}}$
\begin{align}
A^\bullet = \bigoplus_{p\in\mathbb{Z}} A^p,
\end{align}
called a $\mathbb{Z}$--graded vector space, together with a bilinear product
\begin{align}
\cdot:\ A^p\times A^q \longrightarrow A^{p+q},
\end{align}
which is associative and admits a unit $1\in A^0$.  
If, for all homogeneous $a\in A^p$ and $b\in A^q$, one has
\begin{align}
a\cdot b = (-1)^{pq}\,b\cdot a,
\end{align}
then $A^\bullet$ is called a \emph{graded commutative algebra}.
\end{definition}
\begin{proposition}[Graded algebra structure]
For $\alpha\in\Omega^p(M)$ and $\beta\in\Omega^q(M)$ one has
\begin{align}
\alpha\wedge\beta = (-1)^{pq}\,\beta\wedge\alpha,
\qquad
(\alpha\wedge\beta)\wedge\gamma = \alpha\wedge(\beta\wedge\gamma).
\end{align}
Then,
\begin{align}
    \Omega^\bullet(M)\coloneqq\bigoplus_{p\geq0} \Omega^p(M)
\end{align}
forms a graded algebra $(\Omega^\bullet(M),\wedge)$ over $\mathbb R$, where the unity is given by the constant function $1\in\Omega^0(M)$.
\end{proposition}
\begin{proof}
On the basis $\{dx^{i_1}\wedge\cdots\wedge dx^{i_p}\}$ the wedge product is defined by antisymmetrization of the tensor product. Permuting blocks of degrees $p$ and $q$ introduces the sign $(-1)^{pq}$ ($pq$ exchanges of $1$-forms). Associativity holds on products of $1$--forms because both parenthesizations concatenate the same ordered factors, and extends to arbitrary forms by local expansion and bilinearity.
\end{proof}
\begin{definition}[Interior product]
For $X\in\Gamma(TM)$ the \emph{interior product}
\begin{align}
\iota_X:\Omega^p(M)\longrightarrow \Omega^{p-1}(M)
\end{align}
is defined on decomposable forms by
\begin{align}
\iota_X(\alpha_1\wedge\cdots\wedge \alpha_p)
=\sum_{k=1}^p (-1)^{k-1}\,\alpha_k(X)\,\alpha_1\wedge\cdots\wedge\widehat{\alpha_k}\wedge\cdots\wedge\alpha_p,
\end{align}
and extended linearly.
\end{definition}
\begin{proposition}[Properties of $\iota_X$]
For $\alpha\in\Omega^p(M),\ \beta\in\Omega^q(M)$ one has
\begin{align}
\iota_X(\alpha\wedge\beta)=(\iota_X\alpha)\wedge\beta+(-1)^p\,\alpha\wedge(\iota_X\beta).
\end{align}
\end{proposition}
\begin{proof}
Check on decomposable forms: the contraction either acts inside $\alpha$ or inside $\beta$, and the sign arises from moving $\iota_X$ past a block of degree $p$.
\end{proof}
\begin{definition}[Degree of an operator]
Let $(A^\bullet(M),\,\cdot\,)$ be a graded algebra.  
A linear operator $D:A^\bullet\to A^\bullet$ is said to have \emph{degree $k$} if
\begin{align}
D(A^p) \subseteq  A^{p+k}, \qquad \forall\,p.
\end{align}
\end{definition}
\begin{definition}[Exterior derivative]
The \emph{exterior derivative} is the $\mathbb{R}$--linear operator
\begin{align}
d:\Omega^p(M)\longrightarrow \Omega^{p+1}(M)
\end{align}
of degree $1$ such that $df$ is the differential of $f$ (defined in \cref{defcotspace}), i.e. $df=\partial_i f\, dx^i$ for $f\in C^\infty(M)$, and
\begin{align}
d(\alpha\wedge\beta) &= d\alpha\wedge\beta + (-1)^p\,\alpha\wedge d\beta, 
\qquad \alpha\in\Omega^p(M), \\
d^2 &= 0.
\end{align}
\end{definition}
\begin{proposition}[Existence and uniqueness]
There exists a unique operator $d$ with the properties above.  
In coordinates $(x^1,\dots,x^n)$, it is given by
\begin{align}
\omega = \tfrac{1}{p!}\,\omega_I\,dx^I
\quad\Rightarrow\quad
d\omega = \tfrac{1}{p!}\,\partial_j \omega_I\, dx^j\wedge dx^I,
\end{align}
with $dx^I\coloneqq dx^{i_1}\wedge\cdots\wedge dx^{i_p}$.
\end{proposition}
\begin{proof}
\emph{Existence.} On a chart, define for 
\begin{align}
\omega = \tfrac{1}{p!}\,\omega_I\,dx^I \in \Omega^p(U),
\end{align}
\begin{align}
d\omega = \tfrac{1}{p!}\,\partial_j\omega_I\,dx^j\wedge dx^I.
\end{align}
This operator is $\mathbb{R}$--linear, has degree $+1$, reduces to the classical differential on functions, and is a graded derivation. Nilpotency $d^2=0$ follows from the symmetry of second derivatives versus the antisymmetry of the wedge product.
On overlaps of charts, the coordinate formulas agree by the chain rule, so the local definitions glue to a global operator.\\
\emph{Uniqueness.} Suppose $D$ is another operator with the same properties. First, $D$ is local. Indeed, if $\omega$ vanishes near $p$, choose a smooth function $\chi$ supported in that neighborhood and equal to $1$ near $p$. Applying Leibniz to $\chi\omega=0$ gives $(D\omega)_p=0$. Thus $D$ acts on local forms using cutoff extensions. In local coordinates, $D x^i=dx^i$ and $D(dx^i)=D^2x^i=0$. Consequently,
\begin{align}
D\left(\frac{1}{p!}\omega_I\,dx^I\right)
=\frac{1}{p!}\,d\omega_I\wedge dx^I
=d\omega.
\end{align}
This holds in every chart, proving uniqueness.
\end{proof}
\begin{proposition}[Naturality of $d$]
Let $F:N\to M$ be a smooth map and $\omega\in\Omega^p(M)$. Then,
\begin{align}
F^*(d\omega) =  d\bigl(F^*\omega\bigr).
\end{align}
\end{proposition}
\begin{proof}
On a chart, we write $$\omega=\frac{1}{p!}\,\omega_{I}\,dx^{I}.$$ Then,
\begin{align}
F^*\omega
=\frac{1}{p!}\,(\omega_{I}\circ F)\,d(x^{I}\!\circ F)=\frac{1}{p!}\,(\omega_{i_1\dots i_p}\circ F)\,d(x^{i_1}\!\circ F)\wedge\cdots\wedge d(x^{i_p}\!\circ F),
\end{align}
which, using $d(\omega_{I}\circ F)=F^*(d\omega_{I})$ (chain rule) and $d^2=0$, leads to
\begin{align}
d(F^* \omega)
&=\frac{1}{p!}\,d(\omega_{I}\circ F)\wedge d(x^{I}\!\circ F)
=\frac{1}{p!}\,F^* (d\omega_{I})\wedge F^* (dx^{I})
=F^* \!\left(\frac{1}{p!}\,d\omega_{I}\wedge dx^{I}\right)\\
&=F^* (d\omega).
\end{align}
\end{proof}
\begin{remark}
    The exterior derivative $d$ makes $(\Omega^\bullet(M),\wedge)$ into a \textit{differential graded algebra} $(\Omega^\bullet(M),\wedge, d)$.
\end{remark}

\begin{lemma}[Frobenius for a single 1–form]
Let $\alpha\in\Omega^1(U)$ with $\alpha\wedge d\alpha=0$ and $\alpha_p\neq0$ for every $p\in U$. Then, locally
\begin{align}
\alpha  =  g\, df
\end{align}
for some smooth functions $f$ and $g$ with $g$ nowhere vanishing (equivalently, $\alpha=e^{h}\,df$ with $h\in C^\infty(M)$, to be intended locally).
\end{lemma}
\begin{proof}
Since $\alpha\wedge d\alpha=0$, the distribution $D\coloneqq\ker\alpha$ is involutive:
locally choose $Z$ with $\alpha(Z)=1$. For $X,Y\in\Gamma(D)$, evaluating $\alpha\wedge d\alpha$ on $(Z,X,Y)$ gives $d\alpha(X,Y)=0$. Since $d\alpha(X,Y)=-\alpha([X,Y])$, involutivity follows. By Frobenius, there exist local coordinates
$(y,x^2,\dots,x^n)$ such that $D=\mathrm{span}\{\partial_{x^2},\dots,\partial_{x^n}\}$, hence
$\alpha = a\,dy$ with $a\in C^\infty$ and $a\neq 0$. Setting $f:=y$ and $g:=a$ gives $\alpha=g\,df$.
(If desired, shrink to a connected neighborhood and replace $(f,g)$ by $(-f,-g)$ if necessary so that $g>0$, then write $g=e^h$.)
\end{proof}

\begin{theorem}[Frobenius, Pfaffian form]
Let $I\subset T^*M$ be a rank–$k$ subbundle, locally generated on an open set $U\subset M$ by $\theta^1,\dots,\theta^k$ (a local coframe for $I$), and let $\mathcal I$
be the graded algebraic ideal they generate in $\Omega^\bullet(U)$.\footnote{Explicitly, $\mathcal I=\sum_{a=1}^k\theta^a\wedge\Omega^\bullet(U)$. A graded ideal is called a \emph{differential ideal} when it additionally satisfies $d\mathcal I\subset\mathcal I$; this is a condition here, not part of the definition of the generated algebraic ideal.}
The following are equivalent:
\begin{enumerate}[label=\textnormal{(\roman*)}]
\item $d\mathcal I\subset\mathcal I$; equivalently, for each $a=1,\dots,k$ there exist $1$–forms $\eta^a{}_b\in\Omega^1(U)$ $(b=1,\dots,k)$ such that on $U$
\begin{align}
    d\theta^a = \sum_{b=1}^k \theta^b\wedge \eta^a{}_b.
\end{align}

\item $D\coloneqq \ker\theta:=\bigcap_{a=1}^k\ker\theta^a$ is a regular distribution of rank $n-k$ on $U$ and is involutive:
\begin{align}
    [\,\Gamma(D),\Gamma(D)\,]\subset\Gamma(D).
\end{align}

\item Around each point of $U$, after restricting to a smaller neighborhood, there exist coordinates $(x^1,\dots,x^{n-k},y^1,\dots,y^k)$ and an invertible matrix of smooth functions $(\mu^a{}_b)$ such that
\begin{align}
\theta^a  =  \sum_{b=1}^k \mu^a{}_b\,dy^b,\qquad a=1,\dots,k.
\end{align}
\end{enumerate}
In particular, through every point of $U$ there passes a unique maximal connected integral manifold of dimension $n-k$. Its local plaques in sufficiently small adapted charts are the connected components of joint level sets $y^a=c_a$.
\end{theorem}
\begin{proof}
\emph{(i) $\Leftrightarrow$ (ii).}
Let $D=\ker\theta$ and $\alpha\in\Gamma(I)$. For $X,Y\in\Gamma(D)$ we have
\begin{align}
d\alpha(X,Y) = X(\alpha(Y))-Y(\alpha(X))-\alpha([X,Y]) = -\alpha([X,Y]),
\end{align}
since $\alpha(X)=\alpha(Y)=0$. Thus $d\alpha\equiv 0\ \mathrm{mod}\ I$ if and only if $\alpha([X,Y])=0$ for all $\alpha\in I$, i.e. $[X,Y]\in\Gamma(D)$. Regularity of $D$ (rank $n-k$) follows from the constant rank of $I$.\\
\emph{(ii) $\Rightarrow$ (iii).}
By the Frobenius theorem for distributions, after shrinking $U$ about any chosen point, there exist local coordinates $(x^1,\dots,x^{n-k},y^1,\dots,y^k)$ on $U$ such that
\begin{align}
D=\mathrm{span}\{\partial_{x^1},\dots,\partial_{x^{n-k}}\}.
\end{align}
Hence any $1$–form that annihilates $D$ is a linear combination of $dy^1,\dots,dy^k$. In particular, for each $a$ there are smooth functions $\mu^a{}_b$ with
\begin{align}
\theta^a  =  \sum_{b=1}^k \mu^a{}_b\,dy^b.
\end{align}
Since $\{\theta^a\}$ is a coframe for $I$ and $\{dy^b\}$ is a coframe for the annihilator of $D$, the matrix $(\mu^a{}_b)$ is invertible.\\
\emph{(iii) $\Rightarrow$ (i).} Since $\theta^a=\sum_{b=1}^k \mu^a{}_b\,dy^b$ with $(\mu^a{}_b)$ invertible, we have
\begin{align}
d\theta^a
&= \sum_{b=1}^k d\mu^a{}_b \wedge dy^b + \sum_{b=1}^k \mu^a{}_b\, d(dy^b)
= \sum_{b=1}^k d\mu^a{}_b \wedge dy^b \\
&= -\sum_{b=1}^k \sum_{c=1}^k \theta^c \wedge \bigl((\mu^{-1})^b{}_c\, d\mu^a{}_b\bigr)
= \sum_{c=1}^k \theta^c \wedge \eta^a{}_c,
\end{align}
where $\eta^a{}_c \coloneqq -\sum_{b=1}^k (\mu^{-1})^b{}_c\, d\mu^a{}_b \in \Omega^1(U)$.
Thus $d\theta^a$ lies in the graded ideal locally. A partition of unity patches the local expressions on the original open set $U$, and Leibniz then gives $d\mathcal I\subset\mathcal I$.
\end{proof}

\begin{remark}[Normal form for an involutive Pfaffian system]
Let $I\subset T^*M$ be a smooth Pfaffian system of \emph{constant rank} $k$ on an open set $U\subset M$, locally
generated by an $I$–basis $\{\theta^1,\dots,\theta^k\}$, and assume the Frobenius condition
\begin{align}
d\theta^a  \equiv  \sum_{b=1}^k \omega^a{}_b \wedge \theta^b 
\quad\text{on }U \qquad (a=1,\dots,k),
\end{align}
for some $1$–forms $\omega^a{}_b$ on $U$. Then there exist local coordinates
\begin{align}
(x^1,\dots,x^{n-k},\,y^1,\dots,y^k)
\end{align}
on a sufficiently small neighborhood, again denoted $U$, and an invertible smooth matrix $\mu$ such that, after changing the generators,
\begin{align}
\widetilde\theta^a\coloneqq\sum_{b=1}^k(\mu^{-1})^a{}_b\theta^b=dy^a,\qquad a=1,\dots,k.
\end{align}
Equivalently, the map
\begin{align}
y  =  (y^1,\dots,y^k):\ U\longrightarrow \mathbb R^k
\end{align}
is a submersion with
\begin{align}
I  =  \mathrm{span}\{dy^1,\dots,dy^k\}
\qquad\text{and}\qquad
\ker dy_q  =  \bigcap_{a=1}^k \ker\theta^a{}_q\ \ \text{for all }q\in U.
\end{align}
In these coordinates, the local plaques are the connected components of the joint level sets
\begin{align}
L_c  =  \{\,q\in U:\ y^1(q)=c_1,\dots,y^k(q)=c_k\,\},
\end{align}
and their tangent spaces satisfy
\begin{align}
T_qL_c  =  \bigcap_{a=1}^k \ker\theta^a{}_q.
\end{align}
Thus the covectors $dy^1,\dots,dy^k$ vanish on vectors tangent to any integral submanifold, and the transverse
coordinates $y^1,\dots,y^k$ are constant along leaves. This description is \emph{local}: globally a leaf may pass
through several coordinate charts or have nontrivial topology.\\
In the dual normal form one chooses coordinates so that the annihilator distribution
\begin{align}
D  =  \{\,v\in TM:\ \theta^a(v)=0,\ a=1,\dots,k\,\}
\end{align}
is spanned by the coordinate fields (with respect to $(x^1,\dots,x^{n-k})$)
\begin{align}
D  =  \mathrm{span}\{\partial_1\,\dots,\partial_{n-k}\},
\end{align}
so the leaves are given by fixing $(y^1,\dots,y^k)$ as above.
\end{remark}

\begin{definition}[Smooth manifold with boundary]
Fix $n\in\mathbb N$ and let $\mathbb H^n\coloneqq\{x=(x^1,\dots,x^n)\in\mathbb R^n:\ x^n\ge 0\}$ be the closed half--space with the subspace topology from $\mathbb R^n$.
A \emph{smooth manifold with boundary} of dimension $n$ is a pair $(M,[\mathcal A])$ where:
\begin{itemize}
    \item $M$ is a topological space that is Hausdorff and second countable;
    \item $[\mathcal A]$ is a maximal smooth boundary atlas on $M$, i.e. an equivalence class of atlases 
    $\mathcal A=\{(U_i,\varphi_i)\}_{i\in I}$ such that:
    \begin{itemize}
        \item the open sets $\{U_i\}_{i\in I}$ cover $M$;
        \item each $\varphi_i:U_i\to\varphi_i(U_i)\subset\mathbb H^n$ is a homeomorphism onto an open subset of $\mathbb H^n$ (subspace topology);
        \item for every $i,j\in I$ with $U_i\cap U_j\neq\varnothing$, the transition maps
        \begin{align}
            \varphi_j\circ\varphi_i^{-1}:\ \varphi_i(U_i\cap U_j) \longrightarrow  \varphi_j(U_i\cap U_j)
        \end{align}
        are \emph{smooth up to the boundary}: namely, each transition map is the restriction to open subsets of $\mathbb H^n$ of a $C^\infty$--map defined on open neighborhoods in $\mathbb R^n$.
    \end{itemize}
\end{itemize}
Given such an atlas, the \emph{boundary} and the \emph{interior} of $M$ are
\begin{align}
\partial M  \coloneqq  \{\,p\in M \mid \exists (U,\varphi)\in\mathcal A,\ p\in U,\ \varphi(p)\in \partial\mathbb H^n=\{x^n=0\}\,\},\\
\mathring M  \coloneqq  M\setminus \partial M  =  \{\,p\in M \mid \exists (U,\varphi)\in\mathcal A,\ p\in U,\ \varphi(p)\in \{x^n>0\}\,\}.
\end{align}
These definitions are independent of the chosen chart.
\end{definition}

\begin{remark}
Notice that $\partial M$ is a closed subset of $M$ and carries a natural structure of $(n\!-\!1)$--dimensional smooth manifold (without boundary), induced by restricting charts to $\partial\mathbb H^n$. Moreover, $\mathring M$ is an open $n$--dimensional smooth manifold (without boundary). Notice also that this definition excludes the presence of \emph{corners}, where local models are open subsets of  $[0,\infty)^k\times\mathbb R^{n-k}$ for $k>1$.
\end{remark}

\begin{definition}[Orientable, compact, and closed manifold]\label{or_com_clos}
Let $M$ be an $n$--dimensional smooth manifold, possibly with boundary $\partial M$.
\begin{itemize}
\item[--] $M$ is \emph{orientable} iff it admits an atlas $\{(U_i,\varphi_i)\}$ whose transition maps have \emph{positive} Jacobian determinant on overlaps.
For manifolds with boundary we interpret this via the smooth extension of transition maps to open
neighborhoods in $\mathbb R^n$ and require $\det J>0$ at all points of the overlap (interior or boundary). Equivalently, there exists a nowhere–vanishing smooth $n$–form $\omega$ on $M$ (smooth up to the boundary). The charts whose transition maps have positive Jacobian are called \emph{orientation–preserving charts}
and they generate the maximal \emph{orientation–preserving atlas}.

\item[--] $M$ is \emph{compact} if it is compact as a topological space.

\item[--] $M$ is \emph{closed}\footnote{Here “closed” does \emph{not} mean “closed subset of an ambient space”.} if it is compact and without boundary, i.e. $\partial M=\varnothing$.
\end{itemize}
\end{definition}
\begin{remark}[Defining function and outward direction]
Let $M$ be a smooth $n$--dimensional manifold with boundary and $p\in\partial M$. 
A \emph{defining function} near $p$ is a smooth function $\rho\colon U\to\mathbb R$ on a neighborhood $U\subset M$ of $p$ such that
\begin{align}
\rho&>0\quad\text{on }U\setminus\partial M,\\
\partial M\cap U&=\{x\in U:\ \rho(x)=0\},\\
\rho_*|_x&\neq 0\quad\text{for all }x\in\partial M\cap U.
\end{align}
If $\rho,\tilde\rho$ are two defining functions near $p$, then there exist a neighborhood $V\subset U$ of $p$ and a smooth function $f>0$ on $V$ such that $\tilde\rho=f\,\rho$ on $V$. In particular, on $\partial M\cap V$ one has
\begin{align}
\tilde\rho_*|_x  =  f(x)\,\rho_*|_x,
\end{align}
and hence, at $p$,
\begin{align}
\tilde\rho_*|_p(\nu)=f(p)\,\rho_*|_p(\nu)\qquad\text{for all }\nu\in T_pM.
\end{align}
Therefore the sign of $\rho_*|_p(\nu)$ is independent of the choice of $\rho$.\\
A vector $\nu\in T_pM$ is \emph{outward–pointing} iff
\begin{align}
\rho_*|_p(\nu)<0,
\end{align}
and \emph{inward–pointing} iff $\rho_*|_p(\nu)>0$. Equivalent tests:
\begin{enumerate}
\item \textit{Curve criterion for transverse vectors.} If $\rho_*|_p(\nu)\neq0$, then $\nu$ is outward iff there exists a smooth curve $\gamma$ in a local extension of a boundary chart across $\partial M$, with
\begin{align}
\gamma(0)=p,\qquad \gamma'(0)=\nu,\qquad \gamma(t)\notin M \ \text{for all sufficiently small } t>0.
\end{align}
\item \textit{Boundary chart test.} In any boundary chart $\varphi:U\to\mathbb H^n=\{x^n\ge 0\}$ with $\varphi(p)\in\{x^n=0\}$,
write
\begin{align}
\nu=\sum_{i=1}^n \nu^i\,\partial_i\big|_p.
\end{align}
Then
\begin{align}
\text{$\nu$ is outward} \quad\Longleftrightarrow\quad \nu^n<0,
\end{align}
since a standard local choice is $\rho=x^n$ and thus $\rho_*|_p(\nu)=\nu^n$.
\end{enumerate}
With this sign choice, $-\partial_n$ is an outward transverse vector in a boundary chart. A unit normal requires a Riemannian metric $h$; with that additional structure it is $-\operatorname{grad}_h\rho/|\operatorname{grad}_h\rho|_h$ along the boundary.
\end{remark}

\begin{definition}[Integration of top forms]
Let $M$ be an oriented $n$--dimensional smooth manifold. Write $\omega\in\Omega^n_c(M)$ if $\omega$ has compact support, and then define\footnote{Here $\rho_\alpha\omega$ denotes the pointwise product.}
\begin{align}
\int_M \omega = \sum_\alpha \int_{\varphi_\alpha(U_\alpha)} (\varphi_\alpha^{-1})^*(\rho_\alpha\omega),
\end{align}
where $\{(U_\alpha,\varphi_\alpha)\}$ is an oriented atlas and $\{\rho_\alpha\}$ a partition of unity.\footnote{A \emph{partition of unity} subordinate to an open cover $\{U_\alpha\}$ of a smooth manifold $M$ is a locally finite family of smooth functions on $M$ $\{\rho_\alpha\}$ ($\forall\, p\in M\ \exists\, U$ s.t., $\forall \,\alpha$, $\rho_\alpha|_U=0$ apart from finitely many elements) such that
\begin{align}
&0\le \rho_\alpha\le 1,\quad \mathrm{supp}(\rho_\alpha)\subset U_\alpha,\  \text{and}\\
&\sum_{\alpha\in A}\rho_\alpha(p)=1\quad \text{for every }p\in M.
\end{align}
Our definition of a smooth manifolds implies that it is paracompact, and, on these spaces, partitions of unity subordinate to any open cover always exist.}
\end{definition}
\begin{theorem}[Stokes’ theorem]
If $M$ is an oriented compact $n$--dimensional smooth manifold with boundary $\partial M$ and $\omega\in\Omega^{n-1}(M)$, then
\begin{align}
\int_M d\omega = \int_{\partial M} i^*\omega,
\end{align}
where $i\colon\partial M\to M$ is the canonical boundary inclusion and the orientation on $\partial M$ is induced by the outward normal convention.
\end{theorem}
\begin{proof}[Sketch]
In a coordinate cube $[a,b]^n\subset\mathbb{R}^n$ Stokes’ theorem reduces to the fundamental theorem of calculus. By partition of unity one reduces the general case to this local statement, and cancellations along overlapping regions leave only the integral over the boundary of $M$.
\end{proof}
\begin{definition}[The de Rham cohomology]
The sequence
\begin{align}
0 \longrightarrow \Omega^0(M)\xrightarrow{d}\Omega^1(M)\xrightarrow{d}\cdots\xrightarrow{d}\Omega^n(M) \longrightarrow 0
\end{align}
is a cochain complex because $d^2=0$. Its cohomology is the \emph{de Rham cohomology}
\begin{align}
H^p_{\mathrm{dR}}(M)=\frac{\ker(d:\Omega^p(M)\to\Omega^{p+1}(M))}{\mathrm{im}(d:\Omega^{p-1}(M)\to\Omega^p(M))}.
\end{align}
\end{definition}
\begin{remark}[Topological significance]
De Rham’s theorem states that for a smooth manifold $M$, $H^p_{\mathrm{dR}}(M)\cong H^p(M;\mathbb{R})$, the singular cohomology with real coefficients.
\end{remark}

\section{The Lie derivative}
In the following sections, manifolds are understood to have no boundary unless otherwise specified.
The Lie derivative is the most “bare” way to measure the infinitesimal effect of a diffeomorphism on tensor fields, depending only on the flow of a vector field and on the functorial pull--back/pushforward action. In this sense it yields a representation of the Lie algebra $\mathfrak X(M)$ on the tensor algebra $\mathfrak T(M)$ and encodes symmetries via invariance conditions.
\begin{definition}[Local and maximal flow of a vector field]
Let $M$ be a smooth manifold and $X\in\Gamma(TM)$. A \emph{local flow} of $X$ is a smooth map
\begin{align}
\Phi:\ D\subset\mathbb R\times M \longrightarrow M
\end{align}
defined on an open $D$ with $\{0\}\times M\subset D$, with each time set $I(p)=\{t:(t,p)\in D\}$ an interval, such that:
\begin{align}
\Phi(0,p)&=p, \qquad \forall\, p\in M,\\[3pt]
\partial_t \Phi(t,p)&=X_{\Phi(t,p)}, \qquad \forall\, (t,p)\in D,\\[3pt]
\exists\, \text{open } J(t,p)&\ni 0: \{(t+\tau,p):\ \tau\in J(t,p)\}\cup\{(\tau,\Phi(t,p)):\ \tau\in J(t,p)\}\subset D,\\[3pt]
\Phi(t+\tau,p)&=\Phi\bigl(\tau,\Phi(t,p)\bigr), \qquad \forall\, \tau\in J(t,p).
\end{align}
For each $p\in M$, the time set
\begin{align}
I(p)\coloneqq\{t\in\mathbb R:\ (t,p)\in D\}
\end{align}
is an open interval containing $0$, and $t\mapsto\Phi(t,p)$ on $I(p)$ is called the \textit{integral curve} of $X$ with initial value $p$.\\
A local flow $\Phi:D\to M$ is \emph{maximal} if no local flow $\widetilde\Phi:\widetilde D\to M$ satisfies $D\subsetneq\widetilde D$ and $\widetilde\Phi|_D=\Phi$. Its domain then has the form
\begin{align}
D_{\max}=\{(t,p)\in\mathbb R\times M:\ t\in I_{\max}(p)\},
\end{align}
where $I_{\max}(p)$ is the maximal existence interval of $\dot\gamma=X\!\circ\!\gamma$, $\gamma(0)=p$.\\
The field $X$ is \emph{complete} iff $I_{\max}(p)=\mathbb R$ for every $p\in M$, equivalently $D_{\max}=\mathbb R\times M$.
\end{definition}

\begin{proposition}
Let $X\in\Gamma(TM)$ with maximal flow $\Phi:D\subset\mathbb R\times M\to M$. Moreover, for each $t$, define
\begin{align}
U_t\coloneqq\{p\in M:\ (t,p)\in D\}.
\end{align}
Then, the map
\begin{align}
\Phi_t:U_t\longrightarrow U_{-t},\qquad \Phi_t(p)= \Phi(t,p),
\end{align}
is a diffeomorphism with inverse $\Phi_{-t}$. In particular, if $D=\mathbb R\times M$, then $\Phi_t\colon M\to M$ is a diffeomorphism for all $t$.
\end{proposition}
\begin{proof}
Fix $t$ and $p\in U_t$. Let $\gamma_p:I(p)\to M$ be the maximal integral curve with $\gamma_p(0)=p$ and $\gamma_p(s)=\Phi(s,p)$. Define $q\coloneqq \Phi_t(p)=\gamma_p(t)$.\\
\textit{(i) Openness of domains.} By definition of flow, $U_t$ and $U_{-t}$ are open.\\
\textit{(ii) Inverse candidate.} Consider the integral curve $\gamma_q$ through $q$. By uniqueness and time-translation of ODE solutions,
\begin{align}
\gamma_q(s)=\gamma_p(t+s)\qquad\text{whenever both sides are defined.}
\end{align}
In particular $-t\in I(q)$ and
\begin{align}
\Phi_{-t}(q)=\gamma_q(-t)=\gamma_p(0)=p,
\end{align}
so $\Phi_{-t}\circ \Phi_t=\mathrm{id}$ on $U_t$. A symmetric argument gives $\Phi_t\circ \Phi_{-t}=\mathrm{id}$ on $U_{-t}$. Hence $\Phi_t:U_t\to U_{-t}$ is a bijection with inverse $\Phi_{-t}$.\\
\textit{(iii) Smoothness and local invertibility.} Smooth dependence on initial data yields that $\Phi_t$ and $\Phi_{-t}$ are smooth. Differentiating the identities
\begin{align}
\Phi_{-t}\big(\Phi_t(p)\big)=p,\qquad \Phi_t\big(\Phi_{-t}(q)\big)=q,
\end{align}
gives, for all $p\in U_t$,
\begin{align}
d\Phi_{-t}\big|_{\Phi_t(p)}\circ d\Phi_t\big|_{p}=\mathrm{id}_{T_pM},\qquad
d\Phi_t\big|_{\Phi_{-t}(q)}\circ d\Phi_{-t}\big|_{q}=\mathrm{id}_{T_qM}.
\end{align}
Thus $d\Phi_t\big|_{p}$ is a linear isomorphism with inverse $d\Phi_{-t}\big|_{\Phi_t(p)}$. By the inverse function theorem, $\Phi_t$ is a local diffeomorphism at every $p\in U_t$; since $\Phi_{-t}$ is its smooth inverse, $\Phi_t$ is a diffeomorphism from $U_t$ onto $U_{-t}$.
\end{proof}

\begin{definition}[Lie derivative as flow of a vector field]\label{lie1}
Let $M$ be a smooth manifold and $X\in\Gamma(TM)$. Moreover, let $\Phi:D\subset\mathbb R\times M\to M$ be the local flow of $X$. Then, for any function $f\in C^\infty(M)$
\begin{align}
\mathcal L_X f  =  \frac{d}{dt}\Big|_{t=0}\,f\circ\Phi_{t},
\end{align}
and for any vector field $Y\in\Gamma(TM)$
\begin{align}
\mathcal L_X Y  =  \frac{d}{dt}\Big|_{t=0}\,(\Phi_{-t})_*  Y,
\end{align}
where $(\Phi_t)_* $ denotes the natural action of diffeomorphisms on vector fields via the pushforward.
\end{definition}
\begin{definition}[Lie derivative as algebraic operation]\label{lie2}
Let $M$ be a smooth manifold and $X\in\Gamma(TM)$. The \emph{Lie derivative} along $X$ is the unique $\mathbb R$–linear operator
\begin{align}
\mathcal L_X:\ \mathfrak T(M)\longrightarrow \mathfrak T(M)
\end{align}
on the full tensor field algebra $\displaystyle{\mathfrak T(M)\coloneqq\bigoplus_{r,s\ge0}\Gamma(T^{(r,s)}(M))}$ of $M$ which:
\begin{align}
&\text{(i) acts on functions by }\ \mathcal L_X f = X[f],\quad f\in C^\infty(M);\\
&\text{(ii) acts on vector fields by }\ \mathcal L_X Y = [X,Y],\quad Y\in\Gamma(TM);\\
&\text{(iii) is a derivation of degree $0$:} \ \mathcal L_X(T\otimes S)=(\mathcal L_XT)\otimes S + T\otimes(\mathcal L_XS);\\
&\text{(iv) commutes with contractions of indices defined in \cref{opetensors}.}
\end{align}
\end{definition}
\begin{remark}
    Since it is a degree $0$ operator, it maps homogeneous elements of $\mathfrak T(M)$ to others of the same type, i.e. $\mathcal{L}_X\colon \Gamma(T^{(r,s)}(M))\to \Gamma(T^{(r,s)}(M))$.
\end{remark}
\begin{remark}\label{liefromfunvec}
Since $\mathcal L_X$ is an $\mathbb R$–derivation of degree $0$ for $\otimes$ and it
commutes with the canonical contractions, \normalsize{\textbf{its action on arbitrary tensors is uniquely determined by its action on functions and vector fields}}.  
Indeed, for $\alpha\in\Omega^1(M)$ and $Y\in\Gamma(TM)$, commuting with the pairing/natural contraction
$\langle\ ,\ \rangle:T^*M\otimes TM\to C^\infty(M)$ gives\footnote{Notice that we are simply writing $\mathcal L_X\langle \alpha, Y\rangle=\mathcal L_X((\alpha\otimes Y)^i_{\ i})=(\mathcal L_X(\alpha\otimes Y))^i_{\ i}=(\mathcal L_X\alpha\otimes Y)^i_{\ i}+(\alpha\otimes \mathcal L_XY)^i_{\ i}=\langle \mathcal L_X\alpha, Y\rangle+\langle \alpha, \mathcal L_XY\rangle$, where we used the commutation with the contraction of indices, the Leibniz rule, and the linearity of the contraction operation itself (or of $\langle\ ,\ \rangle$ if you want).}
\begin{align}
\langle \mathcal L_X\alpha,  Y\rangle
= \mathcal L_X\langle \alpha,  Y\rangle - \langle \alpha,  \mathcal L_X Y\rangle
= X\big(\alpha(Y)\big) - \alpha([X,Y]),
\end{align}
so $\mathcal L_X$ on $1$–forms is fixed by its values on $C^\infty(M)$ and $\Gamma(TM)$;
Leibniz and $\mathbb R$--linearity then extend it to all tensor types. Equivalently, one may prescribe it on
functions and $1$–forms.
\end{remark}
\begin{proposition}
    \cref{lie1} and \cref{lie2} coincide on functions and vector fields, and thus, by \cref{liefromfunvec}, on all tensor fields.
\end{proposition}
\begin{proof}
On functions we observe that, on every $f\in C^\infty(M)$, we have
\begin{align}
\mathcal L_X f  =  \frac{d}{dt}\Big|_{t=0}\,f\circ\Phi_{t}=Xf,
\end{align}
by the chain rule for the derivative $\frac{d}{dt}$.\\
For proving it on vector fields, notice that the following (typed) identity holds: for smooth $f_t$,
\begin{align}
     \frac{d}{dt}\Big|_{t=0}\big(f_t\circ \Phi_t\big)
      =  \frac{d}{dt}\Big|_{t=0}f_t
      +  df_0\!\left(\frac{d}{dt}\Big|_{t=0}\Phi_t\right)
      =  \frac{d}{dt}\Big|_{t=0}f_t  +  Xf_0.
\end{align}
By definition of pushforward, for every $f\in C^\infty(U)$, we have
\begin{align}
    (\Phi_{-t})_*Y\,f  =  Y(f\circ \Phi_{-t})\circ\Phi_{t}\in C^\infty(U),
\end{align}
where we used that $\Phi_t^{-1}=\Phi_{-t}$ for all $t$. From the identity above, it follows
\begin{align}
    \frac{d}{dt}\Big|_{t=0}(\Phi_{-t})_*  Y\,f
    &= \frac{d}{dt}\Big|_{t=0}\Big(Y(f\circ \Phi_{-t})\circ\Phi_{t}\Big)\\[4pt]
    &= \underbrace{\frac{d}{dt}\Big|_{t=0} Y(f\circ \Phi_{-t})}_{=\,-\,Y(Xf)}
       + \underbrace{\frac{d}{dt}\Big|_{t=0}\big((Yf)\circ \Phi_{t}\big)}_{=\,X(Yf)}\\[4pt]
    &= \,[X,Y]\,f,
\end{align}
where we used $\displaystyle{\frac{d}{dt}\Big|_{t=0}f\circ\Phi_t=Xf}$ and $\displaystyle{\frac{d}{dt}\Big|_{t=0}\Phi_{t}=X}$. Since this holds for all $f$, it follows that
\begin{align}
    \mathcal L_XY=\frac{d}{dt}\Big|_{t=0}(\Phi_{-t})_*  Y  =  [X,Y],
\end{align}
which completes the proof.
\end{proof}

\begin{proposition}[Operational formula for the Lie derivative on tensors]
Let $T$ be an $(r,s)$–tensor field. For $\alpha^1,\dots,\alpha^r\in\Omega^1(M)$ and $Y_1,\dots,Y_s\in\Gamma(TM)$ one has
\begin{align}
\mathcal L_X T(\alpha^1,\dots,\alpha^r;\,Y_1,\dots,Y_s)
&= X\!\big( T(\alpha^1,\dots,\alpha^r;\,Y_1,\dots,Y_s) \big)\\
&\quad - \sum_{i=1}^r T(\alpha^1,\dots,\mathcal L_X\alpha^i,\dots,\alpha^r;\,Y_1,\dots,Y_s)\\
&\quad - \sum_{j=1}^s T(\alpha^1,\dots,\alpha^r;\,Y_1,\dots,\mathcal L_X Y_j,\dots,Y_s).
\end{align}
Here, $\mathcal L_X Y=[X,Y]$ on vector fields and $\mathcal L_X\alpha=d\,\iota_X\alpha+\iota_X d\alpha$ on $1$--forms.
\end{proposition}
\begin{proof}
Apply $\mathcal L_X$ to the scalar function $T(\alpha^1,\dots,\alpha^r;Y_1,\dots,Y_s)$. This evaluation is a contraction of the tensor product of $T$ with its arguments. Compatibility with contractions and Leibniz express its derivative as the evaluation of $\mathcal L_XT$ plus one term for the derivative of each argument. On the scalar function the derivative is $X$. Moving the argument terms to the other side gives the stated formula, with a minus sign for both sums.
\end{proof}
\begin{remark}
    Notice that this result allows the complete equivalence (operationally) between the algebraic and the geometric definitions of the Lie derivative. This is achieved by generalizing \cref{lie1} to tensor fields by the pushforward and the pull-back on the contravariant and covariant slots. This is in agreement with \cref{liefromfunvec}.
\end{remark}
\begin{proposition}[Naturality of the Lie derivative under pushforward/pull--back]
Let $F:M\to N$ be a diffeomorphism. Let $X\in\Gamma(TN)$ and $\widetilde X\in\Gamma(TM)$ be the $F$–related vector fields, i.e.
\begin{align}
dF_p(\widetilde X_p)=X_{F(p)}\qquad(p\in M).
\end{align}
Denote by $\Phi_t$ and $\widetilde\Phi_t$ the local flows of $X$ and $\widetilde X$, respectively.
Then, for all tensor fields $S\in \mathfrak T(M)$ and $T\in\mathfrak T(N)$, one has
\begin{align}
F^*(\mathcal L_X T) &= \mathcal L_{\widetilde X}(F^* T),\\
F_*(\mathcal L_{\widetilde X} S) &= \mathcal L_X(F_* S).
\end{align}
\end{proposition}
\begin{proof}
The flows intertwine with $F$:
\begin{align}
F\circ \widetilde\Phi_t  =  \Phi_t\circ F .
\end{align}
We apply the natural action of diffeomorphisms on vector fields and $1$--forms and then extend this to all tensors. This gives
\begin{align}\label{1}
F^*\circ (\Phi_t)^*  =  (\widetilde\Phi_t)^* \circ F^*
\end{align}
for $1$--forms on $N$, and
\begin{align}\label{2}
F_*\circ (\widetilde\Phi_t)_*  =  (\Phi_t)_*\circ F_*,
\end{align}
for vector fields on $M$.\\[4pt]
Differentiating \eqref{1} at $t=0$ and using the flow definition of the Lie derivative on covariant tensors, we obtain, for every $1$--form $\alpha\in\Omega^1(N)$,
\begin{align}
    \frac{d}{dt}\Big|_{t=0}\bigl(F^*(\Phi_t^*\alpha)\bigr)
     =  \frac{d}{dt}\Big|_{t=0}\bigl((\widetilde\Phi_t)^*(F^*\alpha)\bigr),
\end{align}
hence
\begin{align}
F^{* }\bigl(\mathcal L_X\alpha\bigr)
 =  \mathcal L_{\widetilde X}\bigl(F^{* }\alpha\bigr).
\end{align}
This works verbatim for functions $f\in C^\infty(N)$ (viewed as $0$--forms).\\[4pt]
For vector fields, using the contravariant flow formula
$\displaystyle \mathcal L_X Y=\frac{d}{dt}\big|_{t=0}(\Phi_{-t})_*Y$,
apply \eqref{2} with $-t$ in place of $t$ and then differentiate at $t=0$:
\begin{align}
  \frac{d}{dt}\Big|_{t=0}\bigl(F_*(\widetilde\Phi_{-t})_* Y\bigr)
   =  \frac{d}{dt}\Big|_{t=0}\bigl((\Phi_{-t})_*(F_*Y)\bigr),
\end{align}
which yields, for every $Y\in\Gamma(TM)$,
\begin{align}
F_{* }\bigl(\mathcal L_{\widetilde X}Y\bigr)
 =  \mathcal L_X\bigl(F_{* }Y\bigr).
\end{align}
Finally, both operators $F^*\circ\mathcal L_X$ and $\mathcal L_{\widetilde X}\circ F^*$ are degree-$0$ derivations relative to $F^*$: $D(T\otimes S)=D(T)\otimes F^*S+F^*T\otimes D(S)$. They commute with natural contractions; the analogous statement holds relative to $F_*$ for the pushforward identity. Since they agree on functions ($0$--forms), vector fields, and $1$--forms by the computations above, by \cref{liefromfunvec} they agree on all $(r,s)$--tensors (Leibniz rule and compatibility with contractions). This proves both identities.
\end{proof}
\begin{definition}[Action on tensor fields]
    We denote by $\rho_t$ the generic action of the flows on the tensor fields of mixed type, namely $(\Phi_{-t})_*$ on contravariant slots and $(\Phi_t)^*$ on covariant slots, so that $\rho_t=(\Phi_t)^*$ is the full tensor pull-back and $\frac{d}{dt}|_{t=0}\rho_t=\mathcal L_X$.\footnote{In the literature, by abuse of notation, this is sometimes written directly as $(\Phi_t)^*$ or $(\Phi_t)_*$.}
\end{definition}
\begin{proposition}[Cartan’s magic formula]
For every smooth vector field $X\in\Gamma(TM)$ and every differential form $\omega\in\Omega^\bullet(M)$,
\begin{align}
\mathcal L_X \omega  =  d\,\iota_X \omega  +  \iota_X d\omega.
\end{align}
\end{proposition}
\begin{proof}
Both $\mathcal L_X$ and $d\,\iota_X+\iota_X d$ are $\mathbb R$–linear endomorphisms of $\Omega^\bullet(M)$
and derivations of degree $0$ for the wedge product:
\begin{align}
\mathcal L_X(\alpha\wedge\beta)&=(\mathcal L_X\alpha)\wedge\beta+\alpha\wedge(\mathcal L_X\beta),\\
(d\,\iota_X+\iota_X d)(\alpha\wedge\beta)&=(d\,\iota_X+\iota_X d)\alpha\wedge\beta+\alpha\wedge(d\,\iota_X+\iota_X d)\beta.
\end{align}
Therefore, it suffices to verify the identity on $0$–forms and on $1$–forms, and then extend it to all forms by linearity.\\
\textit{Case of $0$–forms.} For $f\in C^\infty(M)$,
\begin{align}
(d\,\iota_X+\iota_X d)f  =  d(0)+\iota_Xdf = \iota_Xdf = Xf = \mathcal L_X f.
\end{align}
\textit{Case of $1$–forms.} Let $\alpha\in\Omega^1(M)$ and $Y\in\Gamma(TM)$. Then
\begin{align}
\mathcal L_X\alpha(Y)  &=  X(\alpha(Y))-\alpha(\mathcal L_X Y)
 =  X(\alpha(Y))-\alpha([X,Y]),\\
d\,\iota_X\alpha(Y)  &=  Y(\iota_X\alpha) =  Y(\alpha(X)),\\
\iota_X d\alpha(Y)  &=  d\alpha(X,Y).
\end{align}
Using the standard expression for $d\alpha$,
\begin{align}
d\alpha(X,Y)  =  X(\alpha(Y))-Y(\alpha(X))-\alpha([X,Y]),
\end{align}
we obtain
\begin{align}
(d\,\iota_X\alpha+\iota_X d\alpha)(Y)
&= Y(\alpha(X)) + X(\alpha(Y)) - Y(\alpha(X)) - \alpha([X,Y])\\
&= X(\alpha(Y)) - \alpha([X,Y])\\
&= \mathcal L_X\alpha(Y).
\end{align}
Therefore the formula holds on $0$– and $1$–forms, and by the derivation property it extends to all differential forms.
\end{proof}

\begin{remark}\label{opliealg}
Let $\mathcal A$ be any associative real algebra (e.g. $\mathrm{End}_{\mathbb R}(\mathfrak X(M))$ or $\mathrm{End}_{\mathbb R}(\mathfrak T(M))$). 
For $A,B\in\mathcal A$ define the commutator
\begin{align}
[A,B]\coloneqq AB-BA.
\end{align}
Then, for all $A,B,C\in\mathcal A$ and $a\in\mathbb R$,
\begin{align}
&[A,B]=-[B,A],\\
&[A,B+C]=[A,B]+[A,C],\\
&[aA,B]=a[A,B],\\
&[A,BC]=[A,B]\,C+B\,[A,C]\qquad\text{(Leibniz)},\\
&[A,[B,C]]+[B,[C,A]]+[C,[A,B]]=0\qquad\text{(Jacobi)}.
\end{align}
Hence, $(\mathcal A,[\,,\,])$ is a Lie algebra.
\end{remark}

\begin{remark}
Notice that it may be useful to define the notion of a \textit{graded commutator}, for example, when dealing with differential forms. In this setting, for homogeneous operators $A,B$ the commutation relations pick up the sign $(-1)^{|A||B|}$, where $|\cdot|\colon \mathcal A^\bullet\to\{0,1\}$ denotes the parity (degree mod $2$).
\end{remark}

\begin{proposition}[Cartan commutation identities]\label{prop:cartan-comm}
Let $X,Y\in \Gamma(TM)$. Then, the following hold:\footnote{Here $[\,,\,]$ denotes the operator commutator as in \cref{opliealg}. Notice that, in the last two relations, they must be intended on $\Omega^\bullet(M)$; in particular, the commutator is graded.}
\begin{align}
\textnormal{(i)}\quad & [\mathcal L_X,\mathcal L_Y]  =  \mathcal L_{[X,Y]},\\
\textnormal{(ii)}\quad & [\mathcal L_X,\iota_Y]  =  \iota_{[X,Y]},\\
\textnormal{(iii)}\quad & [\mathcal L_X,d]  =  0.
\end{align}
\end{proposition}

\begin{proof}
\textit{(i) On tensors.} By \cref{liefromfunvec}, it suffices to check on functions and vector fields.
For $f\in C^\infty(M)$,
\begin{align}
[\mathcal L_X,\mathcal L_Y]f
= X(Yf)-Y(Xf)
= [X,Y]f
= \mathcal L_{[X,Y]}f.
\end{align}
For $Z\in\Gamma(TM)$,
\begin{align}
[\mathcal L_X,\mathcal L_Y]Z
= [X,[Y,Z]]-[Y,[X,Z]]
= [[X,Y],Z]
= \mathcal L_{[X,Y]}Z,
\end{align}
by the Jacobi identity of the Lie bracket on vector fields. Hence (i) holds on all tensors.\\
\textit{(ii) On differential forms.} First verify on $0$– and $1$–forms; then extend by derivation.
For $f\in C^\infty(M)$,
\begin{align}
[\mathcal L_X,\iota_Y]f
= \mathcal L_X(0)-\iota_Y(Xf)
= 0
= \iota_{[X,Y]}f.
\end{align}
For $\alpha\in\Omega^1(M)$,
\begin{align}
[\mathcal L_X,\iota_Y]\alpha
&= \mathcal L_X\big(\alpha(Y)\big)-\iota_Y\big(\mathcal L_X\alpha\big)\\
&= X\big(\alpha(Y)\big)-\big(d\,\iota_X\alpha+\iota_X d\alpha\big)(Y)\\
&= X\big(\alpha(Y)\big)-Y\big(\alpha(X)\big)-d\alpha(X,Y).
\end{align}
Using $d\alpha(X,Y)=X(\alpha(Y))-Y(\alpha(X))-\alpha([X,Y])$, we obtain
\begin{align}
[\mathcal L_X,\iota_Y]\alpha
= \alpha([X,Y])
= \iota_{[X,Y]}\alpha.
\end{align}
Since $[\mathcal L_X,\iota_Y]$ is a derivation of degree $-1$ and $\iota_{[X,Y]}$ is the unique such derivation agreeing on $0$– and $1$–forms, (ii) follows for all forms.\\
\textit{(iii) On differential forms.} Using Cartan’s identity and $d^2=0$,
\begin{align}
[\mathcal L_X,d]
&= [d\,\iota_X+\iota_X d,\ d]\\
&= d\,\iota_X d + \iota_X d^2 - d^2 \iota_X - d\,\iota_X d\\
&= 0.
\end{align}
\end{proof}
\begin{definition}[Flow loop]\label{defflow}
Let $X,Y\in\mathfrak X(M)$ have local flows $\Phi^X_t,\Phi^Y_s$. For those $s,t$ for which all compositions are defined, set
\begin{align}
\Psi_{s,t} \coloneqq \Phi^X_t \circ \Phi^Y_s \circ \Phi^X_{-t} \circ \Phi^Y_{-s}.
\end{align}
We say that $X$ and $Y$ generate a \emph{flow loop} if every $p\in M$ has a neighborhood $U$ and open time intervals $I,J$ containing $0$ such that all compositions are defined on $U$ and
\begin{align}
\Psi_{s,t}|_U=\mathrm{id}_U\qquad\text{for all }(s,t)\in I\times J.
\end{align}
If the flows are complete, we say they generate a \emph{global flow loop} when
\begin{align}
\Psi_{s,t}=\mathrm{id}_M \qquad \text{for all } s,t\in\mathbb R .
\end{align}
The map $\Psi_{s,t}$ is the \emph{flow commutator}. It measures the failure of commutation on common local domains (with $\mathrm{id}$ denoting the identity on the relevant domain):
\begin{align}
\Psi_{s,t}=\mathrm{id} \quad \Longleftrightarrow \quad
\Phi^X_t\circ\Phi^Y_s=\Phi^Y_s\circ\Phi^X_t.
\end{align}
\end{definition}

\begin{remark}[Isolated loops do not suffice]
The existence of single pairs $(s,t)$ with $\Psi_{s,t}=\mathrm{id}_M$ does \emph{not} imply that  $\Psi_{s,t}$ is a loop; one needs the identity on local neighborhoods for all sufficiently small $(s,t)\in I\times J$.
\end{remark}

\begin{remark}[Geometric meaning of the Lie derivative]\label{liegeoint}
\leavevmode\newline
\vspace{-0.55cm}
\begin{itemize}
\item[--] \textit{Natural transport.} For any vector field $Y\in\Gamma(TM)$, its flow moves points along integral curves and acts functorially on tensor fields via pull--back/pushforward. In this sense, we say that this is a natural way to transport tensor fields on a manifold, one needs no extra structure.
\begin{mdframed}
\normalsize{\textbf{The action of flows on tensor fields is called \textit{natural transport}}}.
\end{mdframed}

\item[--]\textit{Composition of flows and commutativity.}
For $X,Y\in\Gamma(TM)$, consider the composition
\begin{align}
\Psi_{s,t}\coloneqq \Phi_t^X\circ \Phi_s^Y\circ \Phi_{-t}^X\circ \Phi_{-s}^Y.
\end{align}
A second-order Taylor expansion of the pull-back action of $\Psi_{s,t}$ on $f\in C^\infty(M)$ gives
\begin{align}\label{B}
f\circ\Psi_{s,t}=f-st[X,Y]f+O\big(|st|(|s|+|t|)\big),
\end{align}
locally on compact subsets. This proves that a flow loop implies $[X,Y]=0$. Conversely, if $[X,Y]=0$, then
\begin{align}
\frac{d}{dt}(\Phi^X_{-t})_*Y=(\Phi^X_{-t})_*[X,Y]=0.
\end{align}
Thus $\Phi^X_t$ preserves $Y$, and uniqueness of its integral curves implies local commutation of the flows. Therefore, on the appropriate local domains, we have
\begin{align}
    [X,Y]=0 \quad\Longleftrightarrow \quad \Phi^X_t\circ\Phi^Y_s=\Phi^Y_s\circ\Phi^X_t,
\end{align}
which is the commutativity of flows discussed in \cref{defflow}.
\begin{mdframed}
\normalsize{\textbf{This implies: $[X,Y]=0$ if and only if $\Psi_{s,t}$ is a flow loop}}.
\end{mdframed}

\item[--]\textit{Infinitesimal generator.}\footnote{The notion arises from one--parameter group actions and symmetry generators; we will return to it in the next sections.}
The Taylor series at $t=0$ gives
\begin{align}
(\Phi_{-t}^X)_{* }Y  =  Y  +  t\,\mathcal L_X Y  +  O(t^2).
\end{align}
\vspace{-0.65cm}
\begin{mdframed}
\normalsize{\textbf{The Lie derivative is the infinitesimal generator of the natural transport.}}
\end{mdframed}

\item[--] \textit{Cartan relation via geometric construction.}
Generalizing \eqref{B} to the first--order for the action on vector fields (or on tensor fields via $\rho_t$), yields
\begin{align}
(\Psi_{s,t}^{-1})_*=(\Phi_{s}^Y)_{* }\circ(\Phi_{t}^X)_{* }\circ(\Phi_{-s}^Y)_{* }\circ(\Phi_{-t}^X)_{* }  =  \mathrm{id}  -  st\,\mathcal L_{[X,Y]}  +  o(st),
\end{align}
since to mixed order $st$ the map $\Psi_{s,t}$ agrees with the flow of $-[X,Y]$ with parameter $st$. Expanding the operator product on the left gives $\mathrm{id}-st[\mathcal L_X,\mathcal L_Y]+o(st)$, so comparing mixed derivatives at $s=t=0$ yields the identity on tensor fields
\begin{align}
[\mathcal L_X,\mathcal L_Y]  =  \mathcal L_{[X,Y]}.
\end{align}
\vspace{-0.65cm}
\begin{mdframed}
\normalsize{\textbf{At the infinitesimal level, the drag along $X$ then along $Y$ minus the drag along $Y$ then along $X$ equals the drag along $[X,Y]$}}.
\end{mdframed}

\textit{Flatness.}
For a flow loop in the sense of \cref{defflow}, the full local diffeomorphism is the identity, so its induced action on tensors is also the identity. Equivalently, $[X,Y]=0$, and the order of the two local flows does not matter. In particular,
\begin{align}
    [\mathcal L_X,\mathcal L_Y]=0.
\end{align}
\vspace{-0.65cm}
\begin{mdframed}
\normalsize{\textbf{Lie derivatives commute if and only if they generate a flow loop}}.
\end{mdframed}

\item[--] \textit{Invariance.}
A vector field $Y\in\mathfrak X(M)$ is \textit{invariant under the flow of $X$} if and only if\footnote{In general, for tensor fields, we have: let $\rho_t$ be the natural action of the flow $\Phi^X_t$ on tensors by pushforward and pull--back
($\rho_{t+s}=\rho_t\circ\rho_s$, $\rho_0=\mathrm{id}$). By definition, $\mathcal L_X T=\frac{d}{dt}\big|_{0}\rho_t(T)$ and
$\frac{d}{dt}\rho_t(T)=\rho_t(\mathcal L_X T)$. Hence, $\mathcal L_X T=0 \Rightarrow \rho_t(T)=T$ (solution of $\dot T_t=0$), and conversely
$\rho_t(T)=T$ for all $t$ implies $\mathcal L_X T=0$ by differentiating at $t=0$.}

\begin{align}
\mathcal L_X Y  =  0 \Longleftrightarrow (\Phi^X_t)_*Y=Y,
\end{align}
which can be generalized to any tensor field $T\in\mathfrak T(M)$ via $\rho_t$.
\begin{mdframed}
\normalsize{\textbf{The Lie derivative incorporates the symmetry of tensor fields along flows}}.
\end{mdframed}

\item[--] \textit{Mapping of integral curves.} Since, if $\mathcal L_X Y=0$, $Y\in\mathfrak X(M)$ is invariant under the flow of $X$, i.e. $(\Phi^X_t)_*Y=Y$, then, for any integral curve $\gamma\colon I\to M$ of $Y$ (i.e.\ $\dot\gamma(s)\coloneqq Y_{\gamma(s)}$), the transported curve $\widetilde\gamma(s)\coloneqq \Phi^X_t(\gamma(s))$ satisfies
\begin{align}
\dot{\widetilde\gamma}(s)
= \big((\Phi^X_t)_*Y\big)_{\Phi^X_t(\gamma(s))}
= Y_{\Phi^X_t(\gamma(s))}
= Y_{\widetilde\gamma(s)},
\end{align}
so $\widetilde\gamma$ is again an integral curve of $Y$. If also $\mathcal L_Y X=0$, the same holds with $X$ and $Y$ swapped.
\begin{mdframed}
\normalsize{\textbf{The flow of $X$ maps integral curves of $Y$ to integral curves of $Y$ and vice versa, if and only if $[\mathcal L_X,\mathcal L_Y]=0$}}.
\end{mdframed}
\end{itemize}
\end{remark}
\begin{remark}[Coordinate frame]
    Given a frame $(X^1,\dots,X^n)$, we have that:
\begin{align}
    \text{coordinate frame}\quad\Longleftrightarrow\quad \mathcal L_{X^i}X^j=0,
\end{align}
for all $i,j=1,\dots,n$.\footnote{In particular, coordinate frames generate (pair--wise) loops and map their integral curves into their integral curves.}
\end{remark}

\begin{remark}[On a volume form]
    If $\mu$ is a nowhere--vanishing $n$--form (called a \textit{volume form}, see \cref{or_com_clos}), then
\begin{align}
    \mathcal L_X\mu  =  (\operatorname{div}_\mu X)\,\mu.
\end{align}
\end{remark}

\begin{example}
Let $M=\mathbb R^2$ with coordinates $(x,y)$, area form $\mu=dx\wedge dy$, and the \emph{dilation} field
\begin{align}
X  =  x\,\partial_x + y\,\partial_y.
\end{align}
Its flow is $\Phi_t(z)=e^t z$ (homothetic expansion). Then:
\begin{align}
&\text{(a) On functions:}\quad \mathcal L_X f  =  x\,\partial_x f + y\,\partial_y f.\\
&\text{(b) On the area form:}\quad \mathcal L_X(dx\wedge dy)  =  (\operatorname{div} X)\,dx\wedge dy  =  2\,dx\wedge dy,
\end{align}
since $\operatorname{div}X=\partial_x(x)+\partial_y(y)=2$. Therefore, for small $t$,
\begin{align}
(\Phi_t)^* (dx\wedge dy)  =  e^{2t}\,dx\wedge dy
\qquad\Rightarrow\qquad
\frac{d}{dt}\Big|_{t=0}(\Phi_t)^* (dx\wedge dy)  =  2\,dx\wedge dy,
\end{align}
in agreement with $\mathcal L_X\mu=2\mu$: the flow expands areas at instantaneous rate $2$.\\
Finally, take $Y=x\,\partial_x - y\,\partial_y$. Then
\begin{align}
\mathcal L_X Y  =  [X,Y]  =  0,
\end{align}
because both $X$ and $Y$ are linear (diagonal) vector fields whose flows commute. Here $[X,Y]=0$ means that transporting $Y$ along the flow of $X$ and comparing back at the base point produces no first–order change.
\end{example}

\begin{remark}[Algebraic interpretation of the Lie derivative]\label{liealgint}
A (left) representation of a Lie algebra $(\mathfrak g,[\,,\,])$ on a real vector space $V$ is a linear map
\begin{align}
\rho:\ \mathfrak g \longrightarrow \mathrm{End}_{\mathbb R}(V)
\end{align}
such that
\begin{align}
\rho([u,v])  = [\rho(u),\rho(v)] \coloneqq  \rho(u)\circ \rho(v)  -  \rho(v)\circ \rho(u) \qquad \forall\,u,v\in\mathfrak g.
\end{align}
For a smooth manifold $M$, set
\begin{align}
\mathfrak T(M)\coloneqq \bigoplus_{r,s\ge 0}\Gamma\!\bigl(T^{(r,s)}M\bigr),
\end{align}
viewed simply as a real vector space (it is a $C^\infty(M)$--module). Then
\begin{align}
\mathfrak X(M)\ni X \longmapsto \mathcal L_X \in \mathrm{End}_{\mathbb R}\!\bigl(\mathfrak T(M)\bigr)
\end{align}
is a representation, because for all $X,Y\in\mathfrak X(M)$ one has
\begin{align}
[\mathcal L_X,\mathcal L_Y]  =  \mathcal L_{[X,Y]}.
\end{align}
\end{remark}
\begin{proposition}[$\mathcal L_X$ is not $f$--linear]\label{prop:not-Cinfty-linear}
The Lie derivative is $\mathbb R$–linear but, in general, not $C^\infty(M)$–linear, in the vector–field parameter.\footnote{Notice that neither is it in the tensor field one.} In particular,
\begin{align}
\mathcal L_{fX}T  \neq  f\,\mathcal L_X T,
\end{align}
for generic $f\in C^\infty(M)$ and tensor fields $T\in\mathfrak T(M)$.
\end{proposition}
\begin{proof}
Trivial.
\end{proof}

\begin{remark}
A flow of a vector field gives rise to a natural/canonical type of transport of tensor fields on a manifold, the \textit{natural transport}: pushforward on contravariant slots and pull--back on covariant ones. From this canonical construction arises the Lie derivative and its geometric interpretation is the \textit{infinitesimal generator of this transport}, which captures the first--order change produced by dragging a tensor along the flow and comparing it back at its original point: \textit{the Lie derivative emerges naturally from the functoriality of how diffeomorphisms act on tensors}. Then, in \cref{liegeoint}, via this geometric construction, we derived the identity of Cartan calculus
\begin{align}
[\mathcal L_X,\mathcal L_Y]=\mathcal L_{[X,Y]},
\end{align}
which, in the particular case where $[X,Y]=0$ (flow loops), gave the condition
\begin{align}
[\mathcal L_X,\mathcal L_Y]=0.
\end{align}
\textbf{There is no effect in naturally transporting tensor fields around flow loops on a manifold.}\\
By contrast, in \cref{liealgint} we kept the \emph{algebraic} viewpoint behind the same identity, namely that
\begin{align}
X\longmapsto \mathcal L_X
\end{align}
is a Lie algebra representation of $\mathfrak X(M)$ on the tensor fields $\mathfrak T(M)$, understood here as an $\mathbb R$–vector space. Hence, the identity
\begin{align}
[\mathcal L_X,\mathcal L_Y]=\mathcal L_{[X,Y]}
\end{align}
comes into play as the homomorphism property of the representation.\\
Both viewpoints, geometric and algebraic, lead to the same commutation rule. However, the \emph{natural} algebraic structure on $\mathfrak T(M)$ is that of a $C^\infty(M)$–module, and the $\mathbb R$–vector space description is just the result of a “simplification”, or, more formally, of applying the forgetful functor\footnote{Formally, one may consider the forgetful functor $\mathcal F:\mathrm{Mod}_{C^\infty(M)}\to \mathrm{Vec}_{\mathbb R}$.}. In particular, in \cref{prop:not-Cinfty-linear}, the Lie derivative is too coarse of an object to encode $f$–linearity:
\begin{align}
\mathcal L_X\ \text{is not $C^\infty(M)$--linear}.
\end{align}
This motivates the search for a different kind of derivation---one that is $C^\infty(M)$--linear in the vector--field parameter (while retaining Leibniz in the tensor argument) and, therefore, will lead to another geometric interpretation, in particular, it must be associated to a different way of transporting tensor fields on a manifold.
\end{remark}

\section{Parallel transport and affine connection}
\begin{definition}[Vector field along a curve]\label{def:field-along-curve}
Let $\gamma:I\to M$ be a smooth curve. A \emph{vector field along $\gamma$} is a smooth section of the pull--back bundle
$\gamma^*TM\to I$, namely an element $V\in\Gamma(\gamma^*TM)$. Equivalently, it is a smooth map
\begin{align}
V:I\longrightarrow TM,\qquad t\longmapsto V_{\gamma(t)},
\end{align}
such that $V_{\gamma(t)}\in T_{\gamma(t)}M$ for all $t\in I$ (i.e.\ $\pi\circ V=\gamma$, where $\pi:TM\to M$ is the bundle projection). The notation $V_{\gamma(t)}$, also written $V(t)$, retains the parameter $t$: equal points $\gamma(t)=\gamma(s)$ need not give equal vectors. Likewise, expressions $V^k(\gamma(t))$ for fields along curves denote component functions of $t$, not necessarily restrictions of functions on $M$.
\end{definition}

\begin{definition}[Tensor field along a curve]\label{def:tensor-field-along-curve}
Let $\gamma:I\to M$ be a smooth curve. For integers $r,s\geq 0$, a \emph{tensor field of $(r,s)$--type along $\gamma$} is a smooth section of the pull--back bundle
\begin{align}
\gamma^*\big(\otimes^r TM \otimes^s T^*M\big)\longrightarrow I.
\end{align}
Equivalently, it is a smooth map
\begin{align}
T:I\longrightarrow \otimes^r TM \otimes^s T^*M,\qquad
\lambda\longmapsto T_{\gamma(\lambda)},
\end{align}
such that $T_{\gamma(\lambda)}$ lies in $\otimes^r T_{\gamma(\lambda)}M \otimes^s T_{\gamma(\lambda)}^*M$, for all $\lambda\in I$.
The case $r=1$, $s=0$ recovers vector fields along $\gamma$, i.e.\ elements of $\Gamma(\gamma^*TM)$.
\end{definition}

\begin{definition}[Parallel transport]\label{def:parallel-transport}
Let $M$ be a smooth manifold. 
A \emph{parallel transport $P$} is an assignment that associates, to each smooth curve $\gamma:I\to M$ with $I\subset\mathbb R$ an open interval, a family of linear isomorphisms
\begin{align}
P^\gamma_{t,s}:\ T_{\gamma(t)}M \longrightarrow T_{\gamma(s)}M,\qquad (t,s)\in I\times I,
\end{align}
such that, for all $t,u,s\in I$,
\begin{align}
P^\gamma_{t,t}&=\mathrm{id}_{T_{\gamma(t)}M},\\[3pt]
P^\gamma_{t,s}&=P^\gamma_{u,s}\circ P^\gamma_{t,u},
\end{align}
and for every orientation--preserving diffeomorphism $\phi:\tilde I\to I$,\footnote{The construction extends to $C^1$ curves and $C^1$ reparametrizations by the same linear ODE, but we use smooth curves here.}
\begin{align}
P^{\gamma\circ\phi}_{t,s}=P^\gamma_{\phi(t),\,\phi(s)}.
\end{align}
\emph{Regularity.}
\begin{enumerate}
  \item For fixed $\gamma$, the map $(t,s,v)\mapsto P^\gamma_{t,s}(v)$ is smooth on 
  \begin{align}
  I\times I\times_{\gamma} TM  \coloneqq  \{(t,s,v)\in I\times I\times TM \mid \pi(v)=\gamma(t)\}.
  \end{align}
  \item For any smooth family of curves $\Upsilon:(-\varepsilon,\varepsilon)\times I\to M$ with $\Upsilon(0,\cdot)=\gamma$, the assignment $(\varepsilon,t,s,v)\mapsto P^{\Upsilon(\varepsilon,\cdot)}_{t,s}(v)$ is smooth on 
  \begin{align}
 (-\varepsilon,\varepsilon)\times I\times I\times_{\Upsilon} TM  \coloneqq  \{(\varepsilon,t,s,v)\in (-\varepsilon,\varepsilon)\times I\times I\times TM \mid \pi(v)=\Upsilon(\varepsilon,t)\}.
  \end{align}
\end{enumerate}
\emph{Locality along the segment.}
If $\gamma_1,\gamma_2:I\to M$ coincide on the closed subinterval $[a,b]\subset I$, then
\begin{align}
P^{\gamma_1}_{t,s}=P^{\gamma_2}_{t,s}\qquad\text{for all }(t,s)\in [a,b]\times[a,b].
\end{align}
\emph{First--order locality in coordinates.}
For every coordinate chart $(U,x)$ there exist smooth functions $\Gamma^k_{ij}$ on $U$, independent of the chosen curve, such that for every smooth curve $\gamma$, every smooth field $V(t)=V^j(t)\partial_j|_{\gamma(t)}$ along it, and every $t$ with $\gamma(t)\in U$,
\begin{align}
D^P_tV
\coloneqq\lim_{h\to0}\frac{P^\gamma_{t+h,t}V(t+h)-V(t)}{h}
=\left(\frac{dV^k}{dt}(t)+\Gamma^k_{ij}(\gamma(t))\dot\gamma^i(t)V^j(t)\right)\partial_k|_{\gamma(t)}.
\end{align}
Here $\dot\gamma^i=d(x^i\circ\gamma)/dt$. The coefficients give a bilinear correction in $\dot\gamma(t)$ and $V(t)$ in each chart; they are not the components of a globally defined tensor. The expression as a whole is intrinsic by its left-hand side, and its dependence on the curve is only through $\gamma(t)$ and $\dot\gamma(t)$; the field contributes its value and first derivative along the curve. The transformation law of the coefficients is derived in \cref{prop:christoffel-transform}.
\end{definition}

\begin{lemma}\label{lem:P-inverse}
Let $P$ be a parallel transport and let $\gamma:I\to M$ be a smooth curve. For all $s,t\in I$, one has
\begin{align}
P^\gamma_{s,t}\circ P^\gamma_{t,s} = \mathrm{id}_{T_{\gamma(t)}M},\qquad\quad P^\gamma_{t,s}\circ P^\gamma_{s,t} = \mathrm{id}_{T_{\gamma(s)}M},
\end{align}
hence
\begin{align}
P^\gamma_{t,s} = \big(P^\gamma_{s,t}\big)^{-1}.
\end{align}
\end{lemma}
\begin{proof}
By the composition axiom with $u=s$,
\begin{align}
P^\gamma_{t,t}  =  P^\gamma_{s,t}\circ P^\gamma_{t,s}.
\end{align}
Using $P^\gamma_{t,t}=\mathrm{id}_{T_{\gamma(t)}M}$ gives the first identity. Similarly, with $u=t$,
\begin{align}
P^\gamma_{s,s}  =  P^\gamma_{t,s}\circ P^\gamma_{s,t},
\end{align}
and $P^\gamma_{s,s}=\mathrm{id}_{T_{\gamma(s)}M}$ gives the second identity. The inverse statement follows.
\end{proof}

\begin{proposition}[Parallel vector field and local form]\label{coordpartransp}
Let $\gamma:I\to M$ be a smooth curve and $P$ a parallel transport. For each $t\in I$ and $v\in T_{\gamma(t)}M$, there exists a unique smooth vector field $V$ along $\gamma$ solving
\begin{align}
D^P_sV=0,\qquad V(t)=v.
\end{align}
Equivalently,
\begin{align}
V(u)=P^\gamma_{s,u}V(s)\qquad\text{for all }s,u\in I.
\end{align}
In any coordinate chart traversed by $\gamma$, this equation reads
\begin{align}
\frac{dV^k}{ds}(s)+\Gamma^k_{ij}(\gamma(s))\dot\gamma^i(s)V^j(s)=0.
\end{align}
The functions $\Gamma^k_{ij}$ are called the Christoffel symbols associated to $P$ in the chosen chart, and $V$ is called a parallel vector field along $\gamma$.
\end{proposition}
\begin{proof}
Fix $t$ and $v$, and set $V(s)=P^\gamma_{t,s}v$. Regularity gives smoothness, and the identity axiom gives $V(t)=v$. By composition,
\begin{align}
P^\gamma_{s+h,s}V(s+h)=P^\gamma_{s+h,s}P^\gamma_{t,s+h}v=P^\gamma_{t,s}v=V(s),
\end{align}
so $D^P_sV=0$. Composition also gives $V(u)=P^\gamma_{s,u}V(s)$.
Conversely, let $W$ solve $D^P_sW=0$ with $W(t)=v$. Define $F(s)=P^\gamma_{s,t}W(s)\in T_{\gamma(t)}M$. Then
\begin{align}
\frac{F(s+h)-F(s)}{h}
=P^\gamma_{s,t}\frac{P^\gamma_{s+h,s}W(s+h)-W(s)}{h}
\longrightarrow P^\gamma_{s,t}D^P_sW=0.
\end{align}
Thus $F(s)=v$ and $W(s)=P^\gamma_{t,s}v=V(s)$, proving uniqueness and the transport identity. The coordinate equation follows directly from first--order locality in \cref{def:parallel-transport}.
\end{proof}

\begin{definition}[Covariant derivative induced by parallel transport]\label{def:cov-der-from-pt}
Let $P$ be a parallel transport.
For $X,Y\in\Gamma(TM)$ and $p\in M$, choose any smooth curve $\gamma:(-h_0,h_0)\to M$
with $\gamma(0)=p$ and $\dot\gamma(0)=X_p$. The operator
\begin{align}
\nabla^P:\Gamma(TM)\times\Gamma(TM)\longrightarrow\Gamma(TM),
\end{align}
defined by
\begin{align}
(\nabla^P_X Y)_p
\coloneqq
\lim_{h\to0}\frac{P^{\gamma}_{h,0}\,Y_{\gamma(h)}-Y_{\gamma(0)}}{h}
 = 
\left.\frac{d}{dt}\right|_{t=0}\big(P^{\gamma}_{t,0}\,Y_{\gamma(t)}\big),
\end{align}
is called the \emph{covariant derivative} associated to $P$.
\end{definition}

\begin{remark}
Because of the regularity, first--order locality, and segment locality assumptions on $P$, the value
$(\nabla^P_X Y)_p$ is independent of the chosen curve $\gamma$ with $\gamma(0)=p$ and $\dot\gamma(0)=X_p$, and the map $p\mapsto(\nabla^P_X Y)_p$ is smooth. Hence, the operator $\nabla^P\colon\Gamma(TM)\times\Gamma(TM)\to\Gamma(TM)$ is well--defined.
\end{remark}

\begin{proposition}\label{prop:nablaP-derivation}
Let $\nabla^P$ be the covariant derivative associated to a $P$. Then, for all $X,Z,Y,W\in\Gamma(TM)$ and $f,g\in C^\infty(M)$, we have
\begin{align}
\textnormal{(i)} & \nabla^P_{aX+bZ}Y  =  a\,\nabla^P_XY + b\,\nabla^P_ZY \qquad (a,b\in\mathbb R),\\[2pt]
\textnormal{(ii)} & \nabla^P_X(Y+W)  =  \nabla^P_XY + \nabla^P_XW,\\[2pt]
\textnormal{(iii)} & \nabla^P_{fX}Y  =  f\,\nabla^P_XY,\\[2pt]
\textnormal{(iv)} & \nabla^P_X(gY)  =  (Xg) Y  +  g\,\nabla^P_XY,
\end{align}
where $\displaystyle{(Xg)(p)=\frac{d}{dt}\big|_{t=0}\,g(\gamma(t))}$ for any smooth curve $\gamma$ with $\gamma(0)=p$ and $\dot\gamma(0)=X_p$.
\end{proposition}
\begin{proof}
Fix $p\in M$ and choose a coordinate chart around $p$. Applying the first--order axiom to the field $Y\circ\gamma$, where $\gamma(0)=p$ and $\dot\gamma(0)=X_p$, gives
\begin{align}\label{5}
(\nabla^P_XY)_p
=\big(X^i\partial_iY^k+\Gamma^k_{ij}X^iY^j\big)(p)\,\partial_k|_p.
\end{align}
Both terms are $\mathbb R$--linear in $X$ and additive in $Y$, proving (i) and (ii). Replacing $X$ by $fX$ multiplies both terms by $f(p)$, proving (iii). Finally,
\begin{align}
(\nabla^P_X(gY))^k
&=X^i\partial_i(gY^k)+\Gamma^k_{ij}X^i gY^j\\
&=X(g)Y^k+g\big(X^i\partial_iY^k+\Gamma^k_{ij}X^iY^j\big),
\end{align}
which is (iv). Since $p$ is arbitrary, all identities hold globally.
\end{proof}

\begin{remark}
Property (iv) has the \emph{shape} of a Leibniz rule, but at this stage $\nabla^P$ is defined only on vector fields. In particular, we have not yet specified an action of $\nabla^P$ on functions, nor on general tensors.
For this reason we merely view (iv) as a Leibniz--type identity on the $C^\infty(M)$–module $\Gamma(TM)$.
\end{remark}

\begin{proposition}[Coordinate form]\label{prop:coord-form-nablaP}
Let $\nabla^P$ be the covariant derivative associated to $P$ and $X,Y\in\Gamma(TM)$. Then, in local coordinates on $U$, one has
\begin{align}
\big(\nabla^P_X Y\big)^k  =  X^i\,\partial_i Y^k +  \Gamma^k_{ij}\,X^i\,Y^j.
\end{align}
\end{proposition}
\begin{proof}
Fix $p\in U$ and choose $\gamma$ with $\gamma(0)=p$ and $\dot\gamma(0)=X_p$. In the chosen chart, the component derivative of $Y\circ\gamma$ is
\begin{align}
\left.\frac{d}{dt}\right|_{t=0}Y^k(\gamma(t))=X^i(p)\partial_iY^k(p).
\end{align}
Substitution into the first--order axiom gives
\begin{align}
(\nabla^P_XY)_p
=\big(X^i\partial_iY^k+\Gamma^k_{ij}X^iY^j\big)(p)\partial_k|_p,
\end{align}
as claimed.
\end{proof}

\begin{proposition}[Transformation law for Christoffel symbols]\label{prop:christoffel-transform}
Let $(x^1,\dots,x^n)$ and $(x'^1,\dots,x'^n)$ be coordinates on, respectively, open sets $U$ and $V$. Then, on $U\cap V$, the Christoffel symbols transform affinely. In particular, we have
\begin{align}
\Gamma'^k_{ij}
 = 
\frac{\partial x'^k}{\partial x^h}\,\frac{\partial x^l}{\partial x'^i}\,\frac{\partial x^r}{\partial x'^j}\,\Gamma^h_{lr}
 + 
\frac{\partial x'^k}{\partial x^h}\,\frac{\partial^2 x^h}{\partial x'^i\,\partial x'^j}.
\end{align}
\end{proposition}
\begin{proof}
Recall $\partial'_i=\dfrac{\partial x^p}{\partial x'^i}\,\partial_p$. Using $\nabla^P_{fX}Y=f\,\nabla^P_XY$ and
$\nabla^P_X(fY)=X[f]\,Y+f\,\nabla^P_XY$, we compute
\begin{align}
\nabla^P_{\partial'_i}\partial'_j
&=\nabla^P_{\left(\frac{\partial x^p}{\partial x'^i}\partial_p\right)}
           \left(\frac{\partial x^q}{\partial x'^j}\partial_q\right)\\
&=\frac{\partial x^p}{\partial x'^i}\frac{\partial x^q}{\partial x'^j}\,\nabla^P_{\partial_p}\partial_q
   + \partial'_i\!\left(\frac{\partial x^q}{\partial x'^j}\right)\partial_q\\
&=\frac{\partial x^p}{\partial x'^i}\frac{\partial x^q}{\partial x'^j}\,\Gamma^r_{pq}\,\partial_r
   + \frac{\partial^2 x^r}{\partial x'^i\,\partial x'^j}\,\partial_r.
\end{align}
Express $\partial_r=\dfrac{\partial x'^k}{\partial x^r}\,\partial'_k$ and compare coefficients of $\partial'_k$ with
$\nabla^P_{\partial'_i}\partial'_j=\Gamma'^k_{ij}\,\partial'_k$:
\begin{align}
\Gamma'^k_{ij}
=\frac{\partial x'^k}{\partial x^r}\left(
\frac{\partial x^p}{\partial x'^i}\frac{\partial x^q}{\partial x'^j}\,\Gamma^r_{pq}
+ \frac{\partial^2 x^r}{\partial x'^i\,\partial x'^j}\right)
=\frac{\partial x'^k}{\partial x^a}\,\Gamma^a_{bc}\,\frac{\partial x^b}{\partial x'^i}\,\frac{\partial x^c}{\partial x'^j}
+ \frac{\partial x'^k}{\partial x^a}\,\frac{\partial^2 x^a}{\partial x'^i\,\partial x'^j},
\end{align}
as claimed.
\end{proof}

\begin{remark}[Affine terms matter]
By \cref{prop:christoffel-transform}, under a coordinate transformation $x\mapsto x'$ the covariant derivative transforms as
\begin{align}
(\nabla^P_X Y)'^{\,k}
 = 
\frac{\partial x'^k}{\partial x^l}\,(\nabla^P_X Y)^l.
\end{align}
The derivative of the Jacobian coming from $X^i\,\partial_i Y^k$ is exactly canceled by the second--derivative term (the affine term) in the transformation law of $\Gamma$. Hence, $\nabla^P_X Y$ is a well--defined vector field (tensorial), even though neither $\partial_i Y^k$ nor $\Gamma^k_{ij}$ is tensorial on its own.
\end{remark}

\begin{corollary}
Let $P$ be a parallel transport, $\gamma:I\to M$ be a smooth curve and $V$ a vector field along $\gamma$, which in components reads $V_{\gamma(t)}=V^k(\gamma(t))\,\partial_k|_{\gamma(t)}$. Writing $\nabla^P_{\dot\gamma}V\coloneqq D^P_tV$ for the intrinsic derivative along the curve from \cref{def:parallel-transport}, $V$ is parallel if and only if
\begin{align}
    \nabla^P_{\dot\gamma}V=0.
\end{align}
\end{corollary}
\begin{proof}
In components, we have
    \begin{align}
\big(\nabla^P_{\dot\gamma}V\big)^k(\gamma(t))
&=\dot V^k(\gamma(t)) + \Gamma^k_{ij}\big(\gamma(t)\big)\,\dot\gamma^i(t)\,V^j(\gamma(t)),
\end{align}
and, by \cref{coordpartransp}, $V$ is parallel if and only if
\begin{align}
\dot V^k(\gamma(t)) + \Gamma^k_{ij}\big(\gamma(t)\big)\,\dot\gamma^i(t)\,V^j(\gamma(t))=0
\qquad\text{for all $k$.}
\end{align}
Therefore, $\nabla^P_{\dot\gamma}V=0$.
\end{proof}

\begin{remark}[Parallel vs invariant]
Let $P$ be a parallel transport, let $\nabla^P$ be its covariant derivative, and $\gamma:I\to M$ be a smooth curve. 
A vector field $V\in\Gamma(\gamma^*TM)$ along $\gamma$ is \emph{parallel} iff it is invariant under $P$:
\begin{align}
\nabla^P_{\dot\gamma}V=0
\quad\Longleftrightarrow\quad
V_{\gamma(t)}=P^\gamma_{s,t}\,V_{\gamma(s)}\ \ \text{for all }s,t\in I.
\end{align}
This mirrors Lie invariance: if $\Phi_t^X$ is the flow of $X$, then a vector field $Y$ is \emph{flow--invariant} iff
\begin{align}
\mathcal L_X Y=0
\quad\Longleftrightarrow\quad
(\Phi_t^X)_*Y=Y\ \ \text{for all }t\ \text{(on the domain of the flow)}.
\end{align}
Thus, both notions express “no intrinsic change” relative to a chosen comparison rule: 
parallel transport for $\nabla^P$, and the diffeomorphic flow for $\mathcal L_X$.\\ 
\textit{Note: any vector can be parallel transported by $P$, but only those solving 
$\nabla^P_{\dot\gamma}V=0$ are truly parallel; likewise, any field admits a pushforward by $\Phi_t^X$, 
but only those with $\mathcal L_XY=0$ are invariant under the flow}.
\end{remark}

\begin{definition}[Affine connection]\label{def:affine-connection}
An \emph{affine connection} on a smooth manifold $M$ is an $\mathbb R$–bilinear map
\begin{align}
\nabla:\ \Gamma(TM)\times\Gamma(TM)\longrightarrow\Gamma(TM),\qquad (X,Y)\longmapsto \nabla_X Y,
\end{align}
such that for all $X,Y,Z\in\Gamma(TM)$ and $f\in C^\infty(M)$:
\begin{align}
&\textnormal{(i)}\ \ \nabla_{X+Z}Y=\nabla_XY+\nabla_ZY,\quad \nabla_{aX}Y=a\,\nabla_XY\ (a\in\mathbb R),\\
&\textnormal{(ii)}\ \ \nabla_{fX}Y=f\,\nabla_XY,\\
&\textnormal{(iii)}\ \ \nabla_X(fY)=X[f] Y+f\,\nabla_XY.
\end{align}
\end{definition}

\begin{proposition}[Affine connection induces parallel transport]\label{prop:conn-to-PT}
Let $\nabla$ be an affine connection on $M$. Then $\nabla$ defines a unique parallel transport on $M$, obtained by solving the parallel equation along every smooth curve $\gamma:I\to M$.
\end{proposition}
\begin{proof}
Fix a smooth curve $\gamma:I\to M$. The axioms of a connection imply locality by a cutoff argument, so they apply to local coordinate fields. Define its local coefficients by $\nabla_{\partial_i}\partial_j=\Gamma^k_{ij}\partial_k$. Leibniz gives
\begin{align}
(\nabla_XY)^k=X^i\partial_iY^k+\Gamma^k_{ij}X^iY^j.
\end{align}
The calculation in \cref{prop:christoffel-transform}, which uses only these axioms, applies to these coefficients as well. Therefore the expression
\begin{align}
D^\gamma_sV
=\left(\frac{dV^k}{ds}+\Gamma^k_{ij}(\gamma(s))\dot\gamma^i(s)V^j(s)\right)\partial_k|_{\gamma(s)}
\end{align}
for an arbitrary field along $\gamma$ is independent of coordinates.
For $t\in I$ and $v\in T_{\gamma(t)}M$, solve
\begin{align}
\frac{dV^k}{ds}+\Gamma^k_{ij}(\gamma(s))\dot\gamma^i(s)V^j(s)=0,\qquad V(t)=v.
\end{align}
This is a linear ODE with smooth coefficients in each chart. Solutions extend across every compact subinterval of $I$ by a finite subdivision into charts, and agree on overlaps by uniqueness. Hence there is a unique solution on $I$. Define $P^\gamma_{t,s}v=V(s)$.\\
\emph{Identity and composition.} The initial condition gives $P^\gamma_{t,t}=\mathrm{id}$. Both $s\mapsto P^\gamma_{t,s}v$ and $s\mapsto P^\gamma_{u,s}(P^\gamma_{t,u}v)$ solve the same ODE and have the same value at $s=u$. Uniqueness gives
\begin{align}
P^\gamma_{t,s}=P^\gamma_{u,s}\circ P^\gamma_{t,u}.
\end{align}
Linearity in the initial value makes each $P^\gamma_{t,s}$ linear, and composition gives its inverse $P^\gamma_{s,t}$.\\
\emph{Smoothness.} Smooth dependence of linear ODE solutions on initial data, times, and parameters gives both regularity requirements in \cref{def:parallel-transport}. This applies locally to any smooth family of curves, with a finite subdivision covering the relevant segment.\\
\emph{Reparametrization invariance.} If $\phi:\tilde I\to I$ is a smooth orientation--preserving diffeomorphism, the chain rule shows that $V\circ\phi$ is parallel along $\gamma\circ\phi$ whenever $V$ is parallel along $\gamma$. Uniqueness gives
\begin{align}
P^{\gamma\circ\phi}_{\tau,\sigma}=P^\gamma_{\phi(\tau),\phi(\sigma)}.
\end{align}
\emph{Segment locality.} If two curves coincide on a segment, their coordinate ODEs there coincide, so the corresponding transports agree.\\
\emph{First--order locality.} Fix a coordinate chart and let $A(t)^k{}_j=\Gamma^k_{ij}(\gamma(t))\dot\gamma^i(t)$. In the coordinate trivializations at the two endpoints, the matrix of the forward transport is $I-hA(t)+o(h)$ and that of the backward transport is $I+hA(t)+o(h)$. For an arbitrary smooth field $U(t)=U^k(t)\partial_k|_{\gamma(t)}$ along the curve,
\begin{align}
P^\gamma_{t+h,t}U(t+h)
=U(t)+h\left(\frac{dU^k}{dt}(t)+\Gamma^k_{ij}(\gamma(t))\dot\gamma^i(t)U^j(t)\right)\partial_k|_{\gamma(t)}+o(h).
\end{align}
Thus the intrinsic limit in \cref{def:parallel-transport} has exactly the required coordinate expression. All axioms hold. Uniqueness follows from uniqueness of the parallel ODE with the coefficients of $\nabla$.
\end{proof}

\begin{corollary}\label{nablanablaP}
Let $P$ be the parallel transport defined by $\nabla$. Then $\nabla=\nabla^P$.
\end{corollary}
\begin{proof}
Fix smooth vector fields $X,Y$, a point $p$, and a curve $\gamma$ with $\gamma(0)=p$ and $\dot\gamma(0)=X_p$. Use a coordinate chart around $p$ and let $[P^\gamma_{s,0}]$ denote the matrix of transport in the coordinate trivializations. The ODE gives
\begin{align}\label{6}
\left.\frac{d}{ds}\right|_{s=0}[P^\gamma_{s,0}]^k{}_j
=\Gamma^k_{ij}(p)X^i(p).
\end{align}
Consequently,
\begin{align}
(\nabla^P_XY)_p
&=\left.\frac{d}{ds}\right|_{s=0}\big(P^\gamma_{s,0}Y_{\gamma(s)}\big)\\
&=\big(X^i\partial_iY^k+\Gamma^k_{ij}X^iY^j\big)(p)\partial_k|_p
=(\nabla_XY)_p.
\end{align}
Here the field $Y$ is fixed; the expression depends on its first derivatives, not merely on $Y_p$. Since $p,X,Y$ are arbitrary, $\nabla^P=\nabla$.
\end{proof}

\begin{remark}
    In particular, the coordinate expression is given by
    \begin{align}
        (\nabla_X Y)^k=X^i\partial_i Y^k+\Gamma^{k}_{\ ij}X^iY^j,
    \end{align}
    for all $X,Y\in\Gamma(TM)$, or equivalently\footnote{This last coordinate expression is common in the physics literature.}
    \begin{align}
        \nabla_i Y^k=\partial_i Y^k +\Gamma^{k}_{\ ij}Y^j,
    \end{align}
    where $\nabla_i\coloneq \nabla_{\partial_i}$.\footnote{For a general, non--coordinate frame $e_i$, we have $\nabla_{e_i} Y^k=e_iY^k+\Gamma^{k}_{\ ij}Y^j$.}
\end{remark}

\begin{definition}[Covariant derivative as an algebraic operation]\label{def:cov-der-algebraic}
Let $\nabla$ be an affine connection on $M$. For each $X\in\Gamma(TM)$ there exists a unique $\mathbb R$–linear operator
\begin{align}
\nabla_X:\ \mathfrak T(M)\longrightarrow \mathfrak T(M),
\end{align}
called the \emph{covariant derivative along $X$}, characterized by:
\begin{align}
&\textnormal{(i) On functions:}\quad && \nabla_X f = Xf,\qquad f\in C^\infty(M);\\
&\textnormal{(ii) On vector fields:}\quad && \nabla_X Y \ \text{agrees with } \nabla;\\
&\textnormal{(iii) $C^\infty(M)$--linearity in $X$:}\quad && \nabla_{fX}T=f\,\nabla_XT,\qquad f\in C^\infty(M),\ T\in\mathfrak T(M);\\
&\textnormal{(iv) Degree $0$ derivation:}\quad && \nabla_X(T\otimes S)=(\nabla_XT)\otimes S+T\otimes(\nabla_XS);\\
&\textnormal{(v) Compatibility:}\quad && \nabla_X\ \text{commutes with contractions of indices.}
\end{align}
\end{definition}

\begin{remark}
Notice that the covariant derivative is determined by its action of functions and vector fields. The action on general tensor fields is then obtained by the Leibniz rule, the compatibility with contractions, and the fact that it is a degree $0$ operator. Equivalently, we can specify its action of functions and $1$--forms $\omega\in\Omega^1(M)$, given by
\begin{align}
(\nabla_X\omega)(Y)=X[\omega(Y)]-\omega(\nabla_XY).
\end{align}
\end{remark}

\begin{remark}
Parallel transport extends functorially to tensor fields. If $T\in\mathfrak T(M)$ has type $(r,q)$, set
\begin{align}
\big(P^{\gamma}_{t,s} T\big)_{\gamma(s)}\coloneqq
\Big((P^{\gamma}_{t,s})^{\otimes r}\otimes \big((P^{\gamma}_{s,t})^*\big)^{\otimes q}\Big)\,T_{\gamma(t)}.
\end{align}
Its infinitesimal generator is the covariant derivative $\nabla^P$; by \cref{nablanablaP} one has $\nabla^P=\nabla$ (extended to tensors by the Leibniz rule and compatibility with contractions). This yields the equivalence between the geometric (operator defined by parallel-transport) and the algebraic (derivation) definitions of the covariant derivative.
\end{remark}

\begin{remark}[Covariant vs. Lie derivative]
Both $\nabla_X$ and $\mathcal L_X$ extend to $\mathfrak T(M)$ as degree $0$ derivations commuting with contractions,
and they coincide on functions: $\nabla_X f=\mathcal L_X f=Xf$. Nonetheless, there are two key differences:
\begin{enumerate}[label=(\alph*)]
\item \emph{Naturality.} $\mathcal L_X$ is canonical (it comes from the flow of $X$), whereas $\nabla_X$ depends on an additional geometric structure, i.e. a non--canonical choice of an affine connection. For instance, on vector fields $X,Y\in\Gamma(TM)$,
\begin{align}
\mathcal L_X Y = [X,Y], \qquad \nabla_X Y \ \text{is connection-dependent.}
\end{align}
In particular, there is no canonical identity expressing $\mathcal L_XY$ in terms of $\nabla_XY$ without extra geometric structure.
\item \emph{$f$--linearity.} For $f\in C^\infty(M)$,
\begin{align}
\mathcal L_{fX}T &\neq f\,\mathcal L_X T\\[3pt]
\nabla_{fX}T &= f\,\nabla_X T.
\end{align}
\end{enumerate}
\end{remark}

\begin{proposition}[Space of affine connections is affine]
Given two affine connections $\nabla$ and $\nabla'$ on $TM$, we define their difference as\footnote{The difference is defined in the natural way on vector fields.}
\begin{align}
A(X,Y)\coloneqq (\nabla-\nabla')_X Y.
\end{align}
Then, $A$ is a tensor field of $(1,2)$--type. Conversely, for any $A\in\Gamma( T^{(1,2)}(M))$,
\begin{align}
\widehat\nabla_X Y = \nabla_X Y + A(X,Y)
\end{align}
is an affine connection. Hence, the set of affine connections is an affine space modeled on $\Gamma(T^{(1,2)}(M))$.
\end{proposition}
\begin{proof}
(1) \emph{Difference is a tensor.}
Define
\begin{align}
A:\Gamma(TM)\times\Gamma(TM)\longrightarrow\Gamma(TM),\qquad
A(X,Y)\coloneqq (\nabla-\nabla')_X Y .
\end{align}
Then $A$ is $\mathbb R$–bilinear and, using the axioms of an affine connection,
\begin{align}
A(fX,Y)
&= (\nabla-\nabla')_{fX}Y
 = f(\nabla-\nabla')_X Y
 = f\,A(X,Y),\\
A(X,fY)
&= (\nabla-\nabla')_X(fY)
 = X[f]\,Y+f\,\nabla_XY - X[f]\,Y - f\,\nabla'_X Y
 = f\,A(X,Y).
\end{align}
Hence, $A$ is $C^\infty(M)$–bilinear (tensorial) in both slots. By the standard characterization of tensor
fields via $C^\infty(M)$–multilinear maps, $A$ is a smooth section of $T^*M\otimes T^*M\otimes TM$.\\
(2) \emph{Adding a $(1,2)$–tensor preserves the axioms.}
Let $A\in\Gamma(T^*M\otimes T^*M\otimes TM)$ and define
\begin{align}
\widehat\nabla_X Y \coloneqq \nabla_X Y + A(X,Y).
\end{align}
Then, using $C^\infty$–bilinearity of $A$,
\begin{align}
\widehat\nabla_{X+Z}Y &= \nabla_{X+Z}Y + A(X+Z,Y)
= \nabla_XY+\nabla_ZY + A(X,Y)+A(Z,Y),\\
\widehat\nabla_{aX}Y &= a\,\nabla_XY + a\,A(X,Y) \qquad (a\in\mathbb R),\\
\widehat\nabla_{fX}Y &= \nabla_{fX}Y + A(fX,Y)
= f\,\nabla_XY + f\,A(X,Y) = f\,\widehat\nabla_XY,\\
\widehat\nabla_X(fY) &= \nabla_X(fY) + A(X,fY)
= X[f]\,Y + f\,\nabla_XY + f\,A(X,Y)
= X[f]\,Y + f\,\widehat\nabla_XY.
\end{align}
Thus $\widehat\nabla$ is an affine connection.\\
(3) \emph{Affine-space structure.}
Let $\mathrm{Conn}(TM)$ denote the set of affine connections on $TM$ and let
\begin{align}
\Gamma(T^*M\otimes T^*M\otimes TM)\times \mathrm{Conn}(TM)\longrightarrow \mathrm{Conn}(TM),
\qquad (A,\nabla)\longmapsto \nabla+A,
\end{align}
be the natural translation action. It is
\emph{free}, since $\nabla+A=\nabla$ implies $A=0$, and \emph{transitive}, since for any
$\nabla,\nabla'\in\mathrm{Conn}(TM)$ one has the unique $A=\nabla'-\nabla$ with $\nabla'=\nabla+A$ by (1).
Therefore $\mathrm{Conn}(TM)$ is an affine space modeled on
$\Gamma(T^*M\otimes T^*M\otimes TM)$.
\end{proof}

\begin{definition}[Covariant derivative along a curve]\label{defconnpullback}
Let $M$ be a smooth manifold with an affine connection $\nabla$, and let $\gamma:I\to M$ be smooth. Set $E= \gamma^*TM \to I$ and let $\nabla^\gamma\coloneqq \gamma^*\nabla$ denote the pullback connection.
The covariant derivative along $\gamma$ is the operator
\begin{align}
\nabla^\gamma_{\partial_\lambda} : \Gamma(E) \to \Gamma(E),
\end{align}
defined in local coordinates by
\begin{align}
(\nabla^\gamma_{\partial_\lambda}V_\gamma)(\lambda)
=\left(\frac{dV^k}{d\lambda}(\lambda)+\Gamma^k{}_{ij}(\gamma(\lambda))\dot\gamma^i(\lambda)V^j(\lambda)\right)\partial_k|_{\gamma(\lambda)}.
\end{align}
The affine transformation law of the Christoffel symbols makes this expression coordinate independent. It is defined for every section of $\gamma^*TM$, including along constant or self-intersecting curves. When an ambient extension $\tilde V$ exists locally along a parameter interval, this formula agrees with $\nabla_{\dot\gamma}\tilde V$ and is independent of that extension, since its component functions restrict to the same functions $V^k(\lambda)$ throughout that interval. We also write $\nabla_{\dot\gamma}V_\gamma$ for this derivative along the curve.
In particular, since $d\gamma(\partial_\lambda)=\dot\gamma$, one has
\begin{align}
\nabla^\gamma_{\partial_\lambda}\dot\gamma = \nabla_{\dot\gamma}\dot\gamma.
\end{align}
In local coordinates $x^k(\lambda)=\gamma^k(\lambda)$ and $V_\gamma=V^j\partial_j$,
\begin{align}
\bigl(\nabla^\gamma_{\partial_\lambda}V_\gamma\bigr)^k
= \frac{d V^k}{d\lambda} + \Gamma^k{}_{ij}(\gamma(\lambda))\,\dot\gamma^i V^j,
\qquad
\bigl(\nabla^\gamma_{\partial_\lambda}\dot\gamma\bigr)^k
= \frac{d^2 x^k}{d\lambda^2} + \Gamma^k{}_{ij}(\gamma(\lambda))\,\dot\gamma^i\dot\gamma^j.
\end{align}
\end{definition}

\begin{remark}\label{remnablatenscurv}
    The construction in \cref{defconnpullback} extends to tensor fields along a curve. For $T_\gamma\in\Gamma(\gamma^*(T^{(r,s)}M))$, define $(\gamma^*\nabla)_{\partial_\lambda}T_\gamma$ by differentiating its component functions in $\lambda$, adding one $+\Gamma$ term for each contravariant index and one $-\Gamma$ term for each covariant index, all contracted with $\dot\gamma$. This is coordinate independent and satisfies Leibniz and compatibility with contractions. If an ambient extension $\tilde T$ exists on a parameter interval, then
    \begin{align}
        (\gamma^*\nabla)_{\partial_\lambda}T_\gamma(\lambda)
        =\nabla_{\dot\gamma}\tilde T\big|_{\gamma(\lambda)},
    \end{align}
    independently of the extension.
\end{remark}

\begin{definition}[Geodesics]
Let $\gamma$ be a smooth curve and $\nabla$ an affine connection on $M$. If $\gamma$ satisfies
\begin{align}
\nabla_{\dot\gamma}\dot\gamma = 0,
\end{align}
we say that $\gamma$ is a \textit{geodesic} (for $\nabla$). In local coordinates, this is the ODE
\begin{align}
\ddot\gamma^k(\lambda) + \Gamma^k_{ij}(\gamma(\lambda))\,\dot\gamma^i(\lambda)\,\dot\gamma^j(\lambda) = 0.
\end{align}
\end{definition}

\begin{remark}\label{remconnpullback}
By \cref{defconnpullback}, for any smooth curve $\gamma:I\to M$ one has
\begin{align}
\nabla^\gamma_{\partial_\lambda}\dot\gamma
= \nabla_{\,d\gamma(\partial_\lambda)}\dot\gamma
= \nabla_{\dot\gamma}\dot\gamma .
\end{align}
Hence, the geodesic equation can be written equivalently as
\begin{align}
\nabla^\gamma_{\partial_\lambda}\dot\gamma=0
\qquad\Longleftrightarrow\qquad
\nabla_{\dot\gamma}\dot\gamma=0
\quad\text{along }\gamma .
\end{align}
\end{remark}

\begin{proposition}
Let $\nabla$ be an affine connection on $M$. For each $x\in M$ and $v\in T_xM$, there exists an open interval $I_{x,v}\subset\mathbb R$, with $0\in I_{x,v}$, and a unique geodesic $\gamma_{x,v}\colon I_{x,v}\to M$ such that
\begin{align}
\gamma_{x,v}(0)=x,\qquad \dot\gamma_{x,v}(0)=v.
\end{align}
We say that $\gamma_{x,v}$ is maximal in the following sense: if $J\subset\mathbb R$ is an open interval with $0\in J$ and $\tilde\gamma\colon J\to M$ is a geodesic with
\begin{align}
\tilde\gamma(0)=x,\qquad \dot{\tilde\gamma}(0)=v,
\end{align}
then $J\subset I_{x,v}$ and $\tilde\gamma=\gamma_{x,v}|_J$. Equivalently, there is no open interval $J\supset I_{x,v}$ and geodesic $\tilde\gamma\colon J\to M$ with $\tilde\gamma|_{I_{x,v}}=\gamma_{x,v}$.
Moreover, the map $(x,v,\lambda)\mapsto\gamma_{x,v}(\lambda)$ is smooth on its domain.
\end{proposition}
\begin{proof}
In local coordinates, the geodesic equation is the smooth second--order ODE
\begin{align}
\ddot\gamma^k(\lambda)+\Gamma^k_{ij}(\gamma(\lambda))\,\dot\gamma^i(\lambda)\,\dot\gamma^j(\lambda)=0,
\qquad
\gamma(0)=x,\ \dot\gamma(0)=v.
\end{align}
By Picard–Lindelöf there is a unique local solution. On chart overlaps the Christoffel symbols transform so that the equations are equivalent, hence the coordinate solutions agree; this yields a well-defined geodesic on $M$. By standard ODE continuation, the solution extends uniquely until it becomes inextendible as a solution; this defines a unique maximal open interval $I_{x,v}$ containing $0$. Standard ODE theory gives smooth dependence of $\gamma_{x,v}(t)$ on $(x,v,\lambda)$.
\end{proof}

\begin{definition}[Geodesic completeness]
The connection $\nabla$ is geodesically complete if, for every $x\in M$ and $v\in T_xM$, the geodesic $\gamma_{x,v}$ is defined on all of $\mathbb R$.
\end{definition}

\section{Curvature and torsion}
When we discussed the Lie derivative, we noticed that the lack of $C^\infty(M)$--linearity of the derivation (the incapacity to encode the $C^\infty(M)$--module structure) seemed to go along with what we called the “flatness” condition for $\mathcal L$, i.e.
\begin{align}
    [\mathcal L_X,\mathcal L_Y]=0,
\end{align}
for all vector fields $X,Y\in\Gamma(TM)$ such that $[X,Y]=0$. In fact, that is what makes the Lie derivative a Lie algebra representation for $\mathfrak X(M)$ on $\mathfrak T(M)$, seen as real vector spaces, not as $C^\infty(M)$--modules. In other words, the extension of a derivation with all the properties of the Lie derivative plus the $C^\infty(M)$--linearity is precisely the definition of a covariant derivative given in \cref{def:cov-der-algebraic}, and we are about to see that it loses the property of being “flat”.

\begin{definition}[Torsion and curvature]
Let $\nabla$ be an affine connection on a smooth manifold $M$. 
The tensor fields $T\in \Gamma(T^{(1,2)}(M))$ and $R\in \Gamma(T^{(1,3)}(M))$, defined for all $X,Y\in\Gamma(TM)$ by
\begin{align}
T(X,Y)&\coloneqq\nabla_X Y-\nabla_Y X-[X,Y]\\[4pt]
R(X,Y)&\coloneqq[\nabla_X,\nabla_Y] -\nabla_{[X,Y]},
\end{align}
are called, respectively, the \emph{torsion} and \emph{curvature} of $\nabla$.\footnote{Where the brackets can be defined in the sense of \cref{opliealg}.}
\end{definition}

\begin{remark}
The defining formulas are $C^\infty(M)$--linear in every tensor argument. Indeed, using $[fX,Y]=f[X,Y]-Y(f)X$ and Leibniz, the derivative terms cancel to give
\begin{align}
T(fX,Y)&=fT(X,Y),\\
R(fX,Y)Z&=fR(X,Y)Z,\\
R(X,Y)(fZ)&=fR(X,Y)Z+\big(X(Yf)-Y(Xf)-[X,Y]f\big)Z=fR(X,Y)Z.
\end{align}
Skew symmetry gives linearity in the second slot, so these are indeed smooth tensor fields. The tensor fields $T$ and $R$ are skew symmetric in the first two slots; in particular $T(X,Y)=-T(Y,X)$ and $R(X,Y)=-R(Y,X)$.
\end{remark}

\begin{proposition}[Coordinate expressions of torsion and curvature]
Let $(x^1,\dots,x^n)$ be local coordinates on $M$ and write
\begin{align}
\nabla_{\partial_i}\partial_j&=\Gamma^k_{ij}\,\partial_k\\
T(\partial_i,\partial_j)&=T^k{}_{ij}\,\partial_k\\
R(\partial_i,\partial_j,\partial_k)&=R^\ell{}_{ijk}\,\partial_\ell.
\end{align}
Then, the components of $T$ and $R$ are given by
\begin{align}
T^k{}_{ij}&=\Gamma^k_{ij}-\Gamma^k_{ji},\\
R^\ell{}_{ijk}&=\partial_i\Gamma^\ell_{jk}-\partial_j\Gamma^\ell_{ik}
+\Gamma^\ell_{ip}\Gamma^p_{jk}-\Gamma^\ell_{jp}\Gamma^p_{ik}.
\end{align}
In particular, $T^k{}_{ij}=-T^k{}_{ji}$ and $R^\ell{}_{ijk}=-R^\ell{}_{jik}$.
\end{proposition}
\begin{proof}
Since $[\partial_i,\partial_j]=0$ in a coordinate frame,
\begin{align}
T(\partial_i,\partial_j)
=\nabla_{\partial_i}\partial_j-\nabla_{\partial_j}\partial_i-[\partial_i,\partial_j]
=(\Gamma^k_{ij}-\Gamma^k_{ji})\,\partial_k .
\end{align}
For curvature, compute
\begin{align}
\nabla_{\partial_i}\nabla_{\partial_j}\partial_k
&=\nabla_{\partial_i}(\Gamma^\ell_{jk}\,\partial_\ell)
=\partial_i\Gamma^\ell_{jk}\,\partial_\ell+\Gamma^\ell_{jk}\,\Gamma^p_{i\ell}\,\partial_p,\\
\nabla_{\partial_j}\nabla_{\partial_i}\partial_k
&=\partial_j\Gamma^\ell_{ik}\,\partial_\ell+\Gamma^\ell_{ik}\,\Gamma^p_{j\ell}\,\partial_p,\\
\nabla_{[\partial_i,\partial_j]}\partial_k&=0,
\end{align}
and subtract. The stated formulas follow.
\end{proof}

\begin{definition}
    Let $R$ be the curvature of an affine connection $\nabla$. If $R=0$, we say that $\nabla$ is \textit{flat}. Whereas, if $T=0$, $\nabla$ is \textit{torsion--free}.
\end{definition}

\begin{example}[Affine connections on manifolds]\label{exconn}
\leavevmode\newline
\vspace{-0.55cm}
\begin{enumerate}[label=(\roman*)]
\item \textit{The flat Euclidean space $\mathbb R^n$.}
In standard coordinates $x=(x^1,\dots,x^n)$, the connection
\begin{align}
\Gamma^k_{ij}=0\quad \text{(flat choice)},
\end{align}
gives
\begin{align}
T^k{}_{ij}=0,\qquad R^\ell{}_{ijk}=0.
\end{align}

\item \textit{Sphere $S^n$.}
On $U=S^n\setminus\{\text{north pole}\}$ with stereographic coordinates $y=(y^1,\dots,y^n)$,
consider the torsion--free affine connection with
\begin{align}
\Gamma^k_{ij}(y)=-\frac{2}{1+|y|^2}\big(\delta^k_i\,y_j+\delta^k_j\,y_i-\delta_{ij}\,y^k\big),
\qquad |y|^2=\sum_m (y^m)^2.
\end{align}
Then $T^k{}_{ij}=0$ and the curvature components in this chart are
\begin{align}
R^\ell{}_{ijk}(y)= h_{jk}(y)\,\delta^\ell_i-h_{ik}(y)\,\delta^\ell_j,
\qquad
h_{jk}(y)\coloneqq \frac{4}{(1+|y|^2)^2}\,\delta_{jk}.
\end{align}
These local formulas are compatible on overlaps of stereographic charts and thus define a global torsion--free affine connection on $S^n$.

\item \textit{Hyperbolic space $\mathbb H^n$.}
On the unit ball $B^n=\{y\in\mathbb R^n:|y|<1\}$ with coordinates $y=(y^1,\dots,y^n)$,
consider the torsion–free affine connection with
\begin{align}
\Gamma^k_{ij}(y)=\frac{2}{1-|y|^2}\big(\delta^k_i\,y_j+\delta^k_j\,y_i-\delta_{ij}\,y^k\big).
\end{align}
Then $T^k{}_{ij}=0$ and
\begin{align}
R^\ell{}_{ijk}(y)= -\big(h_{jk}(y)\,\delta^\ell_i-h_{ik}(y)\,\delta^\ell_j\big),
\qquad
h_{jk}(y)\coloneqq \frac{4}{(1-|y|^2)^2}\,\delta_{jk}.
\end{align}

\item \textit{A curved connection on $\mathbb R^n$.}
On $\mathbb R^n$, with $n\ge3$ and global coordinates $x=(x^1,\dots,x^n)$, consider the torsion–free connection
\begin{align}
\Gamma^k_{ij}(x)=\delta^k_i\,x_j+\delta^k_j\,x_i-\delta_{ij}\,x^k .
\end{align}
Then $T^k{}_{ij}=0$, while the curvature is nonzero and not constant; in particular, at $x=0$,
\begin{align}
R^\ell{}_{ijk}(0)=-2\big(\delta^\ell_i\,\delta_{jk}-\delta^\ell_j\,\delta_{ik}\big).
\end{align}
\end{enumerate}
\end{example}

\begin{remark}\label{remesconn}
The same smooth manifold can carry many inequivalent affine connections. For instance, in \cref{exconn}, we exhibited two different connections on $\mathbb R^n$, one flat and one curved. The manifold structure itself does not single out a preferred choice: it only imposes regularity and global constraints (coming from the dimension, topology, and smooth structure) on what parallel transports/affine connections can live on it. Within these constraints, different choices of connection yield different affine connections.
\end{remark}

\begin{proposition}[Bianchi identities]
Let $\nabla$ be an affine connection, and let $T$ and $R$ be its torsion and curvature, respectively. Moreover, denote with $\sum_{\mathrm{cyc}\ X,Y,Z}$ the cyclic summation over $(X,Y,Z)$, i.e.
\begin{align}
\sum_{\mathrm{cyc}\ X,Y,Z}F(X,Y,Z)=F(X,Y,Z)+F(Y,Z,X)+F(Z,X,Y).
\end{align}
Then, for all $X,Y,Z,W\in\Gamma(TM)$, one has
\begin{align}
\sum_{\mathrm{cyc}\ X,Y,Z} R(X,Y,Z)
&=
\sum_{\mathrm{cyc}\ X,Y,Z}\big((\nabla_X T)(Y,Z)+T(T(X,Y),Z)\big)\\
\sum_{\mathrm{cyc}\ X,Y,Z} (\nabla_X R)(Y,Z,W)
&=
-\sum_{\mathrm{cyc}\ X,Y,Z} R\big(T(X,Y),Z,W\big).
\end{align}
These relations are called, respectively, first and second Bianchi identity.
\end{proposition}
\begin{proof}
Fix $W$ and use the Jacobi identity for the operators:
\begin{align}
\sum_{\mathrm{cyc}\ X,Y,Z} [\nabla_X,[\nabla_Y,\nabla_Z]]\,W&=0.
\end{align}
Since $[\nabla_Y,\nabla_Z]=\nabla_{[Y,Z]}+R(Y,Z,\cdot)$,
\begin{align}
0
&=\sum_{\mathrm{cyc}\ X,Y,Z}\Big([\nabla_X,\nabla_{[Y,Z]}]W+[\nabla_X,R(Y,Z,\cdot)]W\Big)\\
&=\sum_{\mathrm{cyc}\ X,Y,Z}\Big(\nabla_{[X,[Y,Z]]}W+R(X,[Y,Z],W)+(\nabla_X R)(Y,Z,W)\\
&\phantom{=\sum_{\mathrm{cyc}\ X,Y,Z}\Big(}+R(\nabla_X Y,Z,W)+R(Y,\nabla_X Z,W)\Big).
\end{align}
By Jacobi on vector fields, $\sum_{\mathrm{cyc}\ X,Y,Z}\nabla_{[X,[Y,Z]]}W=0$, hence
\begin{align}
&\sum_{\mathrm{cyc}\ X,Y,Z}(\nabla_X R)(Y,Z,W)=\\
&-\sum_{\mathrm{cyc}\ X,Y,Z}\Big(R(X,[Y,Z],W)+R(\nabla_X Y,Z,W)+R(Y,\nabla_X Z,W)\Big).
\end{align}
With $[Y,Z]=\nabla_Y Z-\nabla_Z Y-T(Y,Z)$, the contributions involving $\nabla_Y Z,\nabla_Z Y$ cancel after the cyclic sum, leaving
\begin{align}
\sum_{\mathrm{cyc}\ X,Y,Z}(\nabla_X R)(Y,Z,W)
&=-\sum_{\mathrm{cyc}\ X,Y,Z} R\big(T(X,Y),Z,W\big),
\end{align}
which is the second Bianchi identity. For the first, abbreviate the cyclic sum over $X,Y,Z$ by $\sum_{\mathrm{cyc}}$. Retaining the second derivatives gives
\begin{align}
\sum_{\mathrm{cyc}}R(X,Y,Z)
&=\sum_{\mathrm{cyc}}\big(\nabla_X(\nabla_YZ-\nabla_ZY)-\nabla_{[Y,Z]}X\big)\\
&=\sum_{\mathrm{cyc}}\big(\nabla_X(T(Y,Z))+\nabla_X[Y,Z]-\nabla_{[Y,Z]}X\big)\\
&=\sum_{\mathrm{cyc}}\big(\nabla_X(T(Y,Z))+T(X,[Y,Z])\big),
\end{align}
where the last equality uses $\sum_{\mathrm{cyc}}[X,[Y,Z]]=0$. Expanding the derivative of $T$ and substituting $[Y,Z]=\nabla_YZ-\nabla_ZY-T(Y,Z)$, the terms involving covariant derivatives of $X,Y,Z$ cancel cyclically. Hence
\begin{align}
\sum_{\mathrm{cyc}}R(X,Y,Z)
&=\sum_{\mathrm{cyc}}\big((\nabla_XT)(Y,Z)-T(X,T(Y,Z))\big)\\
&=\sum_{\mathrm{cyc}}\big((\nabla_XT)(Y,Z)+T(T(X,Y),Z)\big),
\end{align}
using skew symmetry of $T$ and cyclic relabeling in the last step.
\end{proof}

\begin{remark}\label{bianchiT0}
    In the particular case of a torsion--free connection, one has
    \begin{align}
\sum_{\mathrm{cyc}\ X,Y,Z} R(X,Y,Z)
&=0\\
\sum_{\mathrm{cyc}\ X,Y,Z} (\nabla_X R)(Y,Z,W)
&=0,
\end{align}
for all $X,Y,W,Z\in\Gamma(TM)$.
\end{remark}

\begin{definition}[Ricci tensor]
Let $\nabla$ be an affine connection with curvature $R$. The tensor field $\mathrm{Ric}\in\Gamma(T^{(0,2)}(M))$, defined as
\begin{align}
\mathrm{Ric}(Y,Z)\coloneqq \operatorname{tr}\big(W\mapsto R(W,Y,Z)\big),
\end{align}
is called the \emph{Ricci tensor}.
\end{definition}

\begin{remark}
Equivalently, in local coordinates, the Ricci tensor takes the form\footnote{Since physicists like to use coordinates, the Ricci tensor is often denoted just as $R_{ij}$.}
\begin{align}
\mathrm{Ric}_{ij}=R^k{}_{kij}.
\end{align}
\end{remark}

\begin{proposition}[Ricci identity]\label{riccid}
Let $\nabla$ be a torsion--free connection with curvature $R$ and Ricci tensor 
$\mathrm{Ric}$. Then, in local coordinates, we have
\begin{align}
\nabla_a R^a{}_{\,bcd}  = \nabla_b \mathrm{Ric}_{cd}  -  \nabla_c \mathrm{Ric}_{bd}.
\end{align}
\end{proposition}
\begin{proof}
The second Bianchi identity reads
\begin{align}
\nabla_a R^e{}_{\,bcd}  +  \nabla_c R^e{}_{\,abd}  +  \nabla_b R^e{}_{\,cad}  =  0.
\end{align}
Contract $e$ with $a$:
\begin{align}
\nabla_a R^a{}_{\,bcd}  +    \nabla_c R^a{}_{\,abd}  +  \nabla_b R^a{}_{\,cad}  =  0.
\end{align}
By definition of $\mathrm{Ric}$ and $R^a{}_{\,abd}=-R^a{}_{\,bad}$,
\begin{align}
R^a{}_{\,abd}  =  \mathrm{Ric}_{bd}, 
\qquad
R^a{}_{\,cad}  =  -\mathrm{Ric}_{cd}.
\end{align}
Hence, we get
\begin{align}
\nabla_a R^a{}_{\,bcd}  +  \nabla_c \mathrm{Ric}_{bd}  -  \nabla_b \mathrm{Ric}_{cd}  =  0,
\end{align}
which gives the claim.
\end{proof}

\begin{definition}[Holonomy of a connection]
Let $\nabla$ be an affine connection on $M$ and $P$ its associated parallel transport. Here parallel transport is extended to piecewise smooth curves by composing the transports on their smooth segments; reversal gives the inverse map by the ODE and reparametrization. Throughout this section, loops may be piecewise smooth. Moreover, let $\gamma$ denote such a loop, namely $\gamma(0)=\gamma(1)=x\in M$. The \emph{holonomy at $x$} is
\begin{align}
\mathrm{Hol}_x(\nabla)\coloneqq \{\,P^\gamma_{0,1}:\ \gamma \text{ loop based at } x\,\}.
\end{align}
\end{definition}

\begin{remark}
    This object is intended as $\mathrm{Hol}_x\subset GL(T_xM)$, and it forms a group.
\end{remark}
\begin{remark}
    The \emph{restricted holonomy group} $\mathrm{Hol}^0_x(\nabla)$ is the subgroup generated by $P_\gamma$ with $\gamma$ null--homotopic at $x$; it is the identity component of $\mathrm{Hol}_x(\nabla)$. Its Lie algebra is denoted by
\begin{align}
\mathfrak{hol}_x(\nabla)\coloneqq T_{\mathrm{id}}\mathrm{Hol}_x(\nabla)\subset \mathrm{End}(T_xM).
\end{align}
\end{remark}

\begin{proposition}[Base--point independence up to conjugation]
Let $\nabla$ be an affine connection on $M$ with parallel transport $P$.
If $\gamma\colon[0,1]\to M$ is a smooth curve such that $\gamma(0)=x$ and $\gamma(1)=y$, then\footnote{Notice that the sign “$=$” is intended as equivalence of subgroups of $\mathrm{GL}(T_yM)$ and $\mathfrak{gl}(T_yM)$.}
\begin{align}
\mathrm{Hol}_y(\nabla)
&= \operatorname{Ad}_{P^{\gamma}_{0,1}}\!\big(\mathrm{Hol}_x(\nabla)\big)
\coloneqq P^{\gamma}_{0,1}\,\mathrm{Hol}_x(\nabla)\,(P^{\gamma}_{0,1})^{-1}\\[4pt]
\mathfrak{hol}_y(\nabla)
&= \operatorname{Ad}_{P^{\gamma}_{0,1}}\!\big(\mathfrak{hol}_x(\nabla)\big)
\coloneqq P^{\gamma}_{0,1}\,\mathfrak{hol}_x(\nabla)\,(P^{\gamma}_{0,1})^{-1},
\end{align}
where $P^{\gamma}_{0,1}\colon T_xM\to T_yM$ is the fiberwise linear isomorphism induced by $\nabla$ along~$\gamma$.
\end{proposition}
\begin{proof}
Let $\ell$ be a loop based at $y$. Consider the loop based at $x$ obtained by concatenation
$\widetilde\ell\coloneqq \bar\gamma * \ell * \gamma$, where $\bar\gamma$ denotes $\gamma$ traversed backwards.
By the functoriality of parallel transport under concatenation,
\begin{align}
    P^{\widetilde\ell}_{0,1}
= P^{\bar\gamma}_{0,1}\,P^{\ell}_{0,1}\,P^{\gamma}_{0,1}
= (P^{\gamma}_{0,1})^{-1}\,P^{\ell}_{0,1}\,P^{\gamma}_{0,1}.
\end{align}
Hence $P^{\ell}_{0,1} = P^{\gamma}_{0,1}\,P^{\widetilde\ell}_{0,1}\,(P^{\gamma}_{0,1})^{-1}$, 
and varying $\ell$ yields 
$\mathrm{Hol}_y(\nabla) = \operatorname{Ad}_{P^{\gamma}_{0,1}}(\mathrm{Hol}_x(\nabla))$.
Differentiating the conjugation action gives the corresponding identity for the Lie algebras.
\end{proof}

\begin{theorem}[Ambrose--Singer]
Let $\nabla$ be an affine connection on $M$ with curvature $R$, and let $x\in M$.
Moreover, let $\gamma\colon [0,1]\to M$ denote a smooth curve. Then, we have
\begin{align}
\mathfrak{hol}_x(\nabla)
= \operatorname{span}_\gamma\Big\{
P^{\gamma}_{1,0}\,R_{\gamma(1)}(u,v)\,P^{\gamma}_{0,1}
 \Big| \ \gamma(0)=x,\ \,
u,v\in T_{\gamma(1)}M
\Big\}
\subset \mathfrak{gl}(T_xM),
\end{align}
Equivalently,
\begin{align}
\mathfrak{hol}_x(\nabla)
= \operatorname{span}\Big\{
\operatorname{Ad}_{(P^{\gamma}_{0,1})^{-1}}\!\big(R_{\gamma(1)}(u,v)\big)
 \Big| \ 
\gamma(0)=x,\ \ u,v\in T_{\gamma(1)}M
\Big\},
\end{align}
with $\operatorname{Ad}_{Q}(A)\coloneqq Q\,A\,Q^{-1}$.
\end{theorem}

\begin{proposition}\label{propholoomo}
Let $\nabla$ be an affine connection on $M$ with curvature $R$, and let $U\subset M$ be a path--connected open set.
Let $\gamma_0,\gamma_1:[0,1]\to U$ be smooth curves with the same endpoints and homotopic in $U$ with fixed endpoints.\\
If $R= 0$ on $U$, then
\begin{align}
P^{\gamma_0}_{0,1}=P^{\gamma_1}_{0,1}.
\end{align}
In particular, parallel transport on $U$ depends only on the homotopy class (with fixed endpoints) of the curve.\\
Equivalently, if $R= 0$ on $U$, then
\begin{align}
\mathfrak{hol}_x\big(\nabla|_U\big)=\{0\},
\qquad
\mathrm{Hol}^0_x\big(\nabla|_U\big)=\{\mathrm{id}\},
\end{align}
for any $x\in U$.
\end{proposition}
\begin{proof}[Sketch]
By Ambrose--Singer, $R|_U= 0$ implies $\mathfrak{hol}_x(\nabla|_U)=\{0\}$, hence $\mathrm{Hol}^0_x\big(\nabla|_U\big)=\{\mathrm{id}\}$. Parallel transports along homotopic curves differ by an element of $\mathrm{Hol}^0_x$, therefore they coincide.
\end{proof}

\begin{corollary}\label{corhom}
If $R= 0$ on a simply connected open set $U\subseteq M$, then for every $x\in U$ and every loop $\gamma$ based at $x$ with image in $U$,
\begin{align}
P^{\gamma}_{0,1}=\mathrm{id}_{T_xM},
\qquad
\mathrm{Hol}_x\big(\nabla|_U\big)=\{\mathrm{id}\}.
\end{align}
\end{corollary}

\begin{remark}
    In the literature, the condition “$R=0$ in a neighborhood” is also denoted as $R\equiv 0$, to distinguish that from vanishing at a single point.
\end{remark}

\begin{example}[$TS^2$ admits no flat connection]
In \cref{remesconn}, we observed that the same smooth manifold may carry different affine connections, yet these are constrained by the manifold’s own smooth topology. Now that we have basic tools from homology theory, we can treat a classical case: \textit{on the $2$--sphere $S^2$ there is no flat affine connection on the tangent bundle}.\\
Assume, for contradiction, that there exists an affine connection $\nabla$ on $TS^2$ with curvature $R\equiv 0$. Following from \cref{propholoomo}, since $S^2$ is simply connected, parallel transport for $\nabla$ depends only on endpoints; hence for any loop $\gamma$ based at $x\in S^2$ one has
\begin{align}
P^\gamma_{0,1}=\mathrm{id}_{T_xS^2},
\end{align}
so the holonomy is trivial.\footnote{Since $S^2$ is simply connected, one has $\mathrm{Hol}_x(\nabla)=\mathrm{Hol}^0_x(\nabla)$ for every $x\in S^2$.} Fix $x_0\in S^2$ and a basis $(v_1,v_2)$ of $T_{x_0}S^2$. For each $x\in S^2$, choose a smooth curve $\gamma_x$ from $x_0$ to $x$ and define
\begin{align}
E_i(x)\coloneqq P^{\gamma_x}_{0,1}(v_i)\in T_xS^2\qquad (i=1,2).
\end{align}
The trivial holonomy implies that $E_i(x)$ does not depend on the choice of $\gamma_x$ (homotopy invariance), and the definition of parallel transport gives that $x\mapsto E_i(x)$ is smooth. Furthermore, $P^{\gamma_x}_{0,1}$ is an isomorphism, hence $(E_1(x),E_2(x))$ is a basis of $T_xS^2$ for every $x\in S^2$, producing a \textit{global} smooth frame of $TS^2$ and a bundle isomorphism $TS^2\cong S^2\times\mathbb R^2$. This contradicts the non-parallelizability\footnote{Triviality of the tangent bundle.} of $S^2$ (e.g., the hairy ball theorem). Therefore, no flat affine connection on $TS^2$ exists.
\end{example}

\begin{example}[Beware of the “inverse” implication]
Let $M=\mathbb R^2$ with coordinates $(x^1,x^2)$.
Fix the unit circle $\gamma(t)=(\cos(2\pi t),\sin(2\pi t))$, $t\in[0,1]$, and let 
\begin{align}
A=\{\,x\in\mathbb R^2:\ 1-\varepsilon<|x|<1+\varepsilon\,\}
\end{align}
be an annular neighborhood of $\gamma$.
Choose a smooth function $f:\mathbb R^2\to\mathbb R$ such that 
$f\equiv0$ on $A$ but $\partial_2 f\neq0$ somewhere inside the disk $\{|x|<1-\varepsilon\}$.
Define an affine connection by declaring the only nonvanishing Christoffel symbol to be
\begin{align}
\Gamma^{1}_{11}=f(x^1,x^2), \qquad \text{all other }\Gamma^{k}_{ij}=0.
\end{align}
Along $\gamma$ we have $\Gamma^{k}_{ij}(\gamma(t))=0$ for all $t$, so the parallel–transport equation
\begin{align}
\dot V^k(\gamma(t))+\Gamma^{k}_{ij}(\gamma(t))\,\dot\gamma^i(t)\,V^j(\gamma(t))=0
\end{align}
reduces to $\dot V(\gamma(t))=0$. Hence $P^{\gamma}_{0,1}=\mathrm{id}_{T_{\gamma(0)}M}$: the holonomy along $\gamma$ is the identity.
However, the curvature is not identically zero inside the disk. With vanishing torsion,
\begin{align}
R^{\ell}{}_{ijk}
=\partial_i\Gamma^{\ell}_{jk}-\partial_j\Gamma^{\ell}_{ik}
+\Gamma^{\ell}_{ip}\Gamma^{p}_{jk}-\Gamma^{\ell}_{jp}\Gamma^{p}_{ik}.
\end{align}
Since $\Gamma^1_{11}=f$ and all other Christoffel symbols vanish, the only possibly nonzero components are
\begin{align}
R^{1}{}_{1 2 1}=-\,\partial_{2}f,
\qquad
R^{1}{}_{2 1 1}=\,\partial_{2}f,
\end{align}
and all other $R^{\ell}{}_{ijk}$ vanish. In particular, $R\neq0$ wherever $\partial_2 f\neq0$.
Thus this connection satisfies $R\not\equiv0$ on $M$ while the parallel transport along the loop $\gamma$ is the identity.
\end{example}

\begin{proposition}[Flatness from trivial holonomy]
Let $\nabla$ be an affine connection on $M$, and let $U\subseteq M$ be a path--connected open set. Assume that for every $x\in U$ and every piecewise smooth loop $\gamma:[0,1]\to U$ based at $x$ one has
\begin{align}
P^{\gamma}_{0,1}=\mathrm{id}_{T_xM}.
\end{align}
Then, the curvature vanishes on $U$, i.e. $R|_U\equiv 0$.
\end{proposition}
\begin{proof}[Sketch]
If $p,q\in U$ and $\gamma_0,\gamma_1$ are two curves in $U$ from $p$ to $q$, the concatenation $\bar\gamma_1*\gamma_0$ is a loop at $p$, hence
\begin{align}
P^{\gamma_0}_{0,1}=P^{\gamma_1}_{0,1}.
\end{align}
Thus parallel transport between any two points of $U$ is independent of the path. Fix $x_0\in U$ and a basis $(v_1,\dots,v_n)$ of $T_{x_0}M$, and define $E_i(x)\coloneqq P^{\gamma_x}_{0,1}(v_i)$ using any curve $\gamma_x$ from $x_0$ to $x$. Path independence makes $E_i$ well defined and smooth, and $(E_1,\dots,E_n)$ is a global parallel frame on $U$, i.e.\ $\nabla E_i=0$.
Therefore the connection has zero curvature on $U$.
\end{proof}

This result, together with \cref{propholoomo}, implies the following corollary.

\begin{corollary}
    Let $\nabla$ be an affine connection on $M$, and let $U\subseteq M$ be a simply connected open set.
\begin{align}
R|_U\equiv 0 \quad\Longleftrightarrow\quad P^{\gamma}_{0,1}=\mathrm{id}_{T_xM},
\end{align}
for every $x\in U$ and every loop $\gamma$ based at $x$ whose image lies in $U$.
\end{corollary}

\begin{remark}[On the concept of “flatness”]
\leavevmode\newline
\textbf{Geometric intuition.}
\cref{corhom} completes the analogy with the Lie derivative, where the flatness condition is a direct consequence of the naturality of the transport. In other words, the natural transport along flow loops is the identity by construction, whereas parallel transport detects a geometric obstruction given by the curvature. Moreover, in order to see the obstruction given by the curvature, we are not forced to transport vector fields along flow loops, because we have a more structured notion of transporting along curves.
\begin{mdframed}
\normalsize{\textbf{Curvature measures the failure of parallel transporting vector fields along loops back to themselves.}}
\end{mdframed}
\textbf{Algebraic intuition.}
For $X\in\mathfrak X(M)$, the map
\begin{align}
X \longmapsto \nabla_X
\end{align}
acts on $\mathfrak T(M)$ and is $C^\infty(M)$-linear in $X$. For all $X,Y\in\Gamma(TM)$, one has
\begin{align}
[\nabla_X,\nabla_Y]-\nabla_{[X,Y]} = R(X,Y).
\end{align}
In particular, it can no longer be a representation of the real Lie algebra $\mathfrak X(M)$ in general. What holds is: $X \mapsto \nabla_X$ is a \emph{flat} representation of the tangent Lie algebroid on $\mathfrak T(M)$ iff $R=0$.
\begin{mdframed}
\normalsize{\textbf{The covariant derivative is a representation of the tangent Lie algebroid iff its curvature vanishes.}}
\end{mdframed}
\textbf{Comparison with the Lie derivative.}
The Lie derivative yields a strict Lie-algebra representation on the real vector space $\mathfrak T(M)$:
\begin{align}
[\mathcal L_X,\mathcal L_Y]=\mathcal L_{[X,Y]},\qquad
[\mathcal L_X,\mathcal L_Y]=0 \ \text{whenever}\ [X,Y]=0.
\end{align}
This comes at the price of not being $C^\infty(M)$--linear in $X$. Imposing $C^\infty(M)$--linearity led us to obtaining a covariant derivative, and this changed the algebra.
Thus, the commutation of transports, even along commuting vector fields, is obstructed by the curvature, and one has
\begin{align}
[\nabla_X,\nabla_Y]=\nabla_{[X,Y]}\ \text{for all }X,Y
\quad\Longleftrightarrow\quad
R=0.
\end{align}
\begin{figure}[h]
\centering
\includegraphics[width=0.55\textwidth]{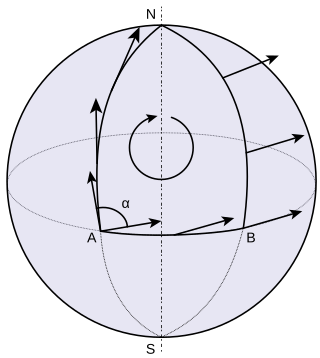}
\caption{Parallel transport along a loop in $S^2$. Since there cannot be a flat connection on the $2$--sphere, the holonomy of a loop need not to be the identity. Source: Wikiwand--Parallel Transport.}
\end{figure}
\end{remark}

\printbibliography[heading=subbibliography]
\chapter{The General Theory of Relativity}
\section{Albert Einstein: The World as I See It}

\textit{“The World as I See It” is the title of a famous book by Albert Einstein. In this book, Einstein sets out his philosophical view of the world he lived in and also comments on the theory of special relativity. One of the authors of these notes, when he was seventeen, might have stolen his grandfather’s old pipe, found a 1989 edition of the book in the village’s small library, and read it on his balcony in the cold of the winter, while trying to smoke some old, dry tobacco.}\\

\noindent Special relativity enlarged the kinematical symmetry of space–time from the Galilean to the Poincaré transformations. Our construction is two--folded. First, it supplied a kinematics that preserves the speed of light $c$;\footnote{Indeed, in special relativity, the symmetry transformations to change from one inertial frame to another are the Poincaré ones.} second, it absorbed the full dynamical content of Maxwell’s electrodynamics into the same symmetries of the free particle: the equations of electromagnetism are invariant in form under the symmetry of free motion.\footnote{Recall from \cref{remsymmforc} that, in Newtonian mechanics, introducing forces typically reduces the symmetry of the equations relative to the free particle. By contrast, special relativity is formulated so that the full electromagnetic interaction is built within the same space--time (Poincaré) symmetry as the free motion.} In this precise sense, special relativity offered a symmetry--based reformulation of mechanics and electromagnetism.\\

\noindent\textbf{Conceptual residues after special relativity.}
\begin{enumerate}
\item \textit{Partiality of the relativity principle.} This construction realizes Einstein’s principle of relativity only for mechanics and electromagnetism. Gravitation is not included in this construction. A prescribed external gravitational field generally breaks the Poincaré symmetry of a fixed background configuration; this does not by itself exclude Poincaré--covariant dynamical theories of gravity on Minkowski space--time.
\item \textit{Epistemic incompleteness of the definition of inertial frames.} Special relativity still relies on the notion of an inertial frame. Operationally, the law of inertia is formulated in such frames, and those frames are in turn identified by the validity of inertial motion. This circularity is not removed at the structural level by special relativity.\footnote{In the Minkowski construction, this fact translates to an \textit{a priori} choice of a fixed geometry given by $\eta$.}
\item \textit{General descriptions.} This issue is a generalization of the two preceding concepts, and it seems even more abstract; it hardly looks like a genuine conceptual problem, yet it is the central point of general relativity and where Einstein’s genius is visible. To quote Richard Feynman: “\textit{Einstein himself, of course, arrived at the same Lagrangian but without the help of a developed field theory, and I must admit that I have no idea how he guessed the final result. We have had troubles enough arriving at the theory — but I feel as though he had done it while swimming underwater, blindfolded, and with his hands tied behind his back!}” Einstein was aware that a mere change of description---a change of perspective on the universe---could not in any case affect physics. Nonetheless, in special relativity, transformations between inertial Cartesian descriptions are the Poincaré transformations. General coordinates are also admissible, but the metric components and connection coefficients then change, so Maxwell equations and relativistic mechanics need not retain their Cartesian component expressions. By contrast, “general” transformations/descriptions spoil the manifest Poincaré form and seem to generate \textit{fictitious forces}, present in one frame and absent in another. Therefore, although \textit{physics itself is not affected by a change of perspective}, it seems that two \textit{arbitrary descriptions}, connected by a general transformation, would \textit{formally} portray the universe---the same physical phenomena---in two “too different” ways.\footnote{Notice that here there is a subtle but important point: changes of description are not necessarily physical transformations among different observers/frames. We will deepen further this topic later on.}
\end{enumerate}

\noindent\textbf{Experimental evidence: Mercury’s perihelion.}

\medskip
\noindent The perihelion of Mercury—the point of closest approach to the Sun—does not remain fixed, but its orbital ellipse slowly rotates within the orbital plane (a secular \emph{apsidal precession}). To set the angular scale on the sky,
\begin{align}
1^{\circ}=60'=3600'', 
\qquad
1''=\frac{\pi}{648000}\ \mathrm{rad}\approx 4.848\times 10^{-6}\ \mathrm{rad}.
\end{align}
The \emph{total} observed secular advance (nineteenth–century ephemerides) is
\begin{align}
\dot\varpi_{\mathrm{observed}}  \approx  5600''/\text{century}.
\end{align}
A theoretical explanation in the Newtonian setting was attempted by subtracting the precession of the equinoxes (an effect from Earth’s axis precession) and the planetary perturbations (mutual Newtonian interactions with other planets). Nonetheless, one is left with the well–known residual:

\begin{align}
\boxed{\ \Delta\dot\varpi  \approx  43''/\text{century}.}
\end{align}
Since Mercury completes about \(415\) revolutions per century, this is
\begin{align}
\Delta\varpi_{\mathrm{per\ orbit}}  \approx  \frac{43''}{415}  \approx  0.103''\ \text{per orbit}.
\end{align}
\textit{How small is $43''$?} The Moon’s apparent diameter is about \(30'=1800''\), so \(43''\) is \(\sim 2.4\%\) of the lunar disk. At arm’s length, the width of a finger subtends \(\sim 1^{\circ}=3600''\); thus \(43''\approx 1/84\) of a finger’s width. A human hair of thickness \(\sim 0.1\,\mathrm{mm}\) viewed from \(60\,\mathrm{cm}\) subtends \(\theta  =  \frac{0.1}{600}\ \mathrm{rad} \approx 1.67\times10^{-4}\ \mathrm{rad}
 \approx 35''.\), i.e. it is essentially a \emph{hair’s breadth on the sky}.\\
The perihelion advance of Mercury thus stood as a small anomaly within an otherwise exceptionally successful Newtonian framework; indeed, the prevailing confidence in Newtonian gravitation motivated explanations \emph{within} the paradigm---such as an intra--Mercurial planet (“Vulcan”) or a non--spherical solar figure---rather than a reformulation of the Newtonian theory. In this light, Einstein's theory that we are about to develop was driven chiefly by conceptual considerations, rather than by an experimental crisis!\\

Einstein's idea was to overcome all these conceptual issues at once, by generalizing the laws of physics to all possible descriptions. It turns out that this is the natural setting to develop a theory for gravity.

The insight behind this fact is highlighted in the famous thought experiment of \textit{Einstein's lift}: consider a small laboratory (“elevator”). Let a \textit{description} be, informally, a local rule that labels\footnote{What kind of geometric object does this process of labeling recall?} events in space--time. In particular, the representation of a worldline $\gamma$ and of its time–derivatives $\dot{\gamma},\ddot{\gamma}$ (the “velocity” and “acceleration”) depends on the chosen description.\\
\textit{Case A: static elevator with a uniform gravitational field.}
In a description $\mathcal D_{\mathrm{stat}}$ where the elevator is at rest and gravity is homogeneous with constant $\mathbf g$, a freely falling test body satisfies
\begin{align}
\ddot{\mathbf{\gamma}}^{\mathcal D_{\mathrm{stat}}}(t)=\mathbf g.
\end{align}
A horizontal light pulse emitted from one wall to the opposite crosses the cabin in time $T\simeq L/c$ (with $L$ the horizontal width). In the small-laboratory approximation $gL/c^2\ll1$, neglecting tidal effects, the equivalence principle predicts the leading vertical drop
\begin{align}
\Delta h^{\mathcal D_{\mathrm{stat}}}
= \tfrac12\, g\,T^2
= \tfrac12\, g\,\frac{L^2}{c^2},
\qquad
\theta^{\mathcal D_{\mathrm{stat}}}\simeq \frac{\Delta h}{L}=\frac{gL}{2c^2},
\end{align}
i.e. the beam “falls” slightly while traversing the elevator. Here $\theta$ is the angle of the chord joining emission and reception; the change in the beam tangent is approximately $gL/c^2$. The following elevator formulas use the same local approximation.\\
\textit{Case B: freely falling elevator.}
Form a new description $\mathcal D_{\mathrm{fall}}$ by re-describing every spatial reading at time $t$ by subtracting the common fall of the elevator:
\begin{align}
\mathcal R_t:\ \text{position}\ \mapsto\ \text{position}-\tfrac12\,\mathbf g\,t^2 .
\end{align}
If $\gamma^{\mathcal D_{\mathrm{stat}}}(t)$ reads a worldline in the static picture, its re-description inside the falling cabin is
\begin{align}
\gamma^{\mathcal D_{\mathrm{fall}}}(t)
 =  \mathcal R_t\big(\gamma^{\mathcal D_{\mathrm{stat}}}(t)\big).
\end{align}
Differentiating twice in $t$ yields the transformation law for the acceleration \emph{as read by the description}:
\begin{align}
\ddot\gamma^{\mathcal D_{\mathrm{fall}}}(t)
= \ddot\gamma^{\mathcal D_{\mathrm{stat}}}(t) - \mathbf g.
\end{align}
Hence free test bodies satisfy
\begin{align}
\ddot\gamma^{\mathcal D_{\mathrm{fall}}}(t)=\mathbf 0,
\end{align}
so inside the falling elevator the same motions appear gravity–free.
Apply the same re-description to the light pulse. In $\mathcal D_{\mathrm{stat}}$ its vertical reading was $h(t)=\tfrac12 g t^2$; in $\mathcal D_{\mathrm{fall}}$ one finds
\begin{align}
h^{\mathcal D_{\mathrm{fall}}}(t)
= h^{\mathcal D_{\mathrm{stat}}}(t) - \tfrac12 g t^2
= 0,
\end{align}
so the beam appears straight across the cabin to the order retained in this approximation.
In other words, the same occurrence is read either as “gravity present” ($\mathcal D_{\mathrm{stat}}$) or as “no gravity” ($\mathcal D_{\mathrm{fall}}$), depending only on the description.\\
\textit{Conversely: an accelerated description manufactures gravity.}
Start from a gravity–free description $\mathcal D_0$ where free bodies have $\ddot\gamma^{\mathcal D_0}=0$ and a light pulse travels straight.
Produce a new description $\widetilde{\mathcal D}$ by re-describing positions via
\begin{align}
\widetilde{\mathcal R}_t:\ \text{position}\ \mapsto\ \text{position}-\tfrac12\,\mathbf a\,t^2 ,
\end{align}
with constant $\mathbf a$ upwards. Then
\begin{align}
\ddot\gamma^{\widetilde{\mathcal D}}(t)
= \ddot\gamma^{\mathcal D_0}(t) - \mathbf a
= -\mathbf a,
\end{align}
so inside $\widetilde{\mathcal D}$ every free body “falls” with acceleration $-\mathbf a$,
exactly mirroring Case A upon the identification $\mathbf g=-\mathbf a$.
\begin{remark}
    Notice that we implicitly assumed that the \textit{inertial mass}, the one in the equation of motion of the free particle, equals the \textit{gravitational mass}, the one that \textit{couples} to gravity. This was a hidden assumption of Newtonian gravity as well. This principle would be replaced by the \textit{universality of the free--fall}.
\end{remark}
These considerations about the dependency of gravity on the description led Einstein to the formulation of his famous \textit{equivalence principle} (EEP).
\begin{mdframed}
\begin{axiom*}[Einstein equivalence principle]
For any event $p$ in space--time and any sufficiently small laboratory around $p$, there exists a (freely--falling) description in which:
\begin{enumerate}\setlength{\itemsep}{0pt}\renewcommand{\labelenumi}{(\roman{enumi})}
\item \textit{Weak equivalence principle (WEP)}: any two test bodies with negligible self--gravity and in absence of other interactions, released with the same initial position and velocity, follow the same worldline independent of the mass or other internal states (universality of the free--fall);
\item \textit{Local Lorentz invariance (LLI)}: the non--gravitational laws reduce at $p$ to the form of special relativity;
\item \textit{Local inertial frame (LIF)}: the metric can be made Minkowskian at $p$ with vanishing first derivatives there; tidal effects encoded by curvature cannot be eliminated;
\item \textit{Local position invariance (LPI)}: for any other event $q$ and for any sufficiently small freely--falling laboratory around it, identically prepared copies of every local non--gravitational experiment performed in the two labs yield identical outcome statistics.
\end{enumerate}
\end{axiom*}
\end{mdframed}
\begin{remark}
Some points in this statement may still appear not fully defined. This is intentional: the reader need not be concerned, as we will formalize these notions shortly in a precise mathematical sense.
\end{remark}
\begin{remark} Given EEP, a single coherent theory covering \textit{all} small laboratories requires that the laws be stated so that changing the description of the very same physical situation does not alter either their form or their predictions.\\
Therefore, in addition, one can assume the following.
\begin{enumerate}\setlength{\itemsep}{2pt}%
\renewcommand{\labelenumi}{\normalfont(\roman{enumi})}
\setcounter{enumi}{4}
\item \textbf{Description--neutrality:} the local relations (equations) describing the laws of physics are invariant in form under a general change of description;
\item \textbf{No hidden scaffolding (Kretschmann 1917):} no fixed, non--dynamical background element is allowed to single out some particular descriptions.
\end{enumerate}
\end{remark}
This is equivalent to stating the following axiom.
\begin{mdframed}
\begin{axiom*}
(v) $+$ (vi)$\;$ The laws of physics are independent of the description.
\end{axiom*}
\end{mdframed}
\begin{remark}
This principle captures the core of general relativity.
\end{remark}
\begin{remark}
Assumption (v) is also known as \textit{principle of general covariance}.\footnote{Notice that, by Kretschmann, this assumption alone does not ensure a description--independent physics.}
\end{remark}

\section{Pseudo--Riemannian manifolds and the Levi--Civita connection}
From the stated EEP, it is clear that we need a notion of space--time distance.\footnote{See for example point (iii), where we want a space--time “scale” of the lab.} We start with naive choice of a distance on $M$, which always exists thanks to the following result.
\begin{theorem}[Urysohn metrization for manifolds]
Let $(M,[\mathcal A])$ be a smooth manifold. Then there exists a distance $d:M\times M\to[0,\infty)$ that induces the given topology of $M$. Equivalently, $M$ is homeomorphic (via the identity on the underlying set) to the metric space $(M,d)$.
\end{theorem}
\begin{example}[Different distances generate the same topology]
\leavevmode\newline
\vspace{-0.55cm}
\begin{enumerate}
\item \textbf{Monotone reparametrization on a general metric space.}
Let $(M,d_1)$ be a metric space and define
\begin{align}
d_2(p,q) := \arctan \big(d_1(p,q)\big).
\end{align}
For $a,b\ge0$, $\arctan(a+b)-\arctan(a)=\int_0^b(1+(a+t)^2)^{-1}dt\le\arctan b$, so $\arctan$ is subadditive and $d_2$ satisfies the triangle inequality. Since $\arctan$ is strictly increasing with $\arctan(0)=0$, for every $p\in M$ and $0<r<\tfrac{\pi}{2}$,
\begin{align}
B_{d_2}(p,r) = \{q:\ d_2(p,q)<r\} = \{q:\ d_1(p,q)<\tan r\} = B_{d_1}(p,\tan r).
\end{align}
Thus $d_1$ and $d_2$ induce the same topology.

\item \textbf{Truncation on $\mathbb{R}^n$.}
Let $d_1(x,y)=\|x-y\|$ and define
\begin{align}
d_2(x,y) := \min\{1,\|x-y\|\}.
\end{align}
For $0<r<1$ one has
\begin{align}
B_{d_2}(x,r) = \{y:\ d_2(x,y)<r\} = \{y:\ \|x-y\|<r\} = B_{d_1}(x,r),
\end{align}
and for $r>1$,
\begin{align}
B_{d_2}(x,r) = \{y:\ d_2(x,y)<r\} = \mathbb{R}^n.
\end{align}
At $r=1$ the ball is $B_{d_1}(x,1)$. The triangle inequality follows from $\min\{1,a+b\}\le\min\{1,a\}+\min\{1,b\}$ for $a,b\ge0$. Hence $d_1$ and $d_2$ generate the same topology.

\item \textbf{Bounded monotone reparametrization on $\mathbb{R}^n$.}
Let $d_1(x,y)=\|x-y\|$ and set
\begin{align}
d_2(x,y) := \frac{\|x-y\|}{1+\|x-y\|}.
\end{align}
The map $f(t)=\tfrac{t}{1+t}$ is subadditive, since $\frac{a+b}{1+a+b}\le\frac{a}{1+a}+\frac{b}{1+b}$ for $a,b\ge0$, so $d_2$ is a distance. It is also a strictly increasing homeomorphism $[0,\infty)\to[0,1)$ with inverse $f^{-1}(s)=\tfrac{s}{1-s}$. Thus, for every $x\in\mathbb{R}^n$ and $0<r<1$, we have
\begin{align}
B_{d_2}(x,r) = \{y:\ d_2(x,y)<r\} = \{y:\ \|x-y\|<\tfrac{r}{1-r}\} = B_{d_1} \Big(x,\frac{r}{1-r}\Big).
\end{align}
Therefore, $d_1$ and $d_2$ induce the same topology.
\end{enumerate}
\end{example}

\begin{remark}
Since different (indeed, infinitely many) choices of a distance $d$ give rise to the same topology on $M$, we may regard them as gauge--equivalent if we retain only the topology they induce. In this sense, we are introducing what goes by the name of a \textit{gauge symmetry}: a freedom in the choice of a space--time object which produces an invariant result. In this case, the space--time object is $d$ and the invariant is the topology of $M$.\\
Since $d$ merely represents the topology already assigned to $M$, it introduces no additional background structure and need not be dynamical. However, any two distances inducing the same topology lie, by our choice, in the same gauge--equivalence class. Hence, $d$ carries no local dynamical degrees of freedom: the gauge class $[d]$ is fixed because the topology on $M$ is fixed. Therefore, every gauge--invariant quantity built from $[d]$ would be topological.\\
Then, if the only structure on our space--time were to be $[d]$, the resulting theory would have only topological content. We will thus not attempt to build gravity out of $d$ considered up to this equivalence, which implies that we do \textit{not} assume, in principle, that gravity is a topological theory.
\end{remark}

\begin{remark}
    GR is indeed a gauge theory: its diffeomorphism symmetry acts on the dynamical geometric fields on a fixed smooth manifold.
\end{remark}
We introduce now the fundamental dynamical object of the theory.
\begin{definition}[Signature of a symmetric bilinear form]
Let $V$ be a finite-dimensional real vector space and $b:V\times V\to\mathbb{R}$ a symmetric, non--degenerate bilinear form.
By Sylvester's law of inertia, there exists a basis $(e_1,\dots,e_n)$ of $V$ such that
\begin{align}
b(e_i,e_j)&=0 \quad\text{for } i\neq j,\\
b(e_i,e_i)&=-1 \ \text{for } i=1,\dots,p,\qquad
b(e_i,e_i)=+1 \ \text{for } i=p+1,\dots,p+q,
\end{align}
with $p+q=n=\dim V$. Equivalently, if $B$ denotes the associated matrix of $b$ in this basis, given by $B_{ij}=b(e_i,e_j),$ then $B$ is diagonal with $p$ entries $-1$ and $q$ entries $+1$.
The ordered pair $(p,q)$ is invariant with respect to the choice of any basis and is called the \emph{signature} of $b$:
\begin{align}
\mathrm{sig}(b)\coloneqq (p,q).
\end{align}
\end{definition}

\begin{definition}[Riemannian metric]
Let $M$ be a smooth manifold. A \emph{Riemannian metric} on $M$ is a smooth section
\begin{align}
g &\in \Gamma(\mathrm{Sym}^2 T^*M),
\end{align}
such that for every $x\in M$ the bilinear form $g_x$ on $T_xM$ is symmetric, non--degenerate, and positive definite. Equivalently, $\mathrm{sig}(g_x)=(0,\dim M)$ for all $x\in M$.
\end{definition}

\begin{definition}[Pseudo--Riemannian metric]
A \emph{pseudo--Riemannian metric} on a smooth manifold $M$ is a smooth section
\begin{align}
g &\in \Gamma(\mathrm{Sym}^2 T^*M),
\end{align}
such that for every $x\in M$ the bilinear form $g_x$ on $T_xM$ is symmetric, non--degenerate, and has signature
\begin{align}
\mathrm{sig}(g_x)=(p,q),\qquad p+q=\dim M.
\end{align}
Equivalently, for every $x\in M$ there is a basis of $T_xM$ in which the matrix of $g_x$ is
$\mathrm{diag}(\underbrace{-1,\dots,-1}_{p},\underbrace{1,\dots,1}_{q})$.
\end{definition}

\begin{definition}[Lorentzian manifold]
A \emph{Lorentzian manifold} is a pair $(M,g)$ where $M$ is a smooth manifold and $g$ is a pseudo--Riemannian metric with
\begin{align}
\mathrm{sig}(g_x)=(1,\dim M-1), \qquad \text{for all } x\in M.
\end{align}
\end{definition}

\begin{remark}
This fixes the $(-,+,\dots,+)$ convention. The opposite convention $(+,-,\dots,-)$ is equally valid and corresponds to $(p,q)=(\dim M-1,1)$.
\end{remark}

\begin{definition}[Cometric]
Let $(M,g)$ be a pseudo--Riemannian manifold. Since $g_x$ is non--degenerate for every $x\in M$, we can define the inverse, called the \emph{cometric}, which is the smooth tensor field of $(2,0)$--type
\begin{align}
g^{-1}\in\Gamma(\mathrm{Sym}^2 TM).
\end{align}
Then, the metric $g$ and the cometric $g^{-1}$ induce the isomorphisms of vector bundles
\begin{align}
g\colon TM&\longrightarrow T^*M\\
X&\longmapsto i_X g, 
\end{align}
and
\begin{align}
g^{-1}\colon T^*M&\longrightarrow TM\\
\alpha&\longmapsto g^{-1}(\alpha).
\end{align}
In local coordinates, if $g=g_{ij}\,dx^i\otimes dx^j$ and $(g^{ij})=(g_{ij})^{-1}$, then
\begin{align}
g^{-1}=g^{ij}\,\partial_i\otimes \partial_j,\qquad i_X g=g_{ij}X^i\,dx^j,\qquad g^{-1}(\alpha)=g^{ij}\alpha_i\,\partial_j.
\end{align}
\end{definition}

\begin{remark}[Raising and lowering of tensor indices]
The bundle isomorphisms
\(
g\colon TM\to T^*M
\)
and
\(
g^{-1}\colon T^*M\to TM
\)
are used to \emph{lower} and \emph{raise} indices of tensor fields. In local
coordinates this amounts to contracting with $g_{ij}$ or $g^{ij}$, e.g.
\begin{align}
X_i = g_{ij}X^j, \qquad \alpha^i = g^{ij}\alpha_j,
\end{align}
and similarly for each index of a the components of a general tensor field.
\end{remark}

\begin{definition}[Orthonormal frames]
Let $\mathrm{sig}(g)= (p,q)$, $p+q=n=\dim M$. A local frame $e=(e_1,\dots,e_n)$ of $TM$ over $U\subset M$ is \emph{$g$--orthonormal} if
\begin{align}
g_x(e_i(x),e_j(x))=\mathrm{diag}(\underbrace{-1,\dots,-1}_{p},\underbrace{1,\dots,1}_{q})_{ij},
\end{align}
for every $x\in U$.
\end{definition}

\begin{remark}
For a Lorentzian manifold $(M,g)$, the condition above becomes
\begin{align}
g_x(e_i(x),e_j(x))=\eta_{ij},
\end{align}
where $\eta=\mathrm{diag}(-1,\underbrace{1,\dots,1}_{\dim M-1})$ is the Minkowskian metric in standard form.
\end{remark}

\begin{proposition}[Indefinite Gram--Schmidt]
Let $(M,g)$ be a pseudo--Riemannian manifold. For every $x\in M$ there exists a neighborhood $U$ and a $g$--orthonormal frame on $U$.
\end{proposition}

\begin{remark}
    In physics, reference frames are space--time objects that allow one to measures separations of events. An orthogonal frame is precisely the mathematical object representing a reference frame. Notice that only linear transformations that preserve this orthogonal structure are considered physical changes of reference frame.
\end{remark}

\begin{definition}[Causal classes and null cones]
Let $(M,g)$ a Lorentzian manifold and $V\in\Gamma(TM)$ be a vector field. Then, $V$ is causally classified as follows:
\begin{align}
\text{timelike} &\iff g_x(V(x),V(x))<0\ \text{for all } x,\\
\text{null (lightlike)} &\iff g_x(V(x),V(x))=0\ \text{for all } x,\\
\text{spacelike} &\iff g_x(V(x),V(x))>0\ \text{for all } x.
\end{align}
A smooth curve $\gamma$ is timelike/null/spacelike if $\dot\gamma$ is a  timelike/null/spacelike vector field. Finally, we define the \emph{null (or light) cone} at $x$ as $\mathcal J_x\coloneqq\{v\in T_xM\setminus\{0\}:\ g_x(v,v)=0\}$, i.e. the set of all null vectors in $T_xM$.
\end{definition}

\begin{remark}[Invariance under coordinate changes]
Let $(x^1,\dots,x^n)$ and $(x'^1,\dots,x'^n)$ be local coordinates around $x\in M$. Write $v=v^i\partial_i=v'^a\partial'_a$ and $g=g_{ij}\,dx^i\otimes dx^j=g'_{ab}\,dx'^a\otimes dx'^b$. Then,
\begin{align}
v'^a&=\frac{\partial x'^a}{\partial x^i}\,v^i, &
g'_{ab}(x')&=\frac{\partial x^i}{\partial x'^a}\,\frac{\partial x^j}{\partial x'^b}\,g_{ij}(x).
\end{align}
Therefore, since $g$ a tensor field, the scalar $g_x(v,v)$ is chart--independent:
\begin{align}
g'_{ab}\,v'^a v'^b
=\Big(\frac{\partial x^i}{\partial x'^a}\,\frac{\partial x^j}{\partial x'^b}\,g_{ij}\Big)
\Big(\frac{\partial x'^a}{\partial x^k}\,v^k\Big)
\Big(\frac{\partial x'^b}{\partial x^\ell}\,v^\ell\Big)
= g_{ij}\,v^i v^\jmath.
\end{align}
Hence, the timelike/null/spacelike classifications and the null--cones at every $x\in M$ $\{v\neq 0:\ g_x(v,v)=0\}$ are intrinsic and do not depend on the chosen coordinates.
\end{remark}

\begin{remark}
    This local view suggests that a theory built out of a Lorentzian metric $g$ will possess a gauge symmetry given by coordinate transformations: the components of the metric itself change, but the scalars constructed with the metric are gauge (coordinate) invariant quantities. This is precisely the gauge symmetry that one has in general relativity.\footnote{We will see in the next section that the purest form of gauge symmetry of general relativity is actually the one of space--time diffeomorphisms, whose local expression is given by coordinate transformations.}
\end{remark}

\begin{definition}[Time orientability]
For each $x\in M$, define the \textit{time-like cone} $\mathcal T_x$ at $x$ as
\begin{align}
  \mathcal T_x \coloneqq \{v\in T_xM:\ g_x(v,v)<0\},
\end{align}
which has two connected components. Then, a Lorentzian manifold $(M,g)$ is time–orientable iff there exists a (global) smooth timelike vector field $T\in\Gamma(TM)$. Such a $T$ selects, at each $x$, one of the two connected components of $\mathcal T_x$, namely the one containing $T_x$. A \emph{time–orientation} $\mathfrak t$ is a smooth choice of one of the two timelike cone components at every point, represented by such a field $T$. Two smooth timelike vector fields define the same time–orientation iff, for every $x\in M$, their values lie in the same connected component of $\mathcal T_x$. Once a time–orientation is fixed, for every $x$, one writes
\begin{align}
    \mathcal T_x = \mathcal T_x^+ \sqcup \mathcal T_x^-,
\end{align}
where we denote by $\mathcal T_x^+$ the component containing $T_x$. We call $\mathcal T_x^+$ the \emph{future} timelike cone and $\mathcal T_x^-$ the \emph{past} timelike cone. A time--like vector $v\in T_xM$ is future–directed (or past--directed) iff it lies in the future (or respectively past) component of the causal cone at $x$.
\end{definition}

\begin{definition}[Cauchy hypersurface and (Cauchy) temporal function]
Let $(M,g)$ be a time–oriented Lorentzian manifold. 
A subset $\Sigma\subset M$ is a \emph{Cauchy hypersurface} if it is a smooth spacelike embedded hypersurface which every inextendible causal curve intersects exactly once. 
A smooth function $t:M\to\mathbb R$ is \emph{temporal} if its gradient is everywhere past--directed timelike\footnote{With our signature, $dt(v)=g(\operatorname{grad}_g t,v)>0$ for every nonzero future--directed causal vector $v$.}. This implies that $t$ is strictly increasing along every future--directed causal curve; strict increase alone does not imply a timelike gradient. 
It is \emph{Cauchy temporal} if, in addition, $t$ is surjective and each level set $t^{-1}(s)$ is a Cauchy hypersurface.
\end{definition}

\begin{theorem}[Bernal–Sánchez splitting theorem]
Let $(M,g)$ be a time–oriented Lorentzian manifold. The following are equivalent:
\begin{enumerate}\setlength{\itemsep}{2pt}\renewcommand{\labelenumi}{(\roman{enumi})}
\item $(M,g)$ is \emph{globally hyperbolic};
\item $(M,g)$ admits a smooth spacelike Cauchy hypersurface $\Sigma$;
\item There exists a \emph{Cauchy temporal function} $t:M\to\mathbb R$ whose level sets $\Sigma_s\coloneqq t^{-1}(s)$ are smooth spacelike Cauchy hypersurfaces;
\item There is a diffeomorphism $\Phi:M\to\Sigma\times\mathbb R$ and smooth data $\beta:\Sigma\times\mathbb R\to(0,\infty)$ and a smooth one–parameter family $\{h_s\}_{s\in\mathbb R}$ of Riemannian metrics on $\Sigma$ such that, writing $t\coloneqq \operatorname{pr}_{\mathbb R} \circ \Phi$,
\begin{align}
g  =  -\,\beta\,dt^2  +  h_t ,
\end{align}
and, in particular, each slice $\Sigma_s=\Sigma\times\{s\}$ is a Cauchy hypersurface.
\end{enumerate}
Moreover, given any (smooth spacelike) Cauchy hypersurface $\Sigma_0\subset M$, one can choose $t$ with $\Sigma_0=t^{-1}(0)$.
\end{theorem}

\begin{remark}
    In particular, if $M$ is globally hyperbolic, then $M\cong \Sigma \times\mathbb R$, where $\Sigma$ is a Cauchy hypersurface. 
\end{remark}

\begin{definition}[Metric pairing on tensors]\label{met_pair}
Let $(M,g)$ be a Lorentzian smooth manifold of dimension $n\ge 2$. The metric $g$ and cometric $g^{-1}$ extend to a pointwise pairing $\langle\cdot,\cdot\rangle_g$ on the full tensor algebra $\mathfrak T(M)$. For $A\in\Gamma(T^*M\otimes TM)$, the \emph{metric adjoint} $A^\dagger\in\Gamma(T^*M\otimes TM)$ is defined by
\begin{align}
g(A^\dagger X,Y)=g(X,AY)\qquad\text{for all }X,Y\in\Gamma(TM).
\end{align}
The \emph{metric trace} is
\begin{align}
\mathrm{tr}_g A \coloneqq \langle A,\mathrm{id}_{TM}\rangle_g = A^i{}_i \quad \text{in }g\text{-orthonormal coordinates.}
\end{align}
\end{definition}

\begin{definition}[Volume density and Hodge star]
Let $(M,g)$ be an oriented pseudo Riemannian manifold of dimension $n$. In positively oriented local coordinates $(x^1,\dots,x^n)$, choose the volume form\footnote{See \cref{or_com_clos}.} $\mathrm{vol}_g\in\Omega^n(M)$ such that
\begin{align}
\mathrm{vol}_g \coloneqq \sqrt{|\det(g_{ij})|}\,dx^1\wedge\cdots\wedge dx^n.
\end{align}
The \emph{Hodge star} $\star:\Omega^k(M)\to \Omega^{n-k}(M)$ is the unique bundle map such that, for all $\alpha,\beta\in\Omega^k(M)$,
\begin{align}
\alpha\wedge\star\beta = \langle \alpha,\beta\rangle_g\,\mathrm{vol}_g,
\end{align}
where the induced pairing on alternating $k$--forms uses the normalization
\begin{align}
\langle\alpha,\beta\rangle_g
=\frac{1}{k!}\,g^{i_1j_1}\cdots g^{i_kj_k}\alpha_{i_1\cdots i_k}\beta_{j_1\cdots j_k}.
\end{align}
Thus it is $1/k!$ times the full tensor contraction from \cref{met_pair}. Finally, if $\mathrm{sig}(g)=(p,q)$, then
\begin{align}
\star\star\big|_{\Omega^k} = (-1)^{\,k(n-k)+p}\,\mathrm{id}_{\Omega^k}.
\end{align}
\end{definition}

\begin{mdframed}
\normalsize\textbf{{From now on, we will adopt the Greek letters for space--time components.}}
\end{mdframed}

\begin{definition}[Levi--Civita connection]
Let $(M,g)$ be a Lorentzian manifold. The \emph{Levi--Civita connection} of $g$ is the unique affine connection $\nabla$ such that the following hold:
\begin{enumerate}\setlength{\itemsep}{0pt}\renewcommand{\labelenumi}{(\roman{enumi})}
\item \textbf{Torsion--free:} for all $X,Y\in\Gamma(TM)$,
\begin{align}
\nabla_XY-\nabla_YX-[X,Y]=0.
\end{align}
In coordinates, it reads $\Gamma^\rho_{\mu\nu}-\Gamma^\rho_{\nu\mu}=0$.
\item \textbf{Metric--compatibility:} for all $X,Y,Z\in\Gamma(TM)$,
\begin{align}
\nabla_X \big(g(Y,Z)\big)=g(\nabla_XY,Z)+g(Y,\nabla_XZ),
\end{align}
i.e. $\nabla_X g=0$ for all $X\in\Gamma(TM)$. In coordinates, it reads $\nabla_\rho g_{\mu\nu}=0$.
\end{enumerate}
The \textit{Christoffel symbols} are given by
\begin{align}
\Gamma^\alpha_{\ \beta\gamma}
= \tfrac12\,g^{\alpha\mu} \left(\partial_\beta g_{\mu\gamma}
+\partial_\gamma g_{\mu\beta}
-\partial_\mu g_{\beta\gamma}\right),
\end{align}
and satisfy $\Gamma^\alpha_{\ \beta\gamma}=\Gamma^\alpha_{\ \gamma\beta}$.\footnote{This condition is given by the fact that the Levi--Civita connection is torsion--free.}
\end{definition}

\begin{remark}
    It is useful to remember that, for any general $\nabla$ (possibly \textit{not} metric--compatible), one also has $\nabla_X \big(g(Y,Z)\big)=X\big(g(Y,Z)\big)$, since $g(Y,Z)\in C^\infty(M)$.
\end{remark}

\begin{definition}[Riemann curvature tensor]
Let $(M,g)$ be a Lorentzian smooth manifold and let $\nabla$ be the Levi--Civita connection of $g$. 
The Riemann curvature tensor is the curvature $R$ of $\nabla$, and it is denoted as $R(g)$.
\end{definition}

\begin{proposition}
Let $(M,g)$ be a Lorentzian smooth manifold and let $\nabla$ be the Levi--Civita connection of $g$. 
The Riemann curvature is completely determined by $g$ and, in local coordinates, it reads
\begin{align}
(R(g))^\rho{}_{\ \mu\nu\sigma}
= \partial_\mu \Gamma^\rho{}_{\nu\sigma}
 - \partial_\nu \Gamma^\rho{}_{\mu\sigma}
 + \Gamma^\rho{}_{\mu\lambda}\Gamma^\lambda{}_{\nu\sigma}
 - \Gamma^\rho{}_{\nu\lambda}\Gamma^\lambda{}_{\mu\sigma}.
\end{align}
Equivalently, lowering the output index (without reordering the three input indices),
\begin{align}
(R(g))_{\alpha\beta\gamma\delta}
= g_{\alpha\lambda}\,(R(g))^\lambda{}_{\ \beta\gamma\delta}
= g_{\alpha\lambda} \left(
\partial_\beta \Gamma^\lambda{}_{\gamma\delta}
- \partial_\gamma \Gamma^\lambda{}_{\beta\delta}
+ \Gamma^\lambda{}_{\beta\mu}\Gamma^\mu{}_{\gamma\delta}
- \Gamma^\lambda{}_{\gamma\mu}\Gamma^\mu{}_{\beta\delta}
\right).
\end{align}
\end{proposition}
\begin{proof}
By definition, using the notation $\nabla_\mu\coloneqq \nabla_{\partial_\mu}$, we have
\begin{align}
(R(g))_{\ \mu\nu}
= \nabla_{\partial_\mu}\nabla_{\partial_\nu}
  -\nabla_{\partial_\nu}\nabla_{\partial_\mu}
  -\nabla_{[\partial_\mu,\partial_\nu]}.
\end{align}
Since $[\partial_\mu,\partial_\nu]=0$ and
$\nabla_{\partial_\mu}\partial_\nu=\Gamma^\rho{}_{\mu\nu}\partial_\rho$, we get
\begin{align}
\nabla_{\partial_\mu}\nabla_{\partial_\nu}\partial_\sigma
&= \partial_\mu\Gamma^\rho{}_{\nu\sigma}\,\partial_\rho
   + \Gamma^\rho{}_{\nu\sigma}\Gamma^\lambda{}_{\mu\rho}\,\partial_\lambda,\\
\nabla_{\partial_\nu}\nabla_{\partial_\mu}\partial_\sigma
&= \partial_\nu\Gamma^\rho{}_{\mu\sigma}\,\partial_\rho
   + \Gamma^\rho{}_{\mu\sigma}\Gamma^\lambda{}_{\nu\rho}\,\partial_\lambda .
\end{align}
Subtracting gives the claim. Lowering the first index with $g$ yields the $(0,4)$ expression above.
\end{proof}

\begin{corollary}[Algebraic symmetries]
For all $X,Y,Z,W\in\Gamma(TM)$, one has
\begin{align}
R(g)(X,Y,Z,W) &= -\,R(g)(Y,X,Z,W) = -\,R(g)(X,Y,W,Z)\\
R(g)(X,Y,Z,W) &= R(g)(Z,W,X,Y)\\
\sum_{\mathrm{cyc}\,X,Y,Z} R(g)(X,Y)Z &= 0 \quad\text{(first Bianchi)},
\end{align}
where $R(g)(X,Y,Z,W)\coloneqq g(R(g)(X,Y)Z,\,W)$. With the component convention above, $(R(g))_{\alpha\beta\gamma\delta}=R(g)(\partial_\beta,\partial_\gamma,\partial_\delta,\partial_\alpha)$. Thus the same symmetries read
\begin{align}
    (R(g))_{\alpha\beta\gamma\delta}&=-(R(g))_{\alpha\gamma\beta\delta}=-(R(g))_{\delta\beta\gamma\alpha}=(R(g))_{\gamma\delta\alpha\beta}\\
    (R(g))_{\alpha[\beta\gamma\delta]}&=0,
\end{align}
with the brackets denoting the anti--symmetrization of the indices.
\end{corollary}

\begin{corollary}[Symmetry of the Ricci tensor]
    Let $(M,g)$ be a Lorentzian smooth manifold and let $\nabla$ be the Levi--Civita connection of $g$. The Ricci tensor for the Riemann curvature, denoted as $\mathrm{Ric}(g)$, is symmetric, i.e. $\mathrm{Ric}(g)(X,Y)=\mathrm{Ric}(g)(Y,X)$, for all $X,Y\in \Gamma(TM)$. Equivalently, $\mathrm{Ric}(g)\in\Gamma(\mathrm{Sym^2}T^*M)$
\end{corollary}

\begin{definition}
    Let $(M,g)$ be a Lorentzian smooth manifold and let $\nabla$ be the Levi--Civita connection of $g$. The \textit{Ricci scalar} $\mathrm{Scal}(g)\in C^\infty(M)$ is defined as
    \begin{align}
        \mathrm{Scal}(g)\coloneqq \mathrm{tr}_g(\mathrm{Ric}(g)).
    \end{align}
\end{definition}

\begin{proposition}\label{contbianc}
    Let $(M,g)$ be a Lorentzian smooth manifold and let $\nabla$ be the Levi--Civita connection of $g$. Then, for the Ricci tensor, it holds
    \begin{align}
        \mathrm{div}_g\big(\mathrm{Ric}(g)-\frac12 \,\mathrm{Scal}(g)\,g\big)=0,
    \end{align}
    where $\nabla$ is regarded as a map $\nabla \colon \Gamma(\otimes^{r}TM\otimes^{s}T^*M) \to  \Gamma(\otimes^{r}TM\otimes^{s+1}T^*M)$, and $\mathrm{div}^j_g\colon\Gamma(\otimes^rTM\otimes^sT^*M)\to\Gamma(\otimes^rTM\otimes^{s-1}T^*M)$ is defined via the trace with respect to the $(s+1)$th and $j$th ($r+1\le j\le r+s$) covariant slots of $\nabla T\in\Gamma(\otimes^rTM\otimes^{s+1}T^*M))$, for every $T\in\Gamma(\otimes^rTM\otimes^sT^*M)$.\footnote{Notice that the index is not specified in the identity above because of symmetry of $\mathrm{Ric}(g)$ and $g$.}\\
    In components, it reads
    \begin{align}
        \nabla_\mu\big((\mathrm{Ric}(g))^\mu_{\ \nu}-\frac12\,\mathrm{Scal}(g)\,\delta^{\mu}_{\ \nu}\big)=0.
    \end{align}
\end{proposition}
\begin{proof}
Write $R_{\mu\nu}=\mathrm{Ric}(g)_{\mu\nu}=R^\rho{}_{\rho\mu\nu}$ and $R=\mathrm{Scal}(g)$. In our convention, the second Bianchi identity reads
\begin{align}
\nabla_\lambda R^\rho{}_{\mu\nu\sigma}
+\nabla_\mu R^\rho{}_{\nu\lambda\sigma}
+\nabla_\nu R^\rho{}_{\lambda\mu\sigma}=0.
\end{align}
Contracting $\rho$ with $\lambda$ gives
\begin{align}
\nabla_\rho R^\rho{}_{\mu\nu\sigma}
=\nabla_\mu R_{\nu\sigma}-\nabla_\nu R_{\mu\sigma}.
\end{align}
Relabeling the free indices, this is
\begin{align}\label{eq:bianchi-contracted-step}
\nabla_\rho R^\rho{}_{\sigma\lambda\nu}
=\nabla_\sigma R_{\lambda\nu}-\nabla_\lambda R_{\sigma\nu}.
\end{align}
The algebraic curvature symmetries give $g^{\lambda\nu}R^\rho{}_{\sigma\lambda\nu}=R^\rho{}_\sigma$. Contracting \eqref{eq:bianchi-contracted-step} with $g^{\lambda\nu}$ and using $\nabla g=0$ yields
\begin{align}
\nabla_\rho R^\rho{}_\sigma
=\nabla_\sigma R-\nabla^\lambda R_{\sigma\lambda}.
\end{align}
By symmetry of Ricci, the two divergence terms are equal, so
\begin{align}
\nabla^\mu R_{\mu\nu}=\tfrac12\nabla_\nu R.
\end{align}
Therefore,
\begin{align}
\nabla^\mu\big(R_{\mu\nu}-\tfrac12R\,g_{\mu\nu}\big)=0,
\end{align}
which is the claim.
\end{proof}
\begin{remark}
    The tensor field defined by $G(g)\coloneqq \mathrm{Ric}(g)-\frac12 \,\mathrm{Scal}(g)\,g\in \Gamma(\mathrm{Sym^2}T^*M)$ is called the \textit{Einstein tensor}. Therefore, the identity above reads 
    \begin{align}
        \mathrm{div}_gG(g)=0
    \end{align}
    Notice that, because of metric compatibility, it also holds
    \begin{align}
        \mathrm{div}_g\big(G(g)+\Lambda g\big)=0,
    \end{align}
    where $\Lambda\in\mathbb R$. The constant $\Lambda$ is called the \textit{cosmological constant}.\footnote{As far as now, this quantity seems to be responsible for the so called \textit{dark energy} in our universe.}
\end{remark}

\begin{mdframed}
\begin{definition}[Axioms of geometry]
\leavevmode\newline
\vspace{-0.55cm}
\begin{enumerate}\setlength{\itemsep}{2pt}
    \item \textbf{Space--time.} Space--time is a connected four--dimensional Lorentzian manifold $(M,g)$; an orientation $\mathfrak o$ and a time–orientation $\mathfrak t$ are fixed.

    \item \textbf{Affine connection.} The space--time manifold $(M,g)$ carries an affine connection $\nabla$, which is the Levi--Civita connection of $g$.

    \item \textbf{Causality and worldlines.} Cones and, for every $x\in M$, admissible causal curves are determined by $g$ and $\mathfrak t$:
    \begin{align}
    \mathcal J_x&=\{v\in T_xM\setminus\{0\}:\ g_x(v,v)=0\},\\
    \mathcal T_x&=\{v\in T_xM:\ g_x(v,v)<0\}=\mathcal T_x^+ \sqcup \mathcal T_x^-.
    \end{align}
    Worldlines are smooth future--directed time--like curves.

    \item \textbf{General covariance with no absolute structure.}
    The equations of the theory are local tensor identities on open sets $U\subseteq M$ (relations in $\mathfrak T(U)$); no additional non--dynamical structure on $M$ is admitted, apart from a fixed orientation $\mathfrak o$ and time–orientation $\mathfrak t$ (discrete; used only for $\mathrm{vol}_g$, $\varepsilon$, $\star$, and future/past orientation).
\end{enumerate}
\end{definition}
\end{mdframed}

\begin{remark}
Notice that axiom $2$ may seem redundant, since every Lorentzian metric determines a unique Levi--Civita connection. However, the axiom specifies that the connection entering the theory is precisely this connection, rather than an independent field. Thus, the connection is entirely determined by $g$, and the gravitational field equations can be written in terms of $g$ alone. In other words, this restricts the space of gravitational fields to the space of Lorentzian metrics.\\
Another formulation treats the connection as an independent field in the variational principle. In the torsion--free metric--affine Palatini formulation of the Einstein--Hilbert action, metric compatibility follows from the connection equation (when matter does not couple to the independent connection). Related first--order formulations will be treated in the next chapters.
\end{remark}

\section{Elements of variational calculus}
\begin{definition}[Fréchet field space of mapping type]\label{def:field-space-mapping}
Let $K$ be a smooth manifold and let $\pi:E\to K$ be a real smooth finite--rank vector bundle.
The space of smooth sections
\begin{align}
\mathcal F \coloneqq \Gamma(E)\subseteq C^\infty(K,E)
\end{align}
endowed with the standard $C^\infty$--Fréchet topology is called a \emph{Fréchet field space of mapping type} or just a \textit{space of fields}.
\end{definition}

\begin{remark}\label{remopfrech}
Any open subset $\mathcal U \subset \Gamma(E)$ (with the $C^\infty$--Fréchet topology) inherits the structure of a Fréchet manifold modeled on $\Gamma(E)$. In particular, a restriction that is open in this topology still gives a Fréchet manifold. Pointwise open conditions, such as non--vanishing or a fixed metric signature, define open subsets when $K$ is compact; on a noncompact base they need not be open in the compact--open $C^\infty$ topology.
\end{remark}

\begin{proposition}[Basic properties]\label{prop:field-space-basic}
Let $\mathcal F$ be a space of fields as in \cref{def:field-space-mapping}, so that $\mathcal F=\Gamma(E)$ for a finite--rank vector bundle $E\to K$. Then, the following hold:
\begin{enumerate}\setlength{\itemsep}{2pt}
    \item $\mathcal F$ is a smooth Fréchet manifold modeled on the Fréchet space $\Gamma(E)$;
    \item for every $\phi\in\mathcal F$ there is a canonical identification
    \begin{align}
        T_\phi \mathcal F \cong \Gamma(E);
    \end{align}
    \item every tangent vector $W\in T_\phi\mathcal F$ can be given as the derivative of some smooth curve at $0$
    \begin{align}
        \Phi:(-\varepsilon,\varepsilon)\to\mathcal F,\qquad \Phi_0=\phi,\quad \dot\Phi_0=W;
    \end{align}
    \item consequently, for any $F\in C^\infty(\mathcal F,\mathbb R)$ and $W\in T_\phi\mathcal F$,
    \begin{align}
        (\delta F)_\phi[W]
        = \frac{d}{ds}\Big|_{s=0} F(\Phi_s),
    \end{align}
    where $\Phi$ is any smooth curve in $\mathcal F$ with $\Phi_0=\phi$ and $\dot\Phi_0=W$, and $\delta$ denotes the de Rham differential on $\mathcal F$.\footnote{This is the usual name for the exterior derivative in the infinite--dimensional case.}
\end{enumerate}
\end{proposition}
\begin{proof}
Since $K$ is a smooth manifold and $E$ has finite rank, we can choose a countable, locally finite trivializing cover $\{U_j\}_{j\in\mathbb N}$ of $E$ and an exhaustion by compacts $K_1 \Subset K_2 \Subset \cdots$ with $K=\bigcup_j K_j$.
Fix a connection on $E$ and a bundle metric, inducing norms on all tensor powers; for $s\in\Gamma(E)$ define seminorms
\begin{align}
    p_{j,m}(s) \coloneqq \sup_{x\in K_j}\big|\nabla^m s(x)\big|.
\end{align}
These seminorms generate the standard $C^\infty$ topology on $\Gamma(E)$, which is metrizable, complete and locally convex; hence $\Gamma(E)$ is a Fréchet space. This proves (1).\\
For (2), let $\phi\in\Gamma(E)$. A tangent vector at $\phi$ is, by definition, the derivative at $s=0$ of a smooth curve of sections
\begin{align}
    s \longmapsto \Phi_s \in \Gamma(E), \qquad \Phi_0 = \phi.
\end{align}
For each $x\in K$, the curve $s\mapsto \Phi_s(x)$ lies in the fiber $E_x$, so its derivative at $0$ is a vertical vector canonically identified with an element of $E_x$. Thus, we get a smooth map
\begin{align}
    x \longmapsto W(x)\in E_x,
\end{align}
i.e. a section $W\in\Gamma(E)$. Conversely, given $W\in\Gamma(E)$, we can define
\begin{align}
    \Phi_s(x) \coloneqq \phi(x) + s\,W(x)
\end{align}
for $|s|$ small; this gives a smooth curve of sections with $\Phi_0=\phi$ and $\dot\Phi_0=W$. Hence,
\begin{align}
    T_\phi \mathcal F \cong \Gamma(E),
\end{align}
which proves (2), and the explicit construction of $\Phi_s$ proves (3).\\
For (4), let $F\in C^\infty(\mathcal F,\mathbb R)$ and $W\in T_\phi\mathcal F$, and pick any smooth curve $\Phi_s$ with $\Phi_0=\phi$ and $\dot\Phi_0=W$. By definition of the de Rham differential $\delta$ on the Fréchet manifold $\mathcal F$,
\begin{align}
    (\delta F)_\phi[W]
    = \frac{d}{ds}\Big|_{s=0} F(\Phi_s),
\end{align}
as claimed.
\end{proof}

\begin{remark}[Variational meaning of the de Rham operator]\label{rem:variational-delta}
The operator $\delta$ appearing in \cref{prop:field-space-basic} is exactly the one we use in the \textit{variational calculus} on the space of fields.
In other words, when we write expressions such as $\delta S(\phi)[W]=0$, the symbol $\delta$ is \emph{nothing else} than the de Rham differential on the Fréchet manifold of fields $\mathcal F=\Gamma(E)$, evaluated on a tangent direction $W\in T_\phi\mathcal F$.\\
The whole construction above is precisely what makes the “physicists’ variation”
\begin{align}
\delta F[\text{field}] = \frac{d}{ds}\Big|_{s=0} F(\text{field} + s\cdot\text{variation})
\end{align}
rigorous, and in this context a “variation” is just a tangent vector to the field space, where the formula in (4) of \cref{prop:field-space-basic} says that the de Rham differential computes it.
\end{remark}

\begin{remark}\label{rem:examples-mapping-type}
An important example of \cref{def:field-space-mapping} is obtained by taking
$E' \coloneqq \Lambda^k T^*K \otimes E$, where $E\to K$ is a vector bundle, so that
\begin{align}
    \mathcal F = \Gamma(E') = \Omega^k(K;E)
\end{align}
is the space of $k$--forms on $K$ with values in $E$.
\end{remark}

\begin{definition}[Space of smooth curves]\label{ass:curves-space}
Let $I=[a,b]\subset\mathbb R$ and set
\begin{align}
\mathcal F \coloneqq C^\infty(I,M).
\end{align}
We refer to this space as the \textit{space of smooth curves}.
\end{definition}

\begin{remark}
    The mapping space
    \begin{align}
    \mathcal F=C^\infty(I,M)
    \end{align}
    is a smooth Fréchet manifold. This is not directly of the form $\Gamma(E)$ with $E$ a \emph{vector} bundle, but it can be realized, as a Fréchet manifold: since $I$ is compact, we have the natural identification
\begin{align}
T\,C^\infty(I,M)\cong C^\infty(I,TM),
\end{align}
and, taking the fiber at $\gamma$, we obtain
\begin{align}
T_\gamma \mathcal F \cong C^\infty(I,TM)_\gamma \cong \Gamma(\gamma^*TM).
\end{align}
One way to see this is to embed $M$ properly and smoothly into some $\mathbb R^N$ and regard $C^\infty(I,M)$ as a closed Fréchet submanifold of $C^\infty(I,\mathbb R^N)=\Gamma(I\times \mathbb R^N)$.\footnote{By the proper version of Whitney's embedding theorem, any finite--dimensional smooth manifold $M$ admits a proper smooth embedding $i:M\hookrightarrow\mathbb R^N$ for some natural $N$. Hence,
    \begin{align}
    C^\infty(I,M)=\big\{\gamma\in C^\infty(I,\mathbb R^N)\ \big|\ \gamma(\lambda)\in i(M)\ \forall\,\lambda\in I\big\},
    \end{align}
    so $C^\infty(I,M)$ is a closed Fréchet submanifold of the section space $C^\infty(I,\mathbb R^N)=\Gamma(I\times\mathbb R^N)$.} In this sense, the tautological form of \cref{def:tautological} fits into the mapping--space framework introduced in \cref{def:field-space-mapping}.
\end{remark}

\begin{definition}[Evaluation map]\label{def:evaluation-map-final}
Let $\mathcal F = C^\infty(I,M)$. The \emph{evaluation map} is
\begin{align}
\mathrm{ev}:\ \mathcal F\times I \longrightarrow M,\qquad
\mathrm{ev}(\gamma,\lambda)\coloneqq \gamma(\lambda).
\end{align}
\end{definition}

\begin{definition}[Tautological $1$--form on $\mathcal F$]\label{def:tautological}
Let $\mathcal F=C^\infty(I,M)$.
Consider the bundle
\begin{align}
\mathcal E \coloneqq \bigsqcup_{\gamma\in\mathcal F} \Gamma(\gamma^*TM) \longrightarrow \mathcal F,
\end{align}
whose fiber over $\gamma$ is the space of vector fields along $\gamma$.
The \emph{tautological one--form} on $\mathcal F$ is the $\mathcal E$--valued $1$--form
\begin{align}
\delta\gamma \in \Omega^1(\mathcal F;\mathcal E)
\end{align}
defined by
\begin{align}
(\delta\gamma)_\gamma[W] \coloneqq V_\gamma \in \Gamma(\gamma^*TM),
\end{align}
where $W\in T_\gamma\mathcal F$ and $V_\gamma$ is the unique vector field along $\gamma$ corresponding to $W$ under the canonical identification
\begin{align}
T_\gamma\mathcal F \cong \Gamma(\gamma^*TM).
\end{align}
\end{definition}

\begin{definition}[Pullback of a tensor field to $\mathcal F\times I$]\label{def:pull--back-tensor}
Let $T\in\Gamma(\otimes^r TM\otimes^s T^*M)$ and let $\mathcal F = C^\infty(I,M)$.
The \emph{pull--back} of $T$ to $\mathcal F\times I$ is
\begin{align}
\mathrm{ev}^*T
\ \in\ \Gamma \big(\mathrm{ev}^*(\otimes^r TM\otimes^s T^*M)\big),
\end{align}
where $\mathrm{ev}:\mathcal F\times I\to M$ is the evaluation map $(\gamma,\lambda)\mapsto \gamma(\lambda)$.
\end{definition}

\begin{remark}
Since $M$ carries an affine connection $\nabla$ (typically the Levi--Civita connection of a metric $g$), it pulls back along $\mathrm{ev}$ to a connection on $\mathrm{ev}^*(\cdots)$, which we denote by
\begin{align}
\mathrm{ev}^*\nabla.
\end{align}
\end{remark}

\begin{definition}[Vertical variation of a pulled--back tensor]\label{def:vertical-variation}
Let $\nabla$ be an affine connection on $M$, let $I=[a,b]\subset\mathbb R$ be compact, and let $\mathcal F=C^\infty(I,M)$ with evaluation map
\begin{align}
\mathrm{ev}:\ \mathcal F\times I \longrightarrow M,\qquad (\gamma,\lambda)\longmapsto \gamma(\lambda).
\end{align}
For a tensor field $T\in\Gamma(\otimes^r TM\otimes^s T^*M)$, consider its pull--back
\begin{align}
\mathrm{ev}^*T \in \Gamma \big(\mathrm{ev}^*(\otimes^r TM\otimes^s T^*M)\big).
\end{align}
The \emph{vertical variation of $T$} is the $1$--form on $\mathcal F$ with values in tensor fields along the curve, defined by
\begin{align}
\delta T \coloneqq \iota_{\delta\gamma}\big(\mathrm{ev}^*(\nabla T)\big).
\end{align}
Here $\delta$ on bundle--valued fields denotes the covariant variation induced by $\mathrm{ev}^*\nabla$, also denoted $\delta_\nabla$; it agrees with the ordinary de Rham variation on scalar functions. More generally, for any section $A$ of a tensor bundle pulled back by $\mathrm{ev}$, including the velocity field, set
\begin{align}
(\delta_\nabla A)_{(\gamma,\lambda)}[W]
=(\mathrm{ev}^*\nabla)_{(W,0)}A.
\end{align}
This distinction matters for higher variations: $\delta_\nabla^2$ is curvature acting on the tensor indices, and need not vanish. Ordinary scalar de Rham variation still satisfies $\delta^2=0$.
\end{definition}

\begin{proposition}[Explicit form of the vertical variation]\label{prop:vertical-variation-explicit}
Let $T$ and $\delta T$ be as in \cref{def:vertical-variation}.
Then, for every $\gamma\in\mathcal F$ and $W\in T_\gamma\mathcal F$, with $V_\gamma=(\delta\gamma)_\gamma[W]\in\Gamma(\gamma^*TM)$, one has
\begin{align}
(\delta T)_\gamma[W] = (\nabla_{V_\gamma} T)^\gamma,
\end{align}
i.e.\ the variation is the covariant derivative of $T$ in the direction of the variational vector field, restricted along $\gamma$.
\end{proposition}

\begin{proof}
Since $\nabla T$ is a $(r,s+1)$--tensor field on $M$, its pull--back $\mathrm{ev}^*(\nabla T)$ is a section of the pulled--back $(r,s+1)$--tensor bundle over $\mathcal F\times I$, with tensor slots in $T_{\gamma(\lambda)}M$ and its dual.
The form $\delta\gamma$ is a $1$--form on $\mathcal F$ with values in vector fields along $\gamma$.
Contracting $\mathrm{ev}^*(\nabla T)$ with $\delta\gamma$ in the last slot gives a $1$--form on $\mathcal F$ with values in tensors along $\gamma$.
Evaluating at $\gamma$ and $W$ amounts to inserting $V_\gamma$ in that last slot:
\begin{align}
(\delta T)_\gamma[W]
= \big(\mathrm{ev}^*(\nabla T)\big)_{(\gamma,\cdot)}[V_\gamma]
= (\nabla_{V_\gamma} T)^\gamma,
\end{align}
because restricting $\mathrm{ev}^*(\nabla T)$ to $\{\gamma\}\times I$ is the same as restricting $\nabla T$ along $\gamma$.
\end{proof}

\begin{lemma}[Chain rule along the interval]\label{lem:chain-rule-I}
Let $d\lambda$ be the standard $1$--form on $I$. Then
\begin{align}
\delta(d\lambda)=0,
\end{align}
and $\delta$ acts only on the $\mathcal F$--variable of $\mathcal F\times I$.
In particular, if $\alpha(\gamma,\lambda)$ is a smooth function on $\mathcal F\times I$, then
\begin{align}
\delta \left(\int_I \alpha(\gamma,\lambda)\,d\lambda\right)
= \int_I \delta\big(\alpha(\gamma,\lambda)\big)\,d\lambda.
\end{align}
\end{lemma}
\begin{proof}
By definition, $\delta$ is the de Rham differential on $\mathcal F$, so it does not act on the $I$--factor; hence $\delta(d\lambda)=0$.
For the integral, let
\begin{align}
G(\gamma)\coloneqq \int_I \alpha(\gamma,\lambda)\,d\lambda.
\end{align}
Then $G:\mathcal F\to\mathbb R$ is smooth (since $I$ is compact and $\alpha$ is smooth). Its differential at $\gamma$ in the direction $W$ is
\begin{align}
(\delta G)_\gamma[W]
= \frac{d}{ds}\Big|_{s=0} \int_I \alpha(\Upsilon_s(\cdot),\lambda)\,d\lambda
= \int_I \frac{d}{ds}\Big|_{s=0} \alpha(\Upsilon_s(\cdot),\lambda)\,d\lambda
= \int_I (\delta \alpha)_\gamma[W](\lambda)\,d\lambda,
\end{align}
where we exchanged derivative and integral by smoothness and compactness of $I$.
\end{proof}

\begin{lemma}[Variation of the velocity]\label{lem:delta-velocity}
Let $\gamma\in\mathcal F=C^\infty(I,M)$ and $W\in T_\gamma\mathcal F$, and let $V_\gamma=(\delta\gamma)_\gamma[W]\in\Gamma(\gamma^*TM)$.
Let $\nabla^\gamma$ be the pull--back connection on $\gamma^*TM$ induced by $\nabla$.
If $\nabla$ is torsion--free (e.g. Levi--Civita), then\footnote{The informal intuition behind this identity is to resemble the commutation rule one has Newtonian mechanics, i.e. $\delta\frac{d}{dt}q^i=\frac{d}{dt}\delta q^i$. Since in general relativity one deals with tensorial quantities, one has to preserve covariance. In particular, $\delta(\partial_\lambda\gamma)=\partial_\lambda(\delta\gamma)$ would make no sense, since $\delta\gamma$ is a $1$--form with values in $\Gamma(\gamma^*TM)$.}
\begin{align}
\big(\delta(\partial_\lambda \gamma)\big)_\gamma[W]
= \nabla^\gamma_{\partial_\lambda} V_\gamma.
\end{align}
\end{lemma}
\begin{proof}
Take a smooth curve of curves $\Upsilon:(-\varepsilon,\varepsilon)\to\mathcal F$ with
$\Upsilon_0=\gamma$ and $\dot\Upsilon_0=W$, and consider
\begin{align}
\Gamma:(-\varepsilon,\varepsilon)\times I \to M,\qquad \Gamma(s,\lambda)=\Upsilon_s(\lambda).
\end{align}
Then,
\begin{align}
\partial_\lambda \gamma(\lambda)=\partial_\lambda \Gamma(0,\lambda), \qquad
V_\gamma(\lambda)=\partial_s \Gamma(s,\lambda)\big|_{s=0}.
\end{align}
By definition of the covariant variation $\delta_\nabla$,
\begin{align}
\big(\delta(\partial_\lambda \gamma)\big)_\gamma[W](\lambda)
= \nabla_{\partial_s}\big(\partial_\lambda \Gamma(s,\lambda)\big)\Big|_{s=0}.
\end{align}
If $\nabla$ is torsion--free, then
\begin{align}
\nabla_{\partial_s}(\partial_\lambda \Gamma) = \nabla_{\partial_\lambda}(\partial_s \Gamma),
\end{align}
so evaluating at $s=0$,
\begin{align}
\big(\delta(\partial_\lambda \gamma)\big)_\gamma[W](\lambda)
= \nabla_{\partial_\lambda} V_\gamma(\lambda)
= \nabla^\gamma_{\partial_\lambda} V_\gamma(\lambda).
\end{align}
\end{proof}

\begin{proposition}[Variation of a tensor contracted with the velocity]\label{prop:variation-tensor-with-velocity}
Let $T\in\Gamma(\otimes^r TM \otimes^s T^*M)$ with $s\ge 2$ and let $\gamma\in\mathcal F$.
Fix $W\in T_\gamma\mathcal F$ and let $V_\gamma=(\delta\gamma)_\gamma[W]$. If $\nabla$ is torsion--free,
then the covariant variation of the tensor field along the curve
\begin{align}
\gamma \longmapsto T(\dot\gamma,\dot\gamma)
\end{align}
is
\begin{align}
\big(\delta\big(T(\dot\gamma,\dot\gamma)\big)\big)_\gamma[W]
= (\nabla_{V_\gamma} T)(\dot\gamma,\dot\gamma)
  + T\big(\nabla^\gamma_{\partial_\lambda} V_\gamma,\dot\gamma\big)
  + T\big(\dot\gamma,\nabla^\gamma_{\partial_\lambda} V_\gamma\big).
\end{align}
\end{proposition}
\begin{proof}
Set
\begin{align}
F(\gamma)(\lambda)\coloneqq T_{\gamma(\lambda)}(\dot\gamma(\lambda),\dot\gamma(\lambda)),
\end{align}
where two specified covariant slots are contracted. For $s\ge2$, this is a field of type $(r,s-2)$ along $\gamma$. In particular, for $(r,s)=(0,2)$ it is a smooth function on $I$, not a real number until evaluated at a parameter or integrated.
We want to compute $(\delta F)_\gamma[W]$ for $\gamma\in\mathcal F$ and $W\in T_\gamma\mathcal F$.
Let $V_\gamma=(\delta\gamma)_\gamma[W]\in\Gamma(\gamma^*TM)$ be the variational vector field along $\gamma$.\\
By the Leibniz rule and compatibility with contractions for the covariant variation (which agrees with de Rham variation for scalar contractions) we have
\begin{align}
\delta\big(T(\dot\gamma,\dot\gamma)\big)
= (\delta T)(\dot\gamma,\dot\gamma)
  + T\big(\delta(\partial_\lambda\gamma),\dot\gamma\big)
  + T\big(\dot\gamma,\delta(\partial_\lambda\gamma)\big).
\end{align}
Now evaluate this $1$--form at $\gamma$ on the direction $W$. For the first term we use the explicit formula for the vertical variation in \cref{prop:vertical-variation-explicit}:
\begin{align}
\big((\delta T)(\dot\gamma,\dot\gamma)\big)_\gamma[W]
= (\delta T)_\gamma[W](\dot\gamma,\dot\gamma)
= \big((\nabla_{V_\gamma} T)^\gamma\big)(\dot\gamma,\dot\gamma)
= (\nabla_{V_\gamma} T)(\dot\gamma,\dot\gamma).
\end{align}
For the second and third terms we use the variation of the velocity, \cref{lem:delta-velocity}, which gives
\begin{align}
\big(\delta(\partial_\lambda\gamma)\big)_\gamma[W]
= \nabla^\gamma_{\partial_\lambda} V_\gamma.
\end{align}
Therefore,
\begin{align}
\big(T(\delta(\partial_\lambda\gamma),\dot\gamma)\big)_\gamma[W]
&= T\big(\nabla^\gamma_{\partial_\lambda} V_\gamma,\dot\gamma\big),\\
\big(T(\dot\gamma,\delta(\partial_\lambda\gamma))\big)_\gamma[W]
&= T\big(\dot\gamma,\nabla^\gamma_{\partial_\lambda} V_\gamma\big).
\end{align}
Putting everything together we obtain
\begin{align}
\big(\delta\big(T(\dot\gamma,\dot\gamma)\big)\big)_\gamma[W]
= (\nabla_{V_\gamma} T)(\dot\gamma,\dot\gamma)
  + T\big(\nabla^\gamma_{\partial_\lambda} V_\gamma,\dot\gamma\big)
  + T\big(\dot\gamma,\nabla^\gamma_{\partial_\lambda} V_\gamma\big),
\end{align}
which is the claimed formula.
\end{proof}

\section{Geodesics}

\begin{mdframed}
    \begin{definition}[Axiom for the free particle]\label{axfreepar}
\leavevmode\newline
\vspace{-0.55cm}
\begin{enumerate}\setlength{\itemsep}{2pt}\setcounter{enumi}{4}
    \item \textbf{Equations of motion.} Let $I=[a,b]\subseteq \mathbb R$ be a compact interval, and let $p,q\in M$. The equations of motion for a free--falling particle are obtained by the \emph{variational principle} with fixed endpoints for the action functional
\begin{align}
    S_{free}[\gamma]=-\tfrac12\int_I (\gamma^*g)(\partial_\lambda,\partial_\lambda)\,d\lambda,
\end{align}
over the space of smooth curves $\gamma\colon I\to M$ with $\gamma(a)=p$ and $\gamma(b)=q$, where $(\gamma^*g)(\partial_\lambda,\partial_\lambda)=g(\dot\gamma,\dot\gamma)\in C^\infty(I)$ for every $\lambda\in I$.
\end{enumerate}
\end{definition}
\end{mdframed}

\begin{remark}
    Notice that, by abuse of notation, in $g(\dot\gamma,\dot\gamma)$ the metric tensor $g$ is evaluated along $\gamma$.
\end{remark}

\begin{remark}
    We will prove later, in \cref{lemincond}, that the quantity $g(\dot\gamma,\dot\gamma)$ is constant along any solution to the variational principle above. Hence, its sign can be read at any parameter value, e.g. at $a$. According to the sign of $g(\dot\gamma(a),\dot\gamma(a))$, we have:\footnote{Provided there exists a curve with the prescribed causal character between $p$ and $q$. In particular, if $(M,g)$ is \emph{globally hyperbolic}, then for any two events $p\le q$ (i.e. $q\in J^+(p)$, timelike or null separated with respect to the chosen time--orientation) there exists a future--directed causal solution joining them; therefore a causal solution exists in that case.}
    \begin{itemize}
        \item $g(\dot\gamma(a),\dot\gamma(a))<0$: the trajectory is \emph{timelike} (causal);
        \item $g(\dot\gamma(a),\dot\gamma(a))=0$: the trajectory is \emph{null} (causal);
        \item $g(\dot\gamma(a),\dot\gamma(a))>0$: the trajectory is \emph{spacelike} (non--causal).
    \end{itemize}
    Therefore, \cref{axfreepar} is not meant to describe only causal trajectories: it yields \emph{all} solutions of the above variational problem, among which physical causal trajectories are selected by their nonzero causal tangent and future orientation; endpoint conditions alone need not select a unique geodesic or causal character.
\end{remark}

\begin{remark}
    Notice that, since $d\lambda$ is a volume form on $I\subseteq R$, we can make formally sense of the  integration of the form $(\gamma^*g)(\partial_\lambda,\partial_\lambda)\,d\lambda\in\Omega^1(I)$. Moreover, note that we can also formulate it differently by splitting into the two cases of the timelike and the null free particle. In the timelike case, we can actually consider of a line--volume element $ds_\gamma\in\Omega^1(I)$ defined by 
    \begin{align}
        ds_\gamma\coloneqq \sqrt{-(\gamma^*g)(\partial_\lambda,\partial_\lambda)}\,d\lambda\in\Omega^1(I),
    \end{align}
    and, consequently, define the action over the space of timelike curves as
    \begin{align}
        S_{free}[\gamma]=-\int_I ds_{\gamma},
    \end{align}
    where $g(\dot\gamma,\dot\gamma)<0$. The variational principle is therefore well--defined since $ds_\gamma\neq 0$. In this case, we introduce a gauge symmetry on $I$ given by the reparametrization $\lambda\mapsto f(\lambda)$, with $f\colon I\to I$ an orientation--preserving smooth diffeomorphism. In fact, we can notice that the line element is invariant under this transformation. This symmetry is encoded by Noether's second theorem and corresponds to constraints in the phase--space.
\end{remark}

\begin{proposition}
Let $S_{free}$ be the action functional defined in \cref{axfreepar}, and let $p,q\in M$.
A smooth curve $\gamma:I\to M$ from $p$ to $q$ is a solution for the variational problem $\delta S_{free}=0$ with fixed endpoints, if and only if $\gamma$ is a geodesic from $p$ to $q$.
\end{proposition}
\begin{proof}
Let $W\in T_\gamma C^\infty(I,M)$ with $W(a)=W(b)=0$. Via the canonical identification $T_\gamma C^\infty(I,M)\cong \Gamma(\gamma^*TM)$, write $V_\gamma\in\Gamma(\gamma^*TM)$ for the vector field along $\gamma$ corresponding to $W$, in particular $\delta\gamma[W]=V_\gamma$ (by \cref{def:tautological}) and $V_{\gamma(a)}=V_{\gamma(b)}=0$.
 Then, omitting to write $\lambda$ in $\gamma$ and $\dot\gamma$, we have\footnote{Here $\delta\big(g(\dot\gamma,\dot\gamma)\big)[W]$ is a smooth function on $I$. Since $I\subseteq\mathbb R$ is orientable, $d\lambda$ gives a global nowhere--vanishing $1$-form on $I$, hence $\delta\big(g(\dot\gamma,\dot\gamma)\big)[W] \,d\lambda\in\Omega^1(I)$ and its integral is well-defined.}
    \begin{align}
        (\delta S_{free})_\gamma[W]
        &=-\tfrac12\int_I \delta\big((\gamma^*g)(\partial_\lambda,\partial_\lambda)\big)[W]\,d\lambda= -\tfrac12\int_I \delta\big(g(\dot\gamma,\dot\gamma)\big)[W]\,d\lambda\\
        &=-\tfrac12\int_I \big((\delta g)_\gamma(\dot\gamma,\dot\gamma)+2g(\delta\dot\gamma,\dot\gamma)\big)[W]\,d\lambda\\
        &=-\tfrac12\int_I (\nabla_{V_\gamma}g)(\dot\gamma,\dot\gamma)\,d\lambda-\int_I \big(g(\delta\dot\gamma,\dot\gamma)\big)[W]\,d\lambda\\
        &= -\int_I \big(g(\delta\dot\gamma,\dot\gamma)\big)[W]\,d\lambda=-\int_I \big(g(\delta(\partial_\lambda\gamma),\dot\gamma)\big)[W]\,d\lambda\\
        &=-\int_I \big(g(\nabla^\gamma_{\partial_\lambda}\delta\gamma,\dot\gamma)\big)[W]\,d\lambda\\
        &=-\int_I \big(\nabla^\gamma_{\partial_\lambda} g(\delta\gamma,\dot\gamma)-g(\delta\gamma,\nabla^\gamma_{\partial_\lambda}\dot\gamma)\big)[W]\,d\lambda\\
        &=-\int_I \big(\partial_\lambda g(\delta\gamma,\dot\gamma)-g(\delta\gamma,\nabla^\gamma_{\partial_\lambda}\dot\gamma)\big)[W]\,d\lambda\\
        &=-\int_I \big(\partial_\lambda g(V_\gamma,\dot\gamma)-g(V_\gamma,\nabla^\gamma_{\partial_\lambda}\dot\gamma)\big)\,d\lambda\\
        &= -\Big[g(V_\gamma,\dot\gamma)\Big]_{a}^{b}+\int_I g(V_\gamma,\nabla^\gamma_{\partial_\lambda}\dot\gamma)\,d\lambda\\
        &=\int_I g(V_\gamma,\nabla^\gamma_{\partial_\lambda}\dot\gamma)\,d\lambda\\
        &=\int_I g\big(V_\gamma,\nabla_{\dot\gamma}\dot\gamma\big)\,d\lambda.
    \end{align}
    Notice that we used \cref{prop:vertical-variation-explicit}, \cref{lem:delta-velocity}, \cref{prop:variation-tensor-with-velocity}, the fact that $\nabla^\gamma_{\partial_\lambda}=\partial_\lambda$ on $C^\infty(I)$, integration by parts, and, in the last equality, \cref{defconnpullback}. Since $W\in T_\gamma C^\infty(I,M)$ is arbitrary (and hence also $V_\gamma\in\Gamma(\gamma^*TM)$), by the fundamental lemma of the calculus of variations, we conclude that
    \begin{align}
        \delta S_{free}=0 \quad \text{if and only if} \quad \nabla_{\dot\gamma}\dot\gamma=0.
    \end{align}
    This is the geodesic equation.
\end{proof}

\begin{remark}[Variational bicomplex]\label{varbicomp}
Informally, the \textit{variational bicomplex}\footnote{Its jet--bundle formulation will be defined in later sections.} separates horizontal and vertical differentials. On the product $\mathcal F\times I$, its mapping--space analogue is the ordinary scalar de Rham bicomplex. Let $\Omega^{p,q}(\mathcal F\times I)$ denote scalar forms with $p$ field--space slots and $q$ interval slots. Then the total differential splits as $d=d_V+d_H$, with $d_V=\delta$ and
\begin{align}
d_H=d\lambda\wedge\partial_\lambda,\qquad
d_V^2=d_H^2=0,\qquad d_Vd_H+d_Hd_V=0,
\end{align}
using total exterior-degree signs. These identities follow from the product de Rham differential. They do not extend unchanged if ordinary differentiation is replaced by a connection on tensor--valued fields. For a two--parameter variation $\Gamma(s,\lambda)$ and a vector field $A$ along it,
\begin{align}
[\nabla^\Gamma_{\partial_s},\nabla^\Gamma_{\partial_\lambda}]A
=R(\partial_s\Gamma,\partial_\lambda\Gamma)A.
\end{align}
Thus covariant variations generally have a curvature obstruction. The identity
\begin{align}
\delta_\nabla(\partial_\lambda\gamma)
=\nabla^\gamma_{\partial_\lambda}(\delta\gamma)
\end{align}
from \cref{lem:delta-velocity} is the torsion--free identity for the velocity; it does not assert flatness of the pulled--back connection.

The scalar bicomplex is represented by the diagram:
This particular variational bicomplex is given by the diagram:

\centering
\begin{tikzpicture}[>=Latex, every node/.style={font=\footnotesize}]
  \matrix (m) [matrix of math nodes,
               row sep=1.0em, column sep=1.2em,
               nodes={text height=1.7ex, text depth=.25ex}] {
    \vdots & \vdots \\
    \Omega^{2,0} (\mathcal F\times I) &
    \Omega^{2,1} (\mathcal F\times I) \\
    \Omega^{1,0} (\mathcal F\times I) &
    \Omega^{1,1} (\mathcal F\times I) \\
    \Omega^{0,0} (\mathcal F\times I) &
    \Omega^{0,1} (\mathcal F\times I) \\
  };

  \draw[->] (m-2-1) -- node[above] {$d_H$} (m-2-2);
  \draw[->] (m-3-1) -- node[above] {$d_H$} (m-3-2);
  \draw[->] (m-4-1) -- node[above] {$d_H$} (m-4-2);

  \draw[->] (m-4-1) -- node[right] {$\delta$} (m-3-1);
  \draw[->] (m-3-1) -- node[right] {$\delta$} (m-2-1);
  \draw[->] (m-2-1) -- node[right] {$\delta$} (m-1-1);

  \draw[->] (m-4-2) -- node[right] {$\delta$} (m-3-2);
  \draw[->] (m-3-2) -- node[right] {$\delta$} (m-2-2);
  \draw[->] (m-2-2) -- node[right] {$\delta$} (m-1-2);
\end{tikzpicture}
\end{remark}

\begin{proposition}\label{lemincond}
    The quantity $g(\dot\gamma,\dot\gamma)$ is constant \emph{on--shell}, i.e. along the solutions of the variational problem for $S_{free}$.
\end{proposition}
\begin{proof}
Let $\gamma : I \to M$ be a geodesic, i.e. a critical curve for $S_{free}$. Therefore, it satisfies
\begin{align}
    \nabla_{\dot\gamma}\dot\gamma=0.
\end{align}
Using metric compatibility $\nabla g=0$ of the Levi--Civita connection, we compute
\begin{align}
\partial_\lambda g(\dot\gamma,\dot\gamma)&=\nabla^\gamma_{\partial_\lambda}g(\dot\gamma,\dot\gamma)=\nabla_{\dot\gamma}g(\dot\gamma,\dot\gamma)\\
&=(\nabla_{\dot\gamma}g)(\dot\gamma,\dot\gamma)+2\,g(\nabla_{\dot\gamma}\dot\gamma,\dot\gamma)\\
&=0+2\,g(0,\dot\gamma)=0,
\end{align}
where we used  the fact that $\nabla^\gamma_{\partial_\lambda}=\partial_\lambda$ on $C^\infty(I)$ and \cref{remnablatenscurv}.
Hence, $g(\dot\gamma,\dot\gamma)$ is constant along $\gamma$, i.e.
\begin{align}
    g(\dot\gamma(\lambda),\dot\gamma(\lambda)) \equiv g(\dot\gamma(a),\dot\gamma(a)) \qquad \forall\,\lambda\in I.
\end{align}
In particular, if $g(\dot\gamma(a),\dot\gamma(a))<0$ (resp.\ $=0$, resp.\ $>0$), then 
$g(\dot\gamma(\lambda),\dot\gamma(\lambda))<0$ (resp.\ $=0$, resp.\ $>0$) for all $\lambda\in I$.
\end{proof}

\begin{remark}[Symmetries of the geodesic action]\label{rem:geodesic-isometries}
Fix a Lorentzian manifold $(M,g)$ and consider the geodesic functional
\begin{align}
    S_{free}[\gamma]
    \coloneqq -\tfrac12 \int_I g(\dot\gamma,\dot\gamma)\,d\lambda.
\end{align}
Let $\varphi\in \mathrm{Diff}(M)$ act on curves by pull--back/composition, i.e.
\begin{align}
    (\varphi^{-1})^*\gamma \coloneqq \varphi\circ \gamma.
\end{align}
Then, this is equivalent to act via pull--back on the metric $g$ by
\begin{align}
    S_{free}[(\varphi^{-1})^*\gamma]
    = -\tfrac12 \int_I g(\varphi_*\dot{\gamma},\varphi_*\dot{\gamma})\,d\lambda
    = -\tfrac12 \int_I (\varphi^*g)(\dot\gamma,\dot\gamma)\,d\lambda.
\end{align}
Hence,
\begin{align}
    S_{free}[(\varphi^{-1})^*\gamma] = S_{free}[\gamma]\quad\text{for all }\gamma
    \quad\Longleftrightarrow\quad \varphi^*g = g.
\end{align}
In other words, the symmetry group of $S_{free}$ is precisely the isometry group
\begin{align}
    \mathrm{Iso}(M,g) \coloneqq \{\varphi\in \mathrm{Diff}(M) \mid \varphi^*g=g\}.
\end{align}
Suppose we have a flow of diffeomorphisms $\Phi^X_t$ generated by $X$. Then, it satisfies $(\Phi^X_t)^*g=g$ (they act as isometries) if and only if
\begin{align}
    \mathcal L_X g = 0,
\end{align}
and $X$ is called a \textit{ Killing vector field}.
\end{remark}

\begin{remark}[Isometries and gauge symmetries]\label{rem:isometries-not-gauge}
For the test--particle action with a fixed background metric, isometries act as global symmetries on the particle trajectories. This differs from the gauge action of diffeomorphisms on the dynamical fields of general relativity. Nevertheless, an isometry is not excluded from that gauge group merely because it fixes $g$: it belongs to the stabilizer of $g$ under the diffeomorphism action. If matter fields are present, it need not fix those fields as well. Which diffeomorphisms are gauge, and which carry physical charges, also depends on boundary or asymptotic conditions. The finite-dimensionality of the isometry group alone does not decide this distinction.
\end{remark}

\begin{remark}[Geodesics and the symmetric part of the connection]
The geodesic equation in local coordinates reads
\begin{align}
    \ddot x^\rho + \Gamma^\rho_{\mu\nu}(x)\,\dot x^\mu \dot x^\nu = 0.
\end{align}
Since $\dot x^\mu \dot x^\nu$ is symmetric in $(\mu,\nu)$, only the symmetric part
\begin{align}
    \Gamma^\rho_{\{\mu\nu\}} \coloneqq \tfrac12\big(\Gamma^\rho_{\mu\nu} + \Gamma^\rho_{\nu\mu}\big)
\end{align}
contributes to the motion. The antisymmetric part
\begin{align}
    \Gamma^\rho_{[\mu\nu]} \coloneqq \tfrac12\big(\Gamma^\rho_{\mu\nu} - \Gamma^\rho_{\nu\mu}\big)
\end{align}
does not appear in the geodesic equation. Notice, however, that it is precisely the antisymmetric part that defines the torsion tensor,
\begin{align}
    T^\rho{}_{\mu\nu} = 2\,\Gamma^\rho_{[\mu\nu]}.
\end{align}
By contrast, the symmetric part $\Gamma^\rho_{\{\mu\nu\}}$ is \emph{not} a tensor, and can be made to vanish in at a point by a suitable choice of coordinates (that we will call normal coordinates).
\end{remark}

\begin{definition}[Star--shaped neighborhood]\label{def:Star--shaped}
Let $V$ be a real vector space and let $U\subset V$ be an open set.\footnote{With respect to the standard topology.} We say that $U$ is \emph{Star--shaped with respect to the origin} if for every $v\in U$ the whole segment
\begin{align}
    [0,1]\ni t \longmapsto t v \in V
\end{align}
is contained in $U$. Equivalently,
\begin{align}
    v\in U  \Longrightarrow  t v \in U \quad \forall\,t\in[0,1].
\end{align}
More generally, $U$ is \emph{Star--shaped with respect to a point} $p\in U$ if
\begin{align}
    v\in U  \Longrightarrow  (1-t)p + t v \in U \quad \forall\,t\in[0,1].
\end{align}
\end{definition}

\begin{definition}[Exponential map]\label{def:exponential-map}
Let $(M,g)$ be a smooth Lorentzian manifold, let $\nabla$ be its Levi--Civita connection, and let $p\in M$. 
For every $v\in T_pM$ consider the unique geodesic $\gamma_v\colon I_v\to M$ solving
\begin{align}
    \gamma_v(0)=p, \qquad \dot\gamma_v(0)=v,
\end{align}
where $I_v\subset \mathbb R$ is the maximal interval of existence. 
Define
\begin{align}
    \mathcal D_p \coloneqq \{ v\in T_pM \mid 1\in I_v \},
\end{align}
i.e. the set of initial velocities for which the corresponding geodesic is defined at time $t=1$. 
The \emph{exponential map at $p$} is the smooth map
\begin{align}
    \exp_p \colon \mathcal D_p \to M, \qquad
    \exp_p(v)\coloneqq \gamma_v(1).
\end{align}
\end{definition}

\begin{lemma}[Local diffeomorphism property of $\exp_p$]\label{lem:exp-local-diffeo}
Let $(M,g)$ be as above and let $p\in M$. Then there exists an open neighborhood $U\subset T_pM$ of $0$ such that:
\begin{enumerate}\setlength{\itemsep}{2pt}
    \item $U$ is Star--shaped with respect to $0$;
    \item for every $v\in U$, the geodesic $\gamma_v$ with initial data $(p,v)$ is defined on $[0,1]$;
    \item the restriction 
    \begin{align}
        \exp_p\big|_U \colon U \longrightarrow \exp_p(U)\subset M
    \end{align}
    is a diffeomorphism onto the open subset $\exp_p(U)\subset M$.
\end{enumerate}
\end{lemma}

\begin{proof}
By standard ODE theory for smooth vector fields on a smooth manifold, applied to the geodesic equation with initial point $p$, there exists an open neighborhood $W\subset T_pM$ of $0$ such that, for every $v\in W$, the unique geodesic $\gamma_v$ with initial data $(p,v)$ is defined at least on the interval $[-1,1]$, and the map
\begin{align}
    \Phi \colon [-1,1]\times W \longrightarrow M,\qquad
    \Phi(t,v)\coloneqq \gamma_v(t),
\end{align}
is smooth.
In particular, for every $v\in W$ we have $1\in[-1,1]\subset I_v$, so $W\subset \mathcal D_p$ and
\begin{align}
    \exp_p(v)=\Phi(1,v)=\gamma_v(1),\qquad v\in W.
\end{align}
Thus $\exp_p$ is smooth on $W$.\\
We now compute the differential of $\exp_p$ at $0$. For $v\in T_pM$ and $t\in\mathbb R$ small, consider the geodesic $\gamma_v$ and the geodesic $\gamma_{t v}$ with initial data $(p,tv)$. The curve
\begin{align}
    s \longmapsto \gamma_v(ts)
\end{align}
is a geodesic with initial data $(p,tv)$, hence by uniqueness of solutions of the geodesic equation we have
\begin{align}
    \gamma_{t v}(s) = \gamma_v(ts), \qquad \text{for all $s$ with $ts\in I_v$}.
\end{align}
In particular, for $|t|$ sufficiently small,
\begin{align}
    \exp_p(tv)
    = \gamma_{t v}(1)
    = \gamma_v(t).
\end{align}
Let $c_v(t)\coloneqq t v\in T_pM$. Then
\begin{align}
    d(\exp_p)_0(v)
    &= \left.\frac{d}{dt}\right|_{t=0} \exp_p(c_v(t))
     = \left.\frac{d}{dt}\right|_{t=0} \exp_p(tv)
     = \left.\frac{d}{dt}\right|_{t=0} \gamma_v(t)
     = \dot\gamma_v(0)
     = v.
\end{align}
Therefore $d(\exp_p)_0 = \mathrm{id}_{T_pM}$. By the inverse function theorem, there exists an open neighborhood $U_0\subset W$ of $0$ such that
\begin{align}
    \exp_p\big|_{U_0}\colon U_0 \longrightarrow \exp_p(U_0)
\end{align}
is a diffeomorphism onto an open neighborhood $\exp_p(U_0)$ of $p$ in $M$.\\
Since $T_pM$ is a finite--dimensional real vector space, we can fix a norm $\|\cdot\|$ and choose $r>0$ such that the open ball
\begin{align}
    B_r(0)\coloneqq \{v\in T_pM:\ \|v\|<r\}
\end{align}
is contained in $U_0$. Set
\begin{align}
    U\coloneqq B_r(0).
\end{align}
Then $U$ is an open neighborhood of $0$, it is Star--shaped with respect to $0$ (by convexity of balls in a normed vector space), and for every $v\in U$ the geodesic $\gamma_v$ is defined on $[-1,1]\supset[0,1]$ because $U\subset U_0\subset W$. Moreover, the restriction $\exp_p|_U$ is still a diffeomorphism onto the open subset $\exp_p(U)\subset M$. This yields the three properties claimed in the statement.
\end{proof}

\begin{theorem}[Expansion in normal coordinates]\label{thm:normal-coord-expansion}
Let $(M,g)$ be a four--dimensional Lorentzian manifold, and let $p\in M$, $\nabla$ be the Levi--Civita connection for $g$.  Fix an orthonormal basis $(e_\mu)_{\mu=0}^3$ of $T_pM$ so that
\begin{align}
    g_p(e_\mu,e_\nu)=\eta_{\mu\nu}.
\end{align}
By Lemma~\ref{lem:exp-local-diffeo} there exists a Star--shaped neighborhood $U\subset T_pM$ of the origin such that
\begin{align}
    \exp_p\big|_U\colon U\longrightarrow \exp_p(U)
\end{align}
is a diffeomorphism onto an open neighborhood $\exp_p(U)$ of $p$ in $M$.
Writing each $v\in U$ as $v=v^\mu e_\mu$, we define the coordinate chart
\begin{align}
    x^\mu \colon \exp_p(U)\to \mathbb R,\qquad 
    x^\mu\!\big(\exp_p(v)\big)\coloneqq v^\mu.
\end{align}
Then $(x^\mu)$ are pseudo--Riemannian normal coordinates around $p$, and the metric components
\begin{align}
      g_{\mu\nu}(x) \coloneqq g\big((\partial_\mu)_x,(\partial_\nu)_x\big)
\end{align}
in this chart satisfy:
\begin{enumerate}\setlength{\itemsep}{3pt}
    \item $x(p)=0$ and
    \begin{align}
        g_{\mu\nu}(0)=\eta_{\mu\nu};
    \end{align}
    \item the first derivatives vanish at the origin:
    \begin{align}
        \partial_\rho g_{\mu\nu}(0)=0;
    \end{align}
    \item the Christoffel symbols vanish at the origin:
    \begin{align}
        \Gamma^\rho_{\mu\nu}(0)=0;
    \end{align}
    \item the second order expansion of the metric at $p$ is
    \begin{align}
        g_{\mu\nu}\!\big(\exp_p(v)\big)
        = \eta_{\mu\nu}
        + \tfrac13\,(R(g))_{\mu\rho\nu\sigma}(p)\,v^\rho v^\sigma
        + O(|v|^3),
    \end{align}
    for $v\in U$ sufficiently small, where $(R(g))_{\mu\rho\nu\sigma}$ are the components of the Riemann curvature tensor at $p$ in the basis $(e_\mu)$, and $O(|v|^3)$ denotes a tensor field of $(0,2)$--type on $U$ whose components are $O(|v|^3)$.
    Equivalently, in terms of the coordinates $x^\mu$, writing $x^\mu=v^\mu$, one has
    \begin{align}
        g_{\mu\nu}(x)
        = \eta_{\mu\nu}
        + \tfrac13\,(R(g))_{\mu\rho\nu\sigma}(p)\,x^\rho x^\sigma
        + O(|x|^3).
    \end{align}
\end{enumerate}
\end{theorem}

\begin{proof}
By the discussion in the proof of Lemma~\ref{lem:exp-local-diffeo}, for every $v\in U$ the curve
\begin{align}
    \gamma_v \colon [0,1]\to M,\qquad \gamma_v(\lambda)\coloneqq \exp_p(\lambda v),
\end{align}
is the unique geodesic with initial data $\gamma_v(0)=p$ and $\dot\gamma_v(0)=v$. In the coordinates defined above, this geodesic has components
\begin{align}
    \lambda \longmapsto (\lambda v^\mu)_{\mu=0}^3,
\end{align}
so all geodesics through $p$ appear as straight radial lines through the origin.\\
At $x=0$ we have $x(p)=0$, and the differential $d(\exp_p)_0=\mathrm{id}_{T_pM}$ implies that the coordinate vector fields at $p$ satisfy
\begin{align}
    (\partial_\mu)_p = e_\mu.
\end{align}
Hence,
\begin{align}
    g_{\mu\nu}(0)=g_p(e_\mu,e_\nu)=\eta_{\mu\nu},
\end{align}
which proves (1).\\
We next show that the Christoffel symbols vanish at $0$. Let $v\in T_pM$, and consider the geodesic
\begin{align}
    \lambda\longmapsto \gamma_v(\lambda)=\exp_p(\lambda v),
\end{align}
whose coordinate representation is $\gamma_v^\mu(\lambda)=\lambda v^\mu$. The geodesic equation in coordinates reads
\begin{align}
    0
    = \ddot\gamma_v^\rho(\lambda)
      + \Gamma^\rho_{\mu\nu}(\gamma_v(\lambda))\,
        \dot\gamma_v^\mu(\lambda)\dot\gamma_v^\nu(\lambda).
\end{align}
Evaluating at $\lambda=0$ and using $\gamma_v(0)=p$, $\dot\gamma_v^\mu(0)=v^\mu$, we obtain
\begin{align}
    0
    = \Gamma^\rho_{\mu\nu}(0)\,v^\mu v^\nu
    \qquad \forall\,v\in T_pM.
\end{align}
Thus $\Gamma^\rho_{(\mu\nu)}(0)=0$. Since the Levi--Civita connection is torsion--free, $\Gamma^\rho_{\mu\nu}=\Gamma^\rho_{\nu\mu}$, and hence
\begin{align}
    \Gamma^\rho_{\mu\nu}(0)=0,
\end{align}
proving (3).\\
Metric compatibility gives, in coordinates,
\begin{align}
  0 = (\nabla_\rho g)_{\mu\nu}
    = \partial_\rho g_{\mu\nu}
      - \Gamma^\sigma_{\rho\mu}g_{\sigma\nu}
      - \Gamma^\sigma_{\rho\nu}g_{\mu\sigma},
\end{align}
so
\begin{align}
    \partial_\rho g_{\mu\nu}
    = g_{\sigma\nu}\,\Gamma^\sigma_{\rho\mu}
      + g_{\mu\sigma}\,\Gamma^\sigma_{\rho\nu}.
\end{align}
Evaluating at $x=0$ and using $g_{\mu\nu}(0)=\eta_{\mu\nu}$ and $\Gamma^\sigma_{\rho\mu}(0)=0$ we obtain
\begin{align}
    \partial_\rho g_{\mu\nu}(0)=0,
\end{align}
which proves (2).\\
For the second order term, we differentiate once more. Differentiating the previous identity and using that $\partial_\sigma g_{\mu\nu}(0)=0$, we get at $x=0$
\begin{align}
    \partial_\sigma \partial_\rho g_{\mu\nu}(0)
    = \eta_{\lambda\nu}\,\partial_\sigma \Gamma^\lambda_{\rho\mu}(0)
      + \eta_{\mu\lambda}\,\partial_\sigma \Gamma^\lambda_{\rho\nu}(0).
\end{align}
Set $A^\lambda{}_{\rho\mu\sigma}=\partial_\sigma\Gamma^\lambda_{\rho\mu}(0)$. Torsion--freeness makes $A$ symmetric in $\rho,\mu$. The radial geodesic identity $\Gamma^\lambda_{\rho\mu}(x)x^\rho x^\mu=0$, differentiated at the origin, implies
\begin{align}
A^\lambda{}_{\rho\mu\sigma}
+A^\lambda{}_{\mu\sigma\rho}
+A^\lambda{}_{\sigma\rho\mu}=0.
\end{align}
With our convention $R(\partial_i,\partial_j)\partial_k=R^\lambda{}_{ijk}\partial_\lambda$, curvature at the origin gives
\begin{align}
(R(g))^\lambda{}_{\sigma\rho\mu}(0)
=A^\lambda{}_{\rho\mu\sigma}-A^\lambda{}_{\sigma\mu\rho}.
\end{align}
Adding this identity to the one with $\rho$ and $\mu$ exchanged, and using the cyclic identity for $A$, yields
\begin{align}
\partial_\sigma\Gamma^\lambda_{\rho\mu}(0)
=\tfrac13\big((R(g))^\lambda{}_{\sigma\rho\mu}(0)
+(R(g))^\lambda{}_{\sigma\mu\rho}(0)\big).
\end{align}
Substituting into the metric-compatibility identity and using the algebraic curvature symmetries gives
\begin{align}
\partial_\sigma\partial_\rho g_{\mu\nu}(0)
=\tfrac13\big((R(g))_{\mu\rho\nu\sigma}(0)
+(R(g))_{\nu\rho\mu\sigma}(0)\big).
\end{align}
Now view the metric components as functions of $v\in U$ via the normal coordinate map:
\begin{align}
  g_{\mu\nu}(v)
  \coloneqq g_{\mu\nu}\big(\exp_p(v)\big).
\end{align}
The Taylor expansion of $g_{\mu\nu}(v)$ at $v=0$ is
\begin{align}
    g_{\mu\nu}(v)
    &= g_{\mu\nu}(0) + \partial_\rho g_{\mu\nu}(0)\,v^\rho
       + \tfrac12\,\partial_\sigma \partial_\rho g_{\mu\nu}(0)\,v^\rho v^\sigma
       + O(|v|^3) \\
    &= \eta_{\mu\nu}
       + \tfrac16\big(
          (R(g))_{\mu\rho\nu\sigma}(0)
          + (R(g))_{\nu\rho\mu\sigma}(0)
        \big)v^\rho v^\sigma
       + O(|v|^3).
\end{align}
Since $v^\rho v^\sigma$ is symmetric in $(\rho,\sigma)$, we can exchange the names of the dummy indices $\rho$ and $\sigma$ in the second term inside the parentheses and use the curvature symmetry
\begin{align}
      (R(g))_{\nu\sigma\mu\rho}=(R(g))_{\mu\rho\nu\sigma}
\end{align}
to find
\begin{align}
  \big(
    (R(g))_{\mu\rho\nu\sigma}(0)
    + (R(g))_{\nu\rho\mu\sigma}(0)
  \big)v^\rho v^\sigma
  = 2\,(R(g))_{\mu\rho\nu\sigma}(0)\,v^\rho v^\sigma.
\end{align}
Therefore,
\begin{align}
  g_{\mu\nu}(v)
  = \eta_{\mu\nu}
    + \tfrac13 (R(g))_{\mu\rho\nu\sigma}(p)\,v^\rho v^\sigma
    + O(|v|^3),
\end{align}
which is (4). Finally, recalling that by definition of the coordinates we have
\begin{align}
      x^\mu\big(\exp_p(v)\big)=v^\mu,
\end{align}
the same expansion can be written as
\begin{align}
  g_{\mu\nu}(x)
  = \eta_{\mu\nu}
    + \tfrac13 (R(g))_{\mu\rho\nu\sigma}(p)\,x^\rho x^\sigma
    + O(|x|^3),
\end{align}
as stated.
\end{proof}

\begin{remark}[Einstein equivalence principle and normal coordinates]
\Cref{thm:normal-coord-expansion} provides the geometric content of the Einstein equivalence principle. At any event $p\in M$ one can choose normal coordinates so that
\begin{align}
    g_{\mu\nu}(p)=\eta_{\mu\nu}, \qquad \Gamma^\rho_{\mu\nu}(p)=0,
\end{align}
hence, in that local freely falling frame,\footnote{Notice that in mathematical terms this is expressed by an orthonormal local frame $(e_\mu)$. Moreover, “physically meaningful” frames are the one respecting the orientation and the time--orientation.} the geodesic equation reduces at $p$ to the Minkowskian inertial equation
\begin{align}
    \ddot x^\rho(p)=0.
\end{align}
The nontrivial gravitational effects reappear at second order through the curvature term
\begin{align}
    g_{\mu\nu}(x) = \eta_{\mu\nu} + \tfrac13 (R(g))_{\mu\rho\nu\sigma}(p)\,x^\rho x^\sigma + O(|x|^3),
\end{align}
which cannot be removed by a change of coordinates. In other words: the EEP allows one to eliminate the non-tensorial part of the connection at a point, but not the curvature terms, which take the name of \textit{tidal effects}.
\end{remark}

\section{The Einstein--Hilbert action}
Albert Einstein arrived at his field equations only at the end of a long and tortuous path. Between 1912 and 1915 he kept oscillating between different candidates, torn between physical requirements and the mathematical demand that the laws should hold for arbitrary descriptions. In Zürich, with the help of Grossmann, he gradually learned the language of differential geometry and tried to cast gravity into that framework, but his first generally covariant attempts were still unsatisfactory: they did not secure a clean built–in notion of energy–momentum conservation for matter, and the overall structure of the theory remained uneasy in his own eyes.\footnote{In particular, they did not respect \cref{contbianc}.} For a while he even stepped back to less general, “restricted” equations, because he feared that full covariance might be incompatible with determinism.\footnote{Einstein’s worry was the so–called \emph{hole argument}: if the field equations were generally covariant, then for a given matter distribution one could apply a diffeomorphism that is the identity outside a “hole” but nontrivial inside, producing a second metric that also solves the equations with the same sources. Taken naively, this seems to destroy uniqueness of the gravitational field in the hole, and thus determinism. Only later did Einstein realize that such diffeomorphism--related metrics represent the same physical situation (gauge invariance).}

The decisive turn came only in the autumn of 1915, when Einstein returned to Berlin and reopened the problem in a series of four short notes to the Prussian Academy, week after week. In these November papers one sees him converging, step by step, towards the final form of the gravitational field equations. At almost the same time, in Göttingen, Hilbert approached the problem from a different, more elegant angle: he wrote down an action principle for gravity coupled to matter and derived a set of field equations that coincided with the ones of Einstein. This near--simultaneity triggered a brief and rather tense “priority race” between the two, with overlapping correspondence and rapidly evolving manuscripts on both sides. Einstein’s final paper, presented on 25 November 1915, appeared in print just before Hilbert’s work, which itself went through last–minute corrections. In later years Hilbert openly acknowledged that, whatever the fine chronology of drafts and proofs, the gravitational field equations and the physical theory they express should be credited to Einstein; what remained his own contribution was the variational viewpoint and the unifying ambition behind his broader “foundations of physics” program.
\begin{definition}[Space of metrics]\label{def:field-space-metrics}
Let $M$ be a smooth manifold. The \emph{space of Lorentzian metrics} on $M$ is
\begin{align}
\mathrm{Met}(M)\coloneqq\big\{g\in\Gamma(\mathrm{Sym}^2T^*M)\mid g_x\text{ is non--degenerate and Lorentzian for every }x\in M\big\}.
\end{align}
If $M$ is compact, this is an open subset of the Fréchet space $\Gamma(\mathrm{Sym}^2T^*M)$, and
\begin{align}
T_g\mathrm{Met}(M)\cong\Gamma(\mathrm{Sym}^2T^*M).
\end{align}
For noncompact $M$, pointwise nondegeneracy is not an open condition in the compact--open $C^\infty$ topology. We will not use that false openness assertion: all local Euler--Lagrange calculations below use compactly supported symmetric variations $h\in\Gamma_c(\mathrm{Sym}^2T^*M)$. For every such $h$, the curve $g+s h$ remains Lorentzian for sufficiently small $|s|$, because the condition is uniform on the compact support and $g$ is unchanged elsewhere. The symbol $\delta$ then denotes the first variation along these curves; equivalently, one may work in Fréchet neighborhoods of metrics on relatively compact regions.
We denote by
\begin{align}
\mathrm{Met}(M)_0\subseteq\mathrm{Met}(M)
\end{align}
the time--orientable Lorentzian metrics. For a fixed $g\in\mathrm{Met}(M)_0$, compactly supported sufficiently small variations also preserve time orientability: a global $g$--timelike field stays timelike on the compact support and is unchanged outside it.
\end{definition}

\begin{definition}[Lorentzian volume form]\label{def:lorentzian-volume-det}
Let $(M,g)$ be an oriented Lorentzian manifold of dimension $n$. In any oriented local chart on $U$ compatible with the given orientation, write
\begin{align}
    g = g_{\mu\nu}\,dx^\mu \otimes dx^\nu,
\end{align}
so that the matrix $(g_{\mu\nu}(x))$ has Lorentzian signature and hence $\det(g_{\mu\nu}(x))<0$ for all $x\in U$. Define on $U$ the $n$--form
\begin{align}
    \mathrm{vol}_g \coloneqq \sqrt{|\det(g_{\mu\nu})|}\, dx^1 \wedge \cdots \wedge dx^n.
\end{align}
These local expressions agree on the overlaps of oriented charts, therefore they glue to a unique smooth $n$--form $\mathrm{vol}_g \in \Omega^n(M)$, the \emph{volume form induced by $g$}.
\end{definition}

\begin{mdframed}
    \begin{definition}[Axiom of field dynamics]\label{axiomfe}
\leavevmode\newline
\vspace{-0.55cm}
\begin{enumerate}\setlength{\itemsep}{2pt}\setcounter{enumi}{5}
    \item \textbf{Field equations.} Let $M$ be a smooth oriented $4$--dimensional manifold without boundary. Space--time $(M,g)$ is a solution to the variational problem for the \textit{Einstein--Hilbert action functional}
\begin{align}
    S_{EH}[g] \coloneqq S_{m}[g]+\frac{1}{16\pi} \int_M \mathrm{Scal}(g)\,\mathrm{vol}_g
\end{align}
over the space of time--orientable Lorentzian metrics $\mathrm{Met}(M)_0$, where $S_{m}:\mathrm{Met}(M)_0\to\mathbb R$ is a smooth local functional, whose first variation at $g$ gives
\begin{align}
    (\delta S_{m})_g[h]
    = \frac12 \int_M \big\langle \mathrm{T}(g), h \big\rangle_g \,\mathrm{vol}_g,
    \qquad h\in \Gamma_c(\mathrm{Sym}^2 T^*M).
\end{align}
On a noncompact space--time, these total actions are regarded as real-valued functionals only when the integrals converge. In general the variational principle is local: for compactly supported $h$, differentiate the integral of the difference of the Lagrangian densities at $g+s h$ and $g$, which has compact support. This gives the same field equations without assuming that the total action is finite.
Here, $\langle \mathrm{T}(g), h \rangle_g$ denotes the pairing on symmetric tensor fields of $(0,2)$--type induced by $g$, i.e.
\begin{align}
    \langle \mathrm{T}(g), h \rangle_g \coloneqq (g^{-1}\otimes g^{-1})(\mathrm{T}(g),h)=g^{\mu\alpha} g^{\nu\beta} \,\mathrm{T}(g)_{\mu\nu} \, h_{\alpha\beta}\in \mathbb R,
\end{align}
and $\mathrm{vol}_g$ is the volume form induced by $g$. The tensor field $\mathrm{T}(g)\in \Gamma(\mathrm{Sym}^2 T^*M)$ defined in this way is called the \emph{energy--momentum} tensor.
\end{enumerate} 
\end{definition}
\end{mdframed}

\begin{remark}
Note that postulating a Lorentzian manifold without boundary is an extra assumption. One can generalize it by adding an extra term in the Einstein--Hilbert action functional.
\end{remark}

\begin{remark}[On existence and uniqueness of solutions]\label{rem:cauchy-gr}
The axiom above characterizes space--times as critical points of the Einstein--Hilbert functional, but it does \emph{not} guarantee that such critical metrics exist, nor that they are unique. In general, the solutions to the variational problem are an underdetermined, diffeomorphism--invariant system of nonlinear PDEs. Existence and (local) uniqueness are obtained instead by posing a Cauchy problem: one prescribes suitable initial data $(\Sigma,\gamma,K)$ on a $3$--dimensional spacelike hypersurface $\Sigma\subset M$, satisfying the Einstein constraint equations, and in vacuum (or for a matter model with an appropriate well--posed coupled hyperbolic system) one proves (Choquet--Bruhat, later Choquet--Bruhat--Geroch) that there exists a unique (up to diffeomorphisms) maximal globally hyperbolic development of these data. Thus, the axiom fixes the \emph{form} of the field equations; well--posedness of general relativity is a separate (hyperbolic) Cauchy theory.
\end{remark}

\begin{proposition}\label{prop:einstein-variational}
Fix a smooth local matter action $S_m$ as in \cref{axiomfe}. Then, the solutions of the variational problem $\delta S_{EH}=0$ over $\mathrm{Met}(M)_0$ are given by the Euler--Lagrange equations
\begin{align}\label{8}
    \mathrm{Ric}(g) - \tfrac12 \mathrm{Scal}(g)\,g = 8\pi\, \mathrm{T}(g),
\end{align}
which are called the Einstein field equations.
\end{proposition}

\begin{proof}
For simplicity, we compute the variation term by term, restricting to metric variations
\begin{align}
    h \in \Gamma_c(\mathrm{Sym}^2 T^*M),
\end{align}
i.e. symmetric tensor fields of $(0,2)$--type with compact support. This is enough to derive the Euler--Lagrange equations pointwise.\\
\textit{(1) Gravitational part.}
Consider the evaluation map
\begin{align}
    \mathrm{ev} : \mathrm{Met}(M)_0\times M \longrightarrow \mathrm{Sym}^2 T^*M,
    \qquad
    \mathrm{ev}(g,x) = g_x,
\end{align}
and the notation $\delta g$ for the metric variation, so that $(\delta g)_g[h]=h$; on a noncompact base we use the compactly supported interpretation specified above.
Since $\delta$ acts only on the metric variable, and $h$ has compact support, we may pass $\delta$ under the integral:
\begin{align}
    \delta  \left( \int_M \mathrm{Scal}(g)\,\mathrm{vol}_g \right)[h]
    = \int_M \delta\big(\mathrm{Scal}(g)\,\mathrm{vol}_g\big)[h].
\end{align}
For the two factors we have the metric variations
\begin{align}
    \delta \mathrm{vol}_g[h] &= \tfrac12 \,\mathrm{tr}_g(h)\,\mathrm{vol}_g,\\
    \delta \mathrm{Scal}(g)[h] &= - \langle \mathrm{Ric}(g), h\rangle_g + \operatorname{div}_g V(h),
\end{align}
where, with $h^{\mu\nu}=g^{\mu\alpha}g^{\nu\beta}h_{\alpha\beta}$,
\begin{align}
V(h)^\mu=\nabla_\nu h^{\mu\nu}-\nabla^\mu(g^{\alpha\beta}h_{\alpha\beta}).
\end{align}
Indeed, differentiation of $g^{\mu\alpha}g_{\alpha\nu}=\delta^\mu_\nu$ gives $\delta g^{\mu\nu}[h]=-h^{\mu\nu}$, while
\begin{align}
\delta\Gamma^\rho_{\mu\nu}[h]
&=\tfrac12g^{\rho\sigma}
(\nabla_\mu h_{\sigma\nu}+\nabla_\nu h_{\sigma\mu}-\nabla_\sigma h_{\mu\nu}),\\
\delta R_{\mu\nu}[h]
&=\nabla_\rho\delta\Gamma^\rho_{\mu\nu}[h]
-\nabla_\mu\delta\Gamma^\rho_{\rho\nu}[h].
\end{align}
Contracting the last identity with $g^{\mu\nu}$ proves the stated variation of scalar curvature. Jacobi's determinant formula gives the volume variation. If $h$ is compactly supported, so is $V(h)$. By the Leibniz rule,
\begin{align}
    \delta\big(\mathrm{Scal}(g)\,\mathrm{vol}_g\big)[h]
    &= \big(\delta \mathrm{Scal}(g)[h]\big)\,\mathrm{vol}_g + \mathrm{Scal}(g)\,\delta \mathrm{vol}_g[h] \\
    &= \big(- \langle \mathrm{Ric}(g), h\rangle_g + \operatorname{div}_g V(h)\big)\,\mathrm{vol}_g
       + \mathrm{Scal}(g)\,\tfrac12 \mathrm{tr}_g(h)\,\mathrm{vol}_g \\
    &= \big\langle -\mathrm{Ric}(g) + \tfrac12 \mathrm{Scal}(g)\,g ,\, h \big\rangle_g \,\mathrm{vol}_g
       + \big(\operatorname{div}_g V(h)\big)\,\mathrm{vol}_g.
\end{align}
The last term is an exact $n$--form (a boundary term) and, since $h$ has compact support, its integral over $M$ vanishes. Therefore,
\begin{align}
    \frac{1}{16\pi} \, \delta  \left( \int_M \mathrm{Scal}(g)\,\mathrm{vol}_g \right)_g[h]
    = -\frac{1}{16\pi} \int_M \big\langle \mathrm{Ric}(g) - \tfrac12 \mathrm{Scal}(g)\,g ,\, h \big\rangle_g \,\mathrm{vol}_g.
\end{align}
\textit{(2) Matter part.}
By assumption,
\begin{align}\label{0}
    (\delta S_m)_g[h]
    = \frac12 \int_M \big\langle \mathrm{T}(g), h \big\rangle_g \,\mathrm{vol}_g,
    \qquad h\in \Gamma_c(\mathrm{Sym}^2 T^*M).
\end{align}
\textit{(3) Sum.}
Since
\begin{align}
    S_{EH}[g] = S_m[g] + \frac{1}{16\pi} \int_M \mathrm{Scal}(g)\,\mathrm{vol}_g,
\end{align}
and $\delta$ is linear,
\begin{align}
    (\delta S_{EH})_g[h]
    &= -\frac{1}{16\pi} \int_M \big\langle \mathrm{Ric}(g) - \tfrac12 \mathrm{Scal}(g)\,g ,\, h \big\rangle_g \,\mathrm{vol}_g
      + \frac12 \int_M \big\langle \mathrm{T}(g), h \big\rangle_g \,\mathrm{vol}_g.
\end{align}
\textit{(4) Euler--Lagrange condition.}
If $(\delta S_{EH})_g[h]=0$ for all $h\in \Gamma_c(\mathrm{Sym}^2 T^*M)$, then
\begin{align}
    \int_M \big\langle \mathrm{Ric}(g) - \tfrac12 \mathrm{Scal}(g)\,g - 8\pi\,\mathrm{T}(g) ,\, h \big\rangle_g \,\mathrm{vol}_g = 0
    \qquad \forall\, h\in \Gamma_c(\mathrm{Sym}^2 T^*M).
\end{align}
By nondegeneracy of the pairing, this implies
\begin{align}
    \mathrm{Ric}(g) - \tfrac12 \mathrm{Scal}(g)\,g = 8\pi\, \mathrm{T}(g).
\end{align}
\end{proof}

\begin{remark}
    In the physics literature, Einstein equations are usually written as
    \begin{align}
        \boxed{R_{\mu\nu}-\tfrac12 R\,g_{\mu\nu}=\frac{8\pi G}{c^4} \,T_{\mu\nu},}
    \end{align}
    where we use now the actual units.
\end{remark}

\begin{remark}[Non--linearity]
    The left--hand side is \emph{non--linear} in $g$: both $\mathrm{Ric}(g)$ and $\mathrm{Scal}(g)$ contain $g^{-1}$ and first derivatives of $g$ inside the Christoffel symbols, and second derivatives of $g$ in the curvature. As a consequence, the field $g$ itself carries energy and momentum (it contributes to curvature), so “gravity itself generates gravity”.
\end{remark}

\begin{proposition}[Diffeomorphism invariance]\label{prop:EH-diffeo-invariant}
Let $S_m:\mathrm{Met}(M)_0\to \mathbb R$ be as defined in \cref{axiomfe} and invariant under orientation--preserving diffeomorphisms, i.e. $S_m[\varphi^*g]=S_m[g]$ for all $g\in\mathrm{Met}(M)_0$ and $\varphi \in \mathrm{Diff}^+(M)$. Then, whenever the total actions are defined, it holds
\begin{align}
    S_{EH}[\varphi^* g] = S_{EH}[g],
\end{align}
for all $g\in \mathrm{Met}(M)_0$.
\end{proposition}
\begin{proof}
The invariance of $S_m$ is given by assumption, i.e. $S_m[\varphi^* g]=S_m[g]$, and, by linearity of the pull-back, we just need to prove the invariance of integral. So, fix $g\in \mathrm{Met}(M)_0$ and $\varphi\in \mathrm{Diff}^+(M)$. By naturality of the curvature, the scalar curvature transforms as
\begin{align}
    \mathrm{Scal}(\varphi^* g) = \varphi^*\big(\mathrm{Scal}(g)\big),
\end{align}
and, since $\varphi$ preserves the orientation, the volume form transforms as
\begin{align}
    \mathrm{vol}_{\varphi^* g} = \varphi^*(\mathrm{vol}_g).
\end{align}
Hence, the integrand of the gravitational part of the Einstein--Hilbert action transforms as
\begin{align}
    \mathrm{Scal}(\varphi^* g)\,\mathrm{vol}_{\varphi^* g}
    = \varphi^*\big(\mathrm{Scal}(g)\big)\,\varphi^*(\mathrm{vol}_g)
    = \varphi^*\big(\mathrm{Scal}(g)\,\mathrm{vol}_g\big).
\end{align}
Using the change--of--variables formula for orientation--preserving diffeomorphisms,
\begin{align}
    \int_M \varphi^* \alpha = \int_M \alpha \qquad \text{for all integrable } \alpha\in \Omega^n(M),
\end{align}
we obtain
\begin{align}
    S_{EH}[\varphi^* g]
    &= S_m[g]+ \frac{1}{16\pi} \int_M \mathrm{Scal}(\varphi^* g)\,\mathrm{vol}_{\varphi^* g} \\
    &= S_m[g]+\frac{1}{16\pi} \int_M \varphi^*\big(\mathrm{Scal}(g)\,\mathrm{vol}_g\big) \\
    &= S_m[g]+\frac{1}{16\pi} \int_M \mathrm{Scal}(g)\,\mathrm{vol}_g \\[2pt]
    &= S_{EH}[g].
\end{align}
\end{proof}

\begin{proposition}[Covariant conservation of the energy--momentum tensor]\label{prop:divT-zero}
Let $S_m:\mathrm{Met}(M)_0\to \mathbb R$ be as defined in \cref{axiomfe} and invariant under orientation--preserving diffeomorphisms. Then, we have\footnote{The divergence operator is defined in \cref{contbianc}.}
\begin{align}
    \operatorname{div}_g \mathrm T(g) = 0.
\end{align}
\end{proposition}
\begin{proof}
Let $X\in \Gamma_c(TM)$ be a smooth vector field with compact support and let $\{\varphi_s\}_{s\in (-\varepsilon,\varepsilon)}$ be its flow, so that
\begin{align}
    \frac{d}{ds}\Big|_{s=0} \varphi_s = X.
\end{align}
By diffeomorphism invariance of $S_m$ we have
\begin{align}
    0
    = \frac{d}{ds}\Big|_{s=0} S_m[\varphi_s^* g]
    = (\delta S_m)_g\big[\mathcal L_X g\big],
\end{align}
where $\mathcal L_X g$ is the Lie derivative of $g$ along $X$. Using \eqref{0} with $h=\mathcal L_X g$ we obtain
\begin{align}\label{eq:stress-diffeo-variation}
    0 = \frac12 \int_M \big\langle \mathrm T(g), \mathcal L_X g \big\rangle_g \,\mathrm{vol}_g.
\end{align}
Now, use the standard identity for the symmetric $(0,2)$--tensor field $\mathrm T$ and a vector field $X$:
\begin{align}\label{eq:stress-lie-contraction}
    \big\langle \mathrm T, \mathcal L_X g \big\rangle_g
    = 2\, \mathrm T^{\mu}{}_{\nu}\, \nabla_\mu X^{\nu},
\end{align}
where $\nabla$ is the Levi--Civita connection of $g$ and we raised an index of $\mathrm T$ with $g$. Since $X$ has compact support, we can integrate by parts:
\begin{align}
    \int_M \mathrm T^{\mu}{}_{\nu}\, \nabla_\mu X^{\nu} \,\mathrm{vol}_g
    = - \int_M (\nabla_\mu \mathrm T^{\mu}{}_{\nu})\, X^{\nu} \,\mathrm{vol}_g,
\end{align}
the boundary term vanishing because $\operatorname{supp} X$ is compact.
Inserting \eqref{eq:stress-lie-contraction} into \eqref{eq:stress-diffeo-variation}, we get
\begin{align}
    0
    &= \frac12 \int_M 2\, \mathrm T^{\mu}{}_{\nu}\, \nabla_\mu X^{\nu} \,\mathrm{vol}_g \\
    &= - \int_M (\nabla_\mu \mathrm T^{\mu}{}_{\nu})\, X^{\nu} \,\mathrm{vol}_g.
\end{align}
Since this holds for every $X\in \Gamma_c(TM)$, by nondegeneracy of the $L^2$--pairing, we conclude
\begin{align}
    \nabla_\mu \mathrm T^{\mu}{}_{\nu} = 0,
\end{align}
that is,
\begin{align}
    \operatorname{div}_g \mathrm T(g) = 0.
\end{align}
\end{proof}

\begin{remark}
The notation $S_m[g]$ above treats matter as a metric-dependent local functional. For a full matter theory $S_m[g,\psi]$ with independent fields $\psi$, the stress--energy tensor is defined by varying $g$ at fixed $\psi$. Diffeomorphism invariance then varies both fields, and implies $\operatorname{div}_gT=0$ on solutions of the matter equations; it does not in general give that conservation law off shell in $\psi$.
\end{remark}

\begin{remark}[Diffeomorphism invariance and contracted Bianchi identity]\label{rem:diff-compatibility}
The key geometric fact is the contracted Bianchi identity
\begin{align}
    \operatorname{div}_g\big(\mathrm{Ric}(g) - \tfrac12 \mathrm{Scal}(g)\,g\big) = 0,
\end{align}
which holds for every Lorentzian metric $g$, independently of any matter content. When we postulate Einstein’s equations
\begin{align}
    \mathrm{Ric}(g) - \tfrac12 \mathrm{Scal}(g)\,g = 8\pi\, \mathrm T(g),
\end{align}
this identity forces the right--hand side to satisfy
\begin{align}
    \operatorname{div}_g \mathrm T(g) = 0.
\end{align}
\cref{prop:divT-zero} shows that this is \emph{exactly} what we get if the matter action $S_m$ is invariant under orientation--preserving diffeomorphisms:
\begin{align}
    \text{$\mathrm{Diff}^+(M)$--invariance of } S_m\Longrightarrow\operatorname{div}_g \mathrm T(g)=0,
\end{align}
so the matter sector is automatically compatible with the geometric identity. In other words, the symmetry of the action under $\mathrm{Diff}^+(M)$ is the mechanism that aligns the matter conservation law with the Bianchi identity; this is why, in general relativity, diffeomorphism invariance is not optional but structurally built in it.
\end{remark}

\begin{remark}[General covariance vs.\ diffeomorphism invariance]\label{rem:covariance-vs-diffeo}
There are three statements that in general are \emph{not} the same:
\begin{enumerate}\setlength{\itemsep}{2pt}
    \item[(a)] \emph{General covariance}: the field equations are written as tensorial (i.e.\ coordinate--independent) relations on open sets $U\subseteq M$.
    \item[(b)] \emph{No fixed background}: no non--dynamical geometric object (metric, connection, volume form, frame, etc.) is kept fixed in the equations, apart from the purely topological data we decided to allow (orientation, time--orientation, signature sector).
    \item[(c)] \emph{Diffeomorphism invariance of the action}: for every orientation--preserving $\varphi\in \mathrm{Diff}^+(M)$,
    \begin{align}
        S[\varphi^* \Phi] = S[\Phi],
    \end{align}
    for all field configurations $\Phi$ (in GR: $\Phi=g$ or $(g,\text{matter})$).
\end{enumerate}
For the Einstein--Hilbert theory, the connection between these statements follows from the \emph{naturality of the Lagrangian}: its density is constructed from the dynamical fields and their derivatives so that
\begin{align}
L[\varphi^*\Phi]=\varphi^*L[\Phi].
\end{align}
Integration then proves (c), subject to convergence and the chosen boundary conditions. The Euler--Lagrange equations transform naturally under the same diffeomorphisms.
However, tensorial equations and the absence of a specified background do not by themselves prove exact invariance of a particular action; boundary terms and the existence of a natural variational formulation also matter. Conversely, covariance under simultaneous transformations of fields and external background data does not turn those background data into dynamical fields. In general relativity, we impose the geometric field content and use a natural action, rather than infer a general logical equivalence between (a), (b), and (c).
\end{remark}

\begin{remark}[Geodesic action vs. full diffeomorphism invariance]\label{rem:geodesic-vs-diffeo}
There is a useful asymmetry between the Einstein--Hilbert sector and the “test--particle” sector. The geodesic equation
\begin{align}
    \nabla_{\dot\gamma}\dot\gamma = 0
\end{align}
is obtained by varying the curve $\gamma:I\to M$ in the functional
\begin{align}
    S_{free}[\gamma]
    = -\tfrac12 \int_I g(\dot\gamma,\dot\gamma)\,d\lambda
\end{align}
\emph{with the metric $g$ kept fixed}. Because $g$ is frozen, $S_{free}$ is not a fully $\mathrm{Diff}(M)$--invariant functional in the GR sense: a general diffeomorphism should act on \emph{all} dynamical fields, but here only $\gamma$ is varied while $g$ is treated as background. This is why, at the level of the particle action alone, one only sees invariance under those diffeomorphisms that are isometries of $g$ (they preserve $g$, so they are actual symmetries of $S_{free}$). This does \emph{not} contradict diffeomorphism invariance of GR. In the full theory the metric is \emph{also} a variable: one first determines $g$ as a solution of Einstein’s equations from the Einstein--Hilbert action, and \emph{then} studies geodesics for that specific solution $g$. So the apparent “background” in the particle action is in fact a dynamical field coming from another variational problem. What breaks full $\mathrm{Diff}^+(M)$ in the geodesic functional is precisely the choice of looking at a \emph{derived} problem, namely the motion of a test body in a fixed geometry.
\end{remark}

\begin{remark}[Universality of the coupling and EEP]\label{rem:universality-EEP}
In the formulation above, the universality of gravity is already built into the action. Indeed:
\begin{enumerate}\setlength{\itemsep}{2pt}
    \item both the gravitational part and the matter part of the theory depend on the \emph{same} geometric variable, the metric $g$;
    \item the matter sector is assumed to be diffeomorphism invariant;
    \item the stress--energy tensor is \emph{defined} as the variational derivative of $S_m$ with respect to $g$.
\end{enumerate}
With the Einstein--Hilbert gravitational action, the metric variation of every matter contribution enters the total source $\mathrm T(g)$ in
\begin{align}
    \mathrm{Ric}(g) - \tfrac12 \mathrm{Scal}(g)\,g = 8\pi\, \mathrm T(g).
\end{align}
Thus all matter included in this formulation contributes to the source obtained by varying the same metric. This is universal sourcing in Einstein's equations, but it does not by itself establish the full Einstein equivalence principle. Diffeomorphism invariance alone also allows nonminimal terms such as $R\psi^2$ or curvature-dependent kinetic couplings, whose local effects cannot be removed by setting $\Gamma=0$ at a point. To implement EEP for nongravitational experiments one needs the corresponding universal local Lorentz-invariant matter coupling assumptions, commonly expressed by minimal coupling to a single metric. For a full action $S_m[g,\psi]$, diffeomorphism invariance gives covariant stress--energy conservation on the matter equations, consistently with the contracted Bianchi identity.
\end{remark}

\begin{remark}[Degree-of-freedom counting in GR]\label{rem:GR-dof}
In this remark we restrict to vacuum Einstein--Hilbert gravity.
In dimension $4$, a Lorentzian metric $g$ has at each point $10$ independent components (a symmetric $4\times 4$ matrix). However, in general relativity not all of these components describe physical degrees of freedom.

\begin{enumerate}\setlength{\itemsep}{2pt}
    \item \textbf{Gauge directions.} The theory is invariant under $\mathrm{Diff}(M)$. Infinitesimally, a diffeomorphism is generated by a vector field $X$, which has $4$ independent components. Acting with $X$ changes $g$ by a Lie derivative,
    \begin{align}
        g \longmapsto g + \mathcal L_X g,
    \end{align}
    and this change is pure gauge. So $4$ combinations of the $10$ metric components do not carry physics: they can be gauged away.

    \item \textbf{Constraints.} The Einstein equations $G_{\mu\nu}=8\pi T_{\mu\nu}$ do \emph{not} determine all remaining components freely. In a $3+1$ split, the projections $G_{\mu\nu}n^\mu n^\nu$ and $G_{\mu\nu}n^\mu e_i^\nu$ (Hamiltonian and momentum constraints), with $n$ the unit normal and $e_i$ tangent to the slice, do not contain second time derivatives of the metric: they are constraint equations restricting the initial data on a Cauchy surface, rather than independent evolution equations. Concretely, this gives $4$ additional relations among the metric components and their first derivatives.

    \item \textbf{Resulting count.} The heuristic spacetime count starts from $10$ components and subtracts $4$ coordinate freedoms\footnote{This is also equivalent to considering the components of $g$ and use the coordinate gauge freedom.} and $4$ constraints, giving
    \begin{align}
        10 - 4 - 4 = 2
    \end{align}
    independent \emph{propagating} degrees of freedom per space point. A precise Hamiltonian count uses the six components of the spatial metric and their six conjugate momenta: four first-class constraints and their four gauge directions leave $12-2\cdot4=4$ physical phase-space components, or two configuration degrees of freedom per point. These $2$ d.o.f.\ are precisely the two polarizations of the gravitational field that appear in the linearized theory on a fixed background (the so called “$+$” and “$\times$” gravitational waves).

    \item \textbf{PDE meaning.} The constraints are the PDE manifestation of diffeomorphism invariance: because $\mathrm{Diff}(M)$ is infinite dimensional, the equations cannot be all evolution equations; some of them must enforce compatibility with the symmetry (they are preserved in time by the Bianchi identity).
\end{enumerate}

This counting is local and on--shell: globally (e.g. in presence of boundaries, asymptotic conditions, or nontrivial topology) extra “global” degrees of freedom or charges can appear, but the local number of propagating modes of GR in $4$ dimension is $2$.
\end{remark}

\printbibliography[heading=subbibliography]

\clearpage
\begingroup
\small

\section*{Bibliography}

As stated in the abstract, these notes were prepared using only
the authors' own knowledge and interactions with ChatGPT models.
The works listed below are recommended for further reading on the
subjects covered in the notes; they are not a list of sources
consulted in preparing the text. They offer complementary
perspectives on classical mechanics, differential geometry,
and general relativity.

\newrefsection[bibliography.bib]

\nocite{
  Arnold1989,
  AbrahamMarsden1978,
  Goldstein2002,
  Lee2013,
  Tu2011,
  BottTu1982,
  Lee2018,
  KobayashiNomizu,
  Nakahara2003,
  Frankel2012,
  ONeill1983,
  LandauLifshitz1975,
  Schutz2009,
  Carroll2004,
  Wald1984,
  MisnerThorneWheeler1973,
  HawkingEllis1973,
  Straumann2013
}

\setlength{\bibitemsep}{2pt}
\renewcommand*{\bibfont}{\small}
\printbibliography[heading=none]

@book{Arnold1989,
  author = {Arnold, Vladimir I.},
  title = {{Mathematical Methods of Classical Mechanics}},
  publisher = {Springer},
  year = {1989},
  edition = {2}
}

@book{AbrahamMarsden1978,
  author = {Abraham, Ralph and Marsden, Jerrold E.},
  title = {{Foundations of Mechanics}},
  publisher = {Benjamin/Cummings},
  year = {1978},
  edition = {2}
}

@book{Goldstein2002,
  author = {Goldstein, Herbert and Poole, Charles and Safko, John},
  title = {{Classical Mechanics}},
  publisher = {Addison-Wesley},
  year = {2002},
  edition = {3}
}

@book{Lee2013,
  author = {Lee, John M.},
  title = {{Introduction to Smooth Manifolds}},
  publisher = {Springer},
  year = {2013},
  edition = {2}
}

@book{Tu2011,
  author = {Tu, Loring W.},
  title = {{An Introduction to Manifolds}},
  publisher = {Springer},
  year = {2011},
  edition = {2}
}

@book{BottTu1982,
  author = {Bott, Raoul and Tu, Loring W.},
  title = {{Differential Forms in Algebraic Topology}},
  publisher = {Springer},
  year = {1982}
}

@book{Lee2018,
  author = {Lee, John M.},
  title = {{Introduction to Riemannian Manifolds}},
  publisher = {Springer},
  year = {2018},
  edition = {2}
}

@book{KobayashiNomizu,
  author = {Kobayashi, Shoshichi and Nomizu, Katsumi},
  title = {{Foundations of Differential Geometry}},
  publisher = {Interscience},
  date = {1963/1969},
  note = {Vols.~I and II (published in 1963 and 1969, respectively)}
}

@book{Nakahara2003,
  author = {Nakahara, Mikio},
  title = {{Geometry, Topology and Physics}},
  publisher = {Institute of Physics Publishing},
  year = {2003},
  edition = {2}
}

@book{Frankel2012,
  author = {Frankel, Theodore},
  title = {{The Geometry of Physics: An Introduction}},
  publisher = {Cambridge University Press},
  year = {2012},
  edition = {3}
}

@book{ONeill1983,
  author = {O'Neill, Barrett},
  title = {{Semi-Riemannian Geometry with Applications to Relativity}},
  publisher = {Academic Press},
  year = {1983}
}

@book{LandauLifshitz1975,
  author = {Landau, Lev D. and Lifshitz, Evgeny M.},
  title = {{The Classical Theory of Fields}},
  publisher = {Pergamon Press},
  year = {1975},
  edition = {4},
  series = {Course of Theoretical Physics},
  volume = {2}
}

@book{Schutz2009,
  author = {Schutz, Bernard},
  title = {{A First Course in General Relativity}},
  publisher = {Cambridge University Press},
  year = {2009},
  edition = {2}
}

@book{Carroll2004,
  author = {Carroll, Sean M.},
  title = {{Spacetime and Geometry: An Introduction to General Relativity}},
  publisher = {Addison-Wesley},
  year = {2004}
}

@book{Wald1984,
  author = {Wald, Robert M.},
  title = {{General Relativity}},
  publisher = {University of Chicago Press},
  year = {1984}
}

@book{MisnerThorneWheeler1973,
  author = {Misner, Charles W. and Thorne, Kip S. and Wheeler, John A.},
  title = {{Gravitation}},
  publisher = {W. H. Freeman},
  year = {1973}
}

@book{HawkingEllis1973,
  author = {Hawking, Stephen W. and Ellis, George F. R.},
  title = {{The Large Scale Structure of Space-Time}},
  publisher = {Cambridge University Press},
  year = {1973}
}

@book{Straumann2013,
  author = {Straumann, Norbert},
  title = {{General Relativity}},
  publisher = {Springer},
  year = {2013},
  edition = {2}
}

\endrefsection

\par
\endgroup
\clearpage

\end{document}